\documentclass[12pt]{article}
\usepackage{booktabs}
\usepackage{threeparttable}
\usepackage{setspace}
\usepackage{makecell}
\usepackage{graphicx}
\usepackage[table,xcdraw]{xcolor} 
\usepackage[round]{natbib}
\usepackage{amsmath}
\usepackage{adjustbox}
\usepackage{tabularx}
\usepackage{multirow}
\usepackage{float}
\usepackage{amssymb}
\usepackage{bbm}
\usepackage{bbold}
\usepackage{subcaption}
\usepackage{geometry}
\usepackage{xfrac}
\usepackage{comment} 
\usepackage{pdflscape}
 \usepackage{tikz}
\usepackage{enumitem}
\usepackage{graphicx} 
\usepackage{listings}
\usepackage{upquote}
\graphicspath{ {./images/} }
\usepackage{natbib}
\newcommand\blfootnote[1]{%
	\begingroup
	\renewcommand\thefootnote{}\footnote{#1}%
	\addtocounter{footnote}{-1}%
	\endgroup
}
\usepackage{mathtools}
\usepackage{amsmath}
\usepackage{amssymb, bm, bbm, amsthm}
\usepackage{amsfonts}

\usepackage{tikz}
\usepackage{scalerel}
\usepackage{pict2e}
\usepackage{tkz-euclide}
\usetikzlibrary{calc}
\usetikzlibrary{patterns,arrows.meta}
\usetikzlibrary{shadows}
\usetikzlibrary{external}

\usepackage{pgfplots}
\pgfplotsset{compat=newest}
\usepgfplotslibrary{statistics}
\usepgfplotslibrary{fillbetween}

\usepackage{xcolor}
\definecolor{Olive}{RGB}{128,128,0}
\usepackage{caption}
\usepackage{nicefrac}
\usetikzlibrary{shapes,decorations,arrows,calc,arrows.meta,fit,positioning}
\tikzset{
	-Latex,auto,node distance =1 cm and 1 cm,semithick,
	state/.style ={ellipse, draw, minimum width = 0.7 cm},
	point/.style = {circle, draw, inner sep=0.04cm,fill,node contents={}},
	bidirected/.style={Latex-Latex,dashed},
	el/.style = {inner sep=2pt, align=left, sloped}
}
\newtheorem{theorem}{Theorem}
\newtheorem{lemma}{Lemma}

\newtheorem{assumption}{Assumption}
\newtheorem{remark}{Remark}
\newtheorem*{condition}{Condition}
\newtheorem{assumptionalt}{Assumption}[assumption]
 
\newenvironment{assumptionp}[1]{
  \renewcommand\theassumptionalt{#1}
  \assumptionalt
}{\endassumptionalt}

\usepackage[colorlinks=true,linkcolor=blue,citecolor=blue]{hyperref}
\begin{document}
\baselineskip 20pt
\setlength\abovedisplayskip{5pt}
    \setlength\belowdisplayskip{5pt}

\title{\bf {Difference-in-Differences with Sample Selection}}
\author{Gayani Rathnayake$^{a,b}$, Akanksha Negi$^a$, Otavio Bartalotti$^{a,c}$, Xueyan Zhao$^a$}
 
\date{First Version: August 29, 2024\\
Current Version: \today}
\maketitle
\blfootnote{$^a$Department of Econometrics and Business Statistics, Monash University, Australia. $^b$ Central Bank of Sri Lanka. $^c$IZA, Germany. Emails: \href{mailto:gayani.rathnayake@monash.edu}{gayani.rathnayake@monash.edu}, \href{mailto:akanksha.negi@monash.edu}{akanksha.negi@monash.edu}, \href{mailto:otavio.bartalotti@monash.edu}{otavio.bartalotti@monash.edu}, \href{mailto:xueyan.zhao@monash.edu}{xueyan.zhao@monash.edu}. We thank Alyssa Carlson, \`{A}ureo de Paula, Desire Kedagni, Pedro Sant'Anna, Rami Tabri, Quentin Brummet, Valentin Verdier, Vitor Possebom, Martin Huber, Giovanni Mellace, Philip Heiler, four anonymous referees, and participants at seminars at the University of Melbourne, 2025 IAAE Annual Meetings, and the 2025 Econometric Society World Congress for their useful suggestions and comments.}

\vspace{-3.5em}

\begin{abstract}
We consider the identification of average treatment effects on the treated (ATT) in difference-in-differences (DiD) settings in the presence of endogenous sample selection. We first establish that the conventional DiD estimand generally fails to recover causally meaningful treatment effects, even if selection and treatment assignment are independent. We then partially identify the ATT for individuals whose outcomes would be observed post-treatment under either counterfactual treatment state, and derive sharp bounds on this parameter under different sets of assumptions on the relationship between sample selection and treatment assignment. These identification results are extended to allow for covariates, repeated cross-section data, and two-by-two comparisons in staggered adoption designs. Furthermore, we present identification results for the ATT of three additional empirically relevant latent groups by imposing outcome mean dominance assumptions that have intuitive appeal in applications. Finally, two empirical illustrations demonstrate the approach's usefulness by revisiting (i) the effect of a job training program on earnings and (ii) the effect of a working-from-home policy on employee performance.
\end{abstract}

\noindent%

\medskip


\vfill

\maketitle

\section{Introduction}
Nonrandom sample selection is a pervasive challenge in empirical research that can compromise the validity of usual causal inference approaches. In many studies, the outcome of interest may only be observed for a non-random subset of the population due to issues such as attrition, survey non-response, and measurement error.\footnote{This is frequently observed in policy evaluation studies that employ a panel survey of individuals before and after a program is implemented (such as cash transfer, poverty alleviation, or a training subsidy program) to evaluate its impact \citep{holzer1993training,bobonis2011impact,asadullah2016evaluating}. Frequently, the follow-up survey will face the problem of non-ignorable attrition.} This can pose a significant problem for difference-in-differences (DiD) methods whose popularity in empirical research has grown over time \citep{goldsmith2024tracking}. In this paper, we address the challenges posed by endogenous sample selection in a DiD setting and propose a partial identification strategy for average treatment effects on the treated based on a latent subpopulation structure under alternative sets of identifying assumptions.

We first demonstrate that na\"{\i}vely applying DiD to units observed in both pre- and post-treatment periods without accounting for sample selection does \textit{not} identify a meaningful causal parameter; either the overall ATT or an adequately weighted average of treatment effects. When the goal is to recover the overall ATT, the canonical DiD estimand will generally be biased unless one imposes restrictive assumptions on untreated outcome trends and treatment effect heterogeneity that are unlikely to be true when sample selection is endogenous. Interestingly, \textit{even if the selection mechanism is independent of treatment assignment}, the bias in na\"{\i}ve DiD does not disappear unless selection is exogenous to the outcome of interest.\footnote{In Section \ref{simulation}, we demonstrate this through a simulation where the na\"{\i}ve DiD estimate exhibits an upward bias, even when the selection and treatment assignment mechanisms are independent.} 

Our first contribution is to propose a partial identification strategy for the average treatment effect on the treated (ATT) for individuals belonging to the latent group whose outcome would be observed regardless of their treatment state ($\tau_{OOO}$, with OOO referring to being ``always-observed'' in the pre-treatment period and in both the untreated and treated post-treatment counterfactual states, respectively) under different assumptions on the sample selection mechanism. The ATT for the always-observed latent group is of interest both as a component of the overall ATT and on its own. As discussed in \cite{bartalotti2023identifying}, $\tau_{OOO}$ can be seen as a measure of the effect of treatment on the intensive margin. Substantively, it captures the effect of the treatment for a stable subpopulation that can be identified based on data observed in both time periods.\footnote{For example, the OOO group could refer to (a) workers with stable labor force attachment in the context of a training program; (b) employees who would stay with a firm irrespective of whether they are offered a work-from-home option, (c) patients who would continue to receive medical care or (d) survive throughout the study period.}  Our approach combines the trimming procedure of \citet{lee2009} to address the identification challenges posed by endogenous sample selection and treatment assignment.

Identification of $\tau_{OOO}$ relies on parallel trends in outcomes (PTO) between the treated and untreated units within the same latent group, and a no-anticipation assumption. The trimming procedure requires knowledge of the latent groups' proportions \citep{imai2008sharp, lee2009, semenova2025generalized}, which we acquire by considering combinations of two assumptions that govern the relationship between selection and treatment assignment: (i) ``ignorability of treatment in potential selection'' (IS) and (ii) ``monotonicity of selection'' (MS) in treatment. IS requires that, conditional on being observed in the pre-treatment period, the probability of being observed post-treatment in each counterfactual state (treated or untreated) is independent of the treatment received. In turn, ``monotonicity of selection'' (MS) requires that treatment has an increasing effect on the probability of selection for all units in the post-treatment period.\footnote{MS is commonly used in the sample selection literature \citep{chen2015bounds, huber2015sharp} and is similar to the LATE monotonicity condition \citep{imbens1994identification}.}

We derive alternative bounds for the ATT of the OOO group under different sets of assumptions, both with and without imposing MS. Without MS, our first result establishes partial identification of the proportions of the ``always-observed'' latent group among the treated (untreated) under no anticipation in selection and IS for potential selection if untreated (treated). Then, $\tau_{OOO}$ can be partially identified if IS holds for both counterfactual treatment states, producing bounds that are general and allow flexible post-treatment selection patterns, including cases in which treatment induces individuals to enter or exit the sample. Our second result tightens the $\tau_{OOO}$ bounds by imposing MS which allows us to relax the IS assumption for one of the two counterfactual treatment states\footnote{Under positive (negative) MS, we only constrain potential selection in the untreated (treated) state, leaving selection in the treated (untreated) state unrestricted.} and point-identifies the latent group proportions. Following \cite{lee2009}, we establish the sharpness of both these bounds.

The second contribution is to extend the partial identification results for $\tau_{OOO}$ to include pre-treatment covariates, whereby we relax the unconditional PTO and IS assumptions to their conditional analogues. These two assumptions, along with alternative MS assumptions, allow us to identify the ATT for the always-observed within each covariate subpopulation. The resulting conditional bounds are then aggregated to obtain the identified set for the unconditional ATT for the OOO group. Importantly, we also consider the relaxed monotonicity framework of \cite{semenova2025generalized}, which permits groups with different observed characteristics to exhibit different monotonicity directions. This provides an intermediate case between the two baseline scenarios of no monotonicity and global monotonicity (discussed in Section \ref{identification}). Formal results and a detailed discussion are deferred to Appendix \ref{sec: covariates}. 

The third contribution of this paper is to extend the bounding approach to identify the ATT for additional latent groups on whom the researcher has limited information compared to the ``always-observed'' type. Similar to the OOO group, these latent subpopulations are defined based on their observed selection status in the pre-treatment period and their counterfactual selection behavior in the post-treatment period under both treated and untreated states. We combine basic support restrictions and cross-group mean dominance assumptions, that impose economically intuitive rankings on potential outcome means across latent types, to obtain informative identified sets for the group-specific ATTs. These parameters are policy-relevant in many empirical settings and can reveal meaningful treatment effect heterogeneity across different selection types. For example, in evaluating job training programs, policymakers may care about effects for NOO (units that were unemployed before treatment but employed post-treatment whether or not they receive training) or NNO (units who were unemployed before treatment but will be employed post-treatment if given training but not otherwise). For workplace policy evaluations (such as work-from-home (WFH) policies), firms may be interested in the effects of WFH for the ONO group (employees who are observed prior to implementation of the policy and will leave the company if WFH is not provided but would stay if it's offered). The shares of these latent groups in the population are identified based on MS and a strengthening of IS to a joint independence assumption that conditions on pre-treatment selection.  

Finally, we show that the overall ATT can be expressed as a weighted average of latent-group-specific ATTs. We obtain bounds on this parameter by combining the identified sets for each group along with appropriate population weights. This delivers partial identification for the overall population ATT, which is the typical estimand in DiD studies.  

We illustrate our approach with two empirical applications. The first evaluates the effect of the National Supported Work training program on the Aid to Families with Dependent Children sample of women \citep{lalonde1986evaluating} using the dataset from \citet{calonico2017women}. Here, we consider the sample selection problem arising from unemployment. For the second application, we evaluate the effect of a WFH policy on employee performance, considering the sample selection problem arising from employee attrition \citep{bloom2015does}. 

While the approach developed here focuses on a two-period panel, we also explore extensions to repeated cross-sections \citep{sant2026difference,abadie2005semiparametric,finkelstein2002effect, meyer1990workers} and staggered adoption DiD settings with multi-period panels \citep{callaway2021difference}. With repeated cross sections, since observations cannot be tracked across time, it becomes impossible to distinguish between individuals observed in both pre- and post-treatment periods from those observed in only one period. We develop identification results for the always-observed latent group by relying on the assumption of no compositional changes \citep{sant2026difference}. Additionally, we discuss how the current identification argument for the always-observed group can be adapted to multi-period staggered adoption settings by considering $2\times 2$ pre- and post-treatment comparisons for cohorts first treated in a specific period. These extensions are provided in appendices \ref{sec:extension_RC} and \ref{sec:extension_multi}, respectively.

This paper contributes to both the DiD and the sample selection literatures in causal inference. In panel data settings, the traditional sample selection literature has primarily focused on parametric or semiparametric models to achieve point identification of treatment effects
\citep{wooldridge1995selection,kyriazidou1997estimation,rochina1999new,semykina2010estimating}. \citet{lechner2016difference} study the implications of panel non-response in the outcome on the parallel trends assumption by comparing ordinary least squares and fixed effects estimates through simulations and applications.\footnote{They conclude that deviation of ordinary least squares and fixed effects estimation indicates nonignorable attrition.} 

Our paper directly relates to a more recent nonparametric, instrument-free strand that pursues partial identification of treatment effects using principal stratification \citep{frangakis2002principal}. Within this framework, \citet{zhang2003estimation}, \citet{zhang2008evaluating}, \citet{lee2009}, and \citet{chen2015bounds} derive bounds for the average treatment effect among always-observed units. \cite{honore2020selection} build on the trimming logic of \cite{lee2009} and obtain tighter bounds by imposing additional structure on the selection model. In contrast, we follow \cite{lee2009} in targeting treatment effects for latent principal strata and extend this framework to a DiD setting. Subsequent work has explored extensions of this framework along other directions. \citet{bartalotti2023identifying} extend this approach to marginal treatment effects for the always-observed group, while \citet{huber2015sharp} derive bounds for additional latent subpopulations that go beyond the always-observed. All these papers focus on the cross-section setting. In contrast, we incorporate pre-treatment information about selection and outcomes via IS and PTO to additionally account for the endogeneity of treatment with respect to outcome and selection. Our paper also builds on \cite{semenova2025generalized}, who incorporates pre-treatment covariates to relax the global monotonicity assumption used in \cite{lee2009} to a weaker \textit{conditional} monotonicity assumption. 
\citet{viviens2025difference} proposes an alternative relaxation of global monotonicity in settings where multiple sources of sample selection are observed and allows the direction of monotonicity to vary across sources.\footnote{\cite{viviens2025difference} provides an example in which students' outcomes are missing because they (i) dropped out of college or (ii) graduated from college. Obtaining a scholarship (treatment) would affect each source of missingness in a different monotone direction.}

A closely related work is \citet{ghanem2024correcting}, which studies attrition using the changes-in-changes (CiC) approach of \cite{athey2006identification} and point identifies treatment effects. Their approach relies on the assumption that the distribution of unobservables affecting outcomes remains stable over time within each treatment-response subgroup and that potential outcomes are strictly monotone in unobserved heterogeneity. We impose weaker restrictions on outcome dynamics and instead leverage restrictions on the selection mechanism to deliver bounds that are robust to a wider class of outcome heterogeneity and more general forms of sample selection beyond just follow-up non-response. Our approach does not require monotonicity between outcomes and unobservables and also delivers bounds for group-specific treatment effects across latent selection types, in addition to bounds for the overall ATT. \cite{viviens2025difference} also studies partial identification of the average and quantile treatment effects in a CiC framework that could be specialized to DiD under the more restrictive conditions needed for CiC. His approach imposes an absorbing-state restriction on missingness under which units not observed at baseline remain permanently unobserved. This restriction limits \citet{viviens2025difference}'s analysis to just four latent groups. His identification results for the always-observed subgroup, including the mixing proportions and bounds, coincide with ours for the monotonicity case. In contrast, we treat baseline non-observability as an integral part of the selection problem, allowing for a richer principal-strata structure that enables the study of additional latent groups. We also incorporate covariates and develop extensions beyond the canonical two-by-two framework.

Concurrently\footnote{We were only made aware after finishing our first draft paper that \citet{shin2024difference} also independently studies the same setting.}, \citet{shin2024difference} also studies missing outcomes in a DiD framework and partially identifies the ATT for the always-observed group using the trimming procedure by \citet{zhang2003estimation} and \citet{lee2009}. Similar to our approach, she considers identification with and without monotonicity of selection. Once Shin’s implicit conditioning on baseline observability is made explicit, her selection assumptions are equivalent to our IS assumption and the identified mixing proportions under each scenario are also equivalent. This implies that both approaches produce the same bounds for the always-observed group. The key difference lies in scope. Our framework allows for baseline non-observability, explicitly models all principal strata, and develops identification results for additional latent groups (ONO, NON, and NOO). In that sense, Shin (2024) can be viewed as a special case of our more general framework. We additionally extend identification for the always-observed group to settings with covariates and relaxed conditional monotonicity assumption in the spirit of \cite{semenova2025generalized}. We further extend the framework to repeated cross-section and staggered adoption settings. These features are not considered in \citet{shin2024difference}. Instead, the latter pursues point identification of the overall ATT using instrumental variables, whereas we bound the overall ATT without instruments.

The remainder of this article is organized as follows. Section \ref{model} introduces the principal stratification framework and identifying assumptions. Section \ref{Canonicalbias} discusses what na\"{\i}ve DiD identifies when sample selection is ignored. Section \ref{identification} develops the identification strategy for the always-observed latent subgroup and derives ATT bounds for this group both with and without the MS assumption. It also presents the corresponding identification extension that incorporates covariates. Section \ref{identification_other} presents results for three additional latent groups under outcome mean dominance assumptions, while Section \ref{estimation} discusses estimation and inference of the proposed bounds. Section \ref{simulation} presents simulation evidence under a range of data-generating processes. Section \ref{application} provides two empirical illustrations, and Section \ref{conclusion} concludes. Additional identification results, extensions, and technical proofs are collected in the appendices. 

\section{Model Framework}\label{model}
Consider a setting with two time periods denoted by $t=0, 1$. Treatment, $D_{t}$, is available only at period $t=1$, such that $D_{0}=0$ for everyone and $D_1 \equiv D$. For each unit, let $Y_{t}^{\ast}(0)$ and $Y_{t}^{\ast}(1)$ be two continuous latent potential outcomes and $Y_{t}^{\ast} = Y_{t}^{\ast}(0)\cdot(1-D) + Y_{t}^{\ast}(1)\cdot D$ be the realized outcome, which is only observed for a non-random subset of the population. To formalize this, let $S_{t}(0)$ and $S_{t}(1)$ be two potential binary selection indicators such that
\begin{equation}\label{s}
    S_{t} = S_{t}(0)\cdot(1-D) + S_{t}(1)\cdot D
\end{equation} and the researcher observes the data vector $(Y_{t},S_{t},D)$ where
\begin{equation}
    Y_{t} = S_{t}\cdot Y_{t}^{\ast} \label{y}
\end{equation} 
and $S_{t}\in {\{1,0}\}$ is the realized selection indicator, which equals one if the outcome for a unit is observed in period `$t$' and zero otherwise. For example, those with $S_{t}(0)=0$ and $S_{t}(1)=1$ are individuals for whom the outcome would not be observed if they are untreated but would be observed if treated.

\begin{assumption}[No anticipation on selection and outcome]\label{no anti} \ 
\begin{equation*}
    S_{0}=S_{0}(0)= S_{0}(1) \ \text{ and } \ Y^\ast_{0}=Y^\ast_{0}(0)= Y^\ast_{0}(1). 
\end{equation*}
\end{assumption}
Assumption \ref{no anti} formalizes the no-anticipation assumptions on selection and potential outcomes in the pre-treatment period. It states that there can be no anticipatory effects of the treatment assignment on sample selection or on the latent potential outcomes at baseline. This is plausible in situations where the treatment is not announced in advance, thereby discouraging individuals from basing their decision to be observed in the sample on whether they will receive the treatment in the future.

We consider the principal stratification framework introduced by \citet{frangakis2002principal} to divide the population into latent subgroups based on the potential sample selection indicators in both periods. This results in sixteen groups, which can be reduced to the eight groups presented in Table \ref{latentgroups} since  $S_{0}(0) = S_{0}(1)$ (Assumption \ref{no anti}).\footnote{Following the nomenclature used in \citet{lee2009}, \citet{huber2015sharp} and \citet{bartalotti2023identifying}, we use ``O'' and ``N'' to denote observed and not observed, respectively.} Let $G$ denote the principal strata or latent group to which a unit belongs, with `g' denoting the group denomination. 

\begin{table}[!h]
\caption{Latent groups based on the sample selection}
    \label{latentgroups}
     \begin{adjustbox}{center=\textwidth}
\begin{tabular}{cccc}
\hline
\multicolumn{1}{l}{\textbf{$S_{0}(0)$}} & \multicolumn{1}{l}{\textbf{$S_{1}(0)$}} & \multicolumn{1}{l}{\textbf{$S_{1}(1)$}} & \multicolumn{1}{l}{$G=g$} \\ \hline
 0                                         & 0                                         & 0                                         & NNN                                      \\
 0                                         & 0                                         & 1                                         & NNO                                      \\
 0                                         & 1                                         & 0                                         & NON                                      \\
 0                                         & 1                                         & 1                                         & NOO                                      \\
 1                                         & 0                                         & 0                                         & ONN                                      \\
 1                                         & 0                                         & 1                                         & ONO                                      \\
 1                                         & 1                                         & 0                                         & OON                                      \\
 1                                         & 1                                         & 1                                         & OOO                             \\ \hline        
\end{tabular}
\end{adjustbox}
\end{table}

Following \citet{lee2009}, we define our target parameter to be the ATT for the subpopulation that is always observed, denoted by OOO, and indicates that selection equals one in all periods and under both counterfactual treatment states. Formally,
\begin{equation}\label{att_ooo}
    \tau_{OOO}=\mathbbm{E}[Y_{1}^{\ast}(1)-Y_{1}^{\ast}(0)|D=1,S_{0}(0)=1,S_{1}(0)=1,S_{1}(1)=1].
\end{equation}
For $\tau_{OOO}$, we require a less restrictive version of parallel trends in outcomes that applies only to the OOO group.
\begin{assumption}[Parallel trends for the OOO group]\label{PT_OOO} \
\begin{equation}\label{stag}
    \mathbbm{E}[Y_{1}^{\ast}(0)-Y_{0}^{\ast}(0)|D=1, OOO]=\mathbbm{E} [Y_{1}^{\ast}(0)-Y_{0}^{\ast}(0)|D=0, OOO]. 
    \end{equation} 
\end{assumption}
Assumption \ref{PT_OOO} states that the changes in the potential outcomes for always-observed (OOO) individuals, in the absence of treatment, would have been the same across the treatment and control groups. This assumption is weaker than requiring parallel trends for the full population of treated and control units, since that also includes other latent types beyond the OOO group. If additional pre-treatment periods are available, one can estimate a placebo DiD using only pre-treatment data. As shown in Appendix \ref{pretrend}, a non-zero estimate may reflect violations of parallel trends within the OOO or ONO group, cross-group trend differentials across latent groups, or any joint combination thereof. Therefore, such a test cannot isolate or falsify the plausibility of parallel trends for the always-observed group alone.\footnote{We thank an anonymous referee for suggesting this approach.} 

While not required for the most general results in Section \ref{sec:ident_nomon}, we consider a monotonicity assumption that is widely used in the literature \citep{lee2009,huber2015sharp,chen2015bounds,bartalotti2023identifying}, which requires that the treatment affect sample selection in only one direction. 
\begin{assumption}[Positive monotone sample selection]\label{monotone}\ 
	\begin{align*}
		\mathbbm{P}[S_{1}(1)\geq S_{1}(0)]&=1. 
	\end{align*}
\end{assumption}
Without loss of generality, Assumption \ref{monotone} assumes that treatment increases the probability of selection or has a non-decreasing effect on sample selection for all individuals. Positive selection implies that there are no individuals whose outcome is observed only when untreated. For example, attending the job training program cannot decrease any individual's employment probability and, thus, does not decrease his/her chance of being observed. Assumption \ref{monotone} rules out the strata NON and OON, that is, individuals that would be observed in period one if untreated but not if treated.\footnote{Symmetric results can be derived under negative monotonicity, which are discussed in Appendix \ref{neg_mono}.} We also assume that our setup has no spillovers and hidden treatment variations. In other words, we assume that the stable unit treatment value assumption holds. 

Finally, we impose the following restriction on the relationship between the potential selection mechanism and treatment assignment.

\begin{assumption}[Ignorability of treatment in potential selection]\label{Partialselection}
\ \\
     (a) Equality in the untreated counterfactual share of individuals observed in period 1 conditional being observed in period 0: 
	\begin{align}
    \mathbbm{P}[S_{1}(0)=1|D=0, S_0=1]=\mathbbm{P}[S_{1}(0)=1|D=1, S_{0}=1].
    \end{align}
    (b) Equality in the treated counterfactual share of individuals observed in period 1 conditional on being observed in period 0:
	\begin{align}
    \mathbbm{P}[S_{1}(1)=1|D=1, S_{0}=1]=\mathbbm{P}[S_{1}(1)=1|D=0, S_{0}=1].
    \end{align}
\end{assumption}
Each part of Assumption \ref{Partialselection} imposes that the counterfactual proportion of individuals observed in the post-treatment period among those observed in the pre-treatment period be the same across the two treatment groups. As we discuss in Section \ref{identification}, under positive (negative) monotonicity, $\tau_{OOO}$ is partially identified under \ref{Partialselection}(a) (\ref{Partialselection}(b)), which restricts only the selection behavior in the untreated (treated) counterfactual.

Assumption 4 is plausible when the unobservables affecting sample selection in the post-treatment period are not systematically related to factors affecting the decision to select into treatment, once we condition on baseline observability. For example, in the WFH application, Assumption 4 is plausible when eligibility is determined based on factors that are not related to employees’ latent retention risk, conditional on being observed at baseline. It is less plausible when WFH is granted selectively based on managerial judgments of burnout risk, outside offers, or other unobserved predictors of attrition.

Note that this assumption only focuses on post-treatment selection behavior of individuals observed in the pre-treatment period ($S_0=1$) but not those who are unobserved in the pre-treatment period ($S_0=0$). Assumption \ref{Partialselection} can also be viewed as selection on lagged outcomes, and its connection to parallel-trends-type restrictions has been noted in \cite{ding2019bracketing}. However, parallel trends for binary outcomes are known to be very restrictive (see \cite{marx2024parallel} and \cite{ghanem2022selection}).

Although Assumption \ref{Partialselection} is inherently untestable, since the relevant counterfactual selection rates across treatment arms (for example, $\mathbbm{P}(S_{1}(1)=1\mid D=0,S_{0}=1)$ and $\mathbbm{P}(S_{1}(0)=1\mid D=1,S_{0}=1)$) are not simultaneously observed, its plausibility can be partially assessed by examining whether pre-treatment selection rates are similar across treatment and control groups. Any significant pre-treatment differences may indicate that the assumption may be less tenable in the post-treatment period. Formally, the null hypothesis corresponding to Assumption 4(a) can be written as
\begin{equation}\label{eq:pretrends_s}
	H_0: \mathbbm{P}\left(S_{0}(0)=1 \mid D=0, S_{-1}=1\right)-\mathbbm{P}\left(S_{0}(0)=1 \mid D=1, S_{-1}=1\right)= 0.
\end{equation}
Both probabilities in \eqref{eq:pretrends_s} are identified from observed data under no-anticipation  for selection in the pre-treatment period. A separate empirical test for Assumption 4(b) is not feasible because, under the same no anticipation in selection assumption,
\begin{equation*}
	\mathbbm{P}(S_0(1)=1\mid D=d, S_{-1}=1)=\mathbbm{P}(S_0(0)=1\mid D=d, S_{-1}=1)\quad \text{for } d\in\{0,1\}.
\end{equation*}	
As a result, testing equality of treated counterfactual selection probabilities across treatment groups produces the same restriction as in \eqref{eq:pretrends_s}. Hence, Assumptions 4(a) and 4(b) imply a testable implication that is identical in the pre-treatment period. 

\begin{assumption}[Random Sampling]\label{rs} \ \\
Assume that $\{(Y_{i0}, Y_{i1}, D_i, S_{i0}, S_{i1}); \ i=1,\ldots,N\}$ are i.i.d draws from an infinite population.
\end{assumption}
Our notion of sample selection is different from the idea of compositional changes that is discussed in \cite{hong2013measuring} and \cite{sant2026difference}. In those papers, selection arises from comparing independently drawn random samples from different time periods (pre- and post-treatment).  This can lead to compositional changes, i.e., the joint distribution of covariates and treatment assignment can vary over time due to sample randomness, creating a ``non-stationarity'' problem that, if unaddressed, produces incorrect estimands for the ATT of interest due to the heterogeneity of the treatment effect. On the other hand, we consider a panel data setting where the same individuals are sampled in both periods. Therefore, sampling for $Y_{i0}$ and $Y_{i1}$ in our framework is stationary and have no compositional changes. Any compositional changes in the observed outcomes can only arise because of outcomes being non-randomly observed due to endogenous selection, which we explicitly model.


\section{What does DiD identify if we ignore sample selection?}\label{Canonicalbias}
Before we delve into identification of $\tau_{OOO}$, it's useful to understand what a na\"{\i}ve DiD estimand that ignores the problem of sample selection identifies. We denote this as $\tau_{\textup{DiDs}}$. With panel data, $\tau_{\textup{DiDs}}$ compares average outcomes over time between the treated and control groups for individuals that are observed in both periods $(S_{0}=1, S_{1}=1)$. Lemma \ref{lemma_dids} shows that $\tau_{\textup{DiDs}}$ 
is biased for the overall ATT, defined as $\tau \equiv \mathbbm{E}[Y_1^*(1) - Y_1^*(0) |D = 1]$.

\begin{lemma}[Bias of  $\tau_{\textup{DiDs}}$]\label{lemma_dids} Under Assumptions \ref{no anti} and \ref{PT_OOO}, the DiD estimand for the observed group in a two-period panel, $\tau_{\textup{DiDs}} \equiv \mathbbm{E}[Y_1-Y_0|D=1, S_0=1, S_1=1] - \mathbbm{E}[Y_1-Y_0|D=0,S_0=1, S_1=1]$, can be decomposed as:
\begin{enumerate}
    \item Relative to the overall ATT, $\tau$.
    \begin{equation*}
	\tau_{\mathrm{DiDs}} = \tau + \Delta_0 + \Delta_{\textup{Het}}
\end{equation*}
where the bias components are,
\begin{align*}
	\Delta_0 \ \equiv \ & \mathbbm{E}[Y_1^\ast(0) - Y_0^\ast(0) |D = 1, S_0 = 1, S_1 = 1] \\
	& - \mathbbm{E}[Y_1^\ast(0) - Y_0^\ast(0) |D = 0, S_0 = 1, S_1 = 1] \\
	\Delta_{\textup{Het}} \ \equiv \ & (\tau_{11} - \tau_{10}) \cdot p_{10} + (\tau_{11} - \tau_{01}) \cdot p_{01} + (\tau_{11} - \tau_{00}) \cdot p_{00}
\end{align*}
where \( \tau_{s_0 s_1} = \mathbbm{E}[Y_1^\ast(1) - Y_1^\ast(0) |D = 1, S_0 = s_0, S_1 = s_1] \) is the treatment effect for the treated among the \( (S_0=s_0, S_1= s_1) \) subpopulation for each $s_t \in \{0,1\}$; $t=0,1$, and \( p_{s_0 s_1} = \mathbbm{P}(S_0 = s_0, S_1 = s_1 |D = 1) \) is the proportion of units belonging to $S_0=s_0, S_1 = s_1$  subpopulation among the treated.
    \item Relative to the ATT for the latent groups.
    \begin{itemize}
        \item[2(a).] Without monotonicity:  
    {\small \begin{align*}
       \tau_{\textup{DiDs}} & = \tau_{OOO}\cdot  p_{OOO1} +\tau_{ONO}\cdot(1-p_{OOO1}) \\ 
        & + \big\{\mathbbm{E}[Y_1^\ast(0)-Y_0^\ast(0)|D=1, ONO] -\mathbbm{E}[Y_1^\ast(0)-Y_0^\ast(0)|D=0, ONO] \big\}\cdot (1-p_{OOO1}) \\
        & +\big\{ \mathbbm{E}[Y_1^\ast(0)-Y_0^\ast(0)|D=0, OOO] - \mathbbm{E}[Y_1^\ast(0)-Y_0^\ast(0)|D=0, ONO]\big\}\cdot p_{OOO1} \\
        & + \big\{\mathbbm{E}[Y_1^\ast(0)-Y_0^\ast(0)|D=0, ONO] - \mathbbm{E}[Y_1^\ast(0)-Y_0^\ast(0)|D=0, OON]\big\} \\
        & + \big\{\mathbbm{E}[Y_1^\ast(0)-Y_0^\ast(0)|D=0, OON] - \mathbbm{E}[Y_1^\ast(0)-Y_0^\ast(0)|D=0, OOO]\big\}\cdot p_{OOO0}.
    \end{align*}}
    \item[2(b).] With positive monotonicity (Assumption \ref{monotone})
    {\small \begin{align*}
        \tau_{\textup{DiDs}} &=p_{OOO1}\cdot \tau_{OOO}+(1-p_{OOO1})\cdot \tau_{ONO}+(1-p_{OOO1})\cdot\bigg\{\bigg(\mathbbm{E}[Y_{1}^{\ast}(0)-Y_{0}^{\ast}(0)|D=1, ONO] \\
        &-\mathbbm{E}[Y_{1}^{\ast}(0)-Y_{0}^{\ast}(0)|D=0, ONO]\bigg) + \bigg(\mathbbm{E}[Y_1^\ast(0)-Y_0^\ast(0)|D=0, ONO] \\
        &-\mathbbm{E}[Y_1^\ast(0)-Y_0^\ast(0)|D=0, OOO]\bigg) \bigg\}  
    \end{align*}}
   \end{itemize}
  where  $p_{OOO1}=\sfrac{\pi_{OOO1}}{(\pi_{OOO1}+\pi_{ONO1})}$, $p_{OOO0}=\sfrac{\pi_{OOO0}}{(\pi_{OOO0}+\pi_{OON0})}$, and $\pi_{gd}=\mathbbm{P}[G=g, D=d]  = \mathbbm{P}[S_0(0)=s,S_1(0)=s',S_1(1)=s'',D=d]$.
  \end{enumerate}
\end{lemma}
The proof is presented in Appendix \ref{L1}.
The first part of Lemma \ref{lemma_dids} shows that $\tau_{\textup{DiDs}}$ is biased for the overall ATT. The bias components are given by \( \Delta_0 \) and \(\Delta_{\text{Het}}\). The term \( \Delta_0 \) represents differential trends in untreated potential outcomes among treated and untreated groups that are observed in both periods. The second bias term \( \Delta_{\text{Het}} \) captures treatment effect heterogeneity in the ATTs across the $S_0=s_0, S_1=s_1$ subpopulations. It becomes clear from this bias expression that assuming trends in untreated potential outcomes to be parallel between the treated and untreated i.e. \(\mathbbm{E}[Y_{1}^\ast(0)-Y_{0}^\ast(0)|D=1, S_0=1, S_1=1]=\mathbbm{E} [Y_{1}^\ast(0)-Y_{0}^\ast(0)|D=0, S_0=1, S_1=1] \), would eliminate $\Delta_0$. However, this would implicitly assume that selection is exogenous with respect to the untreated potential outcome trends, which is misguided in the current framework of endogenous sample selection. Even then, one would still be left with the bias term $\Delta_{\text{Het}}$. 

The second part of Lemma \ref{lemma_dids} establishes that the na\"{\i}ve DiD estimand could also be represented as a weighted average of ATTs for the always-observed (OOO) and the observed-only-when-treated (ONO) subgroups, with weights given by their respective proportions in the treated population, $p_{OOO1}$ and $1-p_{OOO1}$, plus additional terms representing selection bias. These terms include (i) the difference in trends in the absence of treatment between the treated and untreated ONO group and (ii) the cross-group differential in untreated potential outcomes trends across the different latent groups (i.e. ONO, OOO, OON) in the untreated subpopulation.
If we assume parallel trends for the ONO group, the first bias term disappears, and we are left with heterogeneity in untreated potential outcome trends between the three latent groups. Naturally, this implies that if there is no heterogeneity in trends between these groups, i.e. $\mathbbm{E}[Y_1^\ast(0)-Y_0^\ast(0)|D=0, OOO] = \mathbbm{E}[Y_1^\ast(0)-Y_0^\ast(0)|D=0, ONO] = \mathbbm{E}[Y_1^\ast(0)-Y_0^\ast(0)|D=0, OON]$, then $\tau_{\textup{DiDs}}$ has a causal interpretation as a weighted average of group-specific ATTs.

The bias of the na\"{\i}ve DiD simplifies when one imposes positive monotonicity (Assumption \ref{monotone}). In this case, the OON latent stratum disappears from the $D=0$ group and $p_{OOO0}=1$. Just as before, the bias now arises from (i) the difference in trends in the absence of treatment between the treated and untreated ONO group ($\mathbbm{E}[Y_1^\ast(0)-Y_0^\ast(0)|D=1, ONO] - \mathbbm{E}[Y_1^\ast(0)-Y_0^\ast(0)|D=0, ONO]$) and (ii) the difference in untreated potential outcome trends between the OOO and ONO groups ($\mathbbm{E}[Y_1^\ast(0)-Y_0^\ast(0)|D=0, ONO] - \mathbbm{E}[Y_1^\ast(0)-Y_0^\ast(0)|D=0, OOO]$). This implies that even if we assume parallel trends for the ONO group, the DiD estimand will still be biased on account of differential trends between the untreated OOO and ONO groups. The extent of this bias depends on the share of ONO among those assigned treatment and the heterogeneity in average trends in the untreated potential outcomes between the ONO and OOO. A larger share of OOO will cause $\tau_{\textup{DiDs}}$ to be primarily influenced by $\tau_{OOO}$, thereby reducing the impact of the cross-group differences in trends, whereas a larger share of ONO will cause $\tau_{\textup{DiDs}}$ to be primarily influenced by $\tau_{ONO}$ while amplifying the cross-group difference in trends. Notice that even if the selection mechanism is completely independent of the treatment assignment process, na\"{\i}ve DiD would still be biased since selection might still be endogenous to the outcome of interest. If $S_1(0)$ is exogenous to the trends in the untreated potential outcome, then the bias in na\"{\i}ve DiD would disappear and it would give us $p_{OOO1}\cdot \tau_{OOO}+(1-p_{OOO1})\cdot \tau_{ONO}$. 

It is important to note that a value of $p_{OOO1}$ close to one is informative but not sufficient for assessing the overall validity of the na\"{\i}ve DiD estimand in the presence of endogenous sample selection. This is because $p_{OOO1}\rightarrow 1$ indicates that the treated units observed in both periods are composed almost entirely of the OOO units, which eliminates selection bias for this group. As shown in Lemma \ref{lemma_dids}(2b), under positive monotonicity, this is enough to ensure that the na\"{\i}ve DiD estimand recovers $\tau_{OOO}$, because the comparison group is also composed solely of OOO units. However, without monotonicity, the untreated group observed in both periods may still contain a mixture of OOO and OON units, and cross-group differences in untreated potential outcome trends can still generate bias even if $p_{OOO1} \approx 1$.

\section{Identification of ATT for OOO} \label{identification}
This section presents alternative conditions under which we can identify the parameter of interest, $\tau_{OOO}$. First, we discuss identifying the difference in the expected potential outcomes for latent groups. The identification problem arises from the fact that we do not observe the latent group membership directly since we either observe $S_{1}(0)$ or $S_{1}(1)$, but never both. 

It is useful to note that $\tau_{OOO}$ could be identified by a hypothetical DiD estimand for members of the OOO latent group.
\begin{align*}
&\mathbbm{E}[Y_{1}-Y_{0}|D=1, OOO] - \mathbbm{E}[Y_{1}-Y_{0}|D=0, OOO]\\
&=\mathbbm{E}[S_{1}(1)Y_{1}^{\ast}(1)-S_{0}(1)Y_{0}^{\ast}(1)|D=1, OOO] - \mathbbm{E}[S_{1}(0)Y_{1}^{\ast}(0)-S_{0}(0)Y_{0}^{\ast}(0)|D=0, OOO]\\
&=\mathbbm{E}[Y_{1}^{\ast}(1)-Y_{0}^{\ast}(1)|D=1, OOO] - \mathbbm{E}[Y_{1}^{\ast}(0)-Y_{0}^{\ast}(0)|D=0, OOO]\\
&=\mathbbm{E}[Y_{1}^{\ast}(1)-Y_{0}^{\ast}(0)|D=1, OOO] - \mathbbm{E}[Y_{1}^{\ast}(0)-Y_{0}^{\ast}(0)|D=0, OOO], \text{(Assumption \ref{no anti})}\\
&=\mathbbm{E}[Y_{1}^{\ast}(1)-Y_{0}^{\ast}(0)|D=1, OOO] - \mathbbm{E}[Y_{1}^{\ast}(0)-Y_{0}^{\ast}(0)|D=1, OOO],\text{(Assumption \ref{PT_OOO})}\\
&=\mathbbm{E}[Y_{1}^{\ast}(1)-Y_{1}^{\ast}(0)|D=1, OOO]=\tau_{OOO}
\end{align*} 
   
Although $\mathbbm{E}[Y_{1}^{\ast}(d)-Y_{0}^{\ast}(d)|D=d, OOO]$ cannot be generally point identified for $d=\{0,1\}$, it can be partially identified under different combinations of the monotonicity and selection mechanism assumptions. The plausibility of the assumptions required for partial identification depends on the empirical context. We approach this constructively by obtaining bounds for $\tau_{OOO}$ under less informative assumptions that might be valid on a larger range of empirical settings and then moving towards more restrictive assumptions that could be more informative for the parameter of interest. This allows a layered policy analysis \citep{manski2011}, offering various estimates based on different assumptions so that the researcher can explore the information gathered about the parameter of interest by each restriction, as advocated in \cite{tamer2010partial}.

Following the literature, we take advantage of the representation of observed subgroups of individuals as mixtures of latent groups, as shown in Table \ref{tabl obs groups}  \citep{lee2009,chen2015bounds,huber2015sharp,bartalotti2023identifying}. The relationship between observed and latent groups partially identifies $\mathbbm{E}[Y_{1}^{\ast}(d)-Y_{0}^{\ast}(d)|D=d, OOO]$, which we can use to recover $\tau_{OOO}$. 

\begin{table}[H]
\caption{Observed and Latent Groups}
 \label{tabl obs groups}
     \begin{adjustbox}{center=\textwidth}
\begin{tabular}{cccc}
\hline
\textbf{$S_{0}$} & \textbf{$S_{1}$} & \textbf{($D=0$)} & \textbf{($D=1$)} \\ \hline
0                 & 0                  & NNN, NNO             & NNN, NON                 \\
0                 & 1                  & NOO,NON                  & NOO,NNO              \\ 
1                 & 0                  & ONN,ONO              & ONN,OON                  \\
1                 & 1                  & OOO,OON                 & OOO,ONO   \\ \hline         
\end{tabular}
\end{adjustbox}
\end{table}

For instance, consider the group of treated individuals for whom the outcome is observed in both periods $(D=1, S_{0}=1, S_{1}=1)$. Table \ref{tabl obs groups} shows that their observed average outcome reflects a mixture of the potential outcomes for the OOO and ONO latent groups with mixing probabilities corresponding to their relative proportions. 
Then,
{\small \begin{align}\label{mixp}
     &\mathbbm{E}[Y_{1}-Y_{0}|D=1,S_{0}=1, S_{1}=1]= \nonumber \\
    &= \mathbbm{E}[Y_{1}^{\ast}(1)-Y_{0}^{\ast}(1)|D=1,S_{0}=1, (S_{1}(0)=1,S_{1}(1)=1) \quad or \quad (S_{1}(0)=0,S_{1}(1)=1)] \nonumber\\
    &=\mathbbm{E}[Y_{1}^{\ast}(1)-Y_{0}^{\ast}(1)|D=1,OOO]\cdot p_{OOO1} +\mathbbm{E}[Y_{1}^{\ast}(1)-Y_{0}^{\ast}(1)|D=1,ONO]\cdot (1-p_{OOO1}).
\end{align}}
Now, consider the group of control individuals for whom the outcome is observed in both periods $(D=0, S_{0}=1, S_{1}=1)$. Their observed average outcome is a mixture of the potential outcomes for the OOO and OON latent groups,
{\small \begin{align}\label{mixp2}
     &\mathbbm{E}[Y_{1}-Y_{0}|D=0,S_{0}=1, S_{1}=1]= \nonumber \\
    &=\mathbbm{E}[Y_{1}^{\ast}(0)-Y_{0}^{\ast}(0)|D=0,S_{0}=1, (S_{1}(0)=1,S_{1}(1)=1) \quad or \quad (S_{1}(0)=1,S_{1}(1)=0)] \nonumber\\
    &=\mathbbm{E}[Y_{1}^{\ast}(0)-Y_{0}^{\ast}(0)|D=0,OOO]\cdot p_{OOO0}+\mathbbm{E}[Y_{1}^{\ast}(0)-Y_{0}^{\ast}(0)|D=0,OON]\cdot (1-p_{OOO0}).
\end{align}}
We use these mixture representations to bound the expected change in potential outcomes within the always-observed subpopulation by looking at the observed outcomes’ distribution for treated individuals. Specifically, the lower bound for $\mathbbm{E}[Y_{1}^{\ast}(1)-Y_{0}^{\ast}(1)|D=1,OOO]$ is obtained by considering the worst-case scenario in which the OOO group comprises of individuals with the lowest values of $Y_{1}^{\ast}(1)-Y_{0}^{\ast}(1)$ among the subpopulation of treated individuals that has been observed in both periods. This corresponds to the left tail of mass $p_{OOO1}$ of the distribution of changes in outcomes between the pre- and post-treatment periods for treated individuals. The upper bound analogously assumes that the OOO group lies in the right tail of the same distribution, with the highest values of $Y_{1}^{\ast}(1)-Y_{0}^{\ast}(1)$. The intuition is similar to the trimming procedure suggested by \citet{lee2009} and others, where we assume that the OOO group corresponds to the set of treated individuals who either had the lowest observed changes in outcome between the pre- and post-treatment periods (yielding the lower bound) or experienced the highest observed changes (giving us the upper bound).
Hence, $\mathbbm{E}[Y_{1}^{\ast}(1)-Y_{0}^{\ast}(1)|D=1,OOO]$ lies within the interval $[LB_{OOO1},UB_{OOO1}]$ where,
\begin{align}
    LB_{OOO1} &=\mathbbm{E}[Y_{1}-Y_{0}|D=1,S_{0}=1,S_{1}=1, (Y_{1}-Y_{0})\leq F_{\Delta Y|111}^{-1}(p_{OOO1})]\label{LB1}\\
    UB_{OOO1} &=\mathbbm{E}[Y_{1}-Y_{0}|D=1,S_{0}=1,S_{1}=1, (Y_{1}-Y_{0}) > F_{\Delta Y|111}^{-1}(1-p_{OOO1})]. \label{UB1}
\end{align}
Similarly, the conditional distribution $Y_{1}-Y_{0}$ for the untreated individuals observed in both time periods can be trimmed to obtain the bounds for $\mathbbm{E}[Y_{1}^{\ast}(0)-Y_{0}^{\ast}(0)|D=0,OOO]$, which lies within the interval $[LB_{OOO0},UB_{OOO0}]$,
\begin{align}
    LB_{OOO0} &=\mathbbm{E}[Y_{1}-Y_{0}|D=0,S_{0}=1,S_{1}=1,(Y_{1}-Y_{0})\leq F_{\Delta Y|011}^{-1}(p_{OOO0})]\label{LB0}\\
    UB_{OOO0} &=\mathbbm{E}[Y_{1}-Y_{0}|D=0,S_{0}=1,S_{1}=1,(Y_{1}-Y_{0}) > F_{\Delta Y|011}^{-1}(1-p_{OOO0})]. \label{UB0}
\end{align}
In equations \eqref{LB1}-\eqref{UB0} above, $F_{\Delta Y|dss'}^{-1}(.)$ is the quantile function of the distribution of the variable $\Delta Y \equiv Y_1-Y_0$ given $D=d,S_{0}=s,S_{1}=s'$. 

Combining the bounds for $\mathbbm{E}[Y_{1}^{\ast}(1)-Y_{0}^{\ast}(1)|D=1,OOO]$ and $\mathbbm{E}[Y_{1}^{\ast}(0)-Y_{0}^{\ast}(0)|D=0,OOO]$ we find that the parameter of interest $\tau_{OOO}$ is in the interval 
\begin{align*}
    [LB_{OOO1}-UB_{OOO0},UB_{OOO1}-LB_{OOO0}].
\end{align*}
The fundamental aspect of identifying the target parameter is what can be learned about the weights, $p_{OOO0}$ and $p_{OOO1}$. Since we are interested in the always-observed group, a higher share of OOO among the treated individuals for which we have complete data implies that the observed sample provides more information about the changes in outcome for that group. In the extreme case, $p_{OOO1}\rightarrow 1$ and $\mathbbm{E}[Y_{1}^{\ast}(1)-Y_{0}^{\ast}(1)|D=1,OOO]$ is point identified as the observed sample reflects only the OOO type. In the opposite case, $p_{OOO1}\rightarrow 0$ and the observed sample would be uninformative about the always-observed group.

To this end, we consider alternative assumptions that impose different restrictions on the admissible values of the latent mixing proportions, $p_{OOO0}$ and $p_{OOO1}$, thereby yielding more information about $\tau_{OOO}$. 

\subsection{Identification without Monotonicity}\label{sec:ident_nomon}
Initially, consider the case where the researcher is unwilling to assume monotonicity in selection (Assumption \ref{monotone}). We are interested in the unobserved share of always-observed individuals, $\mathbbm{P}[S_0(0)=1, S_1(0)=1, S_1(1)=1, D=d]$. The share of units observed in both periods among each treatment group is informative about the mixing proportions. For the treated group,
\begin{align}\label{pipartialid1}
    &\mathbbm{P}[S_{0}=1, S_{1}=1| D=1]=\mathbbm{P}[S_{0}=1, S_{1}(1)=1| D=1]\\
    &=\mathbbm{P}[S_{0}=1, S_{1}(0)=1, S_{1}(1)=1| D=1]+\mathbbm{P}[S_{0}=1, S_{1}(0)=0, S_{1}(1)=1| D=1] \nonumber\\
    &=\frac{\pi_{OOO1}}{\mathbbm{P}[D=1]}+\frac{\pi_{ONO1}}{\mathbbm{P}[D=1]}.\nonumber
\end{align}
And for untreated observations,
\begin{align}\label{pipartialid2}
    &\mathbbm{P}[S_{0}=1, S_{1}=1| D=0]=\mathbbm{P}[S_{0}=1, S_{1}(0)=1| D=0]\\
    &=\mathbbm{P}[S_{0}=1, S_{1}(0)=1, S_{1}(1)=1| D=0]+\mathbbm{P}[S_{0}=1, S_{1}(0)=1, S_{1}(1)=0| D=0]\nonumber\\
    &=\frac{\pi_{OOO0}}{\mathbbm{P}[D=0]}+\frac{\pi_{OON0}}{\mathbbm{P}[D=0]}.\nonumber
\end{align}
The first equality in the equations above formalize the intuition that we can identify the marginal conditional proportions $\mathbbm{P}[S_{0}=1, S_{1}(d)=1| D=d]$ from observed data. It is useful to express 
{\small \begin{align*}
   \mathbbm{P}[S_{0}=1, S_{1}(0)=1, S_{1}(1)=1| D=d] = \mathbbm{P}[S_{1}(0)=1, S_{1}(1)=1| D=d, S_{0}=1]\cdot \mathbbm{P}[S_0=1|D=d] 
\end{align*}}
where $\mathbbm{P}[S_0=1|D=d]$ is directly observed in the data whereas $\mathbbm{P}[S_{1}(0)=1, S_{1}(1)=1| D=d, S_{0}=1]$ can be partially identified using Fr\'echet bounds \citep{imai2008sharp} as follows:
{\footnotesize
\begin{align}
\mathbbm{P}[S_{1}(0)=1, S_{1}(1)=1|D=d, S_{0}=1]\in& \left[\max\{\mathbbm{P}[S_{1}(0)=1|D=d, S_{0}=1]+\mathbbm{P}[S_{1}(1)=1|D=d, S_{0}=1]-1,0\},\right.\nonumber\\
&\left.\min\{\mathbbm{P}[S_{1}(0)=1|D=d, S_{0}=1],\mathbbm{P}[S_{1}(1)=1|D=d,S_{0}=1]\}\right].
\end{align}}

Note that $\mathbbm{P}[S_{1}(0)=1|D=0, S_{0}=1]$ and $\mathbbm{P}[S_{1}(1)=1|D=1, S_{0}=1]$ are directly identified from the observed data. We consider assumptions restricting the relationship between the selection mechanism and treatment assignment to identify their counterfactual counterparts, $\mathbbm{P}[S_{1}(0)=1|D=1, S_{0}=1]$ and $\mathbbm{P}[S_{1}(1)=1|D=0, S_{0}=1]$. Since the identification of $\mathbbm{E}[Y_{1}^{\ast}(1)-Y_{0}^{\ast}(1)|D=1,OOO]$ depends only on $p_{OOO1}$ and equivalently that of $\mathbbm{E}[Y_{1}^{\ast}(0)-Y_{0}^{\ast}(0)|D=0,OOO]$ solely on $p_{OOO0}$, we consider the assumptions for each term separately.

By combining Assumption \ref{Partialselection} and the information in equations (\ref{pipartialid1}) and (\ref{pipartialid2}), we can identify the missing counterfactual probabilities through the observed proportions for the treated and untreated groups, leading to Lemma \ref{lemma:pipartialid}.
\begin{lemma}\label{lemma:pipartialid}
    (a) Under Assumptions \ref{no anti}, \ref{Partialselection}(a), and \ref{rs} the identified set for $\mathbbm{P}[S_{1}(0)=1, S_{1}(1)=1|D=1, S_{0}=1]$ is given by
    {\footnotesize
 \begin{align}
 \mathbbm{P}[S_{1}(0)=1, S_{1}(1)=1|D=1, S_{0}=1]\in &\left[\max\{\mathbbm{P}[S_{1}=1|D=0, S_{0}=1]+\mathbbm{P}[S_{1}=1|D=1, S_{0}=1]-1,0\},\right. \nonumber\\
 &\left.\min\{\mathbbm{P}[S_{1}=1|D=0, S_{0}=1],\mathbbm{P}[S_{1}=1|D=1, S_{0}=1]\}\right].
 \end{align}}
    (b) Under Assumptions \ref{no anti}, \ref{Partialselection}(b), and \ref{rs} the identified set for $\mathbbm{P}[S_{1}(0)=1, S_{1}(1)=1|D=0, S_{0}=1]$ is given by
    {\footnotesize
 \begin{align}
 \mathbbm{P}[S_{1}(0)=1, S_{1}(1)=1|D=0, S_{0}=1]\in &\left[\max\{\mathbbm{P}[S_{1}=1|D=0, S_{0}=1]+\mathbbm{P}[S_{1}=1|D=1, S_{0}=1]-1,0\},\right.\nonumber\\
&\left.\min\{\mathbbm{P}[S_{1}=1|D=0, S_{0}=1],\mathbbm{P}[S_{1}=1|D=1, S_{0}=1]\}\right].
 \end{align}}
\end{lemma}
The proof of Lemma \ref{lemma:pipartialid} can be found in Appendix \ref{proof:pipartialid}. The restrictive nature of assuming both parts of Assumption \ref{Partialselection} becomes clear as the identified set for 
$\mathbbm{P}[S_{1}(0)=1, S_{1}(1)=1|D=d, S_{0}=1]$  is the same for both treated and control groups in that case, reflecting that the probability of being always-observed is independent of treatment under that assumption. This simplifies the identification of the mixing weights and is similar to scenarios in which the treatment is exogenous \citep{lee2009}, or an instrument is available for selection and treatment \citep{bartalotti2023identifying}. However, the weights will still differ between treated and untreated groups, which remains a challenge for identification in this setting that is not yet addressed in the previous literature.

Lemma \ref{lemma:pipartialid} can be used to obtain the range of possible values that $p_{OOO1}$ and $p_{OOO0}$ can take. For any value $v_d$ in the identified set for $\mathbbm{P}[S_{1}(0)=1, S_{1}(1)=1|D=d, S_{0}=1]$, the $p_{gd}$ associated with it is given by $p_{gd}(v_d)=\frac{v_d}{\mathbbm{P}[S_{1}=1| D=d, S_{0}=1]}$. 
As previously discussed, higher values for $p_{OOO1}$ and $p_{OOO0}$ indicate that a larger share of the observed - treated and untreated, respectively - population belongs to the always-observed latent groups, thus providing more information and tighter bounds for the target parameters. Hence, we only need to focus on the scenario that generates the wider bounds, that is, the smallest $p_{OOO1}(v_1)$ and $p_{OOO0}(v_0)$ \citep{bartalotti2023identifying}. Since $v_d$ has a monotone relationship to the mixture weights, the relevant case is obtained at the lower bound of each of the identified sets for $\mathbbm{P}[S_{1}(0)=1, S_{1}(1)=1|D=d, S_{0}=1]$ described in Lemma \ref{lemma:pipartialid}, which we call $v^{l}_{d}$ for $d=0,1$.

Evaluating equations (\ref{LB1})-(\ref{UB1}) at the least favorable values for $p_{OOO1}(v_1)$ yields,
\begin{align}
    LB_{OOO1}(v^{l}_{1}) &=\mathbbm{E}[Y_{1}-Y_{0}|D=1,S_{0}=1,S_{1}=1, (Y_{1}-Y_{0})\leq F_{\Delta Y|111}^{-1}(p_{OOO1}(v^{l}_{1}))]\label{LB1low}\\
    UB_{OOO1}(v^{l}_{1}) &=\mathbbm{E}[Y_{1}-Y_{0}|D=1,S_{0}=1,S_{1}=1, (Y_{1}-Y_{0}) > F_{\Delta Y|111}^{-1}(1-p_{OOO1}(v^{l}_{1}))]. \label{UB1low}
\end{align}
Similarly, for the bounds for $\mathbbm{E}[Y_{1}^{\ast}(0)-Y_{0}^{\ast}(0)|D=0,OOO]$ based on equations (\ref{LB0})-(\ref{UB0}), evaluated at the smallest admissible value for $p_{OOO0}(v_0)$ yields,
\begin{align}
    LB_{OOO0}(v^{l}_{0}) &=\mathbbm{E}[Y_{1}-Y_{0}|D=0,S_{0}=1,S_{1}=1, (Y_{1}-Y_{0})\leq F_{\Delta Y|011}^{-1}(p_{OOO0}(v^{l}_{0}))]\label{LB0low}\\
    UB_{OOO0}(v^{l}_{0}) &=\mathbbm{E}[Y_{1}-Y_{0}|D=0,S_{0}=1,S_{1}=1,(Y_{1}-Y_{0}) > F_{\Delta Y|011}^{-1}(1-p_{OOO0}(v^{l}_{0}))]. \label{UB0low}
\end{align}
Note that when $p_{OOOd}(v_d^l)=0$, the trimming regions become empty. This occurs when $\mathbbm{P}[S_{1}=1\mid D=0,S_{0}=1]+\mathbbm{P}[S_{1}=1\mid D=1,S_{0}=1]\le 1$ which implies that $F^{-1}_{\Delta Y| d11}(p_{OOOd}(v_d^l)) = -\infty$ and $F^{-1}_{\Delta Y| d11}(1 - p_{OOOd}(v_d^l)) = +\infty$ . Consequently, the trimmed expectations based on the empty sets $\{\Delta Y\le -\infty\}$ in \eqref{LB1low} and \eqref{LB0low} and $\{\Delta Y> +\infty\}$ in \eqref{UB1low} and \eqref{UB0low}, will be undefined. To avoid this issue, we follow the literature and assume that the true mixing proportion is strictly positive \citep{semenova2025generalized,shin2024difference,lee2009}.\footnote{If the estimated proportions are zero even though $p_{OOOd}>0$, this indicates that one cannot rule out the presence of OOO units in either the treated or control groups in the sample. In such cases, one would need additional information about the relative shares of OOO individuals within each group to achieve identification. A natural next step is to impose the monotonicity restriction in Assumption \ref{monotone}. In this case, the proportions are point-identified, and the corresponding lower and upper bounds for $\tau_{OOO}$ are well-defined.} 
Combining the results above, we now propose partial identification of $\tau_{OOO}$.
\begin{theorem}[Bounds for $\tau_{OOO}$]\label{nomono_bound} Under Assumptions \ref{no anti},\ref{PT_OOO}, \ref{Partialselection}(a), \ref{Partialselection}(b), \ref{rs}, and $p_{OOOd}>0$ bounds on the treatment effect on the treated for the always-observed group ($\tau_{OOO}$) lies in the interval $[LB_{\tau_{OOO}}, UB_{\tau_{OOO}}]$,
\begin{equation*}
    \begin{split}
        LB_{\tau_{OOO}} &=\mathbbm{E}[Y_{1}-Y_{0}|D=1,S_{0}=1,S_{1}=1, (Y_{1}-Y_{0})\leq F_{\Delta Y|111}^{-1}(p_{OOO1}(v^{l}_{1}))]\\
        &-\mathbbm{E}[Y_{1}-Y_{0}|D=0,S_{0}=1,S_{1}=1,(Y_{1}-Y_{0}) > F_{\Delta Y|011}^{-1}(1-p_{OOO0}(v^{l}_{0}))],\\
        UB_{\tau_{OOO}} &= \mathbbm{E}[Y_{1}-Y_{0}|D=1,S_{0}=1,S_{1}=1, (Y_{1}-Y_{0}) > F_{\Delta Y|111}^{-1}(1-p_{OOO1}(v^{l}_{1}))]\\
        &-\mathbbm{E}[Y_{1}-Y_{0}|D=0,S_{0}=1,S_{1}=1, (Y_{1}-Y_{0})\leq F_{\Delta Y|011}^{-1}(p_{OOO0}(v^{l}_{0}))]
    \end{split}
\end{equation*}
where, \begin{align*}
    p_{OOO1}(v^{l}_{1})=\frac{\max\{\mathbbm{P}[S_{1}=1|D=0, S_{0}=1]+\mathbbm{P}[S_{1}=1|D=1, S_{0}=1]-1,0\}}{\mathbbm{P}[S_{1}=1| D=1, S_{0}=1]},\\
    p_{OOO0}(v^{l}_{0})=\frac{\max\{\mathbbm{P}[S_{1}=1|D=0, S_{0}=1]+\mathbbm{P}[S_{1}=1|D=1, S_{0}=1]-1,0\}}{\mathbbm{P}[S_{1}=1| D=0, S_{0}=1]}.
\end{align*}
The bounds in Theorem \ref{nomono_bound} are sharp.
\end{theorem}
Proof of Theorem \ref{nomono_bound} can be found in Appendix \ref{proof:nomono_bound}. The partial identification results in Theorem \ref{nomono_bound} allow somewhat flexible patterns of potential selection into the sample. All latent group types are possible, and treatment is allowed to induce individuals to join or leave the sample in the post-treatment period since monotonicity in selection is not assumed. Nevertheless, to achieve identification, we imposed substantial restrictions on the relationship between the selection mechanism and treatment assignment through assumptions \ref{Partialselection}(a) and  \ref{Partialselection}(b). 

\subsection{Identification with Monotonicity}\label{section:mono}
In specific applications, monotonicity in sample selection may be a plausible assumption. In the previous section, we saw that with just \ref{Partialselection}(a) or \ref{Partialselection}(b), we can partially identify $p_{OOO0}$ and $p_{OOO1}$, respectively. It is worth investigating how much leverage monotonicity alone has in terms of bounding the target parameter, $\tau_{OOO}$. 
\begin{lemma}\label{lemma:mono_weights} Under Assumptions \ref{no anti} and \ref{monotone} (positive monotonicity), we obtain $p_{OOO0}=1$ and $p_{OOO1} = \frac{\mathbbm{P}[S_1(0)=1|D=1, S_0=1]}{\mathbbm{P}[S_1=1|D=1, S_0=1]}$.
\end{lemma}
Proof can be found in Appendix \ref{proof:mono_weights}.

Positive monotonicity rules out the NON and OON strata. Hence, all untreated individuals observed in both periods are from the ``always-observed'' latent group and $p_{OOO0}=1$. Therefore, $\mathbbm{E}[Y_{1}^{\ast}(0)-Y_{0}^{\ast}(0)|D=0,OOO]$ is point identified by $\mathbbm{E}[Y_{1}-Y_{0}|D=0,S_{0}=1,S_{1}=1]$.

On the other treatment arm, individuals observed in both periods are still a mixture of OOO and ONO types. However, monotonicity guarantees that, $\mathbbm{P}[S_1(0)=1, S_1(1)=1|D=1, S_0=1] = \mathbbm{P}[S_1(0)=1|D=1, S_0=1]$ which means that we can focus on the values $p_{OOO1}$ can take over all possible $\mathbbm{P}[S_1(0)=1|D=1, S_0=1]$. Under positive monotonicity, the probability of selection in the treated counterfactual is always higher than the probability of selection in the untreated counterfactual, and 
$$0<\mathbbm{P}[S_1(0)=1|D=1, S_0=1]\leq \mathbbm{P}[S_1(1)=1|D=1, S_0=1]=\mathbbm{P}[S_1=1|D=1, S_0=1].$$
Even though monotonicity significantly constraints the possible values that $\mathbbm{P}[S_1(0)=1|D=1, S_0=1]$ can take, this information does not help us in learning about the proportion $p_{OOO1}=\frac{\mathbbm{P}[S_1(0)=1|D=1, S_0=1]}{\mathbbm{P}[S_1=1|D=1, S_0=1]}$, as it can still take any value in the unit interval. 

To be able to partially identify $p_{OOO1}$ and $\tau_{OOO}$ we need to complement monotonicity with restrictions on $\mathbbm{P}[S_1(0)=1|D=1, S_0=1]$ that shrink its possible range to the interior of $[0, \mathbbm{P}[S_1=1|D=1, S_0=1]]$. A natural choice is to consider Assumption \ref{Partialselection}(a), which point identifies $p_{OOO1}$ by assuming $\mathbbm{P}[S_1(0)=1|D=1, S_0=1]=\mathbbm{P}[S_1(0)=1|D=0, S_0=1]$, as we show in Section \ref{section:MonoPTSa}.

Alternatively, one can use a weaker version of this ignorability assumption, say, $\mathbbm{P}[S_{1}(0)=1|D=0, S_{0}=1]\leq \mathbbm{P}[S_{1}(0)=1|D=1, S_{0}=1]$. Intuitively, this condition requires that the probability of selection into the sample in the absence of treatment be at least as strong for the treated group as observed in the untreated group, allowing for ``stronger trends'' among the treated. This puts a floor on the lowest value possible for $p_{OOO1}\in\left[\frac{\mathbbm{P}[S_1=1|D=0, S_0=1]}{\mathbbm{P}[S_1=1|D=1, S_0=1]},1\right]$, which can then be used to construct identified sets for $\tau_{OOO}$ in a similar way to that described in Theorem \ref{nomono_bound}. However, the least favorable bounds in this case do not improve over those derived using Assumption \ref{Partialselection}(a).

\subsection{Identification under Monotonicity and Assumption \ref{Partialselection}(a)}\label{section:MonoPTSa}
As discussed in Section \ref{section:mono}, positive monotonicity rules out latent groups NON and OON, and $p_{OOO0}=1$ and $p_{OOO1} = \frac{\mathbbm{P}[S_1(0)=1|D=1, S_0=1]}{\mathbbm{P}[S_1=1|D=1, S_0=1]}$ (Lemma \ref{lemma:mono_weights}). Since $\mathbbm{E}[Y_{1}^{\ast}(1)-Y_{0}^{\ast}(0)|D=0,OOO]$ is point identified in that case, there is no need for assumption \ref{Partialselection}(b).\footnote{In the case of negative monotonicity, latent groups NNO and ONO are ruled out. This results in $p_{OOO1}=1$, point identification for $\mathbbm{E}[Y_{1}^{\ast}(1)-Y_{0}^{\ast}(0)|D=1,OOO]$. The identification results under this scenario are discussed in Appendix \ref{neg_mono}.} 
 
As suggested in the previous section, we can obtain point identification of $p_{OOO1}$ by combining positive monotonicity in selection and Assumption \ref{Partialselection}(a). Then, $p_{OOO0}=1$ and $p_{OOO1}=\frac{\mathbbm{P}[S_{1}=1|S_{0}=1,D=0]}{\mathbbm{P}[S_{1}=1|S_{0}=1,D=1]}$. 

With point identified $p_{OOO0}$ and $p_{OOO1}$ we propose the partial identification of $\tau_{OOO}$.
\begin{theorem}[Bounds for $\tau_{OOO}$ under positive monotonicity]\label{mono_bound} Under  Assumptions \ref{no anti}, \ref{PT_OOO}, \ref{monotone}, \ref{Partialselection}(a), \ref{rs}, and $p_{OOO1}>0$, bounds on the treatment effect on the treated for the always-observed group ($\tau_{OOO}$) lies in the interval $[LB_{\tau_{OOO}}, UB_{\tau_{OOO}}]$ where,
\begin{equation*}
    \begin{split}
        LB_{\tau_{OOO}} &= LB_{OOO1}-\mathbbm{E}[Y_{1}-Y_{0}|D=0,S_{0}=1,S_{1}=1],\\
        UB_{\tau_{OOO}} &=UB_{OOO1}-\mathbbm{E}[Y_{1}-Y_{0}|D=0,S_{0}=1,S_{1}=1]
    \end{split}
\end{equation*}
and, 
\begin{align*}
    LB_{OOO1} &=\mathbbm{E}[Y_{1}-Y_{0}|D=1,S_{0}=1,S_{1}=1, (Y_{1}-Y_{0})\leq F_{\Delta Y|111}^{-1}(p_{OOO1})]\\
    UB_{OOO1} &=\mathbbm{E}[Y_{1}-Y_{0}|D=1,S_{0}=1,S_{1}=1, (Y_{1}-Y_{0}) > F_{\Delta Y|111}^{-1}(1-p_{OOO1})]
\end{align*}
with $p_{OOO1}=\frac{\mathbbm{P}[S_{1}=1|S_{0}=1,D=0]}{\mathbbm{P}[S_{1}=1|S_{0}=1,D=1]}$. The bounds derived in Theorem \ref{mono_bound} are sharp \citep{lee2009}. 
\end{theorem}
The proof of Theorem \ref{mono_bound} is given in Appendix \ref{proof Theorem mono_bound}. The identified set for $\tau_{OOO}$ under the assumptions of Theorem \ref{mono_bound} is more informative since, by construction, point identification of $\mathbbm{E}[Y_{1}^{\ast}(1)-Y_{0}^{\ast}(1)|D=0,OOO]$ tightens the overall bounds for $\tau_{OOO}$. Similarly, the proportion of the always-observed among the treated, $p_{OOO1}=\frac{\mathbbm{P}[S_{1}=1|S_{0}=1,D=0]}{\mathbbm{P}[S_{1}=1|S_{0}=1,D=1]}$, is the upper bound for $p_{OOO1}$ obtained under the conditions for Lemma \ref{lemma:pipartialid}. 
Since higher shares of always-observed individuals imply more informative identified sets about that group, monotonicity leads to tighter bounds for $\tau_{OOO}$ as well. 

\begin{remark} Identification of ATT for always-observed latent group in a scenario where only repeated cross-section data are available is presented in Appendix \ref{sec:extension_RC}. A discussion of how the current identification argument for the always-observed group can be adapted to the setting of staggered treatment adoption when multiple periods are available is presented in Appendix \ref{sec:extension_multi}. 
\end{remark}

\begin{remark} Intuitively, the na\"{\i}ve DiD is always within the identified set for $\tau_{OOO}$ based on the trimming approach discussed in theorems \ref{nomono_bound}-\ref{mono_bound}. This is because $\tau_{\textup{DiDs}}$ is obtained by comparing expectations of functions of the outcomes, while the bounds for $\tau_{OOO}$ are based on trimmed expectations of the same functions of those variables. Since the overall mean will always be between the trimmed means, $\tau_{\textup{DiDs}}$ will always be in the identified set for $\tau_{OOO}$. This is not unique to the DiD case and applies more generally to Lee-type trimming bounds. A detailed discussion is presented in Appendix \ref{Appendix: Trimmed}.
\end{remark}

\subsection{Generalization with Covariates}\label{identification_x}
Suppose that the researcher observes a vector of pre-treatment covariates $X$. We extend the partial identification results for the OOO group to explicitly incorporate covariates in the analysis. Specifically, we impose conditional versions of the parallel trends assumption, which is standard in the conditional DiD literature \citep{abadie2005semiparametric,sant2020doubly,caetano2024difference}, along with a conditional ignorability assumption which assumes independence between treatment and selection conditional on pre-treatment characteristics and observability. These assumptions allow identification of the ATT within each covariate subpopulation, denoted as $\tau_{OOO}(X)$. 

We establish results under two cases: (i) no monotonicity and (ii) a relaxed form of monotonicity following the framework in \cite{semenova2025generalized}. In the latter case, we allow for covariate-specific monotonicity in selection patterns. This is accomplished by partitioning the covariate space into regions of positive monotonicity, negative monotonicity, and no monotonicity. 

For each case of identification, the lower and upper bounds for the conditional ATT are derived using the same trimming logic as applied earlier to each subpopulation of $X$. Bounds for the unconditional ATT, $\tau_{OOO}$, are then obtained by aggregating conditional bounds over the distribution of covariates for the always-observed units in the treated group. The lower bound is then given by:
\begin{align}\label{nomonox_lb_true}
    &LB_{\tau_{OOO}} = \int_{X} \left(LB_{OOO1}(X) - UB_{OOO0}(X)\right) \cdot dF(X|D=1, S_0=1, S_1(0)=1, S_1(1)=1) \nonumber \\
    & = \int_{X} \left(LB_{OOO1}(X) - UB_{OOO0}(X)\right) \cdot \frac{f(D=1, S_0=1, S_1(0)=1, S_1(1)=1|X)\cdot f({X})}{f(D=1, S_0=1, S_1(0)=1, S_1(1)=1)}dX \nonumber \\
    & = \mathbbm{E}\bigg[\left(LB_{OOO1}(X) - UB_{OOO0}(X)\right) \frac{\pi_{OOO1}(X)}{\pi_{OOO1}}\bigg] 
\end{align}
Similar results follow for the upper bound, where,
\begin{align}\label{nomonox_ub_true}
UB_{\tau_{OOO}} &= \mathbbm{E}\bigg[\left(UB_{OOO1}(X) - LB_{OOO0}(X)\right) \frac{\pi_{OOO1}(X)}{\pi_{OOO1}}\bigg].
\end{align}
Detailed exposition of the partial identification arguments, along with formal statements of the results for both cases, is provided in Appendix \ref{sec: covariates}. In addition, we also provide moment-based representations of the bounds for the unconditional ATT for OOO in Appendix Section \ref{identification_mom}. Estimation of the conditional bounds proceeds by discretizing the covariate space and then estimating the cell-specific bounds in the case of no monotonicity and classifying individuals into positive, negative, or no monotonicity regions before estimating the bounds in each region, for the case of relaxed monotonicity. This is discussed in Appendix Section \ref{estimation_x}. Results for the two empirical applications incorporating covariates are presented in Appendix Section \ref{application_x}. A joint test of monotonicity is also provided in Section \ref{monotonicity_test}.


\section{Identification of ATT for Other Latent Groups}\label{identification_other}
So far, the discussion has focused on identifying $\tau_{OOO}$, the ATT for the always-observed group, which often accounts for a large proportion of the population in many applications. However, in specific applications, policymakers may also be interested in identifying the treatment effect for other latent groups. For example, in evaluating the effects of a training program on earnings, policymakers are interested in the impacts on those unemployed before treatment (e.g., NOO and NNO latent groups). In other cases, the ONO latent group might be of interest. For instance, when considering the impact of working from home (WFH) on employee performance, the company's management may be interested in the effect on the productivity of employees who leave the company if WFH is not provided, but would stay if WFH is provided (i.e. ONO latent group).

This section studies the identification of $\tau_{g}$, the ATT for latent group $g$, for $g\in \{ONO, NOO ,$ $ NNO\}$.\footnote{As discussed in section \ref{section:mono}, positive MS rules out NON and OON latent groups. Furthermore, there is no information on ONN and NNN groups in either treatment arm in the post-treatment period. We therefore focus our attention on partially identifying the remaining groups.} Since less information is available for these groups relative to the OOO group, we introduce additional cross-group mean dominance assumptions to obtain informative bounds. These assumptions are admissible for many empirical situations. We also restrict the support of the potential outcomes to be bounded such that $Y_t^\ast(0),Y_t^\ast(1)\in y=[Y_t^{LB}, Y_t^{UB}]$, where $-\infty<Y_t^{LB}<Y_t^{UB}<\infty$ for $t=0,1$.

To consider $\tau_{ONO}$, $\tau_{NOO}$, and $\tau_{NNO}$, we extend the within-group potential outcomes parallel trends provided in Assumption \ref{PT_OOO} to include these groups.

\begin{assumptionp}{\ref{PT_OOO}$(g)$}[Parallel trends for latent group $g\in \{ONO, NOO, $ $ NNO\}$]\label{PT_group} \ 
    \begin{equation*}
        \mathbbm{E}[Y_{1}^{\ast}(0)-Y_{0}^{\ast}(0)|D=1, G=g]=\mathbbm{E} [Y_{1}^{\ast}(0)-Y_{0}^{\ast}(0)|D=0, G=g]. 
    \end{equation*}
\end{assumptionp}
Next, we introduce cross-group mean dominance assumptions to aid the identification of the ATT for these latent groups. 
\begin{assumption}[Outcome mean dominance]\label{mean dominance} \ 
\begin{enumerate}
    \item[(a)] For $\tau_{ONO}$: \ \ $\mathbbm{E}[Y_{1}^{\ast}(0)|D=0, ONO] \leq \mathbbm{E}[Y_{1}^{\ast}(0)|D=0, OOO]$
    \item[(b)] For $\tau_{NNO}$:\begin{align*}
      &\text{(i)} \quad \mathbbm{E}[Y_{0}^{\ast}(0)|D=d, NNO] \leq \mathbbm{E}[Y_{0}^{\ast}(0)|D=d, ONO], d\in \{0,1\}\\
      &\text{(ii)} \quad \mathbbm{E}[Y_{1}^{\ast}(0)|D=0, NNO] \leq \mathbbm{E}[Y_{1}^{\ast}(0)|D=0, NOO]
    \end{align*} 
    \item[(c)] For $\tau_{NOO}$: \ \ $\mathbbm{E}[Y_{0}^{\ast}(0)|D=d, NOO] \leq \mathbbm{E}[Y_{0}^{\ast}(0)|D=d, OOO], d \in \{0,1\}$
\end{enumerate}
\end{assumption}
Our mean dominance assumptions compare the untreated potential outcomes of a latent group in a specific treatment arm $D=d$ at a given point in time to its closest latent counterpart. Specifically, the inequality restrictions posit that selection into being observed is (weakly) positively correlated with the untreated potential outcomes for group $D=d$ at time $t$. For instance, invoking Assumption 6(a) helps obtain a tighter upper bound on the counterfactual expectation $\mathbbm{E}[Y_{1}^{\ast}(0)\mid D=0, \text{ONO}]$, thereby narrowing the overall bounds for $\tau_{ONO}$. This is achieved by comparing the ONO group to the OOO group, whose untreated potential outcome mean in the post-treatment period can be point identified using $\mathbbm{E}[Y_1 \mid D=0, S_0=1, S_1=1]$, under positive MS. Similar arguments apply to the other latent groups, where the mean dominance assumptions help to refine the theoretical bounds on counterfactual means. Because each latent group's selection behavior varies across treatment states and time periods, the mean dominance assumptions are stratum-specific.

In the context of the job training example, all these assumptions imply that individuals with higher attachment to the labor force or those less prone to be unemployed in some period/treatment scenario have better wages on average than peers with lower attachment in similar situations (time period, treatment counterfactual, etc.). As the always-observed group will be employed irrespective of training, assuming their potential wages to be higher than those of the other groups is reasonable. The justifiability of these assumptions depends on the empirical context, and researchers need to consider them carefully.

To identify bounds for $\tau_{ONO}$, $\tau_{NOO}$, and $\tau_{NNO}$, we introduce a stronger version of IS (Assumption \ref{Partialselection}). It imposes independence on the joint counterfactual selection distribution rather than only relating to the marginal distributions.
\begin{assumptionp}{\ref{Partialselection}$(Joint)$}[Joint independence between selection and treatment assignment]\label{inde_conditional} \ 
	\begin{equation*}
	    (S_{1}(0),S_{1}(1)) \perp D\left|S_{0}\right.
	\end{equation*}
\end{assumptionp}
Assumption \ref{inde_conditional} states that conditional on the initial period selection status, the joint counterfactual selection mechanism is independent of treatment assignment. This is a stronger assumption than its marginal version in Assumption \ref{Partialselection}, and is a sufficient condition for the latter. Under this assumption, the observed selection probabilities conditional on initial period selection and treatment enable us to identify all latent group proportions.\footnote{See Lemma \ref{lemma: pro 4(joint)} and its proof in Appendix \ref{prop_condi}.} 

To derive the ATT bounds for these latent groups, decompose $\tau_{g}$ as follows,
\begin{align}\label{dec_group}
\tau_{g} &= \mathbbm{E}[Y_{1}^{\ast}(1)-Y_{1}^{\ast}(0)|D=1, G=g] \nonumber \\
&= \mathbbm{E}[Y_{1}^{\ast}(1)-Y_{0}^{\ast}(1)+Y_{0}^{\ast}(1)-Y_{1}^{\ast}(0)|D=1, G=g]\nonumber \\
&= \mathbbm{E}[Y_{1}^{\ast}(1)-Y_{0}^{\ast}(1)|D=1, G=g]-\mathbbm{E}[Y_{1}^{\ast}(0)-Y_{0}^{\ast}(1)|D=1, G=g] \nonumber \\
&= \mathbbm{E}[Y_{1}^{\ast}(1)-Y_{0}^{\ast}(1)|D=1, G=g]-\mathbbm{E}[Y_{1}^{\ast}(0)-Y_{0}^{\ast}(0)|D=1, G=g] \text{  (Assumption \ref{no anti})}\nonumber \\
&= \mathbbm{E}[Y_{1}^{\ast}(1)-Y_{0}^{\ast}(1)|D=1, G=g]-\mathbbm{E}[Y_{1}^{\ast}(0)-Y_{0}^{\ast}(0)|D=0, G=g] \text{  (Assumption \ref{PT_group})}.
\end{align}

\subsection{Identification of ATT for ONO Group}
The treatment effect for the ONO group ($\tau_{ONO}$) can be further decomposed using Equation (\ref{dec_group}) as,
\begin{align*}
\tau_{ONO} &=\mathbbm{E}[Y_{1}^{\ast}(1)-Y_{0}^{\ast}(1)|D=1, ONO]-\mathbbm{E}[Y_{1}^{\ast}(0)-Y_{0}^{\ast}(0)|D=0, ONO]\\
&= \mathbbm{E}[Y_{1}^{\ast}(1)-Y_{0}^{\ast}(1)|D=1, ONO]-\mathbbm{E}[Y_{1}^{\ast}(0)|D=0, ONO]+ \mathbbm{E}[Y_{0}^{\ast}(0)|D=0, ONO]
\end{align*}
As explained in Section \ref{identification}, we can use the group of treated individuals for whom the outcome is observed in both periods to partially identify $\mathbbm{E}[Y_{1}^{\ast}(1)-Y_{0}^{\ast}(1)|D=1, ONO]$. Similarly, we can use the group of untreated individuals for whom the outcome is observed in the first period only ($D = 0, S_0 = 1, S_1 = 0$) to partially identify $\mathbbm{E}[Y_{0}^{\ast}(0)|D=0, ONO]$. Identification of $\mathbbm{E}[Y_{1}^{\ast}(0)|D=0, ONO]$ combines the theoretical upper and lower bound of the outcome distribution \citep{huber2015sharp} and the mean dominance Assumption \ref{mean dominance}(a).
\begin{theorem}[Bounds for $\tau_{ONO}$ under positive monotonicity]\label{ONO_bound} Under Assumptions \ref{no anti}, \ref{PT_group}, \ref{monotone}, \ref{inde_conditional}, \ref{rs}, \ref{mean dominance}(a), and $p_{gd}>0$ bounds on the treatment effect on the treated for the ONO group ($\tau_{ONO}$) lies in the interval $[LB_{\tau_{ONO}}, UB_{\tau_{ONO}}]$,
\begin{equation*}
    \begin{split}
        LB_{\tau_{ONO}} &= LB_{ONO1}-\mathbbm{E}[Y_{1}|D=0,S_{0}=1, S_{1}=1]+LB^{0}_{ONO0},\\
        UB_{\tau_{ONO}} &= UB_{ONO1}-Y_{1}^{LB}+UB^{0}_{ONO0}
    \end{split}
\end{equation*}
where $Y_{1}^{LB}$ is the theoretical lower bound of potential outcomes in the post-treatment period, 
\begin{align*}
    LB_{ONO1} &=\mathbbm{E}[Y_{1}-Y_{0}|D=1,S_{0}=1,S_{1}=1, (Y_{1}-Y_{0})\leq F_{\Delta Y|111}^{-1}(1-p_{OOO1})],\\
    UB_{ONO1} &=\mathbbm{E}[Y_{1}-Y_{0}|D=1,S_{0}=1,S_{1}=1, (Y_{1}-Y_{0}) > F_{\Delta Y|111}^{-1}(p_{OOO1})], \\
    p_{OOO1}&=\frac{\mathbbm{P}[S_{1}=1|S_{0}=1,D=0]}{\mathbbm{P}[S_{1}=1|S_{0}=1,D=1]},
\end{align*}
and,
\begin{align*}
    LB^{0}_{ONO0} &=\mathbbm{E}[Y_{0}|D=0,S_{0}=1,S_{1}=0, Y_{0}\leq F_{Y_0|010}^{-1}(p_{ONO0})]\\
    UB^{0}_{ONO0} &=\mathbbm{E}[Y_{0}|D=0,S_{0}=1,S_{1}=0, Y_{0} > F_{Y_0|010}^{-1}(1-p_{ONO0})] \\
    p_{ONO0}&=1-\frac{\mathbbm{P}[S_1=0|S_0=1,D=1]}{\mathbbm{P}[S_1=0|S_0=1,D=0]}.
\end{align*}
\end{theorem}

Proof of Theorem \ref{ONO_bound} is given in the Appendix \ref{proof:ONO_bound}. 
\subsection{Identification of ATT for NNO Group}
The ATT for NNO group ($\tau_{NNO}$) also can be further decomposed using equation (\ref{dec_group}) as,
\begin{align*}
\tau_{NNO} &= \mathbbm{E}[Y_{1}^{\ast}(1)-Y_{0}^{\ast}(1)|D=1, NNO]-\mathbbm{E}[Y_{1}^{\ast}(0)-Y_{0}^{\ast}(0)|D=0, NNO]\\
&= \mathbbm{E}[Y_{1}^{\ast}(1)|D=1, NNO]-\mathbbm{E}[Y_{0}^{\ast}(1)|D=1, NNO]\\
&-\mathbbm{E}[Y_{1}^{\ast}(0)|D=0, NNO]+\mathbbm{E}[Y_{0}^{\ast}(0)|D=0, NNO].
\end{align*}
We can use the group of treated individuals for whom the outcome is not observed in the pre-treatment period but observed in the post-treatment period ($D = 1, S_0 = 0, S_1 = 1$) to partially identify $\mathbbm{E}[Y_{1}^{\ast}(1)|D=1, NNO]$. The other terms, $\mathbbm{E}[Y_{0}^{\ast}(1)|D=1, NNO]$, $\mathbbm{E}[Y_{0}^{\ast}(0)|D=0, NNO]$ and $\mathbbm{E}[Y_{1}^{\ast}(0)|D=0, NNO]$ can be partially identified by imposing the theoretical upper and lower bounds of the respective outcome distributions \citep{huber2015sharp} and tighten these by imposing outcome mean dominance assumptions \ref{mean dominance}b.(i) and \ref{mean dominance}b.(ii), respectively.
\begin{theorem}[Bounds for $\tau_{NNO}$ under positive monotonicity]\label{NNO_bound} Under Assumptions \ref{no anti}, \ref{PT_group}, \ref{monotone}, \ref{inde_conditional}, \ref{rs}, \ref{mean dominance}(b), and $p_{gd}>0$, bounds on the treatment effect on the treated for the NNO group ($\tau_{NNO}$) lies in the interval $[LB_{\tau_{NNO}}, UB_{\tau_{NNO}}]$,
\begin{equation*}
    \begin{split}
        LB_{\tau_{NNO}} &= LB_{NNO1}-LB^{0}_{ONO1}-\mathbbm{E}[Y_{1}|D=0,S_{0}=0, S_{1}=1]+Y_{0}^{LB},\\
         UB_{\tau_{NNO}} &= UB_{NNO1}-Y_{0}^{LB}- Y_{1}^{LB}+LB^{0}_{ONO0}
    \end{split}
\end{equation*}
where $Y_{0}^{LB}$ and $Y_{1}^{LB}$ are the theoretical lower bounds of the potential outcomes in the pre-treatment period and the post-treatment period, respectively. Furthermore,
\begin{align*}
     LB_{NNO1} &=\mathbbm{E}[Y_{1}|D=1,S_{0}=0,S_{1}=1, Y_{1}\leq F_{Y_{1}|101}^{-1}(p_{NNO1})]\\
    UB_{NNO1} &=\mathbbm{E}[Y_{1}|D=1,S_{0}=0,S_{1}=1, Y_{1} > F_{Y_{1}|101}^{-1}(1-p_{NNO1})]\\
    p_{NNO1}&=1-\frac{\mathbbm{P}[S_1=1|S_0=0,D=0]}{\mathbbm{P}[S_1=1|S_0=0,D=1]}
\end{align*}
and,
\begin{equation*}
    LB^{0}_{ONO1} =\mathbbm{E}[Y_{0}|D=1,S_{0}=1,S_{1}=1, Y_{0}\leq F_{Y_{0}|111}^{-1}(1-p_{OOO1})]
\end{equation*}
with $p_{OOO1}=\frac{\mathbbm{P}[S_{1}=1|S_{0}=1,D=0]}{\mathbbm{P}[S_{1}=1|S_{0}=1,D=1]}$. Finally,
\begin{align*}
    LB^{0}_{ONO0} &=\mathbbm{E}[Y_{0}|D=0,S_{0}=1,S_{1}=0, Y_{0}\leq F_{Y_{0}|010}^{-1}(p_{ONO0})]
\end{align*}
with $p_{ONO0}=1-\frac{\mathbbm{P}[S_1=0|S_0=1,D=1]}{\mathbbm{P}[S_1=0|S_0=1,D=0]}$.
\end{theorem}
Proof of Theorem \ref{NNO_bound} is given in the Appendix \ref{proof:NNO_bound}. 

\subsection{Identification of ATT for NOO Group}

The ATT for NOO group ($\tau_{NOO}$) can be decomposed using equation (\ref{dec_group}) as follows,
\begin{align*}
\tau_{NOO} &= \mathbbm{E}[Y_{1}^{\ast}(1)-Y_{0}^{\ast}(1)|D=1, NOO]-\mathbbm{E}[Y_{1}^{\ast}(0)-Y_{0}^{\ast}(0)|D=0, NOO]\\
&= \mathbbm{E}[Y_{1}^{\ast}(1)|D=1, NOO]-\mathbbm{E}[Y_{0}^{\ast}(1)|D=1, NOO]-\mathbbm{E}[Y_{1}^{\ast}(0)|D=0, NOO]\\
&+\mathbbm{E}[Y_{0}^{\ast}(0)|D=0, NOO]. 
\end{align*}
We can use the group of treated individuals for whom the outcome is not observed in the pre-treatment period but observed in the post-treatment period ($D = 1, S_0 = 0, S_1 = 1$) to partially identify $\mathbbm{E}[Y_{1}^{\ast}(1)|D=1, NOO]$. The term, $\mathbbm{E}[Y_{1}^{\ast}(0)|D=0, NOO]$, can be point identified using $\mathbbm{E}[Y_{1}|D=0,S_{0}=0,S_{1}=1]$ which considers the untreated individuals not observed in the pre-treatment period but observed in the post-treatment period ($D=0, S_0=0, S_1=1$). Under positive monotonicity, this observed group is composed exclusively of the NOO latent subgroup. The remaining terms, $\mathbbm{E}[Y_{0}^{\ast}(1)|D=1, NOO]$ and $\mathbbm{E}[Y_{0}^{\ast}(0)|D=0, NOO]$,  can be partially identified by imposing the theoretical upper and lower bounds of the respective outcome distributions \citep{huber2015sharp} where we tighten them by imposing outcome mean dominance Assumption \ref{mean dominance}(c).
\begin{theorem}[Bounds for $\tau_{NOO}$ under positive monotonicity]\label{NOO_bound} Under the Assumptions \ref{no anti}, \ref{PT_group}, \ref{monotone}, \ref{inde_conditional}, \ref{rs}, \ref{mean dominance}(c) and $p_{gd}>0$, bounds on the treatment effect on the treated for the NOO group ($\tau_{NOO}$) lies in the interval $[LB_{\tau_{NOO}}, UB_{\tau_{NOO}}]$,
\begin{equation*}
    \begin{split}
        LB_{\tau_{NOO}} &= LB_{NOO1}-LB^{0}_{OOO1}-\mathbbm{E}[Y_{1}|D=0,S_{0}=0, S_{1}=1]+Y_{0}^{LB},\\
        UB_{\tau_{NOO}} &= UB_{NOO1}-Y_{0}^{LB}-\mathbbm{E}[Y_{1}|D=0,S_{0}=0, S_{1}=1]+\mathbbm{E}[Y_{0}|D=0,S_{0}=1, S_{1}=1]],
    \end{split}
\end{equation*}        
where,
\begin{align*}
    LB_{NOO1} &=\mathbbm{E}[Y_{1}|D=1,S_{0}=0,S_{1}=1, Y_{1}\leq F_{Y_{1}|101}^{-1}(1-p_{NNO1})]\\
    UB_{NOO1} &=\mathbbm{E}[Y_{1}|D=1,S_{0}=0,S_{1}=1, Y_{1} > F_{Y_{1}|101}^{-1}(p_{NNO1})] 
\end{align*}
with $p_{NNO1}=1-\frac{\mathbbm{P}[S_1=1|S_0=0,D=0]}{\mathbbm{P}[S_1=1|S_0=0,D=1]}$ and,
\begin{align*}
    LB^{0}_{OOO1} &=\mathbbm{E}[Y_{0}|D=1,S_{0}=1,S_{1}=1, Y_{0}\leq F_{Y_{0}|111}^{-1}(p_{OOO1})]
\end{align*}
with $p_{OOO1}=\frac{\mathbbm{P}[S_{1}=1|S_{0}=1,D=0]}{\mathbbm{P}[S_{1}=1|S_{0}=1,D=1]}$. Finally, $Y_{0}^{LB}$ is the theoretical lower bound of potential outcomes in the pre-treatment period.
\end{theorem}
Proof of Theorem \ref{NOO_bound} is given in the Appendix \ref{proof:NOO_bound}. 

\subsection{Identification of Overall ATT}
After partially identifying the ATT for each of these latent groups, we provide a way to partially identify the overall ATT, $\tau$, which relies on combining the identified sets of the ATT of these different latent groups with their appropriate population weights. 

\begin{theorem}[Partial Identification of $\tau$]\label{theorem:partialid_tau}
    Under the Assumptions required for Theorems \ref{mono_bound}-\ref{NOO_bound}, the overall ATT in the population, $\tau = \mathbbm{E}[Y_1^\ast(1)-Y_1^\ast(0)\mid D=1]$, is bounded as follows: 
    \begin{equation}
        \tau \in \left[\tau_{LB}, \ \tau_{UB}\right] 
    \end{equation}
    where 
    {\footnotesize  \begin{align*}
        \tau_{LB} &= \sum_{g\in\{OOO,ONO\}} LB_{\tau_g} \cdot p_{g1}\cdot p_{11} +\sum_{g\in\{NOO,NNO\}} LB_{\tau_g} \cdot p_{g1}\cdot p_{01}+ LB_{\tau_{{ONN}}} \cdot p_{10}+LB_{\tau_{{NNN}}} \cdot p_{00}\\
		\tau_{UB} &= \sum_{g\in\{OOO,ONO\}} UB_{\tau_g} \cdot p_{g1}\cdot p_{11} +\sum_{g\in\{NOO,NNO\}} UB_{\tau_g} \cdot p_{g1}\cdot p_{01}+ UB_{\tau_{{ONN}}} \cdot p_{10}+UB_{\tau_{{NNN}}} \cdot p_{00}
    \end{align*}}
    and $p_{s_0s_1} \equiv \mathbbm{P}(S_0=s_0,S_1=s_1 \mid D=1)$ denotes the joint probability of observed selection in both periods for the treated group, $p_{g1}$ denotes the mixing proportions, $LB_{\tau_g}$ and $UB_{\tau_g}$ denote the group-specific bounds given in Theorems \ref{mono_bound}–\ref{NOO_bound}, along with theoretical support restrictions for the ONN and NNN groups. 
\end{theorem}
Proof can be found in Appendix \ref{proof:partialid_tau}. Theorem \ref{theorem:partialid_tau} provides a useful identification result about a parameter that is the typical target of a DiD analysis. 
The expressions for the lower and upper bounds give a detailed and transparent description of the sources of heterogeneity and information affecting the overall ATT. This can help researchers gauge how important each group is. For example, the weights assigned to each identified set and the width of the corresponding group-specific interval together indicate the relative influence of each group on the overall bounds.    

\section{Estimation and Inference} \label{estimation}
The estimation of the bounds defined in Theorem \ref{nomono_bound} and Theorem \ref{mono_bound} are based on the sample analogues of the population counterparts.  To estimate the bounds defined in Theorem \ref{nomono_bound} we first have to estimate the mixing proportions $p_{OOO1}(v^{l}_{1})$ and $p_{OOO0}(v^{l}_{0})$.
Formally, we have,
\begin{align*}
        \hat{p}_{OOO1}(v^{l}_{1})=\frac{\max\{\mathbbm{\hat{P}}[S_{1}=1|D=0, S_{0}=1]+\mathbbm{\hat{P}}[S_{1}=1|D=1, S_{0}=1]-1,0\}}{\mathbbm{\hat{P}}[S_{1}=1| D=1, S_{0}=1]},\\
        \hat{p}_{OOO0}(v^{l}_{0})=\frac{\max\{\mathbbm{\hat{P}}[S_{1}=1|D=0, S_{0}=1]+\mathbbm{\hat{P}}[S_{1}=1|D=1, S_{0}=1]-1,0\}}{\mathbbm{\hat{P}}[S_{1}=1| D=0, S_{0}=1]}
    \end{align*}
    where,
    \begin{align*}
        \mathbbm{\hat{P}}[S_{1}=1|D=0, S_{0}=1] &= \frac{\sum_{i=1}^{n}S_{i0}\cdot S_{i1}\cdot (1-D_{i})}{\sum_{i=1}^{n}S_{i0}\cdot (1-D_{i})}\\
        \mathbbm{\hat{P}}[S_{1}=1|D=1, S_{0}=1] &= \frac{\sum_{i=1}^{n}S_{i0}\cdot S_{i1}\cdot D_{i}}{\sum_{i=1}^{n}S_{i0}\cdot D_{i}}.
    \end{align*}
    With these estimated mixing proportions, the bounds for $\tau_{OOO}$ under Theorem \ref{nomono_bound} can be estimated as follows,
\begin{equation*}
    \begin{split}
    \widehat{LB}_{\tau_{OOO}} 
    &=\frac{\sum_{i=1}^n (Y_{i1}-Y_{i0}) \cdot S_{i0}\cdot S_{i1}\cdot D_{i}  \cdot I\left\{(Y_{i1}-Y_{i0}) \leqslant \hat{y}_{\hat{p}_{OOO1}(v^{l}_{1})}\right\}}{\sum_{i=1}^n S_{i0}\cdot S_{i1}\cdot D_{i}  \cdot I\left\{(Y_{i1}-Y_{i0}) \leqslant \hat{y}_{\hat{p}_{OOO1}(v^{l}_{1})}\right\}}\\
    &-\frac{\sum_{i=1}^n (Y_{i1}-Y_{i0}) \cdot S_{i0}\cdot S_{i1}\cdot (1-D_{i})  \cdot I\left\{(Y_{i1}-Y_{i0}) > \hat{y}_{1-\hat{p}_{OOO0}(v^{l}_{0})}\right\}}{\sum_{i=1}^n S_{i0}\cdot S_{i1}\cdot (1-D_{i})  \cdot I\left\{(Y_{i1}-Y_{i0}) > \hat{y}_{1-\hat{p}_{OOO0}(v^{l}_{0})}\right\}}\\
    \widehat{UB}_{\tau_{OOO}}
    &=\frac{\sum_{i=1}^n (Y_{i1}-Y_{i0}) \cdot S_{i0}\cdot S_{i1}\cdot D_{i}  \cdot I\left\{(Y_{i1}-Y_{i0}) > \hat{y}_{1-\hat{p}_{OOO1}(v^{l}_{1})}\right\}}{\sum_{i=1}^n S_{i0}\cdot S_{i1}\cdot D_{i}  \cdot I\left\{(Y_{i1}-Y_{i0}) > \hat{y}_{1-\hat{p}_{OOO1}(v^{l}_{1})}\right\}}\\
    &-\frac{\sum_{i=1}^n (Y_{i1}-Y_{i0}) \cdot S_{i0}\cdot S_{i1}\cdot (1-D_{i})  \cdot I\left\{(Y_{i1}-Y_{i0}) \leqslant \hat{y}_{\hat{p}_{OOO0}(v^{l}_{0})}\right\}}{\sum_{i=1}^n S_{i0}\cdot S_{i1}\cdot (1-D_{i})  \cdot I\left\{(Y_{i1}-Y_{i0}) \leqslant \hat{y}_{\hat{p}_{OOO0}(v^{l}_{0})}\right\}}
    \end{split}
\end{equation*}
where $\hat{y}_{\hat{p}_{OOO0}(v^{l}_{0})}$ and $\hat{y}_{1-\hat{p}_{OOO0}(v^{l}_{0})}$ are $\hat{p}_{OOO0}(v^{l}_{0})$-th and $(1-\hat{p}_{OOO0}(v^{l}_{0}))$-th quantile of the conditional distribution $Y_{1}-Y_{0}$ for the untreated individuals observed in both time periods. Similarly, $\hat{y}_{\hat{p}_{OOO1}(v^{l}_{1})}$ and $\hat{y}_{1-\hat{p}_{OOO1}(v^{l}_{1})}$ are $\hat{p}_{OOO1}(v^{l}_{1})$-th and $(1-\hat{p}_{OOO1}(v^{l}_{1}))$-th quantile of the conditional distribution $Y_{1}-Y_{0}$ for the treated individuals observed in both time periods. In general, the relevant $q$-th quantile  of the conditional distribution $Y_{1}-Y_{0}$ for the treated individuals observed in both time periods is calculated as,
\begin{equation*}
    \hat{y}_q  = \min \left\{y: \frac{\sum_{i=1}^n S_{i0}\cdot S_{i1}\cdot D_{i} \cdot I\left\{(Y_{i1}-Y_{i0})\leqslant y\right\}}{\sum_{i=1}^n S_{i0}\cdot S_{i1}\cdot D_{i}} \geqslant q\right\}, \quad \text{ where I($\cdot$) is an indicator function}.
\end{equation*}

The bounds for $\tau_{OOO}$ under Theorem \ref{mono_bound} can be estimated similarly. First, estimate the required mixing proportion $p_{OOO1}$ as follows,
\begin{align} \label{Estimation_p}
    \hat{p}_{OOO1} &=\frac{\mathbb{\hat{P}}[S_{1}=1|S_{0}=1,D=0]}{\mathbb{\hat{P}}[S_{1}=1|S_{0}=1,D=1]}, \nonumber \\
    \mathbb{\hat{P}}[S_{1}=1|D=0, S_{0}=1] &= \frac{\sum_{i=1}^{n}S_{i0}\cdot S_{i1}\cdot (1-D_{i})}{\sum_{i=1}^{n}S_{i0}\cdot (1-D_{i})}, \nonumber\\
    \mathbb{\hat{P}}[S_{1}=1|D=1, S_{0}=1] &= \frac{\sum_{i=1}^{n}S_{i0}\cdot S_{i1}\cdot D_{i}}{\sum_{i=1}^{n}S_{i0}\cdot D_{i}}.
\end{align}
Next, the estimated versions of $LB_{OOO1}$ and $UB_{OOO1}$ can be obtained as,
\begin{align*}
     \widehat{LB}_{OOO1} &= \frac{\sum_{i=1}^n (Y_{i1}-Y_{i0}) \cdot S_{i0}\cdot S_{i1}\cdot D_{i}  \cdot I\left\{(Y_{i1}-Y_{i0}) \leqslant \hat{y}_{\hat{p}_{OOO1}}\right\}}{\sum_{i=1}^n S_{i0}\cdot S_{i1}\cdot D_{i}  \cdot I\left\{(Y_{i1}-Y_{i0}) \leqslant \hat{y}_{\hat{p}_{OOO1}}\right\}}\\
    \widehat{UB}_{OOO1} &= \frac{\sum_{i=1}^n (Y_{i1}-Y_{i0}) \cdot S_{i0}\cdot S_{i1}\cdot D_{i}  \cdot I\left\{(Y_{i1}-Y_{i0}) > \hat{y}_{1-\hat{p}_{OOO1}}\right\}}{\sum_{i=1}^n S_{i0}\cdot S_{i1}\cdot D_{i}  \cdot I\left\{(Y_{i1}-Y_{i0}) > \hat{y}_{1-\hat{p}_{OOO1}}\right\}}
\end{align*}
where $\hat{y}_{\hat{p}_{OOO1}}$ and $\hat{y}_{1-\hat{p}_{OOO1}}$ are $\hat{p}_{OOO1}$-th and $(1-\hat{p}_{OOO1})$-th quantile of the conditional distribution $Y_{1}-Y_{0}$ for the treated individuals observed in both time periods. Next, $\mathbbm{E}[Y_{1}-Y_{0}|D=0,S_{0}=1,S_{1}=1]$ (denote as $E_{OOO0}$ for notational ease) will be estimated using its sample analogues
as,
\begin{align*}
     \hat{E}_{OOO0} &= \frac{\sum_{i=1}^n (Y_{i1}-Y_{i0}) \cdot S_{i0}\cdot S_{i1}\cdot (1-D_{i})}{\sum_{i=1}^n S_{i0}\cdot S_{i1}\cdot (1-D_{i})}.
\end{align*}
Finally, the bounds for $\tau_{OOO}$ defined in Theorem \ref{mono_bound} can be estimated as,
\begin{align*}
    	\widehat{LB}_{\tau_{OOO}} &= \widehat{LB}_{OOO1}-\hat{E}_{OOO0}\\
		\widehat{UB}_{\tau_{OOO}} &= \widehat{UB}_{OOO1}-\hat{E}_{OOO0}.
\end{align*}
The bounds for other latent groups defined in Theorems \ref{ONO_bound}, \ref{NNO_bound}, and \ref{NOO_bound} can be estimated similarly. The estimation steps are detailed in Appendix \ref{Estimation_other latent groups}. 

These sample analogue estimators of the bounding functions are functions of conditional probabilities, means, and trimmed means, and include non-smooth functions of auxiliary parameters for which we use plug-in estimates. Fortunately, when the true mixing proportions are strictly positive,\footnote{When the true mixing proportion is zero, the asymptotic behavior of the bounds can be characterized by letting the trimming proportion shrink to zero sufficiently slowly such that, asymptotically, there are enough observations in the trimming region to guarantee that the expectation is well defined. See \cite{andrews2013inference} and references therein for similar approaches.} $\sqrt{n}$-consistency and asymptotic normality of these estimators follow from results in \citet{chen2003estimation}.\footnote{We thank the associate editor for indicating that the proposed estimators are encompassed by the results in \cite{chen2003estimation}. In the supplementary materials Section \ref{Asymptotics}, we verify that the requisite conditions for applying \citet{chen2003estimation} hold, and we derive the resulting asymptotic distributions and their properties.} Hence, we can obtain two types of confidence intervals by applying standard inference procedures. Following \cite{lee2009} and \cite{huber2014treatment}, let $\widehat{LB}_{\tau_g}$ and $ \widehat{UB}_{\tau_g}$ be the estimated bounds for a specific latent group $g$ using the estimation method discussed above and $\hat{\sigma}_{LB_{\tau_g}}$ and $\hat{\sigma}_{{UB}_{\tau_g}}$ denote their respective standard deviations, obtained through bootstrap \citep{huber2014treatment,chen2015bounds}.\footnote{See \cite{bugni2010bootstrap}, \citet{bugni2015specification}, and \cite{andrews2024misspecified}, among others, for inferential methods for partially identified models that solve moment inequalities.} Then, we can compute the first confidence interval as,
\begin{equation*}
    \left[\widehat{LB}_{\tau_g}-1.96\cdot \frac{\hat{\sigma}_{LB_{\tau_g}}}{\sqrt{n}},\widehat{UB}_{\tau_g}+1.96\cdot \frac{\hat{\sigma}_{UB_{\tau_g}}}{\sqrt{n}}\right]
\end{equation*}
which will contain the true bounds with at least 95\% probability. The second option for confidence intervals is based on \citet{imbens2004confidence}. These confidence intervals are focused on covering the true treatment effect with 95\% probability, which are calculated as $[\widehat{LB}_{\tau_g}-C_{n}\cdot \frac{\hat{\sigma}_{LB_{\tau_g}}}{\sqrt{n}},\widehat{UB}_{\tau_g}+C_{n}\cdot \frac{\hat{\sigma}_{UB_{\tau_g}}}{\sqrt{n}}]$ where $C_{n}$ satisfies 
\begin{equation}\label{cn}
    \Phi\left(C_{n}+\sqrt{n}\frac{\widehat{UB}_{\tau_g}-\widehat{LB}_{\tau_g}}{\max(\hat{\sigma}_{LB_{\tau_g}},\hat{\sigma}_{UB_{\tau_g}})}\right)-\Phi(-C_{n})=0.95
\end{equation}
and are uniformly valid.\footnote{\citet{stoye2009more} highlights that the refined confidence interval from Imbens and Manski achieves uniform coverage only under the assumption that the estimator of the identified set length ($\widehat{UB}_{\tau_g}-\widehat{LB}_{\tau_g}$) is super-efficient near point-identification. He establishes a weaker sufficient condition of super-efficiency as being joint-normality of the lower and upper bound estimators along with the bounds being ordered (by construction). Both of these conditions hold in our case. \cite{stoye2009more} proposes a refinement with and without super-efficiency, and \cite{stoye2020simple} extends it to partial identification of a pseudo true parameter.}

\section{Simulation} \label{simulation}
This section presents simulation evidence of the bias induced by sample selection on the standard DiD estimates. It further demonstrates the feasibility of the identification and estimation procedures for partially identifying the ATTs of the different latent groups ($\tau_{OOO}$, $\tau_{ONO}$, $\tau_{NOO}$ and $\tau_{NNO}$) proposed above. The main data generating process (DGP) used in the simulations is as follows,
    \begin{align*}
        Y^\ast_{i0}(0) &= t^{g_i}_0+c_i+ u_{i0} \quad Y^{\ast}_{i1}(d)= t^{g_i}_1 + \tau^{g_i}\ast d+ c_i+ u_{i1}, \quad d=\{0,1\} \\
        S_{i0}(0) &= \mathbbm{1}\{ b_i+ v_{i0}>0\} \quad S_{i1}(d)= \mathbbm{1}\{\zeta\cdot d+ b_i+ v_{i1}>0\}, \quad d=\{0,1\} \\
        D_i &=\mathbbm{1}(a_i+w_{i1}>0)  
    \end{align*}
 where $\bigl(a_i,\;c_i,\;u_{i0},\;v_{i0},\;u_{i1},\;v_{i1}\bigr)'$ is drawn jointly from the six-dimensional truncated multivariate normal distribution, $\mathcal{N}_6\bigl(0,\Sigma_6\bigr)$, whose support is restricted to the hypercube ${[-M,M]^6}$. The covariance matrix $\Sigma_6$ has unit marginal variances, $\operatorname{Cov}(a_i,c_i)=\rho_{ac}$, $\operatorname{Cov}(u_{it},v_{it})=\rho_{u_t v_t}$ for $t=0,1$, and all remaining covariances are set to zero. The variable $c_i$ plays the role of time-invariant unobserved heterogeneity. Finally, $b_i$ and $w_{i1}$ are independent standard normal random variables. Latent group $g$ is defined by the tuple $\bigl(S_{i0},S_{i1}(0),S_{i1}(1)\bigr)\in\{NNN,NNO,NOO,ONN,\\ ONO,OOO\}$. The treatment effects $\tau^{g_i}$ and time effects  $t^{g_i}_t$ are both group-specific, with 
\begin{align*}
    \tau^{g_i}=\tau_{g}\cdot \mathbbm{1}\{g_i=g\}, \quad t^{g_i}_0=t_{0}\cdot \mathbbm{1}\{g_i=g\}, \quad \text{and} \quad t^{g_i}_1=t_{1}\cdot \mathbbm{1}\{g_i=g\}.
\end{align*}
We set $M=5$, $\zeta=1.5$, $\rho_{ac}=0.7$, $\rho_{u_0v_0}=0.7$, $\rho_{u_1v_1}=0.6$, $\tau_g=(-1,1,3,1,4,5)'$ , $t_0=(0,2,3,1,4,5)'$ and $t_1=(1,3,4,2,5,6)'$. The observed data are $\left\{Y_{it},S_{it},D_{i}\right\}_{i=1}^n$ where $Y^\ast_{i1}=Y^\ast_{i1}(0)\cdot (1-D_i)+Y^\ast_{i1}(1)\cdot D_i$, $S_{i1}=S_{i1}(0)\cdot (1-D_i)+S_{i1}(1)\cdot D_i$, and $Y_{it}=S_{it}\cdot Y^\ast_{it}$, for each $t=\{0,1\}$. 

This DGP satisfies Assumptions \ref{no anti}, \ref{PT_OOO}, \ref{PT_group}, \ref{monotone}, \ref{inde_conditional}, and \ref{mean dominance}. The true overall ATT is $\tau=2.8504$ and the latent group-specific ATTs are given by: $\tau_{OOO}=5$, $\tau_{ONO}=4$, $\tau_{ONN}=1$, $\tau_{NOO}=3$, $\tau_{NNO}=1$ and $\tau_{NNN}=-1$. 

\begin{figure}[H]
    \centering
    \caption{Distribution of estimated mixture proportion, $\hat{p}_{OOO1}$, with monotonicity}
   \includegraphics[height=0.29\textheight]{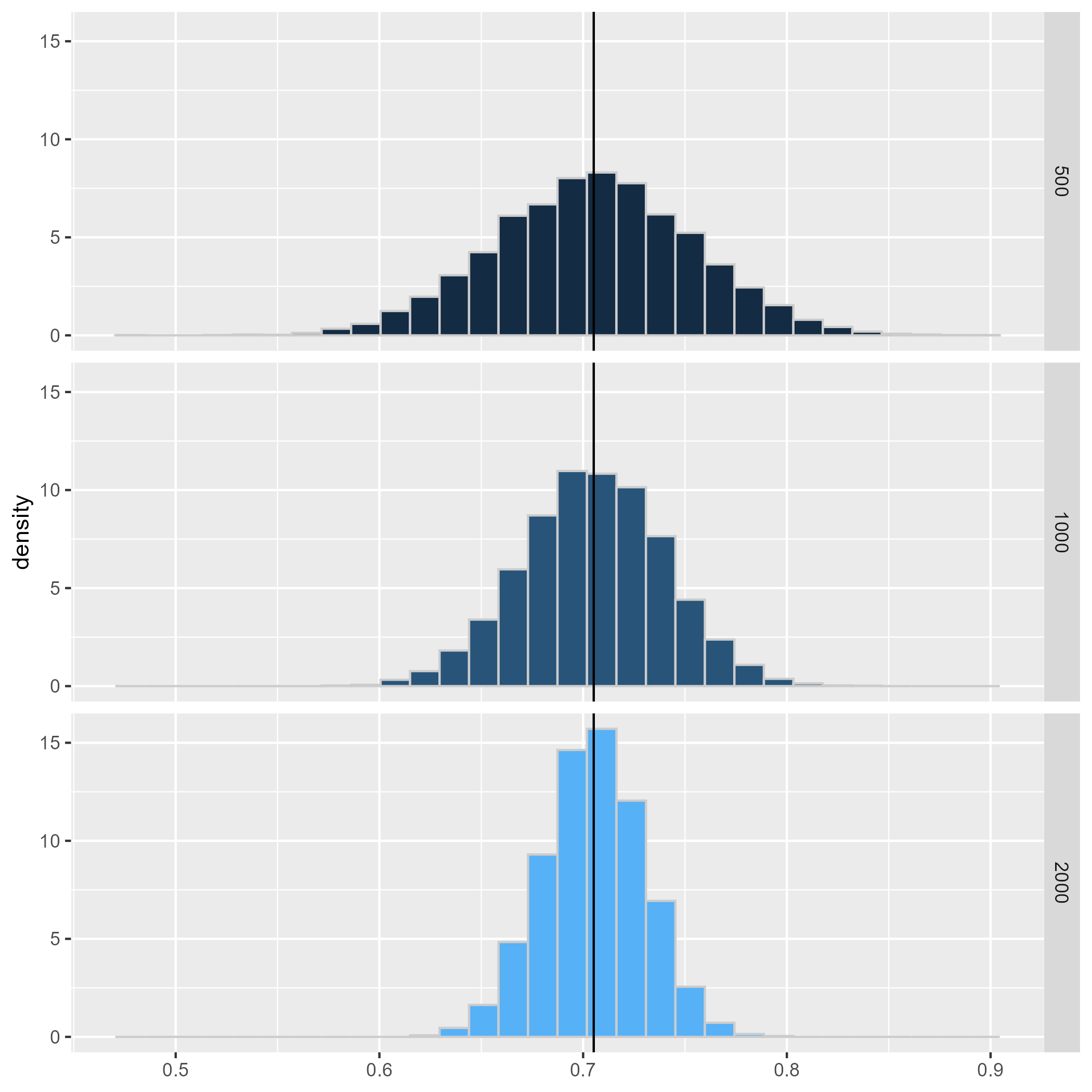}  
    \label{fig:mixp}
    \begin{minipage}{0.7\textwidth}
        \footnotesize \textit{Note:} This figure presents the empirical sampling distribution of $\hat{p}_{OOO1}$ for sample sizes $n \in \{500, 1000, 2000\}$, based on 10,000 replications. The vertical black line is the value of the true mixing proportion $p_{OOO1} = 0.7052$.
    \end{minipage}
\end{figure} 
Under this DGP, the true mixing probability is $p_{OOO1}$ = 0.7052 and the true numerical bounds for $\tau_{OOO}$ are given by $[LB_{\tau_{OOO}}, UB_{\tau_{OOO}}] = [3.6829,5.1969]$. Notice that these contain the true ATT for the OOO group. The mathematical derivation of these numerical bounds can be found in Appendix \ref{numerical bounds}.

We draw random samples of 500, 1000, and 2000 observations from this DGP and compute the empirical distribution of the estimated bounds over 10,000 simulation draws (replications). The bounds for $\tau_{OOO}$ are estimated under two sets of assumptions. The first corresponds to the bounds given in Theorem \ref{nomono_bound}, which we refer to as the \textit{without-monotonicity} scenario. The second corresponds to the bounds presented in Theorem \ref{mono_bound}, which we refer to as the \textit{with-monotonicity} scenario. Figure \ref{fig:mixp} plots the distribution of the estimated mixing probability, $\hat{p}_{OOO1}$, for the \textit{with-monotonicity} scenario, which is centred around the black vertical line, representing the true value of the mixing probability. The average estimated bounds for $\tau_{OOO}$ under the two sets of assumptions are given in Table \ref{simu bounds_OOO_10k}. These are seen to contain the true ATT for OOO of 5. As we discuss in Section \ref{Canonicalbias}, the na\"{\i}ve DiD estimate exhibits an upward bias of around 55.4\% in this relatively simple DGP, which satisfies the full independence assumption \ref{Partialselection}$(Joint)$. The 95\% CI reflects the coverage probability of the true interval, which hovers around 93\%, whereas the Imbens and Manski (IM) 95\% CI is the coverage probability of the true parameter, which hovers around 99\%. Simulation results for the remaining latent groups (ONO, NOO, and NNO) are provided in Appendix \ref{simu_other}.
\begin{table}[htbp]
    \centering
    \caption{Estimated bounds for $\tau_{OOO}$}
    \label{simu bounds_OOO_10k}
    \begin{threeparttable}
    \begin{tabular}{lcccccc}
    \toprule
    &       & \multicolumn{2}{c}{Without monotonicity } &       & \multicolumn{2}{c}{With monotonicity} \\
    \cmidrule{3-4}\cmidrule{6-7}          & $\hat{\tau}_{\textup{DiDs}}$ & $[\widehat{LB}_{\tau_{OOO}}, \widehat{UB}_{\tau_{OOO}}]$ & Coverage &       & $[\widehat{LB}_{\tau_{OOO}}, \widehat{UB}_{\tau_{OOO}}]$ & Coverage \\
    \midrule
    $n=500$   & 4.4291 & [3.3374, 5.5367] &       &       & [3.6868, 5.1909] &  \\
    95\% CI &       &  (2.8325, 6.0332) & 0.9977 &       & (3.2216, 5.6440) & 0.9305 \\
    IM 95\% CI &       & (2.9137, 5.9534) & 0.9998 &       & (3.2963, 5.5712) & 0.9898 \\
    \midrule
    $n=1000$  & 4.4301 & [3.3375, 5.5392] &       &       & [3.6864, 5.1942] &  \\
    95\% CI &       & (2.9811, 5.8884) & 0.9995 &       & (3.3587, 5.5144) & 0.9349 \\
    IM 95\% CI &       & (3.0384, 5.8323) & 1     &       & (3.4114, 5.4629) & 0.9971 \\
    \midrule
    $n=2000$  & 4.4334 & [3.3421, 5.5412] &       &       & [3.6894, 5.1979] &  \\
    95\% CI &       &  (3.0904, 5.7867)  & 1     &       & (3.4581, 5.4241) & 0.9318 \\
    IM 95\% CI &       & (3.1308, 5.7473) & 1     &       & (3.4953, 5.3877) & 0.9994 \\
    \bottomrule
    \end{tabular}%
    \end{threeparttable}
    \begin{tablenotes}[flushleft]
    \item\footnotesize Notes: This table presents estimated bounds for $\tau_{OOO}$ for three sample sizes, $n \in \{ 500, 1000, 2000\}$. These have been averaged across 10,000 replications.  The true simulated bounds are $[LB_{\tau_{OOO}}, UB_{\tau_{OOO}}] = [3.6829,5.1969]$. The 95\% confidence interval (CI) reflects the coverage probability of the true interval. The Imbens and Manski (IM) 95\% CI reports the coverage probability of the true parameter.
    \end{tablenotes}
\end{table}%
\paragraph{Mixing Proportion Approaching One:} We also examine how the bounds for $\tau_{OOO}$ behave under monotonicity (Theorem \ref{mono_bound}) as we increase the mixing proportion to approach 1. To do this, we modify the main DGP given above by varying the parameter $\zeta$ between 0.1 and 3, which generates mixing proportions ranging from 0.9601 to 0.6685, respectively. The results for this case can be found in Appendix Table \ref{tbl:change Mixp}. It becomes clear from these results that as the mixing proportion approaches 1, the estimated bounds get tighter and contract towards the true value of $\tau_{OOO}=5$. 

\paragraph{Violation of Monotonicity:} To further investigate the behavior of the proposed bounds for $\tau_{OOO}$ under the no-monotonicity setting of Theorem \ref{nomono_bound}, we modify the DGP by drawing $\zeta \in \{-1.5, 0, 1.5\}$ at random. This generates positive selection for some units, negative selection for others, and no selection for the rest, such that monotonicity is violated on average for the overall sample. The resulting bounds are reported in Appendix Table \ref{simu bounds nomono OOO} and contain the true ATT for the OOO group. These are seen to be wider compared to the bounds obtained under monotonicity (Table \ref{simu bounds_OOO_10k}). In conclusion, these findings validate our theoretical results on the relationship between proposed no-monotonicity and monotonicity bounds.

\paragraph{Strength of Assumption \ref{Partialselection}:} To assess how different versions of Assumption \ref{Partialselection} affect identification, we modify the DGP so that Assumptions \ref{Partialselection}(a) and \ref{Partialselection}(b) hold, while both monotonicity and Assumption \ref{inde_conditional} fail. We then estimate the \textit{without-monotonicity} bounds for $\tau_{OOO}$ based on Theorem \ref{nomono_bound}. This case is discussed in Appendix section \ref{DGP_4ab} and the results are reported in Appendix Table \ref{simu bounds_OOO_4ab}. We find that these bounds are noticeably wider than the \textit{without-monotonicity} bounds obtained under Assumption \ref{inde_conditional} (Table \ref{simu bounds_OOO_10k}).     

We conduct a similar exercise by modifying the DGP so that Assumption \ref{Partialselection}(a) and monotonicity hold together. We then estimate the \textit{with-monotonicity} bounds given in Theorem \ref{mono_bound}. The estimated bounds provided in Appendix Table \ref{simu bounds_OOO_4a} are again wider than the \textit{with-monotonicity} bounds estimated under Assumption \ref{inde_conditional} that are reported in Table \ref{simu bounds_OOO_10k}. These findings illustrate how strengthening Assumption \ref{Partialselection} narrows the identified set for $\tau_{OOO}$ and produces more informative bounds.

\section{Empirical Illustrations } \label{application} 
In this section, we illustrate our partial identification approach with two empirical applications. First, we revisit the 1970's experiment described in \cite {lalonde1986evaluating} by using the Aid to Families with Dependent Children (AFDC) sample of women from the National Supported Work (NSW) training program. Here, we consider the sample selection problem arising from unemployment (reported zero earnings). For the second application, we consider the study by \cite{bloom2015does} in which an experiment evaluates the effectiveness of a company's working from home policies. In this application, sample selection bias arises from employee attrition. Although both applications were conducted as randomized experiments, outcomes are only observed conditional on being selected into the sample. If selection is endogenous to the outcomes of interest in both the pre- and post-treatment periods, then random assignment does not necessarily ensure that treated and control units remain comparable, especially if individuals enter or exit the sample differently across groups and periods.

\subsection{NSW Training Program for AFDC Women}\label{Application NSW}

NSW  was a temporary employment program implemented in the United States between 1975-1979. It was designed to help individuals from disadvantaged populations find stable employment by offering them structured work experience and counselling in a sheltered environment. The program targeted four disadvantaged socio-economic groups, with qualified applicants being randomly assigned to training. In our empirical analysis, we only consider the AFDC sub-sample of women originally studied in \cite{calonico2017women}, consisting of 1185 individuals out of which 600 received training. We apply the proposed approach to account for sample selection arising from unobserved earnings.

We treat zero earnings as unobserved wages due to individuals' inability to find suitable employment. In this case, $D=1$ if an individual is assigned to receive training and zero otherwise. $Y_{t}^\ast(0)$ and $Y_{t}^\ast(1)$ are potential earnings of an individual and $S_{t}(1)$ and $S_{t}(0)$ are potential indicators for being employed or not. Our treatment effect of interest, $ \tau_{OOO}=\mathbbm{E}[Y_{1}^{\ast}(1)-Y_{1}^{\ast}(0)|D=1, S_{0}=1, S_{1}(0)=1, S_{1}(1)=1]$, captures the effect of training on earnings for the latent group that is ``always employed''. Policy makers often care about this latent group as it reflects the returns to training for workers who would remain employed regardless of training, as they are highly attached to the labor force. Many policies aim to improve outcomes among continuously employed participants. Therefore, this parameter is informative for evaluating the impact of a program within a committed, consistently engaged population, reflecting an important component of the treatment effect on the intensive margin. During the pre-treatment period, 73.3\% of the treated and 74.7\% of the control samples reported zero earnings. The follow-up survey indicates 45.0\% of the treated and 46.2\% of the control group individuals reported zero earnings (see Table \ref{unemployment rate}).\footnote{Tables \ref{summary NSW unemployed base} and \ref{summary NSW unemployed follow up} 
present the covariates used in the analysis along with means for the observed (unobserved) samples for both treatment groups.}

We estimate bounds for $\tau_{OOO}$ under two sets of assumptions. The first considers Assumptions \ref{no anti}, \ref{PT_OOO}, \ref{Partialselection}(a), \ref{Partialselection}(b), and \ref{rs} which we refer to as the \textit{without-monotonicity} scenario and the second under Assumptions \ref{no anti}, \ref{PT_OOO}, \ref{monotone}, \ref{Partialselection}(a), and \ref{rs} which we refer to as the \textit{with-monotonicity} scenario. Intuitively, Assumption \ref{no anti} requires that women's pre-treatment wage offers and the ability/willingness to secure a job (i.e., be employed) are unaffected by their eventual assignment into the training program. For example, if we consider $S(0)$ and $S(1)$ to reflect the decision to work as determined by the interplay of reservation wages and wage offers, the no anticipation assumption rules out scenarios in which, upon learning of their assignment to the training program, women immediately increase their reservation wages or firms increase their wage offers to those slated to participate in the program. As mentioned in \cite{calonico2017women}, this assumption is likely to hold for earnings measured in 1975, before assignment to the program began. Assumption \ref{PT_OOO} will be plausible if, in the absence of treatment, latent wage trends do not differ between the treated and untreated samples of women. Importantly, to partially identify $\tau_{OOO}$, this condition only needs to hold for the always-observed group. In our setting, these are likely women who have higher attachment to the labor force. Given the program's focus on low-income women with dependents and recent unemployment spells, this group is more homogeneous in their employment patterns and wage trends, making the parallel trends assumption credible. Similarly, the ignorability of potential selection (Assumption \ref{Partialselection}) requires that the counterfactual employment behavior of women in the post-treatment period be similar across training assignment groups conditional on pre-treatment unemployment. This is plausible if differences are driven primarily by pre-existing employment patterns rather than differential assignment to training, and rules out that women assigned to the training program would have been, on average, more likely to become unemployed than their counterparts not assigned for training. For example, it rules out the possibility that, in the absence of treatment, reservation wages for the $D=1$ and $D=0$ groups would respond differently to shocks or unobservables, such as having a young child in the household in the post-training period.

In this application, we assume monotonicity operates in the positive direction, implying that receiving training increases the probability of employment and, hence, the likelihood of observing employed individuals in the treated sample compared to the control sample. The results under each set of assumptions are presented in Table \ref{tab:NSW zero}. The bounds derived without Assumption \ref{monotone} are wide and uninformative. Imposing monotonicity substantially tightens the bounds, demonstrating the assumption's informational content. We cannot rule out $\tau_{OOO}$ between \$1,404 and \$1,718. The na\"{\i}ve DiD estimates suggest an increase of around \$1,613 dollars in annual earnings due to training.\footnote{As discussed in Appendix \ref{Appendix: Trimmed}, the na\"{\i}ve DiD is always within the identified set for $\tau_{OOO}$ based on the trimming approach discussed in Theorems \ref{nomono_bound}-\ref{mono_bound}.} 
\begin{table}[htbp]
  \centering
    \caption{Estimated bounds for $\tau_{OOO}$ with and without monotonicity} 
    \label{tab:NSW zero}
    \begin{threeparttable}
    \begin{tabular}{lccc}
    \toprule
    Scenario & [$\widehat{LB}_{\tau_{OOO}}, \widehat{UB}_{\tau_{OOO}}$] & $\hat{p}_{OOO1}$ & $\hat{p}_{OOO0}$ \\
    \midrule
    Without monotonicity & [-8.5230, 11.8245] & [0.4277, 0.9963] & [0.4293, 1] \\
    With monotonicity & [1.4036, 1.7179] & 0.9963 & 1 \\
    Na\"{\i}ve DiD ($\tau_{\textup{DiDs}}$) & 1.6133 &       &  \\
    \bottomrule
    \end{tabular}%
    \begin{tablenotes}[flushleft]
    \item\footnotesize Notes: All earnings are expressed in thousands of 1982 dollars.
    \end{tablenotes}
    \end{threeparttable}    
\end{table}%
While it's interesting to characterize the effects of training for the always-observed group of women, who have relatively high attachment to the labor force and whose estimated ATT could serve as a useful benchmark when treatment effect heterogeneity is limited, it's natural to wonder about the effects for other groups. This is especially true since the NSW demonstration explicitly targeted more vulnerable populations with weaker labor market attachment. Specifically, policymakers may be more interested in estimating the ATT for those i) unemployed before training who will be employed post-treatment irrespective of training (i.e. $\tau_{NOO}$), ii) employed only if they are given training (i.e. $\tau_{NNO}$), or iii) people employed before training who will only be employed post-treatment if they are given training (i.e. $\tau_{ONO}$). The estimated bounds corresponding to these ATTs are presented in Table \ref{tab:NSW other} and use the results in Theorems \ref{ONO_bound}-\ref{NOO_bound}, which impose assumptions \ref{no anti}, \ref{PT_group}, \ref{monotone}, \ref{inde_conditional}, \ref{rs} and \ref{mean dominance}. 

The last column of Table \ref{tab:NSW other} reports the estimated population shares corresponding to each latent group type. The proportion of NOO is estimated to be large (37.4\%), which aligns with the selection criteria for participation in NSW. This also matches with the data where we observe a high fraction of individuals reporting zero earnings in the pre-treatment period and a sizable increase in post-treatment employment among both the treated and control groups. Meanwhile, the share of the NNO latent group reflects the relatively small extensive margin effect of the training program among the individuals who are initially unemployed, implying that training induces only a small increase in employment. The share of OOO is important, at 16.5\%, but also highlights the fact that we can only learn about the impact of the policy with more certainty for a small part of the population. The limited information available for the NOO, NNO, and ONO latent groups is reflected in the wide estimated ATT bounds, which in all cases encompass both the estimated set for $\tau_{OOO}$ and zero. Therefore, we cannot rule out either the case of homogeneous treatment effects relative to OOO or the possibility of a nil effect of the training program for these groups. The overall ATT lies in the interval $\tau \in [-13.5636,15.4139]$. These bounds are estimated based on the results in Theorem \ref{theorem:partialid_tau}. One should also note that, even under these assumptions, bounds for the overall ATT are not very informative given the wide identified sets for the ATT of each latent group and the limited importance of the always-observed group in the population at large. 

\begin{table}[htbp]
  \centering
  \caption{Estimated ATT bounds and CIs for the effect of NSW training on employment}
  \label{tab:NSW other}
  \begin{threeparttable}
    \begin{tabular}{ccccc}
    \toprule
    $\tau_g$ & [$\widehat{LB}_{\tau_g},\widehat{UB}_{\tau_g}$] & 95\% CI & IM 95\% CI & $\widehat{\pi}_g$ \\
    \midrule
    $\tau_{OOO}$ & [1.4036, 1.7179] & (-0.4674, 3.4837) & (-0.1667, 3.1999) & 0.1651 \\
    $\tau_{NOO}$ & [-2.3100, 5.4767] & (-3.6533, 6.7935) & (-3.4374, 6.5819) & 0.3743 \\
    $\tau_{NNO}$ & [-6.9252, 21.3140] & (-7.9650, 26.9658) & (-7.7979, 26.0575) & 0.0092 \\
    $\tau_{ONO}$ & [-12.7556, 37.6753] & (-15.3462, 47.9870) & (-14.9299, 46.3297) & 0.0006 \\
    \bottomrule
    \end{tabular}%
    \begin{tablenotes}[flushleft]
    \item\footnotesize Notes: Earnings are expressed in thousands of 1982 dollars. Estimated bounds are given in square brackets and confidence intervals (CIs) in round brackets. CIs are based on 1000 bootstraps. For Imbens and Manski (IM) CIs, $c_n=1.645$ satisfies the relationship given in Equation (\ref{cn}) for all the cases considered and $\hat{\pi}_{g}=\hat{\pi}_{g1}+\hat{\pi}_{g0}$ represents the estimated proportion of latent group `$g$' in the population.
\end{tablenotes}
\end{threeparttable}
\end{table}%
\subsection{Impact of Work From Home on Employee Performance}\label{Application: WFH}
In this section, we revisit the results of an experiment at Ctrip, a 16,000-employee, NASDAQ-listed Chinese travel agency \citep{bloom2015does}. This was carried out to evaluate the effectiveness of working from home (WFH) policies on employee performance. The experiment was conducted from January 2010 to August 2011 among eligible employees of the airfare and hotel departments at the company's Shanghai call centre, who volunteered to participate. Out of the employees who volunteered, only 49.5\% (249) met the eligibility criteria set by the company. Employees with even-numbered birthdays were assigned to the treated group, where they were allowed to WFH and those with odd birthdays were assigned to the control group with no option to WFH. Accordingly, 52.6\% (131) and 47.4\% (118) were allocated to the treatment and control groups, respectively. We use average individual weekly performance z-scores, a combination of different key performance indicators standardized based on each job type, to evaluate employee performance.\footnote{See \cite{bloom2015does} for a detailed description of the experiment and data collection process.} We will illustrate how our identification strategy can be used to account for selection bias due to employee attrition in the experimental period.\footnote{Table \ref{summary WFH base} and \ref{summary WFH follow up} in Appendix \ref{WFH} reports descriptive statistics for the observed covariates.}
        
In this application, $D=1$ if an employee is eligible to WFH and zero otherwise. $Y_{t}^\ast(0)$ and $Y_{t}^\ast(1)$ are potential average individual weekly performance z-scores of employees under the two treatment scenarios. $S_{t}(1)$ and $S_{t}(0)$ are potential indicators for attrition, where $S_{t}=1$ is the realized sample participation, indicating the particular employee stayed with the company. Our primary focus is the identification of $\tau_{OOO}=\mathbbm{E}[Y_{1}^{\ast}(1)-Y_{1}^{\ast}(0)|D=1, S_{0}=1, S_{1}(0)=1, S_{1}(1)=1]$ which captures the ATT of WFH eligibility on employee performance for the subgroup of employees who will stay with the company irrespective of WFH or not.\footnote{Since the treatment variable $D=1$ indicates eligibility to the program, we could interpret $\tau_{OOO}$ as the intent-to-treat effect of the WFH program on worker's performance for the always-observed group.} This could be interpreted as the ``intensive'' margin effects of the WFH policy, that is, changes in performance for employees that would have been retained even in the absence of WFH.

Table \ref{Employee Attrition} shows the attrition rates during the experimental period. We observe an attrition rate for the control group (34.8\%) that is more than double that for the treated group (16.0\%). 

\begin{table}[htbp]
    \caption{Attrition rates among employees in the experimental period }
    \label{Employee Attrition}
    \begin{adjustbox}{center=\textwidth}
    \begin{tabular}{lccc}
     \toprule
                                & Control & Treated & Total \\
                                \midrule
    Employees who left Ctrip & 41      & 21     & 62   \\
    Total                       & 118      & 131     & 249  \\
    \midrule
    Attrition rates (\%)         & 34.8     & 16.0    & 24.9 \\
    \bottomrule
    \end{tabular}
    \end{adjustbox}
\end{table}
We estimate two sets of bounds for $\tau_{OOO}$. The \textit{without-monotonicity} case imposes assumptions \ref{no anti}, \ref{PT_OOO}, \ref{Partialselection}(a),  \ref{Partialselection}(b), and \ref{rs}, and the \textit{with-monotonicity} scenario under assumptions \ref{no anti}, \ref{PT_OOO}, \ref{monotone}, \ref{Partialselection}(a), and \ref{rs}. Intuitively, Assumption \ref{no anti} requires that employee performance and the decision to remain with the firm before the WFH policy is introduced are unaffected by the assignment to remote work or not. For example, it rules out scenarios where workers who would have otherwise left the firm in the pre-treatment period decide to stay in order to benefit from the WFH scheme. Given the random assignment to WFH based on birth date, it is unlikely that workers could anticipate their assignments and alter pre-treatment quitting behavior. However, we cannot rule out the possibility that workers who were considering quitting might have been more likely to volunteer for the program and delay their decisions until after assignment to the treatment group. This would not violate the no anticipation assumption and is naturally incorporated by the ignorability of potential selection assumption, which only requires the same post-treatment quitting behavior between the two WFH assignment groups, conditional on being observed in the pre-treatment period. For parallel trends on latent outcomes (assumption \ref{PT_OOO}), it is plausible to expect that, in the absence of WFH policy, performance trends for employees would have followed similar paths across treatment assignment groups. This is because the always-observed group comprises of workers with stable attachment to the firm who continue to perform the same roles within the same organizational environment and are evaluated under the same performance criteria. As a result, differences in performance trends are less likely to reflect systematic differences in underlying productivity growth, making parallel trends a reasonable assumption. For the second case, we assume positive monotonicity, implying that WFH employees are at least as likely to stay in the company during the experimental period as those who do not WFH. This can be justified as employees assigned to the WFH treatment can choose whether or not to take advantage of the policies, and are likely better off due to increased convenience and reduced commuting costs if they decide to participate.
  
The results are presented in Table \ref{tab:WFH}. The na\"{\i}ve DiD implies that the overall performance of the treated group is  0.1759 standard deviations higher than what would have been in the absence of the WFH treatment.\footnote{As discussed in Appendix \ref{Appendix: Trimmed}, the na\"{\i}ve DiD is always within the identified set for $\tau_{OOO}$ based on the trimming approach discussed in theorems \ref{nomono_bound}-\ref{mono_bound}.} The more flexible bounds estimated without assuming monotonicity are wide, but rule out a decline in standardized performance larger than 0.33 std. deviations as well as increases above 0.67 std. deviations. These can be tightened by imposing Assumption \ref{monotone}, leading to an identified set ranging between 0.0057 and 0.3806 standard deviations improved performance for the effect of being assigned to WFH among employees who stay with the company regardless of being eligible for the policy.  

\begin{table}[htbp]
  \centering
    \caption{Estimated bounds for $\tau_{OOO}$ with and without monotonicity}
    \label{tab:WFH}
    \begin{tabular}{lccc}
    \toprule
    Scenario & $[\widehat{LB}{\tau_{OOO}}, \widehat{UB}_{\tau_{OOO}}]$ & $\hat{p}_{OOO1}$ & $\hat{p}_{OOO0}$ \\
    \midrule
    Without monotonicity & [-0.3242, 0.6660] & [0.5862, 0.7771] & [0.7543, 1] \\
    With monotonicity & [0.0057, 0.3806] & 0.7771 & 1 \\
    Na\"{\i}ve DiD ($\tau_{\textup{DiDs}}$) & 0.1759 &       &  \\
    \bottomrule
    \end{tabular}%
  \label{tab:addlabel}%
\end{table}%
In contrast to the previous application, all workers in this setting are observed in the initial period, which reduces the number of possible latent types. However, the share of always-observed workers among treated individuals observed in both periods is at most 77.7\%.

Furthermore, as presented in Table \ref{tab:WFH CI}, under monotonicity, the OOO group is estimated to represent close to 65.3\% of the overall population. The ONO latent group is an interesting subtype and represents those workers who would have left the company if WFH were not provided but would stay otherwise. Using results in Theorem \ref{ONO_bound}, we can partially identify the ATT for this group by imposing Assumptions \ref{no anti}, \ref{PT_group}, \ref{monotone}, \ref{inde_conditional}, \ref{rs} and \ref{mean dominance}(a). As can be seen in Table \ref{tab:WFH CI} the ONO encompasses 18.7\% of the population. The estimated identified set for $\tau_{ONO}$ includes zero and a wide range of values. Therefore, we cannot determine the direction of the effect of WFH eligibility (if any) on workers' performance for this latent group without additional information.

\begin{table}[htbp]
  \centering
  \caption{Estimated ATT bounds and CIs for the effect of WFH on employee performance}
    \label{tab:WFH CI}
    \begin{threeparttable}
    \begin{tabular}{ccccc}
    \toprule
    ATT & [$\widehat{LB}_{\tau_g}, \widehat{UB}_{\tau_g}$] & 95\% CI & IM 95\% CI & $\hat{\pi}_g$ \\
    \midrule
    $\tau_{OOO}$   & [0.0057, 0.3806] & (-0.1439, 0.5212) & (-0.1199, 0.4986) & 0.6525 \\
    $\tau_{ONO}$   & [-1.0914, 3.8445] & (-1.3922, 4.0715) & (-1.3439, 4.0350) & 0.1872 \\
    \bottomrule
    \end{tabular}%
    \begin{tablenotes}[flushleft]        
    \item \footnotesize Notes: Bounds are given in square brackets and confidence intervals (CIs) in round brackets. CIs are based on 1000 bootstraps. For Imbens and Manski (IM) CIs, $c_n=1.645$ satisfy the relationship given in Equation (\ref{cn}) for all the cases considered and $\hat{\pi}_{g}=\hat{\pi}_{g1}+\hat{\pi}_{g0}$ represents the estimated proportion of latent group `$g$' in the population.
    \end{tablenotes}
    \end{threeparttable}
\end{table}%
In this particular application, since there is no sample selection in the pre-treatment period and we assume monotonicity, there are only three latent groups, OOO, ONO and ONN. We have estimated bounds for the ATT for both the OOO and ONO groups, which together account for roughly 84.0\% of the total population. The remaining 16.0\% corresponds to workers who would have exited the company regardless of their WFH eligibility status.
The overall ATT lies in the interval $\tau \in [-0.9413,1.7086]$. These bounds are estimated based on the results in Theorem \ref{theorem:partialid_tau}. 

\section{Conclusion} \label{conclusion}
In this article, we address the challenge of endogenous sample selection within the difference-in-differences (DiD) framework. We demonstrate that the standard DiD estimates can be biased when sample selection is ignored, even under the strong assumption of independence between the selection and treatment assignment mechanisms. We propose methods for partial identification of average treatment effects for different latent treated subpopulations by building on the trimming procedure of \cite{lee2009}.

For the OOO-group, identification relies on the insight that individuals observed in both periods are a mixture of two possible latent groups. The mixture proportions are identified under alternative sets of assumptions on the selection and treatment assignment mechanisms. Specifically, we consider scenarios with and without MS. When MS is not imposed, the latent strata proportions are identified under the IS assumption, which assumes counterfactual selection probabilities to be equal between the treated and untreated units, conditional on being observed in the baseline period. In cases where monotonicity holds, we achieve tighter bounds on the ATT for the OOO by imposing ignorability for selection in one direction only. We also discuss identification of ATT bounds for the OOO group when covariates are observed by considering both no monotonicity and a relaxed version of monotonicity.  We also extend the basic identification argument (without covariates) to settings in which the researcher has access only to repeated cross-sectional data and to staggered treatment adoption across multiple time periods. 

Additionally, we present the identified sets for the ATT of other empirically relevant latent groups, such as ONO, NON, and NOO, under MS and outcome mean dominance assumptions. Combining these group-specific identified sets with weights also allows us to partially identify the overall ATT in the population. We illustrate our results through two empirical illustrations: 1) bounding the effects of a job training program and 2) bounding the effects of a work-from-home policy on employee performance. These applications highlight the practical relevance of the proposed bounds in two different empirical settings.

\label{Bibliography}
\pagebreak
\singlespacing
\bibliographystyle{ecta.bst}
\bibliography{mybib}  
\pagebreak
\appendix
\setcounter{footnote}{0} 
\numberwithin{equation}{section}
\numberwithin{table}{section}
\numberwithin{figure}{section}
\numberwithin{lemma}{section}
\numberwithin{assumption}{section}
\numberwithin{theorem}{section}
\section{Identification under Negative Monotonicity}\label{neg_mono}

\begin{assumption}\label{neg_monotone}
    Negative monotone sample selection: 
	\begin{align*}
		\mathbbm{P}[S_{1}(1)\leq S_{1}(0)]&=1 
	\end{align*}
\end{assumption}

\textcolor{blue}{Assumption \ref{neg_monotone}} assumes that treatment decreases the probability of selection for all individuals. Negative selection implies that there are no individuals whose outcomes are observed only when they are treated. For example, consider a context of scholarships and exam attendance. While the policy aims to encourage academic participation, it may in fact reduce the probability that some students sit for final exams. Without the scholarship, students face stronger incentives to perform well in order to secure future funding, so they are more likely to attend exams. With the scholarship in place, however, some students feel financially secure and may exert less effort, skipping exams they otherwise would have taken. In this setting, treatment (receiving the scholarship) lowers selection (exam attendance), illustrating negative monotonicity. Assumption \ref{neg_monotone} rules out the strata NNO and ONO, that is, individuals that would not be observed in period one if untreated but will be observed if treated.

\begin{lemma}\label{lemma:mono_weights_neg} Under Assumptions \ref{no anti} and \ref{neg_monotone} (negative monotonicity), we obtain $p_{OOO1}=1$ and $p_{OOO0} = \frac{\mathbbm{P}[S_1(1)=1|D=0, S_0=1]}{\mathbbm{P}[S_1=1|D=0, S_0=1]}$. 
\end{lemma}
Proof is provided in Appendix \ref{proof:mono_weights_neg}. Negative monotonicity rules out latent groups NNO and ONO, and $p_{OOO1}=1$ and $p_{OOO0} = \frac{\mathbbm{P}[S_1(1)=1|D=0, S_0=1]}{\mathbbm{P}[S_1=1|D=0, S_0=1]}$(Lemma \ref{lemma:mono_weights_neg}).
Since $\mathbbm{E}[Y_{1}^{\ast}(1)-Y_{0}^{\ast}(1)|D=1,OOO]$ is point identified in that case, there is no need for assumption \ref{Partialselection}(a). We can obtain point identification of $p_{OOO0}$ by combining negative monotonicity in selection and Assumption \ref{Partialselection}(b). This point identifies $p_{OOO0}=\frac{\mathbbm{P}[S_1=1|D=1, S_0=1]}{\mathbbm{P}[S_1=1|D=0, S_0=1]}$ by assuming $\mathbbm{P}[S_1(1)=1|D=1, S_0=1]=\mathbbm{P}[S_1(1)=1|D=0, S_0=1]$.

\begin{theorem}[Bounds for $\tau_{OOO}$ under negative monotonicity]\label{mono_bound_negative} Under Assumptions \ref{no anti},\ref{PT_OOO},\ref{neg_monotone}, \ref{Partialselection}(b), \ref{rs}, and $p_{OOO0}>0$, bounds on the treatment effect on the treated for the always-observed group ($\tau_{OOO}$) lies in the interval $[LB_{\tau_{OOO}}, UB_{\tau_{OOO}}]$ where,
\begin{equation*} 
    \begin{split}
        LB_{\tau_{OOO}} &= \mathbbm{E}[Y_{1}-Y_{0}|D=1,S_{0}=1,S_{1}=1]-UB_{OOO0}\\
        UB_{\tau_{OOO}} &=\mathbbm{E}[Y_{1}-Y_{0}|D=1,S_{0}=1,S_{1}=1]-LB_{OOO0}
    \end{split}
\end{equation*}
where,  
\begin{align*}
    LB_{OOO0} &=\mathbbm{E}[Y_{1}-Y_{0}|D=0,S_{0}=1,S_{1}=1, (Y_{1}-Y_{0})\leq F_{\Delta Y|011}^{-1}(p_{OOO0})]\\
    UB_{OOO0} &=\mathbbm{E}[Y_{1}-Y_{0}|D=0,S_{0}=1,S_{1}=1, (Y_{1}-Y_{0}) > F_{\Delta Y|011}^{-1}(1-p_{OOO0})]
\end{align*}
with $p_{OOO0}=\frac{\mathbbm{P}[S_{1}=1|S_{0}=1,D=1]}{\mathbbm{P}[S_{1}=1|S_{0}=1,D=0]}$.
\end{theorem}
Proof is provided in Appendix \ref{proof:mono_bound_negative}.

\section{Extension to Repeated Cross-section}\label{sec:extension_RC}
In many cases, the researcher may have access only to repeated cross-sections of the population of interest. In that case, DiD can still be implemented by comparing the changes in outcomes for the treated and control groups, even though one can no longer match individuals/units over time \citep{sant2026difference,sant2020doubly,abadie2005semiparametric,meyer1990workers,finkelstein2002effect}. That setting is challenging in the presence of sample selection, as we cannot rely on information from subsamples of units observed in both periods or even calculate changes in outcome for specific individuals, as we did in Section \ref{identification}. Nevertheless, we can still partially identify the ATT for a latent group of always-observed individuals, as defined below.\footnote{A formal analysis of any efficiency loss from using repeated cross-sections instead of panel data, as explored in \citet{sant2020doubly}, is left for future work.}
 
Let $T$ be an indicator for the post-treatment period, such that it equals one if a unit is observed in the post-treatment period and zero otherwise. Then, the latent outcome and observed selection is given by
\begin{equation*}
    Y^\ast = T\cdot Y^\ast_1+(1-T)\cdot Y^\ast_0 \text{ and } S = T\cdot S_1+(1-T)\cdot S_0,
\end{equation*} and $Y = S\cdot Y^\ast$. We develop this scenario with the sampling assumption of no compositional changes. This is a frequently imposed assumption with cross-section data \citep{abadie2005semiparametric,callaway2021difference}, and requires sampling observations from the same population across the two periods. Formally, 
\begin{assumption}[Sampling assumption]\label{No_com} Assume that $\{(Y_i, D_i, T_i); \ i=1,2,\ldots, N\}$ consists of independent and identically distributed draws from the following mixture distribution,
\begin{align*}
    \mathbbm{P}[Y\leq y,D=d,S=s,T=t]&= \lambda\cdot t \cdot \mathbbm{P}[Y^\ast_1\leq y, D=d, S_1=s|T=1]\\
    & + (1-\lambda) \cdot (1-t) \cdot \mathbbm{P}[Y^\ast_0\leq y, D=d, S_0=s|T=0]
\end{align*} where $\lambda \equiv \mathbbm{P}[T=1]\in (0,1)$ reflects the proportion of observations sampled in the post-treatment period, with the joint distribution of $D$ being invariant to $T$.
\end{assumption}

The repeated cross-sections DiD estimand that ignores sample selection is 
\begin{equation*}
    \begin{split}
        \tau^{rc}_{\textup{DiDs}} &\equiv \mathbbm{E}[Y|D=1, T=1, S=1]-\mathbbm{E}[Y|D=1, T=0, S=1] \\
        & -\mathbbm{E}[Y|D=0, T=1, S=1] + \mathbbm{E}[Y|D=0, T=0, S=1]
    \end{split}
\end{equation*} and simply compares the observed outcomes for treated and untreated groups in the pre-treatment and post-treatment periods.

Since we don't observe the same individuals in both time periods, we define the latent subgroups for each period $T=t$ as follows: 
\begin{table}[H]
  \centering
  \caption{Latent groups based on sample selection at $T=t$}
    \begin{tabular}{ccc}
    \toprule
    $S_t(0)$ & $S_t(1)$ & $G=g$ \\
    \midrule
    1     & 1     & OO \\
    1     & 0     & ON \\
    0     & 1     & NO \\
    0     & 0     & NN \\
    \bottomrule
    \end{tabular}%
  \label{tab:addlabel}%
\end{table}%
The relationship between the observed and latent groups is given in the table below:
\begin{table}[H]
    \centering
    \caption{Observed and latent groups with repeated cross-sections}
    \label{tab:OO}
    \begin{tabular}{cccccc}
    \toprule
    \multicolumn{3}{c}{T=0} & \multicolumn{3}{c}{T=1} \\
    \midrule
    S     & D=0   & D=1   & S     & D=0   & D=1 \\
    \midrule
    0     & NN    & NN    & 0     & NN, NO & ON, NN \\
    1     & OO    & OO    & 1     & ON, OO & NO, OO \\
    \bottomrule
    \end{tabular}
\end{table}
We define our target parameter to be the ATT for the subpopulation that is always observed, denoted by OO, and indicates that selection equals one in the period considered under both counterfactual treatment states. Formally,
\begin{equation}\label{att_oo}
    \tau_{OO}=\mathbbm{E}[Y_{1}^{\ast}(1)-Y_{1}^{\ast}(0)|D=1,S_{t}(0)=1,S_{t}(1)=1].
\end{equation}
To be able to partially identify the ATT for the OO group and characterize the bias for the na\"{\i}ve DiD estimator in the repeated cross-sections case, we impose a PTO condition similar to Assumption \ref{PT_OOO} in the main text.

\begin{assumption}[Parallel trends for the OO group]\label{pt_rc} 
\footnotesize{
\begin{align*}
   \mathbbm{E}[Y_1^\ast(0)|D=1, T=1, OO] - \mathbbm{E}[Y_0^\ast(0)|D=1, T=0, OO] = \mathbbm{E}[Y_1^\ast(0)|D=0, T=1, OO] - \mathbbm{E}[Y_0^\ast(0)|D=0, T=0, OO]. 
\end{align*}}
\end{assumption}
\begin{lemma}[Bias of  $\tau^{rc}_{\textup{DiDs}}$]\label{lemma_dids_rc} Under Assumptions \ref{no anti}, \ref{monotone}, \ref{No_com} and \ref{pt_rc}, the DiD estimand for the observed group with repeated cross-section, $\tau^{rc}_{\textup{DiDs}}$, can be decomposed as
\begin{align*}
   \tau^{rc}_{\textup{DiDs}} &=q_{OO11}\cdot\tau_{OO}+(1-q_{OO11})\cdot\tau_{NO} +(1-q_{OO11})\cdot \big(\mathbbm{E}[Y_1^\ast(0)|D=1,T=1, NO]\\
   &-\mathbbm{E}[Y_1^\ast(0)|D=1, T=1, OO]\big)
\end{align*}
  where  $q_{OO11}=\sfrac{\pi_{OO11}}{(\pi_{OO11}+\pi_{NO11})}$ , $\pi_{gdt}=\mathbbm{P}[G=g, D=d, T=t]  = \mathbbm{P}[S_1(0)=s',S_1(1)=s'', D=d, T=t]$. 
\end{lemma}
Proof is provided in Appendix \ref{proof:lemma_dids_rc}. Similar to the case of a two-period panel, the na\"{\i}ve DiD applied to repeated cross-section involves a weighted average of treatment effects for the always-observed (OO) and observed-when-treated (NO) plus a term which reflects the average differences in outcomes in the absence of treatment between the OO and NO groups among those treated in the post-treatment period. Since we can no longer track individuals observed in the post-treatment period in the baseline, the OO and NO groups average over $S_0$. Similar as before, the extent of the bias here depends on two factors; the share of OO in the $D=1, T=1$ group and how different the average untreated potential outcomes are between the two latent types. If $S_1(0)$ is exogenous to $Y_1^\ast(0)$, the bias disappears and $\tau^{rc}_{\textup{DiDs}}$ has a causal interpretation. Note, however, that this assumption is stronger than requiring selection to be exogenous to the trends of the untreated potential outcome (as in the case when panel data is available).
\subsection{Bounds for $\tau_{OO}$}
To bound the ATT for OO, we start by showing that a hypothetical DiD estimand for the OO group would correctly recover $\tau_{OO}$.
{\small \begin{align*}
    &\mathbbm{E}[Y|D=1, T=1, OO] - \mathbbm{E}[Y|D=1, T=0, OO] - \mathbbm{E}[Y|D=0, T=1, OO]+ \mathbbm{E}[Y|D=0, T=0, OO] \\
    & = \mathbbm{E}[Y_1^\ast(1)|D=1, T=1, OO] - \mathbbm{E}[Y_0^\ast(1)|D=1, T=0, OO] \\
    &- \mathbbm{E}[Y^\ast_1(0)|D=0, T=1, OO]+ \mathbbm{E}[Y^\ast_0(0)|D=0, T=0, OO] \\
    & = \mathbbm{E}[Y_1^\ast(1)|D=1, T=1, OO] - \mathbbm{E}[Y_0^\ast(0)|D=1, T=0, OO] \tag{Assumption \ref{no anti}}\\
    &- \mathbbm{E}[Y^\ast_1(0)|D=0, T=1, OO]+ \mathbbm{E}[Y^\ast_0(0)|D=0, T=0, OO] \\
    & = \mathbbm{E}[Y_1^\ast(1)-Y_1^\ast(0)|D=1, T=1, OO] = \tau_{OO} \tag{Assumption \ref{pt_rc}}
\end{align*}}
In order to bound $\tau_{OO}$ we can partially identify each of the four expectations in the second and third lines of the display above.
First, consider the group of individuals observed to be treated in the post-treatment period.
\begin{align}\label{mixp_rcs111}
    \mathbbm{E}[Y|D=1, T=1, S=1] &= \mathbbm{E}[Y^\ast_1(1)|D=1, T=1, OO]\cdot q_{OO11} \nonumber \\
    &+\mathbbm{E}[Y^\ast_1(1)|D=1, T=1, NO]\cdot (1-q_{OO11})
\end{align} 
Second, consider the group that is observed to be untreated in the post-treatment period. 
\begin{align}\label{mixp_rcs011}
    \mathbbm{E}[Y|D=0, T=1, S=1] &= \mathbbm{E}[Y^\ast_1(0)|D=0, T=1, OO]\cdot q_{OO01} \nonumber \\
    &+\mathbbm{E}[Y^\ast_1(0)|D=0, T=1, ON]\cdot (1-q_{OO01})
\end{align} where $q_{OO01} = \frac{\pi_{OO01}}{\pi_{OO01}+\pi_{ON01}}$.
While the two groups above do not point identify $\mathbbm{E}[Y^\ast_1(1)|D=1, T=1, OO]$ and $\mathbbm{E}[Y^\ast_1(0)|D=0, T=1, OO]$, they provide the opportunity to partially identify these quantities using the trimming procedure of \cite{lee2009}, as it will be shown below.

For the groups observed in the pre-treatment period, the desired expectations are point identified.
\begin{align*}
    \mathbbm{E}[Y|D=1, T=0, S=1] &= \mathbbm{E}[Y^\ast_0(1)|D=1, T=0, S_0(1)=1] \\
    & = \mathbbm{E}[Y^\ast_0(0)|D=1, T=0, OO] \tag{Assumption \ref{no anti}}
\end{align*}
Similarly, if we consider the untreated group observed in the pre-treatment period, then 
\begin{align*}
    \mathbbm{E}[Y|D=0, T=0, S=1] &= \mathbbm{E}[Y^\ast_0(0)|D=0, T=0, S_0(0)=1] \\
    & = \mathbbm{E}[Y^\ast_0(0)|D=0, T=0, OO] \tag{Assumption \ref{no anti}}.
\end{align*}
Finally, bound the $\mathbbm{E}[Y^\ast_1(1)|D=1, T=1, OO]$ by the same trimming strategy in the main text,
\begin{align*}
    LB_{OO11} &= \mathbbm{E}[Y|D=1, T=1, S=1, Y\leq F^{-1}_{Y|111}(q_{OO11})] \\
    UB_{OO11} &= \mathbbm{E}[Y|D=1, T=1, S=1, Y>F^{-1}_{Y|111}(1-q_{OO11})]
\end{align*} where $F^{-1}_{Y|111}(\cdot)$ is the quantile function of the distribution of $Y$ given $D=1$, $T=1$, $S=1$. Similarly, we can obtain the lower and upper bounds for $\mathbbm{E}[Y^\ast_1(0)|D=0, T=1, OO]$ as
\begin{align*}
    LB_{OO01} &= \mathbbm{E}[Y|D=0, T=1, S=1, Y\leq F^{-1}_{Y|011}(q_{OO01})] \\
    UB_{OO01} &= \mathbbm{E}[Y|D=0, T=1, S=1, Y>F^{-1}_{Y|011}(1-q_{OO01})].
\end{align*} Combining the bounds for $\mathbbm{E}[Y^\ast_1(0)|D=0, T=1, OO]$ and $\mathbbm{E}[Y^\ast_1(1)|D=1, T=1, OO]$ along with point-identification of the other two expectations, we can say that $\tau_{OO}$ lies in the interval
\begin{align}
    &\big[LB_{OO11}-UB_{OO01}-\mathbbm{E}[Y|D=1, T=0, S=1]+\mathbbm{E}[Y|D=0, T=0, S=1], \nonumber \\
    &UB_{OO11}-LB_{OO01}-\mathbbm{E}[Y|D=1, T=0, S=1]+\mathbbm{E}[Y|D=0, T=0, S=1]\big]. \label{UB_tauOO}
\end{align}
The remaining challenge is the identification of the mixing weights for the trimming procedure, which we tackle under different assumptions below.
\subsection{Identification of Mixing Weights and $\tau_{OO}$ without Monotonicity}
We now shift attention to identifying the mixing weights $q_{OO11}$ and $q_{OO01}$ from the observed data. Consider the probability, 
{\small \begin{align*}
    \mathbbm{P}[S=1|D=1, T=1] &= \mathbbm{P}[S_1(1)=1|D=1, T=1]\\
    & = \mathbbm{P}[S_1(0)=1, S_1(1)=1|D=1, T=1]+\mathbbm{P}[S_1(0)=0, S_1(1)=1|D=1, T=1] \\
    & = \frac{\pi_{OO11}}{\mathbbm{P}[D=1, T=1]}+\frac{\pi_{NO11}}{\mathbbm{P}[D=1, T=1]} 
\end{align*}}
Next, consider the probability
{\small \begin{align*}
    \mathbbm{P}[S=1|D=0, T=1] &= \mathbbm{P}[S_1(0)=1|D=0, T=1]\\
    & = \mathbbm{P}[S_1(0)=1, S_1(1)=1|D=0, T=1]+\mathbbm{P}[S_1(0)=1, S_1(1)=0|D=0, T=1]\\
     & = \frac{\pi_{OO01}}{\mathbbm{P}[D=0, T=1]}+\frac{\pi_{NO01}}{\mathbbm{P}[D=0, T=1]}
\end{align*}}
The probability $\mathbbm{P}[S_1(0)=1, S_1(1)=1|D=d, T=1]$ can be partially identified using the following Fr\'echet bounds:
{\small \begin{align*}
    \mathbbm{P}[S_1(0)=1, S_1(1)=1|D=d, T=1] &\in [\max\big\{\mathbbm{P}[S_1(0)=1|D=d, T=1]+\mathbbm{P}[S_1(1)=1|D=d, T=1]-1, 0\big\}, \\
    &\min\big\{\mathbbm{P}[S_1(0)=1|D=d, T=1], \mathbbm{P}[S_1(1)=1|D=d, T=1]\big\} ].
\end{align*}}
Note that the probabilities of potential selection into the sample $\mathbbm{P}[S_1(0)=1|D=0, T=1]$ and $\mathbbm{P}[S_1(1)=1|D=1, T=1]$ can be directly point identified from the observed data. However, we still need to identify their counterfactual versions $\mathbbm{P}[S_1(0)=1|D=1, T=1]$ and $\mathbbm{P}[S_1(1)=1|D=0, T=1]$, respectively. Similarly to the case for panel data discussed in Section \ref{identification}, we consider assumptions about the relationship of the potential sample selection across treatment groups.
\begin{assumption}[Ignorability of treatment in potential selection]\label{Partialselection_rcs} \
    \begin{itemize}
        \item[(a)] $\mathbbm{P}[S_1(0)=1|D=1, T=1] = \mathbbm{P}[S_1(0)=1|D=0,T=1]$. OR parallel trends in selection as follows:
        $\mathbbm{P}[S_1(0)=1|D=1, T=1]-\mathbbm{P}[S_0(0)=1|D=1, T=0] = \mathbbm{P}[S_1(0)=1|D=0,T=1]-\mathbbm{P}[S_0(0)=1|D=0,T=0]$. 
        \item[(b)] $\mathbbm{P}[S_1(1)=1|D=1, T=1] = \mathbbm{P}[S_1(1)=1|D=0,T=1]$. OR parallel trends in selection as follows:
        $\mathbbm{P}[S_1(1)=1|D=1, T=1]-\mathbbm{P}[S_0(1)=1|D=1, T=0] = \mathbbm{P}[S_1(1)=1|D=0,T=1]-\mathbbm{P}[S_0(1)=1|D=0,T=0]$. 
    \end{itemize} 
\end{assumption}
Note that Assumption \ref{Partialselection_rcs}(a) allows us to identify $\mathbbm{P}[S_1(0)=1|D=1, T=1]$ as $\mathbbm{P}[S_1(0)=1|D=0, T=1]=\mathbbm{P}[S_1=1|D=0, T=1]$, or if we use the parallel trends in selection version of the assumption, 
\begin{align*}
  \mathbbm{P}[S_1(0)=1|D=1, T=1] & = \mathbbm{P}[S_1(0)=1|D=0,T=1]-\mathbbm{P}[S_0(0)=1|D=0,T=0] \\
  & +\mathbbm{P}[S_0(0)=1|D=1, T=0] \\
  &= \mathbbm{P}[S_1=1|D=0,T=1]-\mathbbm{P}[S_0=1|D=0,T=0]\\
  &+\mathbbm{P}[S_0=1|D=1, T=0].  \tag{Assumption \ref{no anti}}
\end{align*}
Similarly, Assumption \ref{Partialselection_rcs}(b) identifies $\mathbbm{P}[S_1(1)=1|D=0, T=1]$ as $\mathbbm{P}[S_1(1)=1|D=1, T=1]=\mathbbm{P}[S_1=1|D=1, T=1]$ or using the parallel trends in selection, we get
\begin{align*}
  \mathbbm{P}[S_1(1)=1|D=0, T=1] &= \mathbbm{P}[S_1(1)=1|D=1, T=1]-\mathbbm{P}[S_0(1)=1|D=1, T=0]\\
  &+\mathbbm{P}[S_0(1)=1|D=0,T=0] \\
  & = \mathbbm{P}[S_1=1|D=1, T=1]-\mathbbm{P}[S_0=1|D=1, T=0]\\
  &+\mathbbm{P}[S_0=1|D=0,T=0]. \tag{Assumption \ref{no anti}}
\end{align*}
This leads to the partial identification of $\mathbbm{P}[S_1(0)=1, S_1(1)=1|D=d, T=1]$ as follows, depending on which version of the assumption we impose:
\begin{lemma}\label{lemma:pipartialid_rcs}
    (a) Under Assumptions \ref{no anti}
    and \ref{Partialselection_rcs}(a), 
    {\footnotesize
     \begin{align}
        \mathbbm{P}[S_1(0)=1, S_1(1)=1|D=1, T=1] &\in [\max\big\{\mathbbm{P}[S_1=1|D=0,T=1]+\mathbbm{P}[S_1=1|D=1, T=1]-1, 0\big\}, \nonumber \\
        &\min\big\{\mathbbm{P}[S_1=1|D=0,T=1], \mathbbm{P}[S_1=1|D=1, T=1]\big\} ]
    \end{align} 
    OR
    \begin{align}
        \mathbbm{P}[S_1(0)=1, S_1(1)=1|D=1, T=1] &\in [\max\big\{\mathbbm{P}[S_1=1|D=0,T=1]-\mathbbm{P}[S_0=1|D=0,T=0] \nonumber \\
        &+\mathbbm{P}[S_0=1|D=1, T=0]+\mathbbm{P}[S_1=1|D=1, T=1]-1, 0\big\}, \nonumber\\
        & \min\big\{\mathbbm{P}[S_1=1|D=0,T=1]-\mathbbm{P}[S_0=1|D=0,T=0] \nonumber \\
        &+\mathbbm{P}[S_0=1|D=1, T=0], \mathbbm{P}[S_1=1|D=1, T=1]\big\} ]
    \end{align}}
    (b) Under Assumptions \ref{no anti}
    and \ref{Partialselection_rcs}(b), 
    {\footnotesize \begin{align}
         \mathbbm{P}[S_1(0)=1, S_1(1)=1|D=0, T=1] &\in [\max\big\{\mathbbm{P}[S_1=1|D=0, T=1]+\mathbbm{P}[S_1=1|D=1, T=1]-1, 0\big\}, \nonumber \\
         &\min\big\{\mathbbm{P}[S_1=1|D=0, T=1], \mathbbm{P}[S_1=1|D=1, T=1]\big\} ]
         \end{align}
         OR
         \begin{align}
         \mathbbm{P}[S_1(0)=1, S_1(1)=1|D=0, T=1] &\in [\max\big\{\mathbbm{P}[S_1=1|D=0, T=1]+\mathbbm{P}[S_1=1|D=1, T=1] \nonumber\\
         &-\mathbbm{P}[S_0=1|D=1, T=0]+\mathbbm{P}[S_0=1|D=0,T=0]-1, 0\big\}, \nonumber \\
         &\min\big\{\mathbbm{P}[S_1=1|D=0, T=1], \mathbbm{P}[S_1=1|D=1, T=1] \nonumber \\
         &-\mathbbm{P}[S_0=1|D=1, T=0]+\mathbbm{P}[S_0=1|D=0,T=0]\big\} ]
         \end{align}}
\end{lemma}
Proof of Lemma \ref{lemma:pipartialid_rcs} can be found in Appendix \ref{proof:pipartialid_rcs}. Similar to the panel data case, we can use Lemma \ref{lemma:pipartialid_rcs} to bound $q_{OO11}$ and $q_{OO01}$. Let $\nu_{d}$ denote any value in the identified set for $\mathbbm{P}[S_1(0)=1, S_1(1)=1|D=d, T=1]$ and $q_{gd1}(\nu_d) = \frac{\nu_d}{\mathbbm{P}[S_1(d)=1|D=d, T=1]}$. Then, the least favorable values of $LB_{OO11}$ and $UB_{OO11}$ are obtained at the lowest value of the proportion of always-observed among the treated in the post-treatment period, denoted $\nu_{1}^{l}$. These are given by
\begin{align}
    LB_{OO11}(\nu^l_{1}) &= \mathbbm{E}[Y|D=1, T=1, S=1, Y\leq F^{-1}_{Y|111}(q_{OO11}(\nu^l_{1}))] \label{LB1_rcs}\\
    UB_{OO11}(\nu^l_{1}) &= \mathbbm{E}[Y|D=1, T=1, S=1, Y>F^{-1}_{Y|111}(1-q_{OO11}(\nu^l_{1}))] \label{UB1_rcs}
\end{align}
Similarly, the least favorable values of $LB_{OO01}$ and $UB_{OO01}$ are obtained at the lowest value of the proportion of always-observed among the untreated in the post-treatment period $q_{OO01}$, which we denote $\nu_{0}^{l}$. These are given by
\begin{align}
    LB_{OO01}(\nu^l_{0}) &= \mathbbm{E}[Y|D=0, T=1, S=1, Y\leq F^{-1}_{Y|011}(q_{OO01}(\nu^l_{0}))] \label{LB0_rcs} \\
    UB_{OO01}(\nu^l_{0}) &= \mathbbm{E}[Y|D=0, T=1, S=1, Y>F^{-1}_{Y|011}(1-q_{OO01}(\nu^l_{0}))] \label{UB0_rcs}
\end{align}
Combining \eqref{LB1_rcs}-\eqref{UB0_rcs} together with the bounds for $\tau_{OO}$ given in \eqref{UB_tauOO} along with the identified sets for the probabilities, we can formally state the partial identification results for $\tau_{OO}$. 
\begin{theorem}[Bounds for $\tau_{OO}$]\label{nomono_bound_rcs} Under Assumptions \ref{no anti},\ref{No_com}, \ref{pt_rc}, \ref{Partialselection_rcs}(a), \ref{Partialselection_rcs}(b), and $q_{OOd1}>0$, bounds on the treatment effect on the treated for the always-observed group in the post-treatment period ($\tau_{OO}$) lies in the interval $[LB_{\tau_{OO}}, UB_{\tau_{OO}}]$ where,
{\footnotesize\begin{equation*}
    \begin{split}
        LB_{\tau_{OO}} &=\mathbbm{E}[Y|D=1, T=1, S=1, Y\leq F^{-1}_{Y|111}(q_{OO11}(\nu^l_{1}))]-\mathbbm{E}[Y|D=0, T=1, S=1, Y\leq F^{-1}_{Y|011}(q_{OO01}(\nu^l_{0}))]\\
        &-\mathbbm{E}[Y|D=1, T=0, S=1]+\mathbbm{E}[Y|D=0, T=0, S=1],\\
        UB_{\tau_{OO}} &= \mathbbm{E}[Y|D=1, T=1, S=1, Y>F^{-1}_{Y|111}(1-q_{OO11}(\nu^l_{1}))]-\mathbbm{E}[Y|D=0, T=1, S=1, Y\leq F^{-1}_{Y|011}(q_{OO01}(\nu^l_{0}))]\\
        &-\mathbbm{E}[Y|D=1, T=0, S=1]+\mathbbm{E}[Y|D=0, T=0, S=1]
    \end{split}
\end{equation*}}
where, {\scriptsize \begin{align*}
    q_{OO11}(\nu^{l}_{1})&=\frac{\max\{\mathbbm{P}[S_1=1|D=0,T=1]+\mathbbm{P}[S_1=1|D=1, T=1]-1, 0\}}{\mathbbm{P}[S_{1}=1| D=1, T=1]},\\
    q_{OO01}(\nu^{l}_{0})&=\frac{\max\{\mathbbm{P}[S_1=1|D=0, T=1]+\mathbbm{P}[S_1=1|D=1, T=1]-1, 0\}}{\mathbbm{P}[S_{1}=1| D=0, T=1]}.\\
    OR \\
    q_{OO11}(\nu^{l}_{1})&=\frac{\max\big\{\mathbbm{P}[S_1=1|D=0,T=1]-\mathbbm{P}[S_0=1|D=0,T=0]+\mathbbm{P}[S_0=1|D=1, T=0]
    +\mathbbm{P}[S_1=1|D=1, T=1]-1, 0\big\}}{\mathbbm{P}[S_{1}=1| D=1, T=1]},\\
    q_{OO01}(\nu^{l}_{0})&=\frac{\max\{\mathbbm{P}[S_1=1|D=0, T=1]+\mathbbm{P}[S_1=1|D=1, T=1]-\mathbbm{P}[S_0=1|D=1, T=0]+\mathbbm{P}[S_0=1|D=0,T=0]-1, 0\}}{\mathbbm{P}[S_{1}=1| D=0, T=1]}.
\end{align*}}
\end{theorem}
The proof of Theorem \ref{nomono_bound_rcs} is given in Appendix \ref{proof:nomono_bound_rcs}.

\subsection{Identification of Mixing Weights and $\tau_{OO}$ with Monotonicity}

In this subsection, we consider the identification of the mixing weights and $\tau_{OO}$ in the repeated cross-section case under the monotonicity in selection assumption. Without loss of generality, assume positive monotonicity, which rules out the `ON' stratum for those who are observed to be untreated in the post-treatment period. This helps to point identify the mixing weights as described in the result below.
\begin{lemma}\label{lemma:mono_weights_rcs}  Under Assumptions \ref{no anti} and \ref{monotone} (positive monotonicity), we obtain $q_{OO11}= \frac{\mathbbm{P}[S_1(0)=1|D=1, T=1]}{\mathbbm{P}[S_1(1)=1|D=1, T=1]}$ and $q_{OO01} = 1$.
\end{lemma}
Proof of Lemma \ref{lemma:mono_weights_rcs} is given in Appendix \ref{proof:mono_weights_rcs}. Given the results in Lemma \ref{lemma:mono_weights_rcs}, $\mathbbm{E}[Y_1^\ast(0)|D=0, T=1, OO]$ is the point identified since the group that is untreated in the post-treatment period only contains the OO subgroup. This also means that we no longer require Assumption \ref{Partialselection_rcs}(b) since $q_{OO01}=1$ under positive monotonicity.

When it comes to identifying $q_{OO11}= \frac{\mathbbm{P}[S_1(0)=1|D=1, T=1]}{\mathbbm{P}[S_1(1)=1|D=1, T=1]}$, MS does not provide enough information, even though it implies that $0<\mathbbm{P}[S_1(0)=1|D=1, T=1] \leq \mathbbm{P}[S_1(1)|D=1, T=1]= \mathbbm{P}[S_1|D=1, T=1]$. If we additionally assume Assumption \ref{Partialselection_rcs}(a) then $q_{OO11} = \frac{\mathbbm{P}[S_1=1|D=0, T=1]}{\mathbbm{P}[S_1=1|D=1, T=1]}$ OR 
{\small \begin{equation*}
    \begin{split}
        q_{OO11} &= \frac{\mathbbm{P}[S_1(0)=1|D=0,T=1]-\mathbbm{P}[S_0(0)=1|D=0,T=0]+\mathbbm{P}[S_0(0)=1|D=1, T=0]}{\mathbbm{P}[S_1=1|D=1, T=1]} \\
        & = \frac{\mathbbm{P}[S_1=1|D=0,T=1]-\mathbbm{P}[S_0=1|D=0,T=0]+\mathbbm{P}[S_0=1|D=1, T=0]}{\mathbbm{P}[S_1=1|D=1, T=1]}
    \end{split}
\end{equation*}}

We therefore propose bounds for $\tau_{OO}$ under monotonicity as follows: 
\begin{theorem}[Bounds for $\tau_{OO}$ under positive monotonicity]\label{mono_bound_rcs} Under Assumptions \ref{no anti},\ref{monotone},\ref{No_com}, \ref{pt_rc}, \ref{Partialselection_rcs}(a), and $q_{OO11}>0$, bounds on the treatment effect on the treated for the always-observed group ($\tau_{OO}$) lies in the interval $[LB_{\tau_{OO}}, UB_{\tau_{OO}}]$ where,
{\footnotesize\begin{equation*}
    \begin{split}
        LB_{\tau_{OO}} &= LB_{OO11}-\mathbbm{E}[Y|D=0, T=1, S=1]-\mathbbm{E}[Y|D=1, T=0, S=1]+\mathbbm{E}[Y|D=0, T=0, S=1],\\
        UB_{\tau_{OO}} &=UB_{OO11}-\mathbbm{E}[Y|D=0, T=1, S=1]-\mathbbm{E}[Y|D=1, T=0, S=1]+\mathbbm{E}[Y|D=0, T=0, S=1]
    \end{split}
\end{equation*}}
where, 
\begin{align*}
    LB_{OO11} &=\mathbbm{E}[Y|D=1,T=1, S=1, Y\leq F_{Y|111}^{-1}(q_{OO11})]\\
    UB_{OO11} &=\mathbbm{E}[Y|D=1,T=1, S=1, Y>F_{Y|111}^{-1}(1-q_{OO11})]
\end{align*}
where $q_{OO11}=\frac{\mathbbm{P}[S_{1}=1|D=0, T=1]}{\mathbbm{P}[S_{1}=1|D=1, T=1]}$ OR $q_{OO11} = \frac{\mathbbm{P}[S_1=1|D=0,T=1]-\mathbbm{P}[S_0=1|D=0,T=0]+\mathbbm{P}[S_0=1|D=1, T=0]}{\mathbbm{P}[S_1=1|D=1, T=1]}$.
\end{theorem}
The proof of Theorem \ref{mono_bound_rcs} is given in Appendix \ref{proof Theorem mono_bound_rcs}.

\section{Extension to Staggered Adoption}\label{sec:extension_multi}
Recent developments in the DiD literature have highlighted the challenges and opportunities associated with having treatment rolled out in a staggered fashion, that is, in which different parts of the treated population receive treatment at different periods \citep{callaway2021difference,goodman2021difference, sun2021estimating, borusyak2024revisiting}.
While a full consideration of endogenous sample selection for DiD with staggered treatment adoption is beyond the scope of this paper, the framework developed in Section \ref{model} can be usefully adapted to comparisons of periods ($2\times 2$) in the staggered setting.

Consider an arbitrary number of time periods, $t=0, 1,\ldots, T$. Following this literature, we assume that the treatment is irreversible. Let $D_t$ denote treatment status at time period $t$. 
\begin{assumption}[Irreversibility of treatment]\label{Irr} Assume $D_0=0$ and for $t=1,\ldots, T$, $D_{t-1} = 1$ implies $D_{t}=1$.
\end{assumption}
Assumption \ref{Irr} states that no one is treated in period 0, and once a unit gets treated, it remains treated in all subsequent periods. In other words, treatment adoption is staggered.  Let $\Gamma_\gamma$ be a binary variable that is equal to one if the unit is first treated in period $\gamma$, and define $C$ to be a binary variable that is equal to one if the unit is not treated in any period (i.e. never treated). We define the latent type based on selection in period 0, and the pair of counterfactual selection states in any post-treatment period $t\geq \gamma$, that is, by $(S_0 = s, S_t(0)=s',S_t(1)=s'')$. Let $G_t$ denote this latent type and $g$ be the type denomination. 

The parameter of interest, then, is the average treatment effect for the latent type $g$ in period $t$ among units that are first treated in period $\gamma$.
\begin{equation}\label{att_stag}
    \tau_{g,\gamma,t}=\mathbbm{E}[Y_{t}^{\ast}(\gamma)-Y_{t}^{\ast}(0)| G_t=g ,\Gamma_{\gamma}=1]
\end{equation} Notice that equation \eqref{att_stag} is \citet{callaway2021difference}'s disaggregated group-time average treatment effect parameter, which has been extended to accommodate sample selection. Therefore, in the current setting, ``group'' is a joint indication of the first time a unit gets treated and its latent type based on the tuple $(S_0 = s, S_t(0)=s',S_t(1)=s'')$. There are two different types of treatment effect dynamics that one can explore in this setup. The first is to study the evolution of ATT for the OOO group first treated in period $\gamma$ across subsequent post-treatment periods. For example, with three periods (i.e. $t=0,1,2$) and for the group that gets treated in period 1, one can identify ATT for OOOs, namely, $\tau_{OOO,\gamma=1,t=1}$,$\tau_{OOO,\gamma=1,t=2}$ and $\tau_{OOO,\gamma=1,t=3}$. The second type of dynamic effect examines the cross-sectional variation in the ATTs for OOOs across groups first treated in distinct post-treatment periods $\gamma$, during overlapping post-treatment period windows. For example, with two adoption timing groups $\gamma = 1, 2$ and three periods, $t=0,1,2$, one can study variation in $\tau_{OOO,\gamma=1,t=2}$ and $\tau_{OOO,\gamma=2,t=2}$. Note that the OOOs may be compositionally very different across calendar time and across the two adoption timing groups.

Following the same logic as for the canonical case discussed in Assumption \ref{PT_OOO}, consider the following parallel trends assumption.
\begin{assumption}[Conditional parallel trend for outcome for latent group $(g, \gamma)$ in period $t$]\label{stag_parallel}
\begin{equation*}
  \mathbbm{E}[Y_{t}^{\ast}(0)-Y_{0}^{\ast}(0)|G_t=g ,\Gamma_{\gamma}=1]=\mathbbm{E} [Y_{t}^{\ast}(0)-Y_{0}^{\ast}(0)|G_t=g,C=1] \text{ such that } t\geq \gamma.
\end{equation*}
\end{assumption}
Assumption \eqref{stag_parallel} generalizes the parallel trend assumption to a setting with multiple time periods and states that the evolution of trends in the untreated potential outcomes between any post-treatment period $t$ and period $0$ are the same for latent type $g$ between units that first receive treatment in period $\gamma$ and those that are never treated.

Assume ignorability of selection between the group that first receives treatment in period $\gamma$, $\Gamma_\gamma=1$ and the never-treated group, $C=1$. 
\begin{assumption}\label{Partialselection_stag} 
 Ignorability of selection in period $t\geq \gamma$:
\\
     (a) Equality in the untreated counterfactual share of individuals observed in period $t\geq \gamma$ conditional on selection status for period 0. 
	\begin{align}
    \mathbbm{P}[S_{t}(0)=1|C=1, S_0]=\mathbbm{P}[S_{t}(0)=1|\Gamma_\gamma=1, S_{0}].
    \end{align}
 \\   
    (b) Equality in the treated counterfactual share of individuals observed in period $t\geq \gamma$ conditional on selection status for period 0. 
	\begin{align}
    \mathbbm{P}[S_{t}(1)=1|\Gamma_\gamma=1, S_{0}]=\mathbbm{P}[S_{t}(1)=1|C=1, S_{0}]
    \end{align}
  
\end{assumption}
Then, under assumptions \ref{no anti}, \ref{Irr}, \ref{stag_parallel}, \ref{Partialselection_stag}(a), \ref{Partialselection_stag}(b), $p_{OOO,\Gamma_\gamma, t}(v^{l}_{1t})>0$ and $p_{OOO, C, t}(v^{l}_{0t})>0$, one can bound the treatment effect for the always-observed group among those first treated in period $\gamma$ in any post-treatment period $t\geq \gamma$ ($\tau_{OOO, \gamma, t}$). In other words,  $\tau_{OOO, \gamma, t}$ lies in the interval $[LB_{\tau_{OOO, \gamma, t}}, UB_{\tau_{OOO}, \gamma, t}]$,
\begin{equation*}
    \begin{split}
        LB_{\tau_{OOO,\gamma, t}} &=\mathbbm{E}[Y_{t}-Y_{0}|\Gamma_\gamma=1,S_{0}=1,S_{t}=1, (Y_{t}-Y_{0})\leq F_{\Delta Y|111}^{-1}(p_{OOO,\Gamma_\gamma, t}(v^{l}_{1t}))]\\
        &-\mathbbm{E}[Y_{t}-Y_{0}|C=1,S_{0}=1,S_{t}=1,(Y_{t}-Y_{0}) > F_{\Delta Y|011}^{-1}(1-p_{OOO, C, t}(v^{l}_{0t}))],\\
        UB_{\tau_{OOO,\gamma, t}} &= \mathbbm{E}[Y_{t}-Y_{0}|\Gamma_\gamma=1,S_{0}=1,S_{t}=1, (Y_{t}-Y_{0}) > F_{\Delta Y|111}^{-1}(1-p_{OOO, \Gamma_\gamma, t}(v^{l}_{1t}))]\\
        &-\mathbbm{E}[Y_{t}-Y_{0}|C=1,S_{0}=1,S_{t}=1, (Y_{t}-Y_{0})\leq F_{\Delta Y|011}^{-1}(p_{OOO, C, t}(v^{l}_{0t}))]
    \end{split}
\end{equation*}
where, 
\begin{align*}
    p_{OOO,\Gamma_\gamma, t}(v^{l}_{1t})=\frac{\max\{\mathbbm{P}[S_{t}=1|C=1, S_{0}=1]+\mathbbm{P}[S_{t}=1|\Gamma_\gamma=1, S_{0}=1]-1,0\}}{\mathbbm{P}[S_{t}=1| \Gamma_\gamma=1, S_{0}=1]},\\
    p_{OOO, C, t}(v^{l}_{0t})=\frac{\max\{\mathbbm{P}[S_{t}=1|C=1, S_{0}=1]+\mathbbm{P}[S_{t}=1|\Gamma_\gamma=1, S_{0}=1]-1,0\}}{\mathbbm{P}[S_{t}=1| C=1, S_{0}=1]}.
\end{align*}

A similar approach could be used if one is willing to impose monotonicity (Assumption \ref{monotone} for post-treatment period $t\geq \gamma$) to implement the bounds in Theorem \ref{mono_bound}. Note that the relevant proportions for the latent groups in period $t$ are obtained by comparing the shares of observed individuals in the initial period and those in the post-treatment period $t$. Without loss of generality, if we assume positive Monotonicity, then under assumptions \ref{no anti}, \ref{Irr},\ref{stag_parallel}, \ref{Partialselection_stag}(a), and $p_{OOO,\Gamma_\gamma,t}>0$, we have $\tau_{OOO, \gamma, t}$ lies in the interval $[LB_{\tau_{OOO, \gamma, t}}, UB_{\tau_{OOO}, \gamma, t}]$,
\begin{equation*}
    \begin{split}
        LB_{\tau_{OOO,\gamma, t}} &= LB_{OOO,\Gamma_\gamma,t}-\mathbbm{E}[Y_{t}-Y_{0}|C=1,S_{0}=1,S_{t}=1],\\
        UB_{\tau_{OOO,\gamma, t}} &=UB_{OOO,\Gamma_\gamma,t}-\mathbbm{E}[Y_{t}-Y_{0}|C=1,S_{0}=1,S_{t}=1]
    \end{split}
\end{equation*}
where, 
\begin{align*}
    LB_{OOO,\Gamma_\gamma,t} &=\mathbbm{E}[Y_{t}-Y_{0}|\Gamma_\gamma=1,S_{0}=1,S_{t}=1, (Y_{t}-Y_{0})\leq F_{\Delta Y|111}^{-1}(p_{OOO,\Gamma_\gamma,t})]\\
    UB_{OOO,\Gamma_\gamma,t} &=\mathbbm{E}[Y_{t}-Y_{0}|\Gamma_\gamma=1,S_{0}=1,S_{t}=1, (Y_{t}-Y_{0}) > F_{\Delta Y|111}^{-1}(1-p_{OOO,\Gamma_\gamma,t})]
\end{align*}
where $p_{OOO,\Gamma_\gamma,t}=\frac{\mathbbm{P}[S_{t}=1|S_{0}=1,C=1]}{\mathbbm{P}[S_{t}=1|S_{0}=1,\Gamma_\gamma=1]}$.
\section{Relationship between Trimmed Bound Estimands and Untrimmed ``na\"{\i}ve'' Estimates }\label{Appendix: Trimmed}

This appendix formalizes the intuition that the bounds for the $\tau_{OOO}$ proposed in Section \ref{identification} will always include the ``na\"{\i}ve'' $\tau_{\textup{DiDs}}$. This is because $\tau_{\textup{DiDs}}$ is obtained by comparing expectations (means) of functions of the outcomes, while the bounds for $\tau_{OOO}$ are based on trimmed means of the same functions of those variables. Since the overall mean will always be between the trimmed means, $\tau_{\textup{DiDs}}$ will always be in the identified set for $\tau_{OOO}$. This is not unique to the DiD case and applies more generally to Lee-type trimming bounds.
\subsection{The Relationship between Trimmed and Untrimmed Means}
First, consider the expected value of a variable $Y$ for a population of interest:
\begin{equation}
\mathbbm{E}[Y]=\int^{\infty}_{-\infty}y f_{Y}(y)dy 
\end{equation}

If we trim the top (bottom) of the distribution of $Y$ at an arbitrary quantile $p$, we can obtain the trimmed mean that excludes the highest (lowest) $p\%$ of that population:
\begin{eqnarray}
\mathbbm{E}[Y|Y<y_{1-p}]=\int^{y_{1-p}}_{-\infty}y f_{Y|Y<y_{1-p}}(y)dy=\int^{y_{1-p}}_{-\infty}y \frac{f_{Y}(y)}{(1-p)}dy \\
\mathbbm{E}[Y|Y>y_{p}]=\int^{\infty}_{y_{p}}y f_{Y|Y>y_{p}}(y)dy=\int^{\infty}_{y_{p}}y \frac{f_{Y}(y)}{(1-p)}dy
\end{eqnarray}
where $y_{p}$ is the $p^{th}$ quantile of the distribution of $Y$.
Naturally, we can describe the relationship between the trimmed and untrimmed variables as,
\begin{eqnarray*}
\mathbbm{E}[Y]=\int^{y_{1-p}}_{-\infty}y f_{Y}(y)dy+\int_{y_{1-p}}^{\infty}y f_{Y}(y)dy \\
\geq\int^{y_{1-p}}_{-\infty}y f_{Y}(y)dy+\int_{y_{1-p}}^{\infty}y_{1-p} f_{Y}(y)dy \\
=\int^{y_{1-p}}_{-\infty}y f_{Y}(y)dy+y_{1-p}\int_{y_{1-p}}^{\infty}f_{Y}(y)dy \\
=(1-p)\int^{y_{1-p}}_{-\infty}y \frac{f_{Y}(y)}{(1-p)}dy+py_{1-p} \\
\geq(1-p)\mathbbm{E}[Y|Y<y_{1-p}]+p\mathbbm{E}[Y|Y<y_{1-p}]\\
=\mathbbm{E}[Y|Y<y_{1-p}]
\end{eqnarray*}
Similarly,
\begin{eqnarray*}
\mathbbm{E}[Y]=\int^{y_{p}}_{-\infty}y f_{Y}(y)dy+\int_{y_{p}}^{\infty}y f_{Y}(y)dy \\
\leq\int^{y_{p}}_{-\infty}y_{p} f_{Y}(y)dy+\int_{y_{p}}^{\infty}y f_{Y}(y)dy \\
=y_{p}\int^{y_{p}}_{-\infty}f_{Y}(y)dy+\int_{y_{p}}^{\infty}y f_{Y}(y)dy \\
=py_{p}+(1-p)\int_{y_{p}}^{\infty}y \frac{f_{Y}(y)}{(1-p)}dy \\
\leq p\mathbbm{E}[Y|Y>y_{p}]+(1-p)\mathbbm{E}[Y|Y>y_{p}]\\
=\mathbbm{E}[Y|Y>y_{p}]
\end{eqnarray*}

Hence,
\begin{equation}
   \mathbbm{E}[Y|Y<y_{1-p}] \leq \mathbbm{E}[Y]\leq \mathbbm{E}[Y|Y>y_{p}].
\end{equation}

\subsection{Relationship to Lee (2009) Bounds Estimands}

The ``na\"{\i}ve'' estimands for the ATT/ATE are given by expectations over the observed population for a particular variable. For example, under the canonical DiD assumptions, the DiD estimand is given by:
\begin{eqnarray}
    \tau_{\textup{DiDs}}=\mathbbm{E}[Y_1-Y_0|D=1, S_0=1, S_1=1] - \mathbbm{E}[Y_1-Y_0|D=0,S_0=1, S_1=1]
\end{eqnarray}
As seen in Section \ref{identification}, the upper and lower bounds for $\tau_{OOO}$ are combinations of trimmed means of the same variable, $(Y_1-Y_0)$. Focusing on the more general result in Theorem \ref{nomono_bound}:
 \begin{align*}
   LB_{OOO1} &=\mathbbm{E}[Y_{1}-Y_{0}|D=1,S_{0}=1,S_{1}=1, (Y_{1}-Y_{0})\leq F_{(Y_{1}-Y_{0}) Y|111}^{-1}(p_{OOO1})]\\
   UB_{OOO1} &=\mathbbm{E}[Y_{1}-Y_{0}|D=1,S_{0}=1,S_{1}=1, (Y_{1}-Y_{0}) > F_{(Y_{1}-Y_{0})|111}^{-1}(1-p_{OOO1})]\\
   LB_{OOO0} &=\mathbbm{E}[Y_{1}-Y_{0}|D=0,S_{0}=1,S_{1}=1,(Y_{1}-Y_{0})\leq F_{(Y_{1}-Y_{0})|011}^{-1}(p_{OOO0})]\\
   UB_{OOO0} &=\mathbbm{E}[Y_{1}-Y_{0}|D=0,S_{0}=1,S_{1}=1,(Y_{1}-Y_{0}) > F_{(Y_{1}-Y_{0})|011}^{-1}(1-p_{OOO0})]\\
   \tau_{OOO} &\in [LB_{OOO1}-UB_{OOO0},UB_{OOO1}-LB_{OOO0}].
\end{align*}

Based on the discussion above, we know that 
 \begin{eqnarray}
  LB_{OOO1} \leq\mathbbm{E}[Y_1-Y_0|D=1, S_0=1, S_1=1]\leq  UB_{OOO1} \nonumber\\
  LB_{OOO0} \leq \mathbbm{E}[Y_1-Y_0|D=0,S_0=1, S_1=1]\leq UB_{OOO0} \nonumber \\ \nonumber \\
  \implies LB_{OOO1}-UB_{OOO0}\leq \tau_{\textup{DiDs}} \leq UB_{OOO1}-LB_{OOO0}
 \end{eqnarray}
In practice, if one is to use the sample analogue of these estimands to estimate them, the $\hat{\tau}_{\textup{DiDs}}$ will be within the estimated identified set for $\tau_{OOO}$ in any given dataset.

It is important to keep in mind that this discussion does NOT imply that the unconditional ATT is necessarily within the identified set for $\tau_{OOO}$. Intuitively, if there is no sample selection, the ATT and $\tau_{OOO}$ are point identified and their estimands will coincide. When we introduce sample selection, $\tau_{\textup{DiDs}}$ becomes a weighted average of $\tau_{OOO}$ and other terms, as discussed in Section \ref{Canonicalbias}. Interestingly, the sample selection drives the na\"{\i}ve DiD estimand away from $\tau_{OOO}$ slower than the bounds for it expand, leading it to be always within the identified set for $\tau_{OOO}$. Therefore, one cannot hope to rule out the estimated value, $\hat{\tau}_{\textup{DiDs}}$ by considering the identified set for $\tau_{OOO}$ (or any set for which the bounds are trimmed averages of the same variables that are used untrimmed to construct $\tau_{\textup{DiDs}}$).

\section{Identification of $\tau_{OOO}$ with Covariates}\label{sec: covariates}
We start by relaxing the parallel trend assumption to its conditional analogue. 
\begin{assumption}[Conditional parallel trends and overlap for OOO]\label{PT_OOO_cond} \ 
\begin{enumerate}
    \item[(a)] Conditional parallel trends in outcomes:
    \begin{equation}\label{stag}
        \mathbbm{E}[Y_{1}^{\ast}(0)-Y_{0}^{\ast}(0)|D=1, X, OOO]=\mathbbm{E} [Y_{1}^{\ast}(0)-Y_{0}^{\ast}(0)|D=0, X, OOO]. 
    \end{equation} 
    \item[(b)] Overlap: $\pi_{OOO1}>0$ and $\pi_{OOO0}(X)>0 \text{, } \forall \text{ } X\in \mathcal{X}$
\end{enumerate}
where $\pi_{OOO1}(X)= \mathbbm{P}[S_0(0)=1,S_1(0)=1,S_1(1)=1,D=1|X]$.
\end{assumption}
Other $\pi_{gd}(X)$'s are defined similarly. Part (a) relaxes unconditional parallel trends for the OOO group to its conditional version. This allows for covariate-specific heterogeneity in the untreated potential outcome trends of the always-observed group. Such conditional parallel trend assumptions are standard in the DiD literature, see e.g. \cite{abadie2005semiparametric,sant2020doubly,caetano2024difference}. Part (b) imposes an overlap condition which ensures we observe some OOO individuals in the treated group and that for every value of X in the population, we observe some OOO individuals in the untreated group.
  
The assumption on the selection mechanism is also relaxed. In particular, we consider a conditional IS assumption, which states that post-treatment selection is independent of the treatment status once we condition on covariates and selection in the pre-treatment period.
\begin{assumption}[Conditional ignorability of treatment in potential selection]\label{Partialselection_cond} \ \\
    (a) Equality in the untreated counterfactual share of individuals observed in period 1, conditional on pre-treatment selection and covariates. 
    \begin{align*}
        \mathbbm{P}[S_{1}(0)=1|D=0, S_0=1, X]=\mathbbm{P}[S_{1}(0)=1|D=1, S_{0}=1, X]
    \end{align*}
    (b) Equality in the treated counterfactual share of individuals observed in period 1, conditional on pre-treatment selection and covariates.  
    \begin{align*}
        \mathbbm{P}[S_{1}(1)=1|D=1, S_{0}=1, X]=\mathbbm{P}[S_{1}(1)=1|D=0, S_{0}=1, X].
    \end{align*}
\end{assumption}
Finally, we also assume the availability of a random sample as follows:
\begin{assumption}[Random Sampling]\label{rsx} \ \\
Assume that $\{(Y_{i0}, Y_{i1}, D_i, S_{i0}, S_{i1}, X_i); i=1,\ldots,N\}$ are i.i.d draws from an infinite population.
\end{assumption}
By the same argument presented in Section \ref{identification}, one can identify the conditional ATT based on the infeasible (hypothetical) DiD by considering only the always-observed latent population. Consider the observed conditional differences in outcomes for the OOO latent group, given as,
\begin{align}\label{did_cond}
    &\mathbbm{E}[Y_1-Y_0|D=1, X, OOO] -\mathbbm{E}[Y_1-Y_0|D=0, X, OOO] \nonumber \\
    & = \mathbbm{E}[Y^\ast_1(1)-Y^\ast_1(0)|D=1, X, OOO]+\mathbbm{E}[Y_1^\ast(0)-Y_0^\ast(0)|D=1, X, OOO]\nonumber \\
    &-\mathbbm{E}[Y_1^\ast(0)-Y_0^\ast(0)|D=0, X, OOO]  \nonumber \\
    & = \mathbbm{E}[Y^\ast_1(1)-Y^\ast_1(0)|D=1, X, OOO] \equiv \tau_{OOO}(X) 
\end{align}
where we rely on the no-anticipation in outcomes (Assumption \ref{no anti}) and conditional parallel trends in outcomes (Assumption \ref{PT_OOO_cond}).

Similar to the unconditional case, while we cannot point-identify the conditional means $\mathbbm{E}[Y^\ast_1(d)-Y^\ast_0(d)|D=d, X, OOO]$, we can partially identify them. Consider the group of treated individuals whose outcomes are observed in both periods, i.e. $D=1, S_0=1, S_1=1$. Then,
\begin{small}	
\begin{align}\label{diffy_cond_treat}
    \mathbbm{E}[Y_1-Y_0|D=1, X, S_0=1,&S_1=1]
    = \mathbbm{E}[Y^\ast_1(1)-Y^\ast_0(1)|D=1, X, OOO]\cdot \frac{\pi_{OOO1}(X)}{\pi_{OOO1}(X)+\pi_{ONO1}(X)} \nonumber \\
    & + \mathbbm{E}[Y^\ast_1(1)-Y^\ast_0(1)|D=1, X, ONO]\cdot \frac{\pi_{ONO1}(X)}{\pi_{OOO1}(X)+\pi_{ONO1}(X)}.
\end{align}
\end{small}
Let $p_{OOO1}(X) = \frac{\pi_{OOO1}(X)}{\pi_{OOO1}(X)+\pi_{ONO1}(X)}$. Next, consider the group of untreated individuals who are observed in both periods 
\begin{small}	
\begin{align}\label{diffy_cond_untreat}
    \mathbbm{E}[Y_1-Y_0|D=0, X, S_0=1,&S_1=1]
    = \mathbbm{E}[Y^\ast_1(0)-Y^\ast_0(0)|D=0, X, OOO]\cdot \frac{\pi_{OOO0}(X)}{\pi_{OOO0}(X)+\pi_{OON0}(X)} \nonumber \\
        &+ \mathbbm{E}[Y^\ast_1(0)-Y^\ast_0(0)|D=0, X, OON]\cdot \frac{\pi_{OON0}(X)}{\pi_{OOO0}(X)+\pi_{OON0}(X)}.
\end{align}
\end{small}
Define $p_{OOO0}(X) = \frac{\pi_{OOO0}(X)}{\pi_{OOO0}(X)+\pi_{OON0}(X)}$. Based on the trimming argument from Section \ref{identification}, the conditional means $\mathbbm{E}[Y^\ast_1(1)-Y^\ast_0(1)|D=d, X, OOO]$ for $d=\{0,1\}$ are bounded between the interval $[LB_{OOOd}(X), UB_{OOOd}(X)]$ where 
{\small \begin{equation}\label{bounds1_cond_nomono}
\begin{split}
    LB_{OOOd}(X) & = \mathbbm{E}[Y_1-Y_0|D=d, X, S_0=1, S_1=1, Y_1-Y_0\leq F^{-1}_{\Delta Y|X, d11}(p_{OOOd}(X))], \\
    UB_{OOOd}(X) & = \mathbbm{E}[Y_1-Y_0|D=d, X, S_0=1, S_1=1, Y_1-Y_0 > F^{-1}_{\Delta Y|X, d11}(1-p_{OOOd}(X))]
\end{split}
\end{equation}}
and $F_{\Delta Y|X, dss'}^{-1}(.)$ is the conditional quantile function of $\Delta Y \equiv Y_1-Y_0$ where we are conditioning on $D=d,X,S_{0}=s,S_{1}=s'$.

Combining the bounds for the two conditional means, we obtain the bound for $\tau_{OOO}(X)$, 
\begin{align*}
    \tau_{OOO}(X)\in  [LB_{OOO1}(X)-UB_{OOO0}(X), \ UB_{OOO1}(X)-LB_{OOO0}(X)].
\end{align*}
The bounds on the conditional ATT for the OOO can be of independent interest and provide useful insights about heterogeneity of the treatment effects across different values of the covariates. As shown in equations \eqref{nomonox_lb_true} and \eqref{nomonox_ub_true}, one can get bounds on the unconditional treatment effect, $\tau_{OOO}$, by integrating the conditional bounds over the distribution of $X$ for the always-observed population in the treated group. 

While the bounds themselves are identified from the data, we must also identify the mixing proportion functions $p_{OOOd}(X)$ for both treatment groups and for all values in the support of $X$ in order to partially identify the unconditional ATT.

\subsection{Identification without Monotonicity}

This subsection follows closely the discussion for the unconditional case in Section \ref{sec:ident_nomon}. Here we are interested in the unobserved share of always-observed individuals conditional on the value of the covariates $X$, $\mathbbm{P}[S_0(0)=1, S_1(0)=1, S_1(1)=1, D=d|X]$. For the treated group,
\begin{align}\label{pipartialid1_cov}
    &\mathbbm{P}[S_{0}=1, S_{1}=1| D=1,X] \nonumber \\
    &=\mathbbm{P}[S_{0}=1, S_{1}(0)=1, S_{1}(1)=1| D=1,X]+\mathbbm{P}[S_{0}=1, S_{1}(0)=0, S_{1}(1)=1| D=1,X] \nonumber\\
    &=\frac{\pi_{OOO1}(X)}{\mathbbm{P}[D=1|X]}+\frac{\pi_{ONO1}(X)}{\mathbbm{P}[D=1|X]}.
\end{align}
And for untreated observations,
\begin{align}\label{pipartialid2_cov}
    &\mathbbm{P}[S_{0}=1, S_{1}=1| D=0,X] \nonumber \\ 
    &=\mathbbm{P}[S_{0}=1, S_{1}(0)=1, S_{1}(1)=1| D=0, X]+\mathbbm{P}[S_{0}=1, S_{1}(0)=1, S_{1}(1)=0| D=0,X]\nonumber\\
    &=\frac{\pi_{OOO0}(X)}{\mathbbm{P}[D=0|X]}+\frac{\pi_{OON0}(X)}{\mathbbm{P}[D=0|X]}.
\end{align}

Following the arguments laid out on Section \ref{sec:ident_nomon}, express 
{\small \begin{align*}
   \mathbbm{P}[S_{0}=1, S_{1}(0)=1, S_{1}(1)=1| D=d, X] = \mathbbm{P}[S_{1}(0)=1, S_{1}(1)=1| D=d, S_{0}=1,X]\cdot \mathbbm{P}[S_0=1|D=d,X] 
\end{align*}}
where $\mathbbm{P}[S_0=1|D=d,X]$ is directly observed in the data whereas $\mathbbm{P}[S_{1}(0)=1, S_{1}(1)=1| D=d, S_{0}=1,X]$ can be partially identified using Fr\'echet bounds \citep{imai2008sharp} as follows:
{\footnotesize
\begin{align*}
\mathbbm{P}[S_{1}(0)=1, S_{1}(1)=1|D=d, S_{0}=1,X]\in& \bigg[\max\{\mathbbm{P}[S_{1}(0)=1|D=d, S_{0}=1,X]+\mathbbm{P}[S_{1}(1)=1|D=d, S_{0}=1,X]-1,0\},\\
&\min\{\mathbbm{P}[S_{1}(0)=1|D=d, S_{0}=1,X],\mathbbm{P}[S_{1}(1)=1|D=d,S_{0}=1,X]\}\bigg].
\end{align*}}

Combining Assumption \ref{Partialselection_cond} and the information in equations (\ref{pipartialid1_cov}) and (\ref{pipartialid2_cov}), we identify the missing counterfactual probabilities in Lemma \ref{lemma:pipartialid_cond}.

\begin{lemma}\label{lemma:pipartialid_cond}
    (a) Under assumptions \ref{no anti} and \ref{Partialselection_cond}(a), the identified set for $\mathbbm{P}[S_{1}(0)=1, S_{1}(1)=1|D=1,S_{0}=1, X]$ is given by 
    {\footnotesize
 \begin{align}
 \mathbbm{P}[S_{1}(0)=1, S_{1}(1)=1|D=1,S_{0}=1, X]&\in \left[\max\{\mathbbm{P}[S_{1}=1|D=0, S_{0}=1,X]+\mathbbm{P}[S_{1}=1|D=1, S_{0}=1,X]-1,0\},\right. \nonumber\\
 &\left.\min\{\mathbbm{P}[S_{1}=1|D=0, S_{0}=1, X],\mathbbm{P}[S_{1}=1|D=1, S_{0}=1,X]\}\right]
 \end{align}}
    (b) Under assumptions \ref{no anti} and \textcolor{blue}{\ref{Partialselection_cond}(b)}, the identified set for $\mathbbm{P}[S_{1}(0)=1, S_{1}(1)=1|D=1,S_{0}=1, X]$ is given by
    {\footnotesize
 \begin{align}
 \mathbbm{P}[S_{1}(0)=1, S_{1}(1)=1|D=0, S_{0}=1, X]&\in \left[\max\{\mathbbm{P}[S_{1}=1|D=0, S_{0}=1,X]+\mathbbm{P}[S_{1}=1|D=1, S_{0}=1,X]-1,0\},\right.\nonumber\\
&\left.\min\{\mathbbm{P}[S_{1}=1|D=0, S_{0}=1,X],\mathbbm{P}[S_{1}=1|D=1, S_{0}=1,X]\}\right]
 \end{align}}
\end{lemma}
The proof of Lemma \ref{lemma:pipartialid_cond} is identical to the proof of Lemma \ref{lemma:pipartialid} once we account for the additional conditioning variable $X$, and is therefore omitted. Let $p^l_{OOOd}(X)$ and $\pi^l_{OOOd}(X)$ for $d=\{0,1\}$ denote the lower end of the identified set for $p_{OOOd}(X)$ and $\pi_{OOOd}(X)$ for all $X \in \mathcal{X}$, respectively.
Theorem \ref{nomono_bound_cond} identifies the bounds for $\tau_{OOO}$ that can be obtained without requiring any monotonicity on the sample selection mechanism.
  
\begin{theorem}[Bounds for $\tau_{OOO}$ conditional on covariates]\label{nomono_bound_cond} Under assumptions \ref{no anti},\ref{PT_OOO_cond}, \ref{Partialselection_cond}(a), \ref{Partialselection_cond}(b), and $p^l_{OOOd}(X)>0$, bounds on the treatment effect on the treated for the always-observed group ($\tau_{OOO}$) lies in the interval $[LB_{\tau_{OOO}}, UB_{\tau_{OOO}}]$,
\begin{equation*}
    \begin{split}
        LB_{\tau_{OOO}} &=\mathbbm{E}\bigg[\left(LB_{OOO1}(X) - UB_{OOO0}(X)\right)\dfrac{\pi^l_{OOO1}(X)}{\pi^l_{OOO1}} \bigg]\\
        UB_{\tau_{OOO}} &=\mathbbm{E}\bigg[\left(UB_{OOO1}(X)-LB_{OOO0}(X)\right)\dfrac{\pi^l_{OOO1}(X)}{\pi^l_{OOO1}}\bigg]
    \end{split}
\end{equation*}
where $\pi^l_{OOO1}(X)  = p^l_{OOO1}(X)\cdot \mathbbm{P}(S_0=1, S_1=1, D=1|X), \quad \pi^l_{OOO1} = \mathbbm{E}[\pi^l_{OOO1}(X)]$ and  
\begin{align*}
    LB_{OOO1}(X) & = \mathbbm{E}[Y_1-Y_0|D=1, X, S_0=1, S_1=1, Y_1-Y_0 \leq F^{-1}_{\Delta Y|X, 111}(p^{l}_{OOO1}(X))] \\
    UB_{OOO1}(X) & = \mathbbm{E}[Y_1-Y_0|D=1, X, S_0=1, S_1=1, Y_1-Y_0 > F^{-1}_{\Delta Y|X, 111}(1-p_{OOO1}^l(X))] \\
    LB_{OOO0}(X) & = \mathbbm{E}[Y_1-Y_0|D=0, X, S_0=1, S_1=1, Y_1-Y_0 \leq F^{-1}_{\Delta Y|X, 011}(p^l_{OOO0}(X))] \\
    UB_{OOO0}(X) & = \mathbbm{E}[Y_1-Y_0|D=0, X, S_0=1, S_1=1, Y_1-Y_0 > F^{-1}_{\Delta Y|X, 011}(1-p^l_{OOO0}(X))]
    \end{align*}
    \begin{align*}	
    p^l_{OOO1}(X)&=\frac{\max\{\mathbbm{P}[S_{1}=1|D=0, S_{0}=1, X]+\mathbbm{P}[S_{1}=1|D=1, S_{0}=1, X]-1,0\}}{\mathbbm{P}[S_{1}=1| D=1, S_{0}=1, X]},\\
    p^l_{OOO0}(X)&=\frac{\max\{\mathbbm{P}[S_{1}=1|D=0, S_{0}=1, X]+\mathbbm{P}[S_{1}=1|D=1, S_{0}=1, X]-1,0\}}{\mathbbm{P}[S_{1}=1| D=0, S_{0}=1, X]}.
\end{align*}
\end{theorem}
Lemma \ref{lemma:pipartialid_cond} and Theorem \ref{nomono_bound_cond} provide identification for the mixing proportions and ATT for the always-observed latent group for different values of covariates. The proof of identification of the conditional bounds $LB_{OOOd}$ and $UB_{OOOd}$ is identical to the proof of Theorem \ref{nomono_bound}, with the only difference being the additional conditioning on X. Identification of the unconditional bounds then follows directly from equations \eqref{nomonox_lb_true} and \eqref{nomonox_ub_true} for $LB_{\tau_{OOO}}$ and $UB_{\tau_{OOO}}$, respectively.

\subsection{Relaxed Conditional Monotonicity} \label{relax mono}
We relax monotonicity by allowing treatment to influence selection in a different direction within each partition of the population defined by $X$. Even if selection in each partition increases (or decreases) in the treatment state relative to the control state, the composition of the population across $X$ could involve different directions of monotonicity. To allow for this, we follow the work of \cite{semenova2025generalized} and partition the covariate set into regions where the impact of treatment on selection is positive, regions where it is negative, and regions where treatment may not impact selection. Therefore, the following assumption allows for relaxations of the unconditional monotonicity.
\begin{assumption}\label{monotone_cond}
    Conditional monotone sample selection: The covariate set $\mathcal{X}$ can be partitioned into sets, $\mathcal{X}_{+}$, $\mathcal{X}_{-}$, and $\mathcal{X}_{0}$ such that 
    \begin{enumerate}
        \item[(a)] (Positive conditional monotonicity): $X \in \mathcal{X}_{+} \text{ such that } \mathbbm{P}[S_{1}(1)> S_{1}(0)|X]=1$.
        \item[(b)] (Negative conditional monotonicity): $X \in \mathcal{X}_{-} \text{ such that } \mathbbm{P}[S_{1}(1) < S_{1}(0)|X]=1$. 
        \item[(c)] (No conditional monotonicity): $X \in \mathcal{X}_{0} \text{ such that } \mathbbm{P}[S_{1}(1) = S_{1}(0)|X]=1$. 
    \end{enumerate}
\end{assumption}

First, consider the subset, $X \in \mathcal{X}_{+}$, where we have positive monotonicity. Similar to the unconditional case, the proportion of always-observed subpopulation is point-identified under positive conditional monotonicity and ignorability of selection assumptions for both the treated and untreated groups, taking the value one for the latter. Recall that
\begin{align*}
    p_{OOO1}(X)& = \frac{\pi_{OOO1}(X)}{\pi_{OOO1}(X)+\pi_{ONO1}(X)} \\
    & = \frac{\mathbbm{P}(S_1=1|D=0, X, S_0=1)}{\mathbbm{P}(S_1=1|D=1, X, S_0=1)} \tag{By Assumptions \ref{Partialselection_cond}(a) \& \ref{monotone_cond}(a)} 
\end{align*}
While for the untreated group, 
\begin{align*}
    p_{OOO0}(X)	& = \frac{\pi_{OOO0}(X)}{\pi_{OOO0}(X)+\pi_{OON0}(X)} \\
    & = \frac{\mathbbm{P}(S_1(0)=1|D=0, X, S_0=1)}{\mathbbm{P}(S_1(0)=1|D=0, X, S_0=1)} \tag{By Assumption \ref{monotone_cond}(a)}\\
    & = 1
\end{align*}
Using the fact that $p_{OOO0}(X)=1$ in Equation \eqref{diffy_cond_untreat}, it implies that the conditional mean of $Y_1-Y_0$ for the observed subsample of untreated units point-identifies $\mathbbm{E}[Y_1^\ast(0) - Y_0^\ast(0)|D=0, X, OOO]$. 

Combining this with the bounds for $\mathbbm{E}[Y_1^\ast(1) - Y_0^\ast(1)|D=1, X, OOO]$ presented in Equation \eqref{bounds1_cond_nomono} we can bound $\tau_{OOO}(X)$. Therefore, for $X \in \mathcal{X}_{+}$, $	\tau_{OOO}(X) \in [LB_{OOO1}(X)-\mathbbm{E}[Y_1-Y_0|D=0,X,S_{0}=1,S_{1}=1],  \quad UB_{OOO1}(X)-\mathbbm{E}[Y_1-Y_0|D=0,X,S_{0}=1,S_{1}=1]$.

Results for the covariate set where selection is negative can be treated in an analogous manner. Following the same arguments as above, for any $X \in \mathcal{X}_{-}$ we can show that under Assumptions \ref{Partialselection_cond}(b) and \ref{monotone_cond}(b), $p_{OOO1}(X) = 1$ and 
\begin{equation*}
    p_{OOO0}(X) = \frac{\mathbbm{P}(S_1=1|D=1, X, S_0=1)}{\mathbbm{P}(S_1=1|D=0,X, S_0=1)} 
\end{equation*}
Symmetrically, using the fact that $p_{OOO1}(X)=1$ in Equation \eqref{diffy_cond_treat} implies that the conditional mean of $Y_1-Y_0$ for the observed subsample of treated units point-identifies $\mathbbm{E}[Y_1^\ast(1) - Y_0^\ast(1)|D=1, X, OOO]$. Combining this with bounds for $\mathbbm{E}[Y_1^\ast(0) - Y_0^\ast(0)|D=0, X, OOO]$ given in Equation \eqref{bounds1_cond_nomono}, we can say that for covariate sets where selection is negative i.e. $X \in \mathcal{X}_{-}$, the conditional parameter is bounded between $	\tau_{OOO}(X) \in \big[\mathbbm{E}[Y_1-Y_0|D=1, X, S_0=1, S_1=1] - UB_{OOO0}(X), \quad \mathbbm{E}[Y_1-Y_0|D=1, X, S_0=1, S_1=1]- LB_{OOO0}(X) \big]$.
	
Finally, selection behavior is not impacted by treatment on the covariate set $X \in \mathcal{X}_{0}$ and therefore the conditional means of $Y_1-Y_0$ for the $D=d$ group which is observed in both periods point-identifies the counterfactual means, $\mathbbm{E}[Y_1^\ast(d)-Y_0^\ast(d)|D=d, X, OOO]$ for $d=0,1$. Therefore, $\tau_{OOO}(X)=\mathbbm{E}[Y_1-Y_0|D=1, X, S_0=1, S_1=1] - \mathbbm{E}[Y_1-Y_0|D=0, X, S_0=1, S_1=1]$. As pointed out by \cite{semenova2025generalized}, this boundary case is likely to be of small importance in applications but can still be incorporated into this approach.

Aggregating the bounds across the covariate-specific regions gives us Theorem \ref{Gen_bound_cond}.
\begin{theorem}[Bounds for $\tau_{OOO}$ under conditional monotonicity]\label{Gen_bound_cond} Under Assumptions \ref{no anti},\ref{PT_OOO_cond}, \ref{Partialselection_cond}, \ref{monotone_cond}, and $p_{OOOd}^r(X)>0$, bounds on the treatment effect on the treated for the always-observed group ($\tau_{OOO}$) are  given by the interval $[LB_{\tau_{OOO}}, UB_{\tau_{OOO}}]$ where the lower and upper bounds are given by:
 \begin{align*}
    \mathrm{LB}_{\tau_{OOO}} =\;
    & \sum_{r=\{+, -, 0\}}\mathbbm{E}_{\mathcal{X}_{r}}\left[ (LB^r_{OOO1}(X) - UB^r_{OOO0}(X))\dfrac{\pi^r_{OOO1}(X)}{\pi^r_{OOO1}}\right]\\
    \mathrm{UB}_{\tau_{OOO}} =\;
    & \sum_{r=\{+, -, 0\}}\mathbbm{E}_{\mathcal{X}_{r}}\left[ (UB^r_{OOO1}(X) - LB^r_{OOO0}(X))\dfrac{\pi^r_{OOO1}(X)}{\pi^r_{OOO1}}\right]
\end{align*}
where $\pi^r_{OOO1}(X) = p_{OOO1}^r(X)\cdot \mathbbm{P}(S_0=1, S_1=1, D=1|X), \text{ and } \pi^r_{OOO1}= \mathbbm{E}[\pi^r_{OOO1}(X)], \ r\in \{+,-,0\}.$ The region-specific bounds are:
    \begin{enumerate}
        \item Positive selection ($r=+$): 
        \begin{align*}
            LB^{+}_{OOO1}(X) & = \mathbbm{E}[Y_1-Y_0|D=1, X, S_0=1, S_1=1, Y_1-Y_0 \leq 	F^{-1}_{\Delta Y|X, 111}(p^{+}_{OOO1}(X))], \\
            UB^{+}_{OOO1}(X) & = \mathbbm{E}[Y_1-Y_0|D=1, X, S_0=1, S_1=1, Y_1-Y_0 > 	F^{-1}_{\Delta Y|X, 111}(1-p^{+}_{OOO1}(X))], \\
            LB^{+}_{OOO0}(X) & = UB^{+}_{OOO0}(X) = \mathbbm{E}[Y_1-Y_0\mid D=0, X, S_0=1, S_1=1] 
        \end{align*}
        \item Negative selection ($r=-$): 
        \begin{align*}
            LB^{-}_{OOO1}(X) & = UB^{-}_{OOO1}(X) = \mathbbm{E}[Y_1-Y_0\mid D=1, X, S_0=1, S_1=1] \\
            LB^{-}_{OOO0}(X) & = \mathbbm{E}[Y_1-Y_0|D=0, X, S_0=1, S_1=1, Y_1-Y_0 \leq 	F^{-1}_{\Delta Y|X, 011}(p^{-}_{OOO0}(X))], \\
            UB^{-}_{OOO0}(X) & = \mathbbm{E}[Y_1-Y_0|D=0, X, S_0=1, S_1=1, Y_1-Y_0 > 	F^{-1}_{\Delta Y|X, 011}(1-p^{-}_{OOO0}(X))],
        \end{align*}
        \item No selection ($r=0$):
        \begin{align*}
            LB^{0}_{OOO1}(X) & = UB^{0}_{OOO1}(X) = \mathbbm{E}[Y_1-Y_0\mid D=1, X, S_0=1, S_1=1] \\
            LB^{0}_{OOO0}(X) & = UB^{0}_{OOO0}(X) = \mathbbm{E}[Y_1-Y_0\mid D=0, X, S_0=1, S_1=1] 
        \end{align*}
    \end{enumerate}
    with mixing proportions given by:
    \begin{align*}
        p_{OOO1}^+(X) &= \frac{\mathbbm{P}[S_{1}=1|S_{0}=1,D=0,X]}{\mathbbm{P}[S_{1}=1|S_{0}=1,D=1,X]}, \quad p_{OOO0}^-(X)=\frac{\mathbbm{P}[S_{1}=1|S_{0}=1,D=1,X]}{\mathbbm{P}[S_{1}=1|S_{0}=1,D=0,X]}, \\
    p_{OOO0}^+(X) &= 1, \quad p_{OOO1}^-(X)=1, \quad p_{OOO1}^0(X)=1, \quad p_{OOO0}^0(X)=1, 
    \end{align*}
\end{theorem}
Proof of Theorem \ref{Gen_bound_cond} can be found in Appendix \ref{proof:Gen_bound_cond}.

\subsection{Moment-based Representations}\label{identification_mom}
Following \cite{semenova2025generalized}, the bounds for $\tau_{OOO}$ can be represented as a ratio of two moments. 
 \begin{align*}
    LB_{\tau_{OOO}}&= \frac{\mathbbm{E}[m_L(W, \xi)]}{\pi_{OOO1}} \quad \text{ and } \quad
    UB_{\tau_{OOO}}= \frac{\mathbbm{E}[m_U(W, \xi)]}{\pi_{OOO1}} \quad \text{ with } \\
    \pi_{OOO1}  &= \mathbbm{E}[\min\{s(0,X),s(1,X)\}\cdot \mu_1(X)\cdot \mathbbm{P}(S_0=1\mid X)],\\
    \mu_d(X) &\equiv \mathbbm{P}(D=d \mid X, S_0=1)   
\end{align*}        
where $s(d,X) \equiv \mathbbm{P}(S_1=1\mid D=d, X, S_0=1)$ and the region-specific moment conditions for the lower bound are given as:
{\small \begin{align*}
    m_L(W, \xi)	& = \begin{cases}
        \left(D\cdot \mathbbm{1}\{\Delta Y \leq F^{-1}_{\Delta Y\mid X, 111}(p_{OOO1}(X))\} - (1-D)\cdot \frac{\mu_1(X)}{\mu_0(X)}\right)\cdot S_0\cdot S_1\cdot \Delta Y, \qquad  &X\in \mathcal{X}_{+}\\ 
        \\
        \left(D - (1-D)\cdot \mathbbm{1}\{\Delta Y>F^{-1}_{\Delta Y\mid X, 011}(1-p_{OOO0}(X))\}\cdot \frac{\mu_1(X)}{\mu_0(X)}\right)\cdot S_0\cdot S_1\cdot \Delta Y, \qquad &X\in \mathcal{X}_{-}  \\
        \\
         \left(D - (1-D)\cdot 	\frac{\mu_1(X)}{\mu_0(X)}\right)\cdot S_0\cdot S_1\cdot \Delta Y, \qquad &X\in \mathcal{X}_{0}
    \end{cases} 
\end{align*}}
and the moment conditions for the upper bound are:
{\small \begin{align*}
     m_U(W, \xi) &= \begin{cases}
         \left(D\cdot \mathbbm{1}\{\Delta Y > F^{-1}_{\Delta Y\mid 	X, 111}(1-p_{OOO1}(X))\} - (1-D)\cdot \frac{\mu_1(X)}{\mu_0(X)}\right)\cdot S_0\cdot S_1\cdot \Delta Y, \qquad &X\in \mathcal{X}_{+} \\
         \\
         \left(D - (1-D)\cdot \mathbbm{1}\{\Delta Y\leq F^{-1}_{\Delta Y\mid X, 011}(p_{OOO0}(X))\}\cdot\frac{\mu_1(X)}{\mu_0(X)}\right)\cdot S_0\cdot S_1\cdot \Delta Y, \qquad &X\in \mathcal{X}_{-} \\
         \\
          \left(D - (1-D)\cdot \frac{\mu_1(X)}{\mu_0(X)}\right)\cdot S_0\cdot S_1\cdot \Delta Y, \qquad &X\in \mathcal{X}_{0}.
     \end{cases}
\end{align*}}
The trimming proportions are identified in Theorem \ref{Gen_bound_cond}. 

\subsection{Estimation}\label{estimation_x}
\paragraph{Without Monotonicity:} Estimation of the covariate-adjusted lower bound for $\tau_{OOO}$, $LB_{\tau_{OOO}}$ as given in Theorem \ref{nomono_bound_cond} could proceed in different ways.\footnote{Estimation of the upper bound follows similarly.} We detail here an intuitive approach that relies on partitioning the covariate space into J discrete cells. Since the overall bounds can be written as
\begin{align}\label{nomonox_lb}
    \widehat{LB}_{\tau_{OOO}} =\sum_{j=1}^{J}\left(\widehat{LB}_{OOO1}(j) - \widehat{UB}_{OOO0}(j)\right)\frac{\hat{\pi}^l_{OOO1}(j)}{\hat{\pi}^l_{OOO1}} \cdot \mathbbm{\hat{P}}[J=j],
\end{align}
We focus on estimating cell-specific bounds for $j=1,2,...,J$. First, we estimate the lower bound for the conditional mixing probabilities in each cell $j$ as,
\begin{align*}
    \hat{p}^l_{OOO1}(j)=\frac{\max\{\mathbbm{\hat{P}}[S_{1}=1|D=0, S_{0}=1,J=j]+\mathbbm{\hat{P}}[S_{1}=1|D=1, S_{0}=1,J=j]-1,0\}}{\mathbbm{\hat{P}}[S_{1}=1| D=1, S_{0}=1,J=j]},\\
    \hat{p}^l_{OOO0}(j)=\frac{\max\{\mathbbm{\hat{P}}[S_{1}=1|D=0, S_{0}=1,J=j]+\mathbbm{\hat{P}}[S_{1}=1|D=1, S_{0}=1,J=j]-1,0\}}{\mathbbm{\hat{P}}[S_{1}=1| D=0, S_{0}=1,J=j]}.
\end{align*}
where,
\begin{align*}
    \mathbb{\hat{P}}[S_{1}=1|D=0, S_{0}=1,J=j] &= \frac{\sum_{i=1}^{n}S_{i0}\cdot S_{i1}\cdot (1-D_{i})\cdot I(J_i=j)}{\sum_{i=1}^{n}S_{i0}\cdot (1-D_{i})\cdot I(J_i=j)}\\
    \mathbb{\hat{P}}[S_{1}=1|D=1, S_{0}=1,J=j] &= \frac{\sum_{i=1}^{n}S_{i0}\cdot S_{i1}\cdot D_{i}\cdot I(J_i=j)}{\sum_{i=1}^{n}S_{i0}\cdot D_{i}\cdot I(J_i=j)}
\end{align*}
Then we estimate the bounds for each cell $j$ as follows,
{\small \begin{align}
    \widehat{LB}_{OOO1}(j)
    &=\frac{\sum_{i=1}^n (Y_{i1}-Y_{i0}) \cdot S_{i0}\cdot S_{i1}\cdot D_{i} \cdot I(J_i=j) \cdot I\{(Y_{i1}-Y_{i0}) \leqslant \hat{y}_{\hat{p}^l_{OOO1} \label{LB_Cond_estimate}(j)}\}}{\sum_{i=1}^n S_{i0}\cdot S_{i1}\cdot D_{i} \cdot I(J_i=j)  \cdot I\{(Y_{i1}-Y_{i0}) \leqslant \hat{y}_{\hat{p}^l_{OOO1}(j)}\}}\\
    \widehat{UB}_{OOO0}(j) 
    &=\frac{\sum_{i=1}^n (Y_{i1}-Y_{i0}) \cdot S_{i0}\cdot S_{i1}\cdot (1-D_{i} )\cdot I(J_i=j) \cdot I\{(Y_{i1}-Y_{i0}) > \hat{y}_{1-\hat{p}^l_{OOO0}(j)}\}}{\sum_{i=1}^n S_{i0}\cdot S_{i1}\cdot (1-D_{i} )\cdot I(J_i=j)  \cdot I\{(Y_{i1}-Y_{i0}) > \hat{y}_{1-\hat{p}^l_{OOO0}(j)}\}} \label{UB_Cond_estimate}
\end{align}}
where $\hat{y}_{\hat{p}^l_{OOO1}(j)}$ and $\hat{y}_{1-\hat{p}^l_{OOO0}(j)}$ are the $\hat{p}^l_{OOO1}(j)$-th and $(1-\hat{p}^l_{OOO0}(j))$-th quantile of the conditional distribution $Y_1 - Y_0$ for individuals observed in both time periods who fall into the covariate group $J=j$ for the treated and untreated group respectively.  In general, the relevant $q$-th quantile  of the conditional distribution $Y_1 - Y_0$ for the treated individuals observed in both time periods who falls into the covariate group $J=j$ is calculated as,
\begin{equation*}
    \hat{y}_q  = \min \left\{y: \frac{\sum_{i=1}^n S_{i0}\cdot S_{i1}\cdot D_{i}\cdot I(J_i=j) \cdot I\{(Y_{i1}-Y_{i0})\leqslant y\}}{\sum_{i=1}^n S_{i0}\cdot S_{i1}\cdot D_{i}\cdot I(J_i=j)} \geqslant q\right\}, 
\end{equation*}
where I($\cdot$) is an indicator function. Then we estimate the weights of each cell $j$ as follows,
\begin{align}\label{nomonox_weights_1}
    \frac{\hat{\pi}^l_{OOO1}(j)}{\hat{\pi}^l_{OOO1}}& = \frac{\hat{p}^l_{OOO1}(j)\cdot \hat{\mathbbm{P}}(S_0=1, S_1=1, D=1|j)}{\hat{\mathbbm{E}}[\hat{p}^l_{OOO1}(j)\cdot\mathbbm{P}(S_0=1, S_1=1, D=1|j) ]}.
\end{align}
where,
\begin{align}
    \mathbb{\hat{P}}[S_{0}=1,S_{1}=1, D=1|J=j] &= \frac{\sum_{i=1}^{n}S_{i0}\cdot S_{i1}\cdot D_{i}\cdot I(J_i=j)}{\sum_{i=1}^{n} I(J_i=j)} \label{nomonox_weights_2}\\
    \mathbbm{\hat{P}}[J=j] &=\frac{\sum_{i=1}^{n}I(J_i=j)}{n} \label{nomonox_weights_3}\\
    \hat{\mathbbm{E}}[p^l_{OOO1}(j)\cdot\mathbbm{P}(S_0=1, S_1=1, D=1|j) ]&=\sum_{j=1}^{J}\hat{p}^l_{OOO1}(j) \cdot \mathbb{\hat{P}}[S_{0}=1,S_{1}=1, D=1|J=j] \cdot \nonumber \\
    &\quad \cdot \mathbbm{\hat{P}}[J=j] \label{nomonox_weights_4}.
\end{align}
We use bootstrap to obtain the asymptotic variance of this estimator.\footnote{We have established the bootstrap consistency of $\widehat{LB}_{\tau_{OOO}(j)}\equiv (\widehat{LB}_{OOO1}(j) - \widehat{UB}_{OOO0}(j))$ and the weights $w_j$ estimated by sample proportions, $\widehat{w}_j=\frac{\hat{\pi}^l_{OOO1}(j)}{\hat{\pi}^l_{OOO1}} \cdot \mathbbm{\hat{P}}[J=j]$. The estimator $\widehat{LB}_{\tau_{OOO}(X)}=\sum_{j=1}^{J}\widehat{w}_j\widehat{LB}_{\tau_{OOO}(j)}$ is a smooth linear function of  $\widehat{LB}_{\tau_{OOO}(j)}$ and $\hat{w}_j$. Thus, bootstrap is consistent for inference on  $\widehat{LB}_{\tau_{OOO}(X)}$. See details in Appendix \ref{appendix:Bootstrap}.} 

\paragraph{Estimation with Relaxed Monotonicity:}
The estimation procedure for the covariate-adjusted bounds for $\tau_{OOO}$ given in Theorem \ref{Gen_bound_cond} has many of the same components described in the discussion about estimation without imposing monotonicity in the previous section. The main differences are in the point identification of the conditional mixing proportions, and we need to define the subsets of the covariate space for which monotonicity is positive ($\mathcal{X}_{+}$), negative ($\mathcal{X}_{-}$) or no selection ($\mathcal{X}_{0}$).

First, partition the covariate space into J discrete cells. Then for each $j=1,2,...,J$, estimate,
\begin{align} 
    \mathbbm{\hat{P}}[S_{1}=1|D=0, S_{0}=1,J=j] &= \frac{\sum_{i=1}^{n}S_{i0}\cdot S_{i1}\cdot (1-D_{i})\cdot I(J_i=j)}{\sum_{i=1}^{n}S_{i0}\cdot (1-D_{i})\cdot I(J_i=j)}, \nonumber\\
    \mathbbm{\hat{P}}[S_{1}=1|D=1, S_{0}=1,J=j] &= \frac{\sum_{i=1}^{n}S_{i0}\cdot S_{i1}\cdot D_{i}\cdot I(J_i=j)}{\sum_{i=1}^{n}S_{i0}\cdot D_{i}\cdot I(J_i=j)}.
\end{align}
Then we assign the j-th cell into a covariate space as follows,
\begin{align*}
   j= 
   \begin{cases}
    \quad j\in\mathcal{X}^{+}, & \text{if} \quad 
      \mathbbm{\hat{P}}[S_{1}=1 \mid S_{0}=1,D=1,J=j] > \mathbbm{\hat{P}}[S_{1}=1 \mid S_{0}=1,D=0,J=j],\\
    \quad j\in\mathcal{X}^{-}, & \text{if} \quad \mathbbm{\hat{P}}[S_{1}=1 \mid S_{0}=1,D=1,J=j] < \mathbbm{\hat{P}}[S_{1}=1 \mid S_{0}=1,D=0,J=j], \\
     \quad j\in\mathcal{X}^{0}, & \text{if} \quad \mathbbm{\hat{P}}[S_{1}=1 \mid S_{0}=1,D=1,J=j] = \mathbbm{\hat{P}}[S_{1}=1 \mid S_{0}=1,D=0,J=j], \\
\end{cases}  
\end{align*}
Define,
\begin{align*}
   r(j)= 
\begin{cases}
    +, & \text{if} \quad j\in\mathcal{X}^{+}\\
    -, & \text{if} \quad j\in\mathcal{X}^{-}\\
    0, & \text{if} \quad j\in\mathcal{X}^{0}\\
\end{cases},
   \quad \text{and} \quad d(j)= 
\begin{cases}
    1, & \text{if} \quad j\in\mathcal{X}^{+}\\
    0, & \text{if} \quad j\in\mathcal{X}^{-}\\
    0, & \text{if} \quad j\in\mathcal{X}^{0}\\
\end{cases}
\end{align*}
{\small
 \begin{align}
\hat{p}_{OOOd(j)}^{r(j)}(j) = 
\dfrac{\mathbbm{\hat{P}}[S_{1}=1 \mid S_{0}=1,D=(1-d),J=j]}
      {\mathbbm{\hat{P}}[S_{1}=1 \mid S_{0}=1,D=d,J=j]}, 
\end{align}}
Finally,  $\hat{p}_{OOO(1-d(j))}^{r(j)}(j)=1$ 
 The $j$-th cell-specific bounds are estimated as follows,
 \begin{align}
    \widehat{LB}_{\tau_{OOO}}(j) = 
    \begin{cases}
    \widehat{LB}_{OOO1}(j) - \hat{E}_{OOO0}(j), 
    & \text{if } \quad j\in\mathcal{X}^{+} \\
    \hat{E}_{OOO1}(j)-\widehat{UB}_{OOO0}(j), & \text{if} \quad j\in\mathcal{X}^{-}  \\
    \hat{E}_{OOO1}(j) - \hat{E}_{OOO0}(j), 
    &  \text{if} \quad j\in\mathcal{X}^{0}  \\  
    \end{cases}
\end{align}
where 
\begin{align*}
     \hat{E}_{OOO0}(j) &= \frac{\sum_{i=1}^n (Y_{i1}-Y_{i0}) \cdot S_{i0}\cdot S_{i1}\cdot (1-D_{i})\cdot I(J_i=j)}{\sum_{i=1}^n S_{i0}\cdot S_{i1}\cdot (1-D_{i})\cdot I(J_i=j)}.\\
     \hat{E}_{OOO1}(j) &= \frac{\sum_{i=1}^n (Y_{i1}-Y_{i0}) \cdot S_{i0}\cdot S_{i1}\cdot D_{i}\cdot I(J_i=j)}{\sum_{i=1}^n S_{i0}\cdot S_{i1}\cdot D_{i}\cdot I(J_i=j)},
\end{align*}
and the estimates for $\widehat{LB}_{OOO1}(j)$ and $\widehat{UB}_{OOO0}(j)$ have the same form of the estimators in equations (\ref{LB_Cond_estimate}) and (\ref{UB_Cond_estimate}) with $\hat{p}_{OOOd(j)}^{r(d)}(j)$ replacing $\hat{p}^{l}_{OOOd}(j)$. Likewise, the estimates of the cell-specific weights and the covariate-adjusted lower bound for $\tau_{OOO}$ follow the same structure as in equation \eqref{nomonox_weights_1}, with components defined in \eqref{nomonox_weights_2}–\eqref{nomonox_weights_4} and \eqref{nomonox_lb}, respectively, again substituting   $\hat{p}_{OOOd(j)}^{r(d)}(j)$ for $\hat{p}^{l}_{OOOd}(j)$. The upper bound can be estimated using a similar procedure.

\subsection{Empirical Results with Covariates}\label{application_x}

In this section, we draw on the discussion and methods developed in Section \ref{sec: covariates} to examine how baseline covariates help with the identification of the bounds in both empirical applications. One of the main difficulties in implementing the bounds with covariates is how to consider the covariate space. Following \cite{lee2009}, we divide the covariate space into mutually exclusive groups based on a single baseline covariate proxy index. This covariate index is obtained as a linear combination of all baseline covariates, with the weights determined by the coefficients from a linear regression of the post-treatment outcome on the covariates included in the conditioning set, in other words, the linear regression of $y_{1}$ on pre-treatment $X$.

For the empirical application on the impact of WFH on employee performance (Section \ref{Application: WFH}), the index is constructed using a linear regression of the post-treatment performance z-score on all observed baseline covariates given in Table \ref{summary WFH base}. Next, this baseline ``covariate index'' is discretized into four groups based on the predicted performance potential.\footnote{The four mutually exclusive groups are formed based on whether the predicted post-treatment performance z-score is less than -0.2541, between -0.2541 and -0.0592, between -0.0592 and 0.1029, or greater than 0.1029.}  We present two sets of generalized bounds with covariates in Table \ref{WFH covariate boundsE_newmodel}: one without assuming monotonicity, as derived in Theorem \ref{nomono_bound_cond}, and another assuming positive monotonicity across all four covariate groups, which is a specialized version of the bounds given in Theorem \ref{Gen_bound_cond}.

\begin{table}[htbp]
  \centering
   \caption{Estimated bounds and CIs for the effect of WFH on employee performance using baseline covariates}
    \label{WFH covariate boundsE_newmodel}
    \scalebox{0.88}{\begin{threeparttable}
    \begin{tabular}{ccccc}
    \toprule
          & \multicolumn{2}{c}{With Monotonicity} &       & Without Monotonicity \\
	\cmidrule{2-3}\cmidrule{5-5}    Covariate Group (j) & $[\widehat{LB}_{\tau_{OOO}}(j), \widehat{UB}_{\tau_{OOO}}(j)]$ & [$\widehat{LB}_{\tau_{ONO}}(j), \widehat{UB}_{\tau_{ONO}}(j)$] &       & $[\widehat{LB}_{\tau_{OOO}}(j), \widehat{UB}_{\tau_{OOO}}(j)]$  \\
	\cmidrule{1-3}\cmidrule{5-5}    1     & [-0.1746, 0.2981] & [-0.8511, 3.4476] &       & [-0.5642, 0.5687] \\
	    2     & [0.1857, 0.5135] & [-1.2934, 4.1300] &       & [-0.1653, 0.8137] \\
	    3     & [0.1192, 0.5562] & [-0.5525, 3.9625] &       & [-0.2224, 0.8249] \\
	    4     & [-0.0377, 0.1810] & [-0.7502, 4.0123] &       & [-0.2746, 0.4492] \\
	\cmidrule{1-3}\cmidrule{5-5}    Aggregate bound & [0.0364, 0.3985] & [-0.8299, 3.8403] &       & [-0.2876, 0.6825] \\
	    95\% CI & (-0.1348, 0.5969) & (-1.2747, 5.0699) &       & (-0.4904, 0.9405) \\
	    IM 95\% CI & (-0.1073, 0.5650) & (-1.2032, 4.8723) &       & (-0.4578, 0.8990) \\
	\bottomrule   
    \end{tabular}%
	\begin{tablenotes}[flushleft]
    \footnotesize
    \item Notes: Bounds are given in square brackets and confidence intervals in round brackets. Confidence intervals (CIs) are based on 1000 bootstrap repetitions. For Imbens and Manski (IM) CIs, \(c_n = 1.645\) satisfies the relationship given in Equation \ref{cn} for all cases considered.
    \end{tablenotes}
    \end{threeparttable}}
\end{table}%

Table \ref{WFH covariate boundsE_newmodel} shows that the length of the covariate-adjusted bound for $\tau_{OOO}$ is around 3.40\% narrower without monotonicity and 2.03\% narrower with monotonicity compared to the bounds that do not incorporate pre-treatment covariates into the analysis, as those reported in Table \ref{tab:WFH}. The covariate-adjusted bound for $\tau_{ONO}$ is around 5.38\% narrower than the bounds that do not control for covariates, as those reported in Table \ref{tab:WFH CI}.

Turning our attention to the empirical application of the NSW training program for AFDC women (Section \ref{Application NSW}), using a similar approach as above, we divide the covariate space into four groups based on a proxy for predicted wages for each individual in the post-treatment period, based on the linear regression of that outcome on the baseline covariates given in Table \ref{summary NSW unemployed base}.\footnote{The four mutually exclusive groups are formed based on whether the predicted wage (in thousands of 1982 dollars) in the treatment period is less than 7.1590, between 7.1590 and 7.7571, between 7.7571 and 8.2348, and greater than 8.2348.} Here, we follow the procedure described in Section \ref{estimation_x} to determine which direction of monotonicity holds for different partitions of the covariate space based on the relative share of individuals observed in the post-treatment period conditional on their treatment group and being observed in the pre-treatment period. Hence, we assume positive monotonicity for the two groups with the lowest and highest predicted wage potential, and negative monotonicity for the other two groups. We present two sets of generalized bounds with covariates in Table \ref{NSW covariate bounds}, where one set does not assume monotonicity as given in Theorem \ref{nomono_bound_cond} and the other set relaxes monotonicity, allowing monotonicity to hold in both directions as given in Theorem \ref{Gen_bound_cond}.

\begin{table}[H]
  \centering
    \caption{Estimated bounds and CIs for the effect of NSW training program using baseline covariates }
    \label{NSW covariate bounds}
    \scalebox{0.92}{\begin{threeparttable}
    \begin{tabular}{cccc}
    \toprule
                         & Relaxed Monotonicity &       & Without Monotonicity \\
    \cmidrule{2-2}\cmidrule{4-4}    Covariate Group $(j)$ & [$\widehat{LB}_{\tau_{OOO}}(j), \widehat{UB}_{\tau_{OOO}}(j)$] &       & [$\widehat{LB}_{\tau_{OOO}}(j), \widehat{UB}_{\tau_{OOO}}(j)$]  \\
    \cmidrule{1-1} \cmidrule{2-2} \cmidrule{4-4}    1     & [1.9157, 2.8409] &       & [-9.9061, 14.3277] \\
        2     & [0.7512, 1.0240] &       & [-8.5842, 13.4050] \\
        3     & [0.2761,0.6270] &       & [-8.3800, 8.9011] \\
        4     & [1.7915, 1.2152] &       & [-4.5863, 10.9939] \\
        \midrule
        Aggregate bound & [1.1557, 1.4990] &       & [-7.5635, 11.2723] \\
        95\% CI &        (-0.6823, 3.2228) &       & (-10.4421,14.4533) \\
        IM 95\% CI &     (-0.3869, 2.9377) &       & (-9.9794,13.9420) \\
    \bottomrule
    \end{tabular}%
    \begin{tablenotes}[flushleft]
    \footnotesize
    \item Notes: All earnings are expressed in thousands of 1982 dollars. Bounds are given in square brackets and confidence intervals in round brackets. Confidence intervals (CIs) are based on 1000 bootstrap repetitions. For Imbens and Manski (IM) CIs, \(c_n = 1.645\) satisfies the relationship given in Equation \eqref{cn} for all cases considered.
    \end{tablenotes}
    \end{threeparttable}}
\end{table}%
The length of the covariate-adjusted bound for $\tau_{OOO}$ in the absence of monotonicity is given in the last column of Table \ref{NSW covariate bounds}. This is around 8.03\% narrower than the bounds that do not control for covariates in the analysis, as presented in Table \ref{tab:NSW zero}. The covariate-adjusted bound for $\tau_{OOO}$ under the relaxed monotonicity framework is around 7.16\% narrower than the bounds without covariates assuming positive monotonicity (Table \ref{tab:NSW zero}).

\subsection{Testing Monotonicity}\label{monotonicity_test}

We extend the monotonicity test in \citet{semenova2025generalized} to our setting. First, partition the covariate space into J disjoint cells $C_j$ and index the baseline sample selection state $s\in\{0,1\}$. Under the identifying  assumption $(S_1(0),S_1(1))\perp D|S_0=s,X_i \in C_j$ the cell-specific DiD,

\begin{align*}
    \mu_{sj}=\mathbbm{E}[S_1-S_0|S_0=s,D=1,X_i\in C_j]-\mathbbm{E}[S_1-S_0|S_0=s,D=0,X_i\in C_j]
\end{align*}
identifies the causal mean $\mathbbm{E}[S_1(1)-S_1(0)|S_0=s,X_i\in C_j]$. We test the joint null $H_0:\quad S_i(1)\ge S_i(0)$ a.s. (equivalently $-\mu_{sj}\le0\quad\forall\,s,j$) with the self normalized max-T statistic,

\begin{align*}
T \;=\; \max_{\,s\in\{0,1\},\,1\le j\le J} \frac{-(\hat{\mu}_{sj})}{\hat{\sigma}_{sj}}
\end{align*}
using the self-normalized critical values from \cite{chernozhukov2019inference}. Here, $\hat\mu_{sj}
= \widehat{\Delta S}_{1,s,j} - \widehat{\Delta S}_{0,s,j}$ with $\Delta S_i=S_{i1}-S_{i0}$ and, 
\begin{align*}
    \widehat{\Delta S}_{d,s,j}
&= \frac{1}{n_{d,s,j}}
  \sum_{\substack{i : D_i = d,\\S_{i0}=s,\,X_i\in C_j}}
  \Delta S_i \quad \text{and} \quad n_{d,s,j}=\sum_i I(D_i=d)\cdot I(S_{i0}=s)\cdot I(\,X_i\in C_j)
\end{align*}

As $\Delta S\in \{1,-1,0\}$ we write the within-arm variance in closed form. Define population proportions,
\begin{align*}
    \pi^{d,s,j}_{v}= P\bigl(\Delta S_i = v \mid D_i=d,\;S_{i0}=s,\;X_i\in C_j\bigr), \quad v \in\{1,-1,0\}
\end{align*}
and their empirical analogs $\hat{\pi}^{d,s,j}_{v}$. Then,
\begin{align*}
    \sigma^2_{d,s,j}&=\mathbbm{E}\bigl[(\Delta S_i)^2 \mid d,s,j\bigr]-[\mathbbm{E}(\Delta S_i \mid d,s,j)]^2=\pi^{d,s,j}_{-1} + \pi^{d,s,j}_{+1}-(\pi^{d,s,j}_{+1} - \pi^{d,s,j}_{-1})^2
\end{align*}
and we estimate the variance using the degrees of freedom adjusted plug-in, $\hat{\sigma}^2_{d,s,j}=\frac{n_{d,s,j}}{n_{d,s,j}-1}\cdot (\hat{\pi}^{d,s,j}_{-1} + \hat{\pi}^{d,s,j}_{+1}-(\hat{\pi}^{d,s,j}_{+1} - \hat{\pi}^{d,s,j}_{-1})^2)$.\footnote{Conditioning on $S_0=s$ forces one of the $\pi^{d,s,j}_{v}$ to be zero, e.g., if $s=0$ then ${\pi}^d_{-1}=0$.} The sampling variance of $\hat{\mu}_{sj}$ is estimated as,

\begin{align*}
\widehat{var}(\hat{\mu}_{sj})=\frac{\hat{\sigma}^2_{1,s,j}}{n_{1,s,j}} +\frac{\hat{\sigma}^2_{0,s,j}}{n_{0,s,j}},
\end{align*}
and the test statistic $T$ is formed using $\hat{\sigma}_{sj}=\sqrt{\hat{var}(\hat{\mu}_{sj})}$. Cells with $n_{d,s,j}\leq 1$ do not yield a defined unbiased variance and are treated as missing. $T$ is computed over cells with valid standard errors. Rejection of $H_0$ implies that at least one cell has $\mu_{s,j}\leq 0$ and, under the assumption $(S_1(0),S_1(1))\perp D| S_0=S, X_i \in C_j$, can be interpreted as a violation of monotonicity in sample selection. 

For the WFH application, since the number of covariates is large, we implement this test for the four covariate groups that have been constructed in Section \ref{application_x}. The test does not reject monotonicity at the 5\% level (T = -1.10). Similarly, this test was performed for the four covariate groups in the NSW application in Section \ref{application_x}, which also fails to reject monotonicity at the 5\% level (T = 0.929). As this application has only seven covariates (see Appendix Table \ref{summary NSW unemployed base}), partitioning the covariate space is manageable. 
In doing so, we also converted the continuous variables into categorical variables. For example, the number of children was converted to 4 categories(no children, one child, two children, more than two children), age converted to 4 categories (greater than 50 years, 40 to 50 years, 30 to 40 years, less than 30 years) and years of education converted to 4 categories ($\geq 12$ years, 11 years,10 years, $< 10$ years). This results in 57 of 127 (45\%) divisions producing statistics high enough to reject the null of positive monotonicity at the 5\% significance level if considered in isolation.

If all seven variables are used for partitioning, it yields 207 cells, but only 63 valid cells for which test statistics can be computed, since the number of data points $n_{d,s,j}$ in most cells is too small to yield a valid variance estimate. That is a good example of the difficulties and drawbacks of considering finer partitions of the covariate space when assessing the monotonicity assumption.

\section{Proofs of Results}\label{Appendix}
\subsection{Proof of Lemma \ref{lemma_dids}}\label{L1}
\begin{proof}
Consider the DiD estimand for the observed group $(S_{0}=1,S_{1}=1)$, denoted as $\tau_{\textup{DiDs}}$.
\begin{equation}\label{didA}
    \tau_{\textup{DiDs}} =\mathbbm{E}[Y_{1}-Y_{0}|D=1,S_{0}=1,S_{1}=1]-\mathbbm{E}[Y_{1}-Y_{0}|D=0,S_{0}=1,S_{1}=1]
\end{equation}
\begin{enumerate}
    \item[1.] For the first part of Lemma \ref{lemma_dids}, we decompose $\tau_{\textup{DiDs}}$ into the overall ATT, $\tau$, and remaining terms. Consider the overall average treatment effect for the treated population:
	\begin{equation*}
		\tau = \mathbbm{E}[Y_1^\ast(1) - Y_1^\ast(0)|D = 1],
	\end{equation*}
	where $Y_1^\ast(d)$ denotes the potential outcome under treatment states $d \in \{0,1\}$ at time $t = 1$, and $S_0$, $S_1$ are binary indicators for whether a unit's outcome is observed at times $t=0$ and $t=1$, respectively.
	
	Then $\tau$ admits the following decomposition:
	\begin{align*}
		\tau =\;
		& \mathbbm{E}\left[ Y_1^\ast(1) - Y_1^\ast(0) | D=1,S_0 = 1,S_1 = 1 \right] \cdot \mathbbm{P}(S_0 = 1, S_1 = 1 | D = 1) \\
		+ & \mathbbm{E}\left[ Y_1^\ast(1) - Y_1^\ast(0) | D=1, S_0 = 1,S_1 = 0 \right] \cdot \mathbbm{P}(S_0 = 1,S_1 = 0 | D = 1) \\
		+ & \mathbbm{E}\left[ Y_1^\ast(1) - Y_1^\ast(0) | D=1, S_0 = 0,S_1 = 1 \right] \cdot \mathbbm{P}(S_0 = 0,S_1 = 1 | D = 1) \\
		+ & \mathbbm{E}\left[ Y_1^\ast(1) - Y_1^\ast(0) | D=1, S_0 = 0,S_1 = 0 \right] \cdot \mathbbm{P}(S_0 = 0,S_1 = 0 | D = 1)
	\end{align*}
	In the expression above, replace 
	\begin{align*}
		\mathbbm{P}(S_0 = 1,S_1 = 1 | D = 1) &= 1- \mathbbm{P}(S_0 = 1,S_1 = 0 | D = 1) - \mathbbm{P}(S_0 = 0,S_1 = 1 | D = 1) \\
		& - \mathbbm{P}(S_0 = 0,S_1 = 0 | D = 1)
	\end{align*} 
	which gives us
	\begin{align}\label{DiDs:terms}
		\tau =\;
		& \mathbbm{E}\left[Y_1^\ast(1) - Y_1^\ast(0) | D=1,S_0 = 1,S_1 = 1 \right] \nonumber \\
		+ & \big(\mathbbm{E}\left[Y_1^\ast(1) - Y_1^\ast(0) | D=1, S_0 = 1,S_1 = 0 \right]-\mathbbm{E}\left[Y_1^\ast(1) - Y_1^\ast(0) | D=1,S_0 = 1,S_1 = 1 \right]\big) \nonumber \\ & \times \mathbbm{P}(S_0 = 1,S_1 = 0 | D = 1) \nonumber \\
		+ & \big(\mathbbm{E}\left[ Y_1^\ast(1) - Y_1^\ast(0) | D=1, S_0 = 0,S_1 = 1 \right] -\mathbbm{E}\left[Y_1^\ast(1) - Y_1^\ast(0) | D=1,S_0 = 1,S_1 = 1 \right] \big) \nonumber \\
		& \times \mathbbm{P}(S_0 = 0,S_1 = 1 | D = 1) \nonumber \\
		+ & \big(\mathbbm{E}\left[ Y_1^\ast(1) - Y_1^\ast(0) | D=1, S_0 = 0,S_1 = 0 \right]-\mathbbm{E}\left[Y_1^\ast(1) - Y_1^\ast(0) | D=1,S_0 = 1,S_1 = 1 \right]  \big) \nonumber \\
		&\times \mathbbm{P}(S_0 = 0,S_1 = 0 | D = 1)
	\end{align}
	Let's consider the first term in \eqref{DiDs:terms} above, 
	\begin{align}
		&\mathbbm{E}\left[Y_1^\ast(1) - Y_1^\ast(0) | D=1,S_0 = 1,S_1 = 1 \right] \nonumber \\
		=& \mathbbm{E}\left[Y_1^\ast(1) - Y_0^\ast(1)|D=1,S_0 = 1,S_1 = 1 \right]-\mathbbm{E}\left[Y_1^\ast(0)-Y_0^\ast(0)| D=1,S_0 = 1,S_1 = 1 \right] \nonumber\\
		=& \tau_{\textup{DiDs}}+\mathbbm{E}[Y^\ast_1(0)-Y^\ast_0(0)|D=0, S_0=1, S_1=1]-\mathbbm{E}\left[Y_1^\ast(0)-Y_0^\ast(0)| D=1,S_0 = 1,S_1 = 1 \right] \label{eq: Lemma1ATT2}
	\end{align}
	Plugging \eqref{eq: Lemma1ATT2} in \eqref{DiDs:terms} we get, 
	\begin{align*}
		\tau_{\textup{DiDs}} &= \tau+\mathbbm{E}\left[Y_1^\ast(0)-Y_0^\ast(0)| D=1,S_0 = 1,S_1 = 1 \right] -\mathbbm{E}[Y^\ast_1(0)-Y^\ast_0(0)|D=0, S_0=1, S_1=1] \\
		+ & \big(\mathbbm{E}\left[Y_1^\ast(1) - Y_1^\ast(0) | D=1,S_0 = 1,S_1 = 1 \right]-\mathbbm{E}\left[Y_1^\ast(1) - Y_1^\ast(0) | D=1, S_0 = 1,S_1 = 0 \right]\big) \nonumber \\ & \times \mathbbm{P}(S_0 = 1,S_1 = 0 | D = 1) \nonumber \\
		+ & \big(\mathbbm{E}\left[Y_1^\ast(1) - Y_1^\ast(0) | D=1,S_0 = 1,S_1 = 1 \right]-\mathbbm{E}\left[ Y_1^\ast(1) - Y_1^\ast(0) | D=1, S_0 = 0,S_1 = 1 \right] \big) \nonumber \\
		& \times \mathbbm{P}(S_0 = 0,S_1 = 1 | D = 1) \nonumber \\
		+ & \big(\mathbbm{E}\left[Y_1^\ast(1) - Y_1^\ast(0) | D=1,S_0 = 1,S_1 = 1 \right]-\mathbbm{E}\left[ Y_1^\ast(1) - Y_1^\ast(0) | D=1, S_0 = 0,S_1 = 0 \right] \big) \nonumber \\
		&\times \mathbbm{P}(S_0 = 0,S_1 = 0 | D = 1) 
	\end{align*}
    \item[2.] For the second part of Lemma \ref{lemma_dids} we decompose $\tau_{\textup{DiDs}}$ into the ATT for different latent groups and remaining terms (bias). 
Let $p_{OOO1}=\frac{\pi_{OOO1}}{\pi_{OOO1}+\pi_{ONO1}}$. Naturally, $1-p_{OOO1}= \frac{\pi_{ONO1}}{\pi_{OOO1}+\pi_{ONO1}}$. Now let us consider $\mathbbm{E}[Y_{1}-Y_{0}|D=1,S_{0}=1,S_{1}=1]$ which can be decomposed as follows,
{\small \begin{align}\label{mixpA}
     &\mathbbm{E}[Y_{1}-Y_{0}|D=1,S_{0}=1, S_{1}=1]= \nonumber \\
    &= \mathbbm{E}[Y_{1}^{\ast}(1)-Y_{0}^{\ast}(1)|D=1,S_{0}=1, (S_{1}(0)=1,S_{1}(1)=1) \quad or \quad (S_{1}(0)=0,S_{1}(1)=1)] \nonumber \\
    &=\mathbbm{E}[Y_{1}^{\ast}(1)-Y_{0}^{\ast}(1)|D=1,OOO]\cdot \frac{\pi_{OOO1}}{\pi_{OOO1}+\pi_{ONO1}}\nonumber \\
    &+\mathbbm{E}[Y_{1}^{\ast}(1)-Y_{0}^{\ast}(1)|D=1,ONO]\cdot \frac{\pi_{ONO1}}{\pi_{OOO1}+\pi_{ONO1}}\nonumber \\
    &= \mathbbm{E}[Y_{1}^{\ast}(1)-Y_{0}^{\ast}(1)|D=1,OOO] \cdot p_{OOO1} +\mathbbm{E}[Y_{1}^{\ast}(1)-Y_{0}^{\ast}(1)|D=1,ONO] \cdot (1-p_{OOO1}) \nonumber \\
    & \text{Assumption \ref{no anti}}\nonumber \\
    & = \mathbbm{E}[Y_{1}^{\ast}(1)-Y_{0}^{\ast}(0)|D=1,OOO] \cdot p_{OOO1} +\mathbbm{E}[Y_{1}^{\ast}(1)-Y_{0}^{\ast}(0)|D=1,ONO] \cdot (1-p_{OOO1})
\end{align}}
 Similarly, let $p_{OOO0}=\frac{\pi_{OOO0}}{\pi_{OOO0}+\pi_{OON0}}$. Naturally, $1-p_{OOO0}= \frac{\pi_{OON0}}{\pi_{OOO0}+\pi_{OON0}}$. Then $\mathbbm{E}[Y_{1}-Y_{0}|D=0,S_{0}=1,S_{1}=1]$ can be decomposed,
{\small \begin{align}\label{mixp2A}
     &\mathbbm{E}[Y_{1}-Y_{0}|D=0,S_{0}=1, S_{1}=1]= \nonumber  \\
    &= \mathbbm{E}[Y_{1}^{\ast}(0)-Y_{0}^{\ast}(0)|D=0,S_{0}=1, (S_{1}(0)=1,S_{1}(1)=1) \quad or \quad (S_{1}(0)=1,S_{1}(1)=0)] \nonumber \\
    &=\mathbbm{E}[Y_{1}^{\ast}(0)-Y_{0}^{\ast}(0)|D=0,OOO]\cdot \frac{\pi_{OOO0}}{\pi_{OOO0}+\pi_{OON0}}\nonumber \\
    &+\mathbbm{E}[Y_{1}^{\ast}(0)-Y_{0}^{\ast}(0)|D=0,OON]\cdot \frac{\pi_{OON0}}{\pi_{OOO0}+\pi_{OON0}}\nonumber \\
    &= \mathbbm{E}[Y_{1}^{\ast}(0)-Y_{0}^{\ast}(0)|D=0,OOO] \cdot p_{OOO0} +\mathbbm{E}[Y_{1}^{\ast}(0)-Y_{0}^{\ast}(0)|D=0,OON] \cdot (1-p_{OOO0})
\end{align}}
Combining the results of equation \ref{mixpA} and equation \ref{mixp2A} with equation \ref{didA}, we get
\begin{align*}
    \tau_{\textup{DiDs}} &=\mathbbm{E}[Y_{1}-Y_{0}|D=1,S_{0}=1,S_{1}=1]-\mathbbm{E}[Y_{1}-Y_{0}|D=0,S_{0}=1,S_{1}=1]\\
    &=\mathbbm{E}[Y_{1}^{\ast}(1)-Y_{0}^{\ast}(0)|D=1,OOO] \cdot p_{OOO1} +\mathbbm{E}[Y_{1}^{\ast}(1)-Y_{0}^{\ast}(0)|D=1,ONO] \cdot (1-p_{OOO1}) \\
    & - \mathbbm{E}[Y_{1}^{\ast}(0)-Y_{0}^{\ast}(0)|D=0,OOO] \cdot p_{OOO0} -\mathbbm{E}[Y_{1}^{\ast}(0)-Y_{0}^{\ast}(0)|D=0,OON] \cdot (1-p_{OOO0}) \\
    & = \tau_{OOO}\cdot p_{OOO1} +\tau_{ONO}\cdot (1-p_{OOO1}) \\
    & + \mathbbm{E}[Y_1^\ast(0)-Y_0^\ast(0)|D=1, OOO]\cdot p_{OOO1} +\mathbbm{E}[Y_1^\ast(0)-Y_0^\ast(0)D=1, ONO]\cdot (1-p_{OOO1}) \\
    & - \mathbbm{E}[Y_{1}^{\ast}(0)-Y_{0}^{\ast}(0)|D=0,OOO] \cdot p_{OOO0} -\mathbbm{E}[Y_{1}^{\ast}(0)-Y_{0}^{\ast}(0)|D=0,OON] \cdot (1-p_{OOO0})\\
    & = \tau_{OOO}\cdot p_{OOO1} +\tau_{ONO}\cdot (1-p_{OOO1}) \\
    & + \big\{\mathbbm{E}[Y_1^\ast(0)-Y_0^\ast(0)|D=1, OOO]-\mathbbm{E}[Y_1^\ast(0)-Y_0^\ast(0)|D=0, OOO]\big\}\cdot p_{OOO1} \\
    & + \mathbbm{E}[Y_1^\ast(0)-Y_0^\ast(0)|D=0, OOO]\cdot p_{OOO1} + \mathbbm{E}[Y_1^\ast(0)-Y_0^\ast(0)|D=1, ONO]\cdot (1-p_{OOO1}) \\
    & - \mathbbm{E}[Y_{1}^{\ast}(0)-Y_{0}^{\ast}(0)|D=0,OOO] \cdot p_{OOO0} -\mathbbm{E}[Y_{1}^{\ast}(0)-Y_{0}^{\ast}(0)|D=0,OON] \cdot (1-p_{OOO0})\\
    & \text{Assumption \ref{PT_OOO}}\\  
    & = \tau_{OOO}\cdot p_{OOO1} +\tau_{ONO}\cdot (1-p_{OOO1}) \\ 
    & + \mathbbm{E}[Y_1^\ast(0)-Y_0^\ast(0)|D=0, OOO]\cdot p_{OOO1} + \mathbbm{E}[Y_1^\ast(0)-Y_0^\ast(0)|D=1, ONO]\cdot (1-p_{OOO1}) \\
    & - \mathbbm{E}[Y_{1}^{\ast}(0)-Y_{0}^{\ast}(0)|D=0,OOO] \cdot p_{OOO0} -\mathbbm{E}[Y_{1}^{\ast}(0)-Y_{0}^{\ast}(0)|D=0,OON] \cdot (1-p_{OOO0}) \\
    & = \tau_{OOO}\cdot p_{OOO1} +\tau_{ONO}\cdot (1-p_{OOO1}) \\ 
    & + \big\{\mathbbm{E}[Y_1^\ast(0)-Y_0^\ast(0)|D=1, ONO] -\mathbbm{E}[Y_1^\ast(0)-Y_0^\ast(0)|D=0, ONO] \big\}\cdot (1-p_{OOO1}) \\
    & + \mathbbm{E}[Y_1^\ast(0)-Y_0^\ast(0)|D=0, OOO] \cdot p_{OOO1}+\mathbbm{E}[Y_1^\ast(0)-Y_0^\ast(0)|D=0, ONO]\cdot (1-p_{OOO1}) \\
    & - \mathbbm{E}[Y_{1}^{\ast}(0)-Y_{0}^{\ast}(0)|D=0,OOO] \cdot p_{OOO0} -\mathbbm{E}[Y_{1}^{\ast}(0)-Y_{0}^{\ast}(0)|D=0,OON] \cdot (1-p_{OOO0}) \\
    \end{align*}
    \begin{align*}
    & = \tau_{OOO}\cdot p_{OOO1} +\tau_{ONO}\cdot (1-p_{OOO1}) \\ 
    & + \big\{\mathbbm{E}[Y_1^\ast(0)-Y_0^\ast(0)|D=1, ONO] -\mathbbm{E}[Y_1^\ast(0)-Y_0^\ast(0)|D=0, ONO] \big\}\cdot (1-p_{OOO1}) \\
    & +\big\{ \mathbbm{E}[Y_1^\ast(0)-Y_0^\ast(0)|D=0, OOO] - \mathbbm{E}[Y_1^\ast(0)-Y_0^\ast(0)|D=0, ONO]\big\}\cdot p_{OOO1} \\
    & + \big\{\mathbbm{E}[Y_1^\ast(0)-Y_0^\ast(0)|D=0, ONO] - \mathbbm{E}[Y_1^\ast(0)-Y_0^\ast(0)|D=0, OON]\big\} \\
    & + \big\{\mathbbm{E}[Y_1^\ast(0)-Y_0^\ast(0)|D=0, OON] - \mathbbm{E}[Y_1^\ast(0)-Y_0^\ast(0)|D=0, OOO]\big\}\cdot p_{OOO0}
\end{align*}
If we assume positive monotonicity, then $p_{OOO0}=1$. Therefore, $\tau_{\textup{DiDs}}$ simplifies to just:
\begin{align*}
    \tau_{\textup{DiDs}} &= \tau_{OOO}\cdot p_{OOO1} +\tau_{ONO}\cdot (1-p_{OOO1}) \\ 
    & + \big\{\mathbbm{E}[Y_1^\ast(0)-Y_0^\ast(0)|D=1, ONO] -\mathbbm{E}[Y_1^\ast(0)-Y_0^\ast(0)|D=0, ONO] \big\}\cdot (1-p_{OOO1}) \\
    & +\big\{ \mathbbm{E}[Y_1^\ast(0)-Y_0^\ast(0)|D=0, ONO] -\mathbbm{E}[Y_1^\ast(0)-Y_0^\ast(0)|D=0, OOO] \big\}\cdot (1-p_{OOO1}).
\end{align*}
\end{enumerate}
\end{proof}

\subsection{Testing plausibility of the parallel trends for OOO using pre-treatment data}\label{pretrend}

It is intuitive to consider testing the plausibility of the parallel trends for OOO using pre-treatment data. We show that a placebo DiD test based on pre-treatment periods $t=-1,0,1$, which conditions on individuals who are observed in all three periods, gives a decomposition that consists of other latent types in addition to the OOO type. A nonzero pre-trend can arise from either violations of parallel trends within the always-observed group or from differences in untreated trends across other latent groups. As a result, the placebo DiD test detects only joint violations of assumptions and does not provide a clean, standalone test of the parallel trends assumption for the OOO group.
            Consider the following placebo DiD estimand, constructed from pre-treatment periods $t \in \{-1,0\}$ and conditioning on units observed in all three pre-treatment periods $t=\{-1,0,1\}$:

            \begin{align*}
            \tau^{\textup{pre}}_{\textup{DiDs}} &= \mathbbm{E}[Y_{0} - Y_{-1}\mid D=1, S_{-1}=1, S_0=1, S_1=1] \\
            &- \mathbbm{E}[Y_{0} - Y_{-1}\mid D=0, S_{-1}=1, S_0=1, S_1=1]
            \end{align*}
        
            Let
        	\begin{align*}
        		\tilde{\pi}_{gd} &\equiv  \mathbbm{P}(S_{-1}(0)=1, G=g, D=d) \\
                \tilde{p}_{g0} &\equiv \frac{\tilde{\pi}_{g0}}{\sum_{h\in\{OOO,OON\}}\tilde{\pi}_{h0}},  \qquad
                \tilde{p}_{g1} \equiv \frac{\tilde{\pi}_{g1}}{\sum_{h\in\{OOO,ONO\}}\tilde{\pi}_{h1}} \\
                \tilde{\mu}_{gd} &\equiv \mathbbm{E}[Y_{0}(0)-Y_{-1}(0)\mid S_{-1}(0)=1, G=g, D=d].
        	\end{align*}
            Then, 
        	\begin{align}
        		\mathbbm{E}[Y_{0} - Y_{-1}\mid D=1, S_{-1}=1, S_0=1, S_1=1] 
        		& =  \sum_{g\in\{OOO,ONO\}}\tilde{\mu}_{g1} \cdot \tilde{p}_{g1} \label{g1}  
        	\end{align}
        	  and similarly, 
        	\begin{align}
        		\mathbbm{E}[Y_{0} - Y_{-1}\mid D=0, S_{-1}=1, S_0=1, S_1=1] 
        		& =  \sum_{g\in\{OOO,OON\}}\tilde{\mu}_{g0} \cdot \tilde{p}_{g0} \label{g0}
        	\end{align}
            where the first equality in equations \eqref{g1} and \eqref{g0} follow from no-anticipation in selection and outcomes in the pre-treatment periods. The second equality simply applies the definitions of pre-treatment means and mixing proportions that were introduced above.
            Combining \eqref{g1} and \eqref{g0} gives us
            \begin{align}
            	\tau^{\textup{pre}}_{\textup{DiDs}} 
            	&= \sum_{g\in\{OOO,ONO\}}\tilde{\mu}_{g1} \cdot \tilde{p}_{g1}- \sum_{g\in\{OOO,OON\}}\tilde{\mu}_{g0} \cdot \tilde{p}_{g0} \nonumber \\
            	& = \big(\tilde{\mu}_{OOO1} - \tilde{\mu}_{OOO0}\big)\cdot \tilde{p}_{OOO1}+\big(\tilde{\mu}_{ONO1} - \tilde{\mu}_{ONO0}\big)\cdot \tilde{p}_{ONO1} \nonumber \\
            	&+\big(\tilde{\mu}_{OOO0} - \tilde{\mu}_{ONO0}\big)\cdot \tilde{p}_{OOO1} + \big(\tilde{\mu}_{ONO0} - \tilde{\mu}_{OON0}\big) \nonumber \\
            	& +\big(\tilde{\mu}_{OON0} - \tilde{\mu}_{OOO0}\big)\cdot \tilde{p}_{OOO0} \label{precombine}
            \end{align}
            The first and second terms in equation \eqref{precombine} capture the average untreated potential outcome trends between the treated and untreated $OOO$ and $ONO$ individuals who are also observed in the pre-treatment period i.e. $S_{-1}(0)=1$; weighted by their relative proportions in the treated group. The third, fourth, and fifth terms represent the average cross-group trend differential among the untreated $OOO/ONO$, $ONO/OON$, and $OON/OOO$ latent groups, respectively. 
        
            Looking at this pre-trend decomposition makes it clear that a non-zero placebo DiD estimand, $\tau^{\text{pre}}_{\text{DiDs}}$, indicates that at least one of these components is nonzero, but it does not identify which one. In particular, rejection of the null of no difference based on this placebo test could be due to:
            (a) non-parallel untreated outcome trends for the OOO or ONO groups, or
            (b) differences in trends across untreated latent types (e.g., ONO vs. OON). Therefore, such a test could only pick up joint violations of assumptions and may not be able to conclusively test for the plausibility of Assumption \ref{PT_OOO} (PT for OOO) alone. 


\subsection{Proof of Lemma \ref{lemma:pipartialid}}
\begin{proof}\label{proof:pipartialid} \ 
\begin{itemize}
    \item[(a)] For the treated group, $D=1$, the Fr\'echet bounds are given as
	{\footnotesize	\begin{align}\label{eq:fb1}
 \mathbbm{P}[S_{1}(0)=1, S_{1}(1)=1|D=1, S_{0}=1]
  &\in \left[\max\{\mathbbm{P}[S_{1}(0)=1|D=1, S_{0}=1]+\mathbbm{P}[S_{1}(1)=1|D=1, S_{0}=1]-1,0\},\right.\nonumber\\
 & \quad \left.\min\{\mathbbm{P}[S_{1}(0)=1|D=1, S_{0}=1],\mathbbm{P}[S_{1}(1)=1|D=1, S_{0}=1]\}\right].
\end{align}}
Notice that 
\begin{equation}\label{eq:lemma2_11}
	\mathbbm{P}[S_{1}(1)=1|D=1, S_{0}=1]= \mathbbm{P}[S_{1}=1|D=1, S_{0}=1]
\end{equation}
which means that it can be directly identified from observed data.  Assumption \ref{Partialselection}(a) implies that
\begin{equation}\label{eq:lemma2_12}
	\begin{split}
	\mathbbm{P}[S_{1}(0)=1|D=1, S_{0}=1] &=\mathbbm{P}[S_{1}(0)=1|D=0, S_{0}=1] \\
	& = \mathbbm{P}[S_{1}=1|D=0, S_{0}=1] 
	\end{split}
\end{equation} 
Plugging \eqref{eq:lemma2_11} and \eqref{eq:lemma2_12} into \eqref{eq:fb1}, we get the observable bounds
	{\footnotesize	\begin{align}\label{eq:fb1_obs}
		\mathbbm{P}[S_{1}(0)=1, S_{1}(1)=1|D=1, S_{0}=1]\in \left[\max\{\mathbbm{P}[S_{1}=1|D=0, S_{0}=1]+\mathbbm{P}[S_{1}=1|D=1, S_{0}=1]-1,0\},\right.\nonumber\\
		\left.\min\{\mathbbm{P}[S_{1}=1|D=0, S_{0}=1],\mathbbm{P}[S_{1}=1|D=1, S_{0}=1]\}\right].
\end{align}}
\item[(b)] Analogously, using Assumption \ref{Partialselection}(b) for the $D=0$ group gets us the observable bounds
	{\footnotesize	\begin{align}\label{eq:fb0_obs}
	 \mathbbm{P}[S_{1}(0)=1, S_{1}(1)=1|D=0, S_{0}=1]\in \left[\max\{\mathbbm{P}[S_{1}=1|D=0, S_{0}=1]+\mathbbm{P}[S_{1}=1|D=1, S_{0}=1]-1,0\},\right.\nonumber\\
	\left.\min\{\mathbbm{P}[S_{1}=1|D=0, S_{0}=1],\mathbbm{P}[S_{1}=1|D=1, S_{0}=1]\}\right].
\end{align}} Equations \eqref{eq:fb1_obs} and \eqref{eq:fb0_obs} taken together show the result.
\end{itemize}	
\end{proof}
\subsection{Proof of Theorem \ref{nomono_bound}}\label{proof:nomono_bound}
\begin{proof}
The treatment effect of treated for OOO group ($\tau_{OOO}$) can be decomposed as follows,
\begin{align*}
    \tau_{OOO} &= \mathbbm{E}[Y_{1}^{\ast}(1)-Y_{1}^{\ast}(0)|D=1, OOO]\\
    &= \mathbbm{E}[Y_{1}^{\ast}(1)-Y_{0}^{\ast}(1)+Y_{0}^{\ast}(1)-Y_{1}^{\ast}(0)|D=1, OOO]\\
    &= \mathbbm{E}[Y_{1}^{\ast}(1)-Y_{0}^{\ast}(1)|D=1, OOO]-\mathbbm{E}[Y_{1}^{\ast}(0)-Y_{0}^{\ast}(1)|D=1, OOO]\\
    &\text{Assumption \ref{no anti}}\\
    &= \mathbbm{E}[Y_{1}^{\ast}(1)-Y_{0}^{\ast}(1)|D=1, OOO]-\mathbbm{E}[Y_{1}^{\ast}(0)-Y_{0}^{\ast}(0)|D=1, OOO]\\
    &\text{Assumption \ref{PT_OOO}}\\
    &= \mathbbm{E}[Y_{1}^{\ast}(1)-Y_{0}^{\ast}(1)|D=1, OOO]-\mathbbm{E}[Y_{1}^{\ast}(0)-Y_{0}^{\ast}(0)|D=0, OOO]
    \end{align*}  
    For each $d=0,1$, the decompositions of $\mathbbm{E}[Y_{1}-Y_{0}|D=d,S_{0}=1,S_{1}=1]$ given in equations \eqref{mixp} and \eqref{mixp2} lead to partial identification of $\mathbbm{E}[Y_{1}^{\ast}(1)-Y_{0}^{\ast}(d)|D=d,OOO]$ which lies within the interval $[LB_{gd},UB_{gd}]$ 
    where,
    \begin{align}\label{LBUB}
        LB_{gd} &=\mathbbm{E}[Y_{1}-Y_{0}|D=d,S_{0}=1,S_{1}=1, (Y_{1}-Y_{0})\leq F_{\Delta Y|d11}^{-1}(p_{gd})] \nonumber \\
        UB_{gd} &=\mathbbm{E}[Y_{1}-Y_{0}|D=d,S_{0}=1,S_{1}=1, (Y_{1}-Y_{0}) > F_{\Delta Y|d11}^{-1}(1-p_{gd})]
    \end{align}
    where $F_{\Delta Y|d11}^{-1}(.)$ is the quantile function of the distribution of $Y_{1}-Y_{0}$ given $D=d,S_{0}=1,S_{1}=1$. 
    Using Lemma \ref{lemma:pipartialid}, we know that for any $v_d \in [\max\{\mathbbm{P}[S_1=1|D=d, S_0=1]+\mathbbm{P}[S_1=1|D=1-d, S_0=1]-1, 0\}, \min\{\mathbbm{P}[S_1=1|D=d, S_0=1], \mathbbm{P}[S_1=1|D=1-d, S_0=1]\}]$, one can obtain the associated mixing probabilities/weights as $p_{gd}(v_d) = \frac{v_d}{\mathbbm{P}[S_1=1|D=d, S_0=1]}$. Evaluating equations (\ref{LBUB}) for each value of $d$ at the least favorable values for $p_{gd}(v_d)$ yields,
    \begin{align}
        LB_{gd}(v^{l}_{d}) &=\mathbbm{E}[Y_{1}-Y_{0}|D=d,S_{0}=1,S_{1}=1, (Y_{1}-Y_{0})\leq F_{\Delta Y|d11}^{-1}(p_{gd}(v^{l}_{d}))]\label{LBlow}\\
        UB_{gd}(v^{l}_{d}) &=\mathbbm{E}[Y_{1}-Y_{0}|D=1,S_{0}=1,S_{1}=1, (Y_{1}-Y_{0}) > F_{\Delta Y|d11}^{-1}(1-p_{gd}(v^{l}_{d}))]. \label{UBlow}
    \end{align} where $v_d^l = \max\{\mathbbm{P}[S_1=1|D=d, S_0=1]+\mathbbm{P}[S_1=1|D=1-d, S_0=1]-1, 0\}$ is the lower bound of the identified set $\mathbbm{P}[S_{1}(0)=1, S_{1}(1)=1|D=d, S_{0}=1]$ for $d=0,1$. Therefore, combining the bounds for the treated and untreated subpopulations along with point identification of the weights implies that $\tau_{OOO} \in [LB_{\tau_{OOO}}, UB_{\tau_{OOO}}]$ where 
    \begin{align*}
        LB_{\tau_{OOO}}= LB_{OOO1}(v^l_1)-UB_{OOO0}(v^l_0) \text{ and } \ UB_{\tau_{OOO}} = UB_{OOO1}(v^l_1)-LB_{OOO0}(v^l_0).
    \end{align*} 
\end{proof}
\subsubsection{Proof of Sharpness of Bounds given in Theorem \ref{nomono_bound} } \label{proof_nomo bound sharp}
  \begin{proof}
      Let $OOO\equiv\{S_0=1,S_1(0)=1,S_1(1)=1\}$ and $\Delta Y\equiv Y_1-Y_0$. Under Assumptions \ref{no anti},\ref{Partialselection}(a) and \ref{Partialselection}(b), fix the least  favorable always observe shares $p_1=p_{OOO1}(v^{l}_{1})\in (0,1]$ and $p_0=p_{OOO0}(v^{l}_{0})\in (0,1]$. For ease of notation, let $F_d$ denote the cdf of $\Delta Y$ in the observed selected sample $\{S_0=1,S_1=1,D=d\},d\in \{0,1\}$. Define the trimmed means,
      \begin{align*}
          L_d &=\mathbbm E[\Delta Y|D=d,S_0=1,S_1=1,\Delta Y \leq q_d(p_d)]\\
          U_d &=\mathbbm E[\Delta Y|D=d,S_0=1,S_1=1,\Delta Y > q_d(1-p_d)]
      \end{align*}
      where $q_d(.)=F_d^{-1}(.)$. Then the identified interval for $\tau_{OOO}$ is exactly the interval $[L_1-U_0,U_1-L_0]$ as given in Theorem  \ref{nomono_bound}. Equivalently, every value in  $[L_1-U_0, U_1-L_0]$ is attainable by a DGP satisfying the assumptions and reproducing the observed $(F_1,F_0)$ and the shares $(p_1,p_0)$ and no value outside the interval is attainable.

      The proof follows a similar approach to that of \citet{lee2009}. We will use the mixture decomposition argument in the treated and the control groups separately. For a given group $d\in\{0,1\}$,the observed selected sample distribution $F_d$ is generated by a mixture of two latent groups,
      \begin{align*}
    F_d(\Delta y)= p_{d}\cdot N_d(\Delta y)+(1-p_{d})\cdot M_d(\Delta y)
\end{align*}

Where $N_d$ is the latent cdf of $\Delta Y^*_d$ in the always-observed latent group OOO, and the $M_d$ is the latent cdf of the complementary stratum (i.e., ONO when $d=1$, and OON when $d=0$).The object of interest is the mean $\mu_d=\mathbbm E_{N_d}[\Delta Y^*_d]$ where $\tau_{OOO}=\mu_1-\mu_0$. 

Then by \citet{lee2009} lemma 1 we have,
{\footnotesize
\begin{align*}
    U_d=\frac{1}{p_{d}}\cdot \int_{\Delta y_{1-p_{d}}}^{\infty}\Delta y\cdot dF_d(\Delta y)&\geq \int_{-\infty}^{\infty} \Delta y\cdot dN_d(\Delta y)=\mu_d
\end{align*}}

$p_{d}$ is uniquely determined by data and it follows from the \citet{lee2009} lemma 1 that conditional on $p_{d}$, $U_d$ is a sharp upper bound. Therefore, $U_d$ is the maximum possible value for $\mu_d$. A similar argument will follow for the lower bound. Thus $L_d$ is the minimum possible value for $\mu_d$. Hence,
\begin{align*}
    L_d\leq \mu_d \leq U_d
\end{align*} 

Now we prove the attainability of this interval. Let $y_d^+=q_d(1-p_d)$ and $y_d^-=q_d(p_d)$. Let $N_d^U$ be the conditional distribution of $\Delta Y$ on the upper tail, $Up\equiv\{\Delta Y >y_d^+\}$ and $N_d^L$ be the conditional distribution of $\Delta Y$ on the lower tail, $Low\equiv\{\Delta Y <y_d^-\}$. Then,
\begin{align*}
    N_d^L(\Delta y)=\frac{\mathbbm{P}(\Delta Y\leq \Delta y,Low)}{\mathbbm{P}(Low)}
\end{align*}
In the case when $\Delta y < y_d^-$, then $\{\Delta Y < y_d^-\}\subseteq \{Low\}$ so $\mathbbm{P}(\Delta Y\leq \Delta y,Low)=\mathbbm{P}(\Delta Y\leq \Delta y)=F_d(\Delta y)$. Hence, $N_d^L(\Delta y)=\frac{F_d(\Delta y)}{p_d}$. In the case when $\Delta y \geq y_d^-$ the $\mathbbm{P}(\Delta Y\leq \Delta y,Low)=\mathbbm{P}(Low)=p_d$. Hence  $N_d^L(\Delta y)=1$. Combining the two notation in one line with $a\wedge b=min\{\ a,b\}$, We have,
\begin{align*}
    N_d^L(\Delta y)=\frac{F_d(\Delta y)\wedge p_d}{p_d}
\end{align*}
as if $\Delta y < y_d^-$ then $F_d(\Delta y)\leq p_d$ then the numerator is $F_d(\Delta Y)$ and for $\Delta y \geq y_d^-$, $F_d(\Delta y)\geq p_d$ and the numerator is $p_d$. similarly,

\begin{align*}
    N_d^U(\Delta y)=\frac{\mathbbm{P}(\Delta Y\leq \Delta y,Up)}{\mathbbm{P}(Up)}
\end{align*}
In the case when $\Delta y < y_d^+$, then $\{\Delta y > y_d^+\}$ is disjoint from $\{\Delta Y\leq \Delta y\}$. so $\mathbbm{P}(\Delta Y\leq \Delta y,Up)=0$. Hence, $N_d^L(\Delta y)=0$. In the case when $\Delta y \geq y_d^+$ the $\mathbbm{P}(\Delta Y\leq \Delta y,Up)=\mathbbm{P}(y_d^+<\Delta Y< \Delta y)=F_d(\Delta y)-F_d(y_d^+)=F_d(\Delta y)-(1-p_d)$. Hence, $N_d^U(\Delta y)=\frac{F_d(\Delta y)-(1-p_d)}{p_d}$. Combining the two notation in one line with $(x)_+=max\{\ x,0\}$, We have,
\begin{align*}
    N_d^U(\Delta y)=\frac{\big (F_d(\Delta Y)-(1- p_d)\big )_+}{p_d}
\end{align*}
as if $\Delta y < y_d^+$ then $F_d(\Delta y)\leq F_d(y_d^+)=1-p_d$ then the $F_d(\Delta Y)-(1- p_d)$ is negative $\big (F_d(\Delta Y)-(1- p_d)\big )_+=0$. when $\Delta y > y_d^+$ then $F_d(\Delta y)\geq F_d(y_d^+)=1-p_d$ and $F_d(\Delta Y)-(1- p_d)$ is positive, then  $\big (F_d(\Delta Y)-(1- p_d)\big )_+=F_d(\Delta Y)-(1- p_d)$.

We define the complementary other groups by residuals,
\begin{align*}
    M_d^U(\Delta y)=\frac{F_d(\Delta y)-p_d\cdot N_d^U(\Delta y)}{1-p_d} \quad M_d^L(\Delta y)=\frac{F_d(\Delta y)-p_d\cdot N_d^L(\Delta y)}{1-p_d}
\end{align*}
Direct check shows for all $\Delta y$,
\begin{align*}
    p_d\cdot N_d^U(\Delta y)+(1-p_d) \cdot M_d^U(\Delta y)=F_d(\Delta y)\\
    p_d\cdot N_d^L(\Delta y)+(1-p_d) \cdot M_d^L(\Delta y)=F_d(\Delta y)
\end{align*}
Moreover, $N_d^U(\Delta y)$ is precisely $F_d$ restricted to the upper tail and re normalized and  $N_d^L(\Delta y)$ is $F_d$ restricted to the lower tail and re normalized. Therefore, we have,
\begin{align*}
    \mathbbm E_{N_d^U}[\Delta Y]=U_d \quad  \mathbbm E_{N_d^L}[\Delta Y]=L_d
\end{align*} 

Thus, endpoint bounds are attainable by a feasible decomposition that reproduces $F_d$

Now we show the attainability of any interior value. Let $N_d^\lambda=\lambda\cdot N_d^U+(1-\lambda)\cdot N_d^L$ and define $M_d^\lambda$ by $p_d \cdot N_d^\lambda+ (1-p_d)\cdot M_d^\lambda=F_d$ (i.e. this puts the remainder of the mass to the other complementary latent type). Then, $\mathbbm E_{N_d^\lambda} [\Delta Y]=\lambda \cdot U_d+(1-\lambda) \cdot L_d$ , moves continuously from $L_d$ to $U_d$ as $\lambda$ goes from 0 to 1. Therefore, every mean in $[L_d, U_d]$ is attainable in any group $d$ while reproducing $F_d$ and keeping the share $p_d$ fixed.

Let $I_d=[L_d, U_d]$  be the attainable interval for OOO mean in group $d$, established above. consider $\tau_{OOO}=\mu_1-\mu_0$ with $\mu_d\in I_d$.  As the subtraction is linear, the set of all achievable differences is $\{\mu_1-\mu_0 ; \mu_1\in I_1, \mu_0\in I_0\}=[L_1-U_0, U_1-L_0]$. We now show that every value in this interval is not only algebraically representable, but attainable by a DGP consistent with the observed data. 

Pick any $\delta \in [L_1-U_0, U_1-L_0]$, choose $\mu^{\lambda_1}_1\in [L_1,U_1]$ and $\mu^{\lambda_0}_0\in [L_0,U_0]$ such that $\mu^{\lambda_1}_1-\mu^{\lambda_0}_0=\delta$ (e.g. take $\mu^{\lambda_1}_1=L_1+t\cdot (U_1-L_1)$,  $\mu^{\lambda_0}_0=L_0+t\cdot (U_0-L_0)$  for a suitable $t$). we already showed above that there exist $\lambda_1,\lambda_0 \in [0,1]$ such that group specific latent OOO cdfs $N_1^{\lambda_1}$ and $N_0^{\lambda_0}$ yield exactly those means $\mu^{\lambda_1}_1$ and $\mu^{\lambda_0}_0$, and together with the corresponding $M_d^{\lambda_d}$ reproduce observed $F_d$ in each group. Finally, we can form a joint distribution of $(\Delta Y_1,\Delta Y_0)$ for the OOO type by coupling the marginals $N_1^{\lambda_1}$ and $N_0^{\lambda_0}$ with any choice of copula. Similarly, the joint distribution of the complementary type can be formed by coupling the marginals $M_1^{\lambda_1}$ and $M_0^{\lambda_0}$ and mixing types within each group with shares $(p_1,(1-p_1))$ and $(p_0,(1-p_0))$. This joint DGP (i) reproduces the observed selected sample distributions $F_1$ and $F_0$ (ii) respects the fixed shares $(p_1,p_0)$ and delivers $\mathbbm E_{N_1^{\lambda_1}}[\Delta Y_1|OOO]=\mu_1^{\lambda_1}$ and $\mathbbm E_{N_0^{\lambda_0}}[\Delta Y_0|OOO]=\mu_0^{\lambda_0}$ and hence $\delta$.

Thus, we have shown that in each treatment and control group, trimmed means $L_d$ and $U_d$ are valid and attainable extrema for the OOO latent group mean given $p_d$. Furthermore, we demonstrate that every difference in $[L_1-U_0, U_1-L_0]$ is realised by some DGP that matches the observed data and shares. Values outside the interval are ruled out, as we showed that $U_d$ is the maximum possible value for $\mu_d$, and $L_d$ is the minimum possible value for $\mu_d$. Therefore, no other valid bounds under the same assumptions can be narrower. Thus, the bounds in Theorem \ref{nomono_bound} are sharp. 
  \end{proof}

\subsection{Proof of Lemma \ref{lemma:mono_weights}}
\begin{proof}\label{proof:mono_weights} \ 
    Recall that in order to bound $\tau_{OOO}$, we need to identify $p_{OOO0}$ and $p_{OOO1}$. We can express, 
    \begin{align*}
        p_{OOO0} &= \frac{\mathbbm{P}[S_1(0)=1, S_1(1)=1|S_0=1, D=0]}{\mathbbm{P}[S_1(0)=1|S_0=1, D=0]} \\
        & = \frac{\mathbbm{P}[S_1(0)=1|S_0=1, D=0]}{\mathbbm{P}[S_1(0)=1|S_0=1, D=0]}= 1
    \end{align*} where second equality follows from positive monotonicity. Next, we can also express $p_{OOO1}$ as
    \begin{align*}
        p_{OOO1} &= \frac{\mathbbm{P}[S_1(0)=1, S_1(1)=1|S_0=1, D=1]}{\mathbbm{P}[S_1(1)=1|S_0=1, D=1]} \\
        & = \frac{\mathbbm{P}[S_1(0)=1|S_0=1, D=1]}{\mathbbm{P}[S_1(1)=1|S_0=1, D=1]}. 
    \end{align*} where, again, second equality applies positive monotonicity. Combining the two together, we get the result.
\end{proof}

\subsection{Proof of Theorem \ref{mono_bound}}\label{proof Theorem mono_bound}
\begin{proof}\
The treatment effect of treated for OOO group ($\tau_{OOO}$) can be decomposed as follows,
\begin{align*}
\tau_{OOO} &= \mathbbm{E}[Y_{1}^{\ast}(1)-Y_{1}^{\ast}(0)|D=1, OOO]\\
&= \mathbbm{E}[Y_{1}^{\ast}(1)-Y_{0}^{\ast}(1)+Y_{0}^{\ast}(1)-Y_{1}^{\ast}(0)|D=1, OOO]\\
&= \mathbbm{E}[Y_{1}^{\ast}(1)-Y_{0}^{\ast}(1)|D=1, OOO]-\mathbbm{E}[Y_{1}^{\ast}(0)-Y_{0}^{\ast}(1)|D=1, OOO]\\
&\text{Assumption \ref{no anti}}\\
&= \mathbbm{E}[Y_{1}^{\ast}(1)-Y_{0}^{\ast}(1)|D=1, OOO]-\mathbbm{E}[Y_{1}^{\ast}(0)-Y_{0}^{\ast}(0)|D=1, OOO]\\
&\text{Assumption \ref{PT_OOO}}\\
&= \mathbbm{E}[Y_{1}^{\ast}(1)-Y_{0}^{\ast}(1)|D=1, OOO]-\mathbbm{E}[Y_{1}^{\ast}(0)-Y_{0}^{\ast}(0)|D=0, OOO]
\end{align*}
The decomposition of $\mathbbm{E}[Y_{1}-Y_{0}|D=1,S_{0}=1,S_{1}=1]$ given in equation \ref{mixp} lead to partial identification of $\mathbbm{E}[Y_{1}^{\ast}(1)-Y_{0}^{\ast}(1)|D=1,OOO]$ which lies within the interval $[LB_{OOO1},UB_{OOOO1}]$ 
where,
\begin{align*}
    LB_{OOO1} &=\mathbbm{E}[Y_{1}-Y_{0}|D=1,S_{0}=1,S_{1}=1, (Y_{1}-Y_{0})\leq F_{\Delta Y|111}^{-1}(p_{OOO1})]\\
    UB_{OOO1} &=\mathbbm{E}[Y_{1}-Y_{0}|D=1,S_{0}=1,S_{1}=1, (Y_{1}-Y_{0}) > F_{\Delta Y|111}^{-1}(1-p_{OOO1})]
\end{align*}
and $F_{\Delta Y|111}^{-1}(.)$ is the quantile function of the distribution of $Y_{1}-Y_{0}$ given $D=1,S_{0}=1,S_{1}=1$.
Under Assumption \ref{monotone} (positive monotonicity) we can point identify $p_{OOO1}=\frac{\mathbbm{P}[S_{1}=1|S_{0}=1,D=0]}{\mathbbm{P}[S_{1}=1|S_{0}=1,D=1]}$ and $p_{OOO0}=1$ (Refer 
Lemma \ref{lemma:mono_weights}).
As $p_{OOO0}=1$ the decomposition of $\mathbbm{E}[Y_{1}-Y_{0}|D=0,S_{0}=1,S_{1}=1$ given in equation \ref{mixp2} lead to point  identification of $\mathbbm{E}[Y_{1}^{\ast}(1)-Y_{0}^{\ast}(0)|D=0,OOO]$ as $\mathbbm{E}[Y_{1}-Y_{0}|D=0,S_{0}=1,S_{1}=1]$.

Combining the bounds for $\mathbbm{E}[Y_{1}^{\ast}(1)-Y_{0}^{\ast}(1)|D=1,OOO]$ and point identification of $\mathbbm{E}[Y_{1}^{\ast}(0)-Y_{0}^{\ast}(0)|D=0,OOO]$ we find that the parameter of interest $\tau_{OOO} \in [LB_{\tau_{OOO}}, UB_{\tau_{OOO}}]$ where
    \begin{align*}
        LB_{\tau_{OOO}} &= LB_{OOO1}-\mathbbm{E}[Y_{1}-Y_{0}|D=0,S_{0}=1,S_{1}=1]\\
        UB_{\tau_{OOO}} &= UB_{OOO1}-\mathbbm{E}[Y_{1}-Y_{0}|D=0,S_{0}=1,S_{1}=1].
    \end{align*}
\end{proof}
\subsubsection{Proof of Sharpness of Bounds given in Theorem \ref{mono_bound}}\label{sharp}
\begin{proof}
We follow the method used by \citet{lee2009} to prove the sharpness of bounds given in Theorem \ref{mono_bound}.
Sharpness implies that $LB_{\tau_{OOO}}$ 
and $UB_{\tau_{OOO}}$ are the largest lower bound and smallest upper bound which are consistent with the observed data. Moreover, any other valid bounds that impose the same assumptions will contain the interval $[LB_{\tau_{OOO}}, UB_{\tau_{OOO}}]$.

We show that $ UB^{\prime}_{OOO1}=\mathbbm{E}[Y_{1}-Y_{0}|D=1,S_{0}=1,S_{1}=1, (Y_{1}-Y_{0}) > F_{\Delta Y|111}^{-1}(1-p_{OOO1})]$ is a sharp upper bound for $\mathbbm{E}[Y_{1}^{\ast}(1)-Y_{0}^{\ast}(1)|D=1,S_{0}=1, S_{1}(0)=1,S_{1}(1)=1]$ and similar argument will follow for the lower bound.

Under Assumptions \ref{monotone} and \ref{Partialselection}(a),
\begin{align*}
    p_{OOO1}&=\frac{\mathbbm{P}[S_{1}=1|S_{0}=1,D=0]}{\mathbbm{P}[S_{1}=1|S_{0}=1,D=1]}\\
    &=\frac{\mathbbm{P}[S_0=1, S_1(0)=1, S_1(1)=1, D=1]}{\mathbbm{P} [S_0=1, S_1(0)=1, S_1(1)=1, D=1]+\mathbbm{P}[S_0=1, S_1(0)=0, S_1(1)=1, D=1]}\\
    &=\frac{\mathbbm{P}[S_1(0)=1, S_1(1)=1|S_0=1,D=1]}{\mathbbm{P}[S_1=1|S_0=1,D=1]}
\end{align*}

Let $F(\Delta y)$ be the cdf  of $Y_{1}-Y_{0}$ conditional on $D=1,S_{0}=1, S_{1}=1$, $N(\Delta y)$ be the cdf of $Y_{1}^{\ast}(1)-Y_{0}^{\ast}(1)$ conditional on $D=1,S_{0}=1, S_{1}(0)=1,S_{1}(1)=1$ and $M(\Delta y)$ be the cdf of $Y_{1}^{\ast}(1)-Y_{0}^{\ast}(1)$ conditional on $D=1,S_{0}=1, S_{1}(0)=0,S_{1}(1)=1$. Then,
\begin{align*}
    F(\Delta y)= p_{OOO1}\cdot N(\Delta y)+(1-p_{OOO1})\cdot M(\Delta y)
\end{align*}
 Then by \citet{lee2009} lemma 1 we have,
{\footnotesize
\begin{align*}
    UB_{OOO1}=\frac{1}{p_{OOO1}}\cdot \int_{\Delta y_{1-p_{OOO1}}}^{\infty}\Delta y\cdot dF(\Delta y)&\geq \int_{-\infty}^{\infty} \Delta y\cdot dN(\Delta y)\\
    &=\mathbbm{E}[Y_{1}^{\ast}(1)-Y_{0}^{\ast}(1)|D=1,S_{0}=1, S_{1}(0)=1,S_{1}(1)=1]
\end{align*}}
 $p_{OOO1}$ is uniquely determined by data and it follows from the \citet{lee2009} lemma 1 that conditional on $p_{OOO1}$, $UB_{OOO1}$ is a sharp upper bound. Therefore, $UB_{OOO1}$ is the maximum possible value for $\mathbbm{E}[Y_{1}^{\ast}(1)-Y_{0}^{\ast}(1)|D=1,S_{0}=1, S_{1}(0)=1,S_{1}(1)=1]$

 To show that any other valid bounds that impose the same assumptions will contain the interval $[LB_{\tau_{OOO}}, UB_{\tau_{OOO}}]$, we show that observed data cannot rule out any $\delta$ strictly within this interval.
 
 We define $\Delta Y= Y_{1}-Y_{0} $ then,
 {\footnotesize
 \begin{align*}
     UB_{\tau_{OOO}} &\geq \mathbbm{E}[\Delta Y|D=1,S_{0}=1,S_{1}=1]-\mathbbm{E}[\Delta Y|D=0,S_{0}=1,S_{1}=1] \\
     &\geq \mathbbm{E}[\Delta Y|D=1,S_{0}=1,S_{1}=1, \Delta Y< \Delta y_{1-p_{OOO1}}]-\mathbbm{E}[\Delta Y|D=0,S_{0}=1,S_{1}=1]
 \end{align*}}
 Therefore, for any $\delta$ between $UB_{\tau_{OOO}}$ and $ \mathbbm{E}[\Delta Y|D=1,S_{0}=1,S_{1}=1]-\mathbbm{E}[\Delta Y|D=0,S_{0}=1,S_{1}=1]$, there exist $\lambda \in [0,1]$ such that,
 {\footnotesize
 \begin{align*}
     \delta= \lambda \cdot UB_{\tau_{OOO}}+ (1-\lambda)\cdot \{\mathbbm{E}[\Delta Y|D=1,S_{0}=1,S_{1}=1, \Delta Y< \Delta y_{1-p_{OOO1}}]-\mathbbm{E}[\Delta Y|D=0,S_{0}=1,S_{1}=1]\}
 \end{align*}}
 Define $g(\Delta y)$ as the density of $\Delta Y$ conditional on $\Delta Y \geq \Delta y_{1-p_{OOO1}},D=1,S_0=1,S_1=1$ and $h(\Delta y)$ as the density of $\Delta Y$ conditional on $\Delta Y < \Delta y_{1-p_{OOO1}},D=1,S_0=1,S_1=1$.
 Then with this $\lambda$, we can construct the density of $Y_{1}^{\ast}(1)-Y_{0}^{\ast}(1)$ conditional on $D=1, S_0=1,S_1(0)=1,S_1(1)=1$ as $\lambda g(\Delta y)+(1-\lambda)h(\Delta y)$ and the density of $Y_{1}^{\ast}(1)-Y_{0}^{\ast}(1)$ conditional on $D=1, S_0=1,S_1(0)=0,S_1(1)=1$ as $(\frac{p_{OOO1}}{1-p_{OOO1}}-\frac{p_{OOO1}}{1-p_{OOO1}}\lambda) g(\Delta y)+(1-(\frac{p_{OOO1}}{1-p_{OOO1}})(1-\lambda))h(\Delta y)$.
 
 The mixture of these two latent densities by construction replicates the observed density $\Delta Y$ conditional on $D=1, S_0=1, S_1=1$. Moreover, by construction mean of this constructed density of $Y_{1}^{\ast}(1)-Y_{0}^{\ast}(1)$ conditional on $D=1, S_0=1,S_1(0)=1,S_1(1)=1$ minus the control mean (i.e. $\mathbbm{E}[\Delta Y|D=0,S_{0}=1,S_{1}=1$) gives us $\delta$. A symmetric argument can be developed for any $\delta$ in between  $ \mathbbm{E}[\Delta Y|D=1,S_{0}=1,S_{1}=1]-\mathbbm{E}[\Delta Y|D=0,S_{0}=1,S_{1}=1]$ and $LB_{\tau_{OOO}}$. Therefore, each $\delta$ within the interval $[LB_{\tau_{OOO}}, UB_{\tau_{OOO}}]$ cannot be ruled out by the observed data.
 
 \end{proof}
  \subsection{Lemma \ref{lemma: pro 4(joint)}}  
 \label{prop_condi}
 \begin{lemma}
    \label{lemma: pro 4(joint)} Under Assumptions \ref{monotone} and \ref{inde_conditional}
    \begin{align*}
         \pi_{OOO0}&=\mathbbm{P}[S_{0}=1, S_{1}=1,D=0]\\
    \pi_{OOO1} &=\mathbbm{P}[S_{1}=1|S_{0}=1,D=0] \cdot \mathbbm{P}[D=1, S_{0}=1]\\
    \pi_{ONO0} &= \mathbbm{P}[S_{0}=1, S_{1}=0,D=0]-\mathbbm{P}[S_{1}=0|S_{0}=1,D=1] \cdot \mathbbm{P}[D=0, S_{0}=1]\\
    \pi_{ONO1} &=\mathbbm{P}[S_{0}=1, S_{1}=1, D=1]- \mathbbm{P}[S_{1}=1|S_{0}=1,D=0] \cdot \mathbbm{P}[D=1, S_{0}=1]\\
    \pi_{ONN0} &=\mathbbm{P}[S_{1}=0|S_{0}=1,D=1] \cdot \mathbbm{P}[D=0, S_{0}=1]\\
    \pi_{ONN1} &=\mathbbm{P}[S_{0}=1, S_{1}=0,D=1]\\
    \pi_{NOO0} &=\mathbbm{P}[S_{0}=0, S_{1}=1,D=0]\\
    \pi_{NOO1} &=\mathbbm{P}[S_{1}=1|S_{0}=0,D=0] \cdot \mathbbm{P}[D=1, S_{0}=0]\\
     \pi_{NNO0} &= \mathbbm{P}[S_{0}=0, S_{1}=0,D=0]-\mathbbm{P}[S_{1}=0|S_{0}=0,D=1] \cdot \mathbbm{P}[D=0, S_{0}=0]\\
    \pi_{NNO1} &=\mathbbm{P}[S_{0}=0, S_{1}=1, D=1]- \mathbbm{P}[S_{1}=1|S_{0}=0,D=0] \cdot \mathbbm{P}[D=1, S_{0}=0]\\
    \pi_{NNN0} &= \mathbbm{P}[S_{1}=0|S_{0}=0,D=1] \cdot \mathbbm{P}[D=0, S_{0}=0]\\
    \pi_{NNN1} &=\mathbbm{P}[S_{0}=0, S_{1}=0,D=1]
    \end{align*}
\end{lemma}
\begin{proof}    
Consider the observed proportions to identify the required principal strata proportions.
		\begin{align*}
			\mathbbm{P}[S_1=1&|S_0=1,D=0]=\mathbbm{P}[S_1(0)=1,S_1(1)=1|S_0(0)=1,D=0]\\
   &+\mathbbm{P}[S_1(0)=1,S_1(1)=0|S_0(0)=1,D=0] \\
			&=\mathbbm{P}[S_1(0)=1,S_1(1)=1|S_0(0)=1,D=0] \quad \text{{(Assump. \ref{monotone})}}\\ 
   &=\frac{\mathbbm{P}[S_0(0)=1,S_1(0)=1,S_1(1)=1,D=0]}{\mathbbm{P}[S_0(0)=1,D=0]}=\frac{\pi_{OOO0}}{\mathbbm{P}[S_0=1,D=0]}\\
   \pi_{OOO0} &= \mathbbm{P}[S_1=1|S_0=1,D=0]\cdot \mathbbm{P}[S_0=1,D=0]\\
             &= \mathbbm{P}[S_0=1,S_1=1,D=0]\\
      \pi_{OOO1} &= \mathbbm{P}[S_0(0)=1,S_1(0)=1,S_1(1)=1,D=1]\\
      &=\mathbbm{P}[S_1(0)=1,S_1(1)=1|S_0(0)=1,D=1]\cdot \mathbbm{P}[S_0(0)=1,D=1] \text{{(Assump. \ref{inde_conditional})}}\\
      &=\mathbbm{P}[S_1(0)=1,S_1(1)=1|S_0(0)=1,D=0]\cdot \mathbbm{P}[S_0=1,D=1]\\
      &=\mathbbm{P}[S_1=1|S_0=1,D=0] \cdot \mathbbm{P}[S_0=1,D=1]
      \end{align*}
      \begin{align*}
      \mathbbm{P}[S_1=1|S_0=1,D=1] &=\mathbbm{P}[S_1(0)=1,S_1(1)=1|S_0(0)=1,D=1]\\
            &+\mathbbm{P}[S_1(0)=0,S_1(1)=1|S_0(0)=1,D=1]\\
            &=\frac{\mathbbm{P}[S_0(0)=1,S_1(0)=1,S_1(1)=1,D=1]}{ \mathbbm{P}[S_0(0)=1,D=1]} \\
            &+\frac{\mathbbm{P}[S_0(0)=1,S_1(0)=0,S_1(1)=1,D=1]}{ \mathbbm{P}[S_0(0)=1,D=1]}\\
            \mathbbm{P}[S_0=1,S_1=1,D=1]  &= \pi_{OOO1}+ \pi_{ONO1}
            \end{align*}
            \begin{equation*}
                 \pi_{ONO1} =  \mathbbm{P}[S_0=1,S_1=1,D=1] -\mathbbm{P}[S_1=1|S_0=1,D=0] \cdot \mathbbm{P}[S_0=1,D=1]
            \end{equation*}
A similar argument can be established to identify the relationship between the observed conditional selection probabilities and other principal strata proportions given in Lemma \ref{lemma: pro 4(joint)}.
\end{proof}

\subsection{Proof of Theorem \ref{ONO_bound}} \label{proof:ONO_bound}
\begin{proof}
The treatment effect of treated for ONO group ($\tau_{ONO}$) can be decomposed as follows,
{\small \begin{align} \label{decom_ONO}
\tau_{ONO} &= \mathbbm{E}[Y_{1}^{\ast}(1)-Y_{1}^{\ast}(0)|D=1, ONO] \nonumber\\
&= \mathbbm{E}[Y_{1}^{\ast}(1)-Y_{0}^{\ast}(1)+Y_{0}^{\ast}(1)-Y_{1}^{\ast}(0)|D=1, ONO]\nonumber\\
&= \mathbbm{E}[Y_{1}^{\ast}(1)-Y_{0}^{\ast}(1)|D=1, ONO]-\mathbbm{E}[Y_{1}^{\ast}(0)-Y_{0}^{\ast}(1)|D=1, ONO] \nonumber\\
&\text{Assumption \ref{no anti}}\nonumber\\
&= \mathbbm{E}[Y_{1}^{\ast}(1)-Y_{0}^{\ast}(1)|D=1, ONO]-\mathbbm{E}[Y_{1}^{\ast}(0)-Y_{0}^{\ast}(0)|D=1, ONO] \nonumber\\
&\text{Assumption \ref{PT_group}}\nonumber\\
&= \mathbbm{E}[Y_{1}^{\ast}(1)-Y_{0}^{\ast}(1)|D=1, ONO]-\mathbbm{E}[Y_{1}^{\ast}(0)-Y_{0}^{\ast}(0)|D=0, ONO]\nonumber\\
&= \mathbbm{E}[Y_{1}^{\ast}(1)-Y_{0}^{\ast}(1)|D=1, ONO]-\mathbbm{E}[Y_{1}^{\ast}(0)|D=0, ONO]+ \mathbbm{E}[Y_{0}^{\ast}(0)|D=0, ONO]
\end{align}}

As explained in section \ref{identification}, we can use the group of treated individuals for whom the outcome is observed in both periods to partially identify $\mathbbm{E}[Y_{1}^{\ast}(1)-Y_{0}^{\ast}(1)|D=1, ONO]$. Hence, $\mathbbm{E}[Y_{1}^{\ast}(1)-Y_{0}^{\ast}(1)|D=1,ONO]$ lies within the interval $[LB_{ONO1},UB_{ONO1}]$ where,
\begin{align}\label{bound_ONO1}
    LB_{ONO1} &=\mathbbm{E}[Y_{1}-Y_{0}|D=1,S_{0}=1,S_{1}=1, (Y_{1}-Y_{0})\leq F_{\Delta Y|111}^{-1}(1-p_{OOO1})] \nonumber\\
    UB_{ONO1} &=\mathbbm{E}[Y_{1}-Y_{0}|D=1,S_{0}=1,S_{1}=1, (Y_{1}-Y_{0}) > F_{\Delta Y|111}^{-1}(p_{OOO1})] 
\end{align}
and $F_{\Delta Y|111}^{-1}(.)$ is the quantile function of the distribution of $Y_{1}-Y_{0}$ given $D=1,S_{0}=1,S_{1}=1$.
Similarly, we can use the group of untreated individuals for whom the outcome is observed in the first period but not observed in the second period ($D = 0, S_0 = 1, S_1 = 0$) to partially identify $\mathbbm{E}[Y_{0}^{*}(0)|D=0, ONO]$. Table \ref{tabl obs groups} shows that their observed average outcome reflects a mixture of the potential outcomes for the ONN and ONO latent groups with mixing probabilities corresponding to their relative proportions. 
{\small \begin{align}\label{mix_10}
     &\mathbbm{E}[Y_{0}|D=0,S_{0}=1, S_{1}=0]= \nonumber \\
    &= \mathbbm{E}[Y_{0}^{\ast}(0)|D=0,S_{0}=1, (S_{1}(0)=0,S_{1}(1)=0) \quad or \quad (S_{1}(0)=0,S_{1}(1)=1)] \nonumber\\
    &=\mathbbm{E}[Y_{0}^{\ast}(0)|D=0,ONN]\cdot \frac{\pi_{ONN0}}{\pi_{ONN0}+\pi_{ONO0}} +\mathbbm{E}[Y_{0}^{\ast}(0)|D=0,ONO]\cdot \frac{\pi_{ONO0}}{\pi_{ONN0}+\pi_{ONO0}}
\end{align}}
For notation simplicity let $p_{ONO0}=\frac{\pi_{ONO0}}{\pi_{ONN0}+\pi_{ONO0}}$. We impose Assumption \ref{inde_conditional} in addition to the positive monotonicity assumption to identify this proportion.\footnote{Lemma \ref{lemma: pro 4(joint)} in Appendix \ref{prop_condi}.} 
 
Hence, $\mathbbm{E}[Y_{0}^{\ast}(0)|D=0, ONO]$ lies within the interval $[LB^{0}_{ONO0},UB^{0}_{ONO0}]$ where,
\begin{align} \label{Bound_ONO0_t_0}
    LB^{0}_{ONO0} &=\mathbbm{E}[Y_{0}|D=0,S_{0}=1,S_{1}=0, Y_{0}\leq F_{Y_{0}|010}^{-1}(p_{ONO0})]\nonumber\\
    UB^{0}_{ONO0} &=\mathbbm{E}[Y_{0}|D=0,S_{0}=1,S_{1}=0, Y_{0} > F_{Y_{0}|010}^{-1}(1-p_{ONO0})]
\end{align}
and $F_{Y_{0}|010}^{-1}(.)$ is the quantile function of the distribution of $Y_{0}$ given $D=0,S_{0}=1,S_{1}=0$.

Identification of $\mathbbm{E}[Y_{1}^{\ast}(0)|D=0, ONO]$ is not straightforward, but we can use the theoretical upper bound and lower bound of the potential outcome distribution \citep{huber2015sharp},
\begin{align*}
    Y_{1}^{LB}\leq \mathbbm{E}[Y_{1}^{\ast}(0)|D=0, ONO] \leq Y_{1}^{UB},
\end{align*}
where $Y_{1}^{LB}$ and $Y_{1}^{UB}$ are the theoretical upper and lower bounds of the potential outcomes in the post-treatment period. As these are wide bounds, we impose Assumption \ref{mean dominance}(a) to shrink the bounds. This assumption implies the potential outcome of always-observed in the untreated state first-order stochastically dominates compliers. Imposing this assumption, we can tighten the upper bound of $\mathbbm{E}[Y_{1}^{\ast}(0)|D=0, ONO]$. Then
   \begin{align}
    Y_{1}^{LB}\leq \mathbbm{E}[Y_{1}^{\ast}(0)|D=0, ONO] \leq \mathbbm{E}[Y_{1}^{\ast}(0)|D=0, OOO]
\end{align}     
where $\mathbbm{E}[Y_{1}^{\ast}(0)|D=0, OOO]$ can be point identified as $\mathbbm{E}[Y_{1}|D=0,S_{0}=1, S_{1}=1]$ under positive monotonicity considering the untreated individuals observed in both periods. Thus,
\begin{align}\label{bound_ONO0}
    Y_{1}^{LB}\leq \mathbbm{E}[Y_{1}^{\ast}(0)|D=0, ONO] \leq \mathbbm{E}[Y_{1}|D=0,S_{0}=1, S_{1}=1]
\end{align} 

Combining the bounds for $\mathbbm{E}[Y_{1}^{\ast}(1)-Y_{0}^{\ast}(1)|D=1, ONO]$, $\mathbbm{E}[Y_{0}^{\ast}(0)|D=0, ONO]$ and $\mathbbm{E}[Y_{1}^{\ast}(0)|D=0, ONO]$ as derived in Equations \ref{bound_ONO1}, \ref{Bound_ONO0_t_0} and \ref{bound_ONO0} with the decomposition of $\tau_{ONO}$ given in Equation \ref{decom_ONO} we find that the parameter of interest $\tau_{ONO}$ is in the interval
{\small \begin{align*}
    [LB_{ONO1}-\mathbbm{E}[Y_{1}|D=0,S_{0}=1, S_{1}=1]+LB^{0}_{ONO0}, UB_{ONO1}-Y_{1}^{LB}+UB^{0}_{ONO0}].
\end{align*}}
\end{proof}
\subsection{Proof of Theorem \ref{NNO_bound}} \label{proof:NNO_bound}
\begin{proof}

The treatment effect of treated for NNO group ($\tau_{NNO}$) can be decomposed as follows,
\begin{align} \label{decom_NNO}
\tau_{NNO} &= \mathbbm{E}[Y_{1}^{\ast}(1)-Y_{1}^{\ast}(0)|D=1, NNO] \nonumber\\
&= \mathbbm{E}[Y_{1}^{\ast}(1)-Y_{0}^{\ast}(1)+Y_{0}^{\ast}(1)-Y_{1}^{\ast}(0)|D=1, NNO] \nonumber\\
&= \mathbbm{E}[Y_{1}^{\ast}(1)-Y_{0}^{\ast}(1)|D=1, NNO]-\mathbbm{E}[Y_{1}^{\ast}(0)-Y_{0}^{\ast}(1)|D=1, NNO] \nonumber\\
&\text{Assumption \ref{no anti}}\nonumber\\
&= \mathbbm{E}[Y_{1}^{\ast}(1)-Y_{0}^{\ast}(1)|D=1, NNO]-\mathbbm{E}[Y_{1}^{\ast}(0)-Y_{0}^{\ast}(0)|D=1, NNO]\nonumber\\
&\text{Assumption \ref{PT_group}}\nonumber\\
&= \mathbbm{E}[Y_{1}^{\ast}(1)-Y_{0}^{\ast}(1)|D=1, NNO]-\mathbbm{E}[Y_{1}^{\ast}(0)-Y_{0}^{\ast}(0)|D=0, NNO]\nonumber\\
&= \mathbbm{E}[Y_{1}^{\ast}(1)|D=1, NNO]-\mathbbm{E}[Y_{0}^{\ast}(1)|D=1, NNO]\nonumber\\
&-\mathbbm{E}[Y_{1}^{\ast}(0)|D=0, NNO]+\mathbbm{E}[Y_{0}^{\ast}(0)|D=0, NNO]
\end{align}
We can use the group of treated individuals for whom the outcome is not observed in the pre-treatment period but observed in the post-treatment period ($D = 1, S_0 = 0, S_1 = 1$) to partially identify $\mathbbm{E}[Y_{1}^{\ast}(1)|D=1, NNO]$. Table \ref{tabl obs groups} shows that their observed average outcome reflects a mixture of the potential outcomes for the NOO and NNO latent groups with mixing probabilities corresponding to their relative proportions.
{\small \begin{align}\label{mix_01}
     &\mathbbm{E}[Y_{1}|D=1,S_{0}=0, S_{1}=1]= \nonumber \\
    &= \mathbbm{E}[Y_{1}^{\ast}(1)|D=1,S_{0}=0, (S_{1}(0)=1,S_{1}(1)=1) \quad or \quad (S_{1}(0)=0,S_{1}(1)=1)] \nonumber\\
    &=\mathbbm{E}[Y_{1}^{\ast}(1)|D=1,NOO]\cdot \frac{\pi_{NOO1}}{\pi_{NOO1}+\pi_{NNO1}}+\mathbbm{E}[Y_{1}^{\ast}(1)|D=1,NNO]\cdot \frac{\pi_{NNO1}}{\pi_{NOO1}+\pi_{NNO1}}
\end{align}}
For notation simplicity let $p_{NNO1}=\frac{\pi_{NNO1}}{\pi_{NOO1}+\pi_{NNO1}}$. This proportion can be point identified imposing Assumption \ref{inde_conditional} in addition to the positive monotonicity assumption.\footnote{Lemma \ref{lemma: pro 4(joint)} in Appendix \ref{prop_condi}.} Hence, $\mathbbm{E}[Y_{1}^{\ast}(1)|D=1, NNO]$ lies within the interval $[LB_{NNO1},UB_{NNO1}]$ where,
\begin{align}\label{bound_NNO1}
    LB_{NNO1} &=\mathbbm{E}[Y_{1}|D=1,S_{0}=0,S_{1}=1, Y_{1}\leq F_{Y_{1}|101}^{-1}(p_{NNO1})] \nonumber\\
    UB_{NNO1} &=\mathbbm{E}[Y_{1}|D=1,S_{0}=0,S_{1}=1, Y_{1} > F_{Y_{1}|101}^{-1}(1-p_{NNO1})]
\end{align}
and $F_{Y_{1}|101}^{-1}(.)$ is the quantile function of the distribution of $Y_{1}$ given $D=1,S_{0}=0,S_{1}=1$.

Identification of $\mathbbm{E}[Y_{0}^{\ast}(1)|D=1, NNO]$ is not straightforward, but we can use the theoretical upper bound and lower bound of the potential outcome distribution \citep{huber2015sharp}
\begin{align*}
    Y_{0}^{LB}\leq \mathbbm{E}[Y_{0}^{\ast}(1)|D=1, NNO] \leq Y_{0}^{UB}
\end{align*}
$Y_{0}^{LB}$ and $Y_{0}^{UB}$ are the theoretical upper and lower bounds of the potential outcomes in the pre-treatment period. We impose an outcome mean dominance Assumption \ref{mean dominance} (b.i), noting that under Assumption \ref{no anti}, $\mathbbm{E}[Y_{0}^{\ast}(1)|D=1, NNO]=\mathbbm{E}[Y_{0}^{\ast}(0)|D=1, NNO]$ to tighten the upper bound of $\mathbbm{E}[Y_{0}^{\ast}(0)|D=1, NNO]$. Then,
\begin{align*}
    Y_{0}^{LB}\leq \mathbbm{E}[Y_{0}^{\ast}(0)|D=1, NNO] \leq \mathbbm{E}[Y_{0}^{\ast}(0)|D=1, ONO].
\end{align*}    
where $\mathbbm{E}[Y_{0}^{\ast}(0)|D=1, ONO]$ can be partially identified considering that the observed average outcome of the treated individuals is a mixture of potential outcomes for the OOO and ONO groups with mixing probability $p_{OOO1}=\frac{\pi_{OOO1}}{\pi_{OOO1}+\pi_{ONO1}}$, which is point identified under Assumption \ref{Partialselection} (a) in addition to positive monotonicity assumption. Hence, $\mathbbm{E}[Y_{0}^{\ast}(0)|D=1, ONO]$ lies within the interval $[LB^{0}_{ONO1},UB^{0}_{ONO1}]$ where,
\begin{align*}
    LB^{0}_{ONO1} &=\mathbbm{E}[Y_{0}|D=1,S_{0}=1,S_{1}=1, Y_{0}\leq F_{Y_{0}|111}^{-1}(1-p_{OOO1})]\\
    UB^{0}_{ONO1} &=\mathbbm{E}[Y_{0}|D=1,S_{0}=1,S_{1}=1, Y_{0} > F_{Y_{0}|111}^{-1}(p_{OOO1})] 
\end{align*}
and $F_{Y_{0}|111}^{-1}(.)$ is the quantile function of the distribution of $Y_{0}$ given $D=1,S_{0}=1,S_{1}=1$. Then,
\begin{align}\label{bound_NNO1t0}
    Y_{0}^{LB}\leq \mathbbm{E}[Y_{0}^{\ast}(0)|D=1, NNO] \leq LB^{0}_{ONO1}
\end{align} 
In order to identify $\mathbbm{E}[Y_{1}^{\ast}(0)|D=0, NNO]$ we first impose theoretical upper and lower bounds:
\begin{align*}
    Y_{1}^{LB}\leq \mathbbm{E}[Y_{1}^{\ast}(0)|D=0, NNO] \leq Y_{1}^{UB}.
\end{align*}
Then we impose an outcome mean dominance Assumption \ref{mean dominance} (b.ii) to tighten the upper bound of $\mathbbm{E}[Y_{1}^{\ast}(0)|D=0, NNO]$.
 \begin{align*}
    Y_{1}^{LB}\leq \mathbbm{E}[Y_{1}^{\ast}(0)|D=0, NNO] \leq \mathbbm{E}[Y_{1}^{\ast}(0)|D=0, NOO]
\end{align*}
where $\mathbbm{E}[Y_{1}^{\ast}(0)|D=0, NOO]$ can be point identified as $\mathbbm{E}[Y_{1}|D=0, S_{0}=0, S_{1}=1]$ under positive monotonicity considering the untreated individuals observed in post-treatment period but not observed in the pre-treatment period. Thus,
\begin{align}\label{bound_NNO0t1}
    Y_{1}^{LB}\leq \mathbbm{E}[Y_{1}^{\ast}(0)|D=0, NNO] \leq \mathbbm{E}[Y_{1}^{\ast}(0)|D=0, NOO]
\end{align}

Identification of $\mathbbm{E}[Y_{0}^{\ast}(0)|D=0, NNO]$ will follow similarly. We first consider the  theoretical upper bound and lower bound of the potential outcome distribution in pre-treatment period, denoted as $Y_{0}^{LB}$ and $Y_{0}^{UB}$, following \citep{huber2015sharp}
\begin{align*}
    Y_{0}^{LB}\leq \mathbbm{E}[Y_{0}^{\ast}(0)|D=0, NNO] \leq Y_{0}^{UB}
\end{align*}
 We impose the outcome mean dominance Assumption \ref{mean dominance}(b)i to tighten the upper bound of $\mathbbm{E}[Y_{0}^{\ast}(0)|D=0, NNO]$. Then
\begin{align*}
    Y_{0}^{LB}\leq \mathbbm{E}[Y_{0}^{\ast}(0)|D=0, NNO] \leq \mathbbm{E}[Y_{0}^{\ast}(0)|D=0, ONO]
\end{align*}     
where $\mathbbm{E}[Y_{0}^{\ast}(0)|D=0, ONO]$ can be partially identified as given in Equation \ref{Bound_ONO0_t_0}. Then
\begin{align}\label{bound_NNO0t0}
    Y_{0}^{LB}\leq E[Y_{0}^{\ast}(0)|D=0, NNO] \leq LB^{0}_{ONO0}
\end{align}

Combining the bounds for $\mathbbm{E}[Y_{1}^{\ast}(1)|D=1, NNO]$, $\mathbbm{E}[Y_{0}^{\ast}(1)|D=1, NNO]$, $\mathbbm{E}[Y_{1}^{\ast}(0)|D=0, NNO]$ and $\mathbbm{E}[Y_{0}^{\ast}(0)|D=0, NNO]$  as given in Equations \ref{bound_NNO1}, \ref{bound_NNO1t0}, \ref{bound_NNO0t1} and \ref{bound_NNO0t0} with the decomposition of $\tau_{NNO}$ as given in Equation \ref{decom_NNO}, we find that the parameter of interest $\tau_{NNO}$ is in the interval, $[LB_{NNO1}-LB^{0}_{ONO1}-\mathbbm{E}[Y_{1}|D=0,S_{0}=0, S_{1}=1]+Y_{0}^{LB}, \quad UB_{NNO1}-Y_{0}^{LB}- Y_{1}^{LB}+LB^{0}_{ONO0}]$.
\end{proof}
\subsection{Proof of Theorem \ref{NOO_bound}} \label{proof:NOO_bound}
\begin{proof}
The treatment effect of treated for NOO group ($\tau_{NOO}$) can be decomposed as follows,
\begin{align} \label{decom_NOO}
\tau_{NOO} &= \mathbbm{E}[Y_{1}^{\ast}(1)-Y_{1}^{\ast}(0)|D=1, NOO] \nonumber\\
&= \mathbbm{E}[Y_{1}^{\ast}(1)-Y_{0}^{\ast}(1)+Y_{0}^{\ast}(1)-Y_{1}^{\ast}(0)|D=1, NOO]\nonumber\\
&= \mathbbm{E}[Y_{1}^{\ast}(1)-Y_{0}^{\ast}(1)|D=1, NOO]-\mathbbm{E}[Y_{1}^{\ast}(0)-Y_{0}^{\ast}(1)|D=1, NOO]\nonumber\\
&\text{Assumption \ref{no anti}}\nonumber\\
&= \mathbbm{E}[Y_{1}^{\ast}(1)-Y_{0}^{\ast}(1)|D=1, NOO]-\mathbbm{E}[Y_{1}^{\ast}(0)-Y_{0}^{\ast}(0)|D=1, NOO]\nonumber\\
&\text{Assumption \ref{PT_group}}\nonumber\\
&= \mathbbm{E}[Y_{1}^{\ast}(1)-Y_{0}^{\ast}(1)|D=1, NOO]-\mathbbm{E}[Y_{1}^{\ast}(0)-Y_{0}^{\ast}(0)|D=0, NOO]\nonumber\\
&= \mathbbm{E}[Y_{1}^{\ast}(1)|D=1, NOO]-\mathbbm{E}[Y_{0}^{\ast}(1)|D=1, NOO]-\mathbbm{E}[Y_{1}^{\ast}(0)|D=0, NOO] \nonumber\\
&+\mathbbm{E}[Y_{0}^{\ast}(0)|D=0, NOO]
\end{align}
We can use the group of treated individuals for whom the outcome is not observed in the pre-treatment period but observed in the post-treatment period ($D = 1, S_0 = 0, S_1 = 1$) to partially identify $\mathbbm{E}[Y_{1}^{\ast}(1)|D=1, NOO]$. Analogous to equation \ref{mix_01}, therefore,  $\mathbbm{E}[Y_{1}^{\ast}(1)|D=1, NOO]$ lies within the interval $[LB_{NOO1},UB_{NOO1}]$ where,
\begin{align}\label{bound_NOO1t1}
    LB_{NOO1} &=\mathbbm{E}[Y_{1}|D=1,S_{0}=0,S_{1}=1, Y_{1}\leq F_{Y_{1}|101}^{-1}(1-p_{NNO1})] \nonumber\\
    UB_{NOO1} &=\mathbbm{E}[Y_{1}|D=1,S_{0}=0,S_{1}=1, Y_{1} > F_{Y_{1}|101}^{-1}(p_{NNO1})]
\end{align}
and $F_{Y_{1}|101}^{-1}(.)$ is the quantile function of the distribution of $Y_{1}$ given $D=1,S_{0}=0,S_{1}=1$.

Identification of $\mathbbm{E}[Y_{0}^{\ast}(1)|D=1, NOO]$ is not straightforward, but we can use the theoretical upper bound and lower bound of the potential outcome distribution  in the pre-treatment period (\citet{huber2015sharp})
\begin{align*}
    Y_{0}^{LB}\leq \mathbbm{E}[Y_{0}^{\ast}(1)|D=1, NOO] \leq Y_{0}^{UB}
\end{align*}
$Y_{0}^{LB}$ and $Y_{0}^{UB}$ are the theoretical upper and lower bounds of the potential outcome distribution in the pre-treatment period. We impose outcome mean dominance Assumption \ref{mean dominance} (c), noting that under Assumption \ref{no anti}, $\mathbbm{E}[Y_{0}^{\ast}(1)|D=1, NOO]=\mathbbm{E}[Y_{0}^{\ast}(0)|D=0, NOO]$ to shrink the bounds, as these are wide bounds. Imposing this assumption, we can tighten the upper bound of $E[Y_{0}^{\ast}(1)|D=1, NOO]$. Then,
\begin{align*}
    Y_{0}^{LB}\leq \mathbbm{E}[Y_{0}^{\ast}(1)|D=1, NOO] \leq \mathbbm{E}[Y_{0}^{\ast}(1)|D=1, OOO]
\end{align*}  
where $\mathbbm{E}[Y_{0}^{\ast}(1)|D=1, OOO]$ can be partially identified considering that the observed average outcome of the treated individuals observed in both periods is a mixture of potential outcomes for the OOO and ONO groups with mixing probability $p_{OOO1}=\frac{\pi_{OOO1}}{\pi_{OOO1}+\pi_{ONO1}}$, which is point identified under Assumption \ref{Partialselection} (a) in addition to positive monotonicity assumption. Hence, $E[Y_{0}^{\ast}(1)|D=1, OOO]$ lies within the interval $[LB^{0}_{OOO1},UB^{0}_{OOO1}]$ where,
\begin{align*}
    LB^{0}_{OOO1} &=\mathbbm{E}[Y_{0}|D=1,S_{0}=1,S_{1}=1, Y_{0}\leq F_{Y_{0}|111}^{-1}(p_{OOO1})]\\
    UB^{0}_{OOO1} &=\mathbbm{E}[Y_{0}|D=1,S_{0}=1,S_{1}=1, Y_{0} > F_{Y_{0}|111}^{-1}(1-p_{OOO1})]
\end{align*}
and $F_{Y_{0}|111}^{-1}(.)$ is the quantile function of the distribution of $Y_{0}$ given $D=1,S_{0}=1,S_{1}=1$. Then
\begin{equation}\label{bound_NOO1t0}
    Y_{0}^{LB}\leq \mathbbm{E}[Y_{0}^{\ast}(1)|D=1, NOO] \leq LB^{0}_{OOO1}
\end{equation}
$\mathbbm{E}[Y_{1}^{\ast}(0)|D=0, NOO]$ can be point identified as $\mathbbm{E}[Y_{1}|D=0,S_{0}=0,S_{1}=1]$ considering the untreated individuals not observed in the pre-treatment period but observed in the post-treatment period ($D=0, S_0=0, S_1=1$) under positive monotonicity.

To identify $\mathbbm{E}[Y_{0}^{\ast}(0)|D=0, NOO]$, we first consider the  theoretical upper bound and lower bound of the potential outcome distribution, following \cite{huber2015sharp},
\begin{align*}
    Y_{0}^{LB}\leq \mathbbm{E}[Y_{0}^{\ast}(0)|D=0, NOO] \leq Y_{0}^{UB}
\end{align*}
$Y_{0}^{LB}$ and $Y_{0}^{UB}$ are the theoretical upper and lower bounds of the potential outcome distribution in the pre-treatment period. We impose the outcome mean dominance Assumption \ref{mean dominance} (c)  to shrink the bounds, as these are wide bounds. Imposing this assumption, we can tighten the upper bound of $\mathbbm{E}[Y_{0}^{\ast}(0)|D=0, NOO]$. Then
\begin{equation*}
    Y_{0}^{LB}\leq \mathbbm{E}[Y_{0}^{\ast}(0)|D=0, NOO] \leq \mathbbm{E}[Y_{0}^{\ast}(0)|D=0, OOO]
\end{equation*}  
where $\mathbbm{E}[Y_{0}^{\ast}(0)|D=0, OOO]$ can be point identified under positive monotonicity assumption as $\mathbbm{E}[Y_{0}|D=0,S_{0}=1, S_{1}=1]$ considering the untreated observed in both periods. Thus,
\begin{align}\label{bound_NOO0t0}
    Y_{0}^{LB}\leq \mathbbm{E}[Y_{0}^{\ast}(0)|D=0, NOO] \leq \mathbbm{E}[Y_{0}|D=0,S_{0}=1, S_{1}=1]
\end{align} 
Combining the bounds for $\mathbbm{E}[Y_{1}^{\ast}(1)|D=1, NOO]$, $\mathbbm{E}[Y_{1}^{\ast}(0)|D=1, NOO]$, and $\mathbbm{E}[Y_{0}^{\ast}(0)|D=0, NOO]$ from Equations \ref{bound_NOO1t1}, \ref{bound_NOO1t0} and \ref{bound_NOO0t0}, along with point identified $\mathbbm{E}[Y_{1}^{\ast}(0)|D=0, NOO]$ and the decomposition of $\tau_{NOO}$ given in Equation \ref{decom_NOO}, implies that the parameter of interest $\tau_{NOO}$ lies in the interval
\begin{align*}
    &\bigg[LB_{NOO1}-LB^{0}_{OOO1}-\mathbbm{E}[Y_{1}|D=0,S_{0}=0, S_{1}=1]+Y_{0}^{LB},\\
    &\quad UB_{NOO1}-Y_{0}^{LB}-\mathbbm{E}[Y_{1}|D=0,S_{0}=0, S_{1}=1]+\mathbbm{E}[Y_{0}|D=0,S_{0}=1, S_{1}=1]\bigg].
\end{align*}
\end{proof}

\subsection{Proof of Theorem \ref{theorem:partialid_tau}}\label{proof:partialid_tau}
\begin{proof}
\begin{align}\label{direct}
    \tau
    &= \mathbbm{E}\left[ Y_1^\ast(1) - Y_1^\ast(0) | D=1,S_0 = 1,S_1 = 1 \right] \cdot \mathbbm{P}(S_0 = 1, S_1 = 1 | D = 1) \nonumber \\
    &\quad + \mathbbm{E}\left[ Y_1^\ast(1) - Y_1^\ast(0) | D=1, S_0 = 1,S_1 = 0 \right] \cdot \mathbbm{P}(S_0 = 1,S_1 = 0 | D = 1) \nonumber \\
    &\quad + \mathbbm{E}\left[ Y_1^\ast(1) - Y_1^\ast(0) | D=1, S_0 = 0,S_1 = 1 \right] \cdot \mathbbm{P}(S_0 = 0,S_1 = 1 | D = 1) \nonumber \\
    &\quad + \mathbbm{E}\left[ Y_1^\ast(1) - Y_1^\ast(0) | D=1, S_0 = 0,S_1 = 0 \right] \cdot \mathbbm{P}(S_0 = 0,S_1 = 0 | D = 1)
\end{align}
In the above expression, consider,
\begin{align*}
    &\mathbbm{E}\left[ Y_1^\ast(1) - Y_1^\ast(0) | D=1,S_0 = 1,S_1 = 1 \right] \\
    &= \mathbbm{E}[Y_{1}^{*}(1)-Y_{1}^{*}(0)|D=1,S_{0}=1, (S_{1}(0)=1,S_{1}(1)=1) \quad or \quad (S_{1}(0)=0,S_{1}(1)=1)] \\
    &=\mathbbm{E}[Y_{1}^{*}(1)-Y_{1}^{*}(0)|D=1,OOO]\cdot \frac{\pi_{OOO1}}{\pi_{OOO1}+\pi_{ONO1}}\\
    &\quad +\mathbbm{E}[Y_{1}^{*}(1)-Y_{1}^{*}(0)|D=1,ONO]\cdot \frac{\pi_{ONO1}}{\pi_{OOO1}+\pi_{ONO1}} \\
    &=\tau_{OOO}\cdot p_{OOO1}+\tau_{ONO} \cdot (1-p_{OOO1})
\end{align*}
A similar argument will follow for other expressions under positive monotonicity,
\begin{align*}
    \mathbbm{E}\left[ Y_1^\ast(1) - Y_1^\ast(0) | D=1, S_0 = 1,S_1 = 0 \right] &= \tau_{ONN}\\
    \mathbbm{E}\left[ Y_1^\ast(1) - Y_1^\ast(0) | D=1, S_0 = 0,S_1 = 1 \right] &= \tau_{NOO}\cdot p_{NOO1}+\tau_{NNO} \cdot (1-p_{NOO1})\\
    \mathbbm{E}\left[ Y_1^\ast(1) - Y_1^\ast(0) | D=1, S_0 = 0,S_1 = 0 \right] &= \tau_{NNN}
\end{align*}
Plugging into Equation \ref{direct} we have,
\begin{align*}
    \tau &= (\tau_{OOO}\cdot p_{OOO1}+\tau_{ONO} \cdot (1-p_{OOO1}))\cdot \mathbbm{P}(S_0 = 1, S_1 = 1 | D = 1)\\
    &\quad +\tau_{ONN}\cdot \mathbbm{P}(S_0 = 1,S_1 = 0 | D = 1)+ \tau_{NNN} \cdot \mathbbm{P}(S_0 = 0,S_1 = 0 | D = 1)\\
    &\quad +(\tau_{NOO}\cdot p_{NOO1}+\tau_{NNO} \cdot (1-p_{NOO1}))\cdot \mathbbm{P}(S_0 = 0,S_1 = 1 | D = 1)
\end{align*}
We  partially identify $\tau_{OOO}$, $\tau_{ONO}$, $\tau_{NNO}$, and $\tau_{NOO}$ through Theorems \ref{mono_bound}, \ref{ONO_bound}, \ref{NNO_bound}, and \ref{NOO_bound} respectively with mixing proportion  $p_{OOO1}$ being point identified in Theorem \ref{mono_bound} and $p_{NOO1}$ being point identified by Lemma \ref{lemma: pro 4(joint)}. We can use theoretical upper bounds and lower bounds to derive bounds for the ATT of ONN and NNN groups. Thus,  $\tau_{ONN}\in[Y_1^{LB}-Y_1^{UB},Y_1^{UB}-Y_1^{LB}]$ and $\tau_{NNN}\in[Y_1^{LB}-Y_1^{UB},Y_1^{UB}-Y_1^{LB}]$. Therefore, we can bound $\tau$ as follows,

\begin{align*}
\tau_{UB}&= \bigg(UB_{\tau_{OOO}}\cdot p_{OOO1}+ UB_{\tau_{ONO}} \cdot (1-p_{OOO1})\bigg)\cdot             \mathbbm{P}(S_0 = 1, S_1 = 1 | D = 1)\\
    &\quad +UB_{\tau_{ONN}}\cdot \mathbbm{P}(S_0 = 1,S_1 = 0 | D = 1)+ UB_{\tau_{NNN}} \cdot \mathbbm{P}(S_0 = 0,S_1 = 0 | D = 1)\\
    &\quad+\bigg(UB_{\tau_{NOO}}\cdot p_{NOO1}+UB_{\tau_{NNO}} \cdot (1-p_{NOO1})\bigg)\cdot \mathbbm{P}(S_0 = 0,S_1 = 1 | D = 1)\\
 \tau_{LB}&= \bigg(LB_{\tau_{OOO}}\cdot p_{OOO1}+LB_{\tau_{ONO}} \cdot (1-p_{OOO1})\bigg)\cdot \mathbbm{P}(S_0 = 1, S_1 = 1 | D = 1)\\
    &\quad+LB_{\tau_{ONN}}\cdot \mathbbm{P}(S_0 = 1,S_1 = 0 | D = 1)+ LB_{\tau_{NNN}} \cdot \mathbbm{P}(S_0 = 0,S_1 = 0 | D = 1)\\
    &\quad+\bigg(LB_{\tau_{NOO}}\cdot p_{NOO1}+LB_{\tau_{NNO}} \cdot (1-p_{NOO1})\bigg)\cdot \mathbbm{P}(S_0 = 0,S_1 = 1 | D = 1)
\end{align*}
where,
\begin{equation*}
\begin{split}
        LB_{\tau_{OOO}} &= LB_{OOO1}-\mathbbm{E}[Y_{1}-Y_{0}|D=0,S_{0}=1,S_{1}=1],\\
        UB_{\tau_{OOO}} &=UB_{OOO1}-\mathbbm{E}[Y_{1}-Y_{0}|D=0,S_{0}=1,S_{1}=1]\\
        LB_{\tau_{ONO}} &= LB_{ONO1}-\mathbbm{E}[Y_{1}|D=0,S_{0}=1, S_{1}=1]+LB^{0}_{ONO0},\\
        UB_{\tau_{ONO}} &= UB_{ONO1}-Y_{1}^{LB}+UB^{0}_{ONO0}\\
         LB_{\tau_{NOO}} &= LB_{NOO1}-LB^{0}_{OOO1}-\mathbbm{E}[Y_{1}|D=0,S_{0}=0, S_{1}=1]+Y_{0}^{LB},\\
        UB_{\tau_{NOO}} &= UB_{NOO1}-Y_{0}^{LB}-\mathbbm{E}[Y_{1}|D=0,S_{0}=0, S_{1}=1]+\mathbbm{E}[Y_{0}|D=0,S_{0}=1, S_{1}=1]]\\
        LB_{\tau_{NNO}} &= LB_{NNO1}-LB^{0}_{ONO1}-\mathbbm{E}[Y_{1}|D=0,S_{0}=0, S_{1}=1]+Y_{0}^{LB},\\
         UB_{\tau_{NNO}} &= UB_{NNO1}-Y_{0}^{LB}- Y_{1}^{LB}+LB^{0}_{ONO0}\\
         LB_{\tau_{ONN}} &=Y_1^{LB}-Y_1^{UB},\\
         UB_{\tau_{ONN}} &=Y_1^{UB}-Y_1^{LB}\\
         LB_{\tau_{NNN}} &=Y_1^{LB}-Y_1^{UB},\\
         UB_{\tau_{NNN}} &=Y_1^{UB}-Y_1^{LB}
    \end{split}
\end{equation*}
and, 
\begin{align*}
    LB_{OOO1} &=\mathbbm{E}[Y_{1}-Y_{0}|D=1,S_{0}=1,S_{1}=1, (Y_{1}-Y_{0})\leq F_{\Delta Y|111}^{-1}(p_{OOO1})]\\
    UB_{OOO1} &=\mathbbm{E}[Y_{1}-Y_{0}|D=1,S_{0}=1,S_{1}=1, (Y_{1}-Y_{0}) > F_{\Delta Y|111}^{-1}(1-p_{OOO1})]\\
    LB_{ONO1} &=\mathbbm{E}[Y_{1}-Y_{0}|D=1,S_{0}=1,S_{1}=1, (Y_{1}-Y_{0})\leq F_{\Delta Y|111}^{-1}(1-p_{OOO1})],\\
    UB_{ONO1} &=\mathbbm{E}[Y_{1}-Y_{0}|D=1,S_{0}=1,S_{1}=1, (Y_{1}-Y_{0}) > F_{\Delta Y|111}^{-1}(p_{OOO1})], \\
    LB^{0}_{ONO0} &=\mathbbm{E}[Y_{0}|D=0,S_{0}=1,S_{1}=0, Y_{0}\leq F_{Y_0|010}^{-1}(p_{ONO0})]\\
    UB^{0}_{ONO0} &=\mathbbm{E}[Y_{0}|D=0,S_{0}=1,S_{1}=0, Y_{0} > F_{Y_0|010}^{-1}(1-p_{ONO0})]\\
     LB_{NOO1} &=\mathbbm{E}[Y_{1}|D=1,S_{0}=0,S_{1}=1, Y_{1}\leq F_{Y_{1}|101}^{-1}(1-p_{NNO1})]\\
    UB_{NOO1} &=\mathbbm{E}[Y_{1}|D=1,S_{0}=0,S_{1}=1, Y_{1} > F_{Y_{1}|101}^{-1}(p_{NNO1})] \\
    LB^{0}_{OOO1} &=\mathbbm{E}[Y_{0}|D=1,S_{0}=1,S_{1}=1, Y_{0}\leq F_{Y_{0}|111}^{-1}(p_{OOO1})]\\
     LB_{NNO1} &=\mathbbm{E}[Y_{1}|D=1,S_{0}=0,S_{1}=1, Y_{1}\leq F_{Y_{1}|101}^{-1}(p_{NNO1})]\\
    UB_{NNO1} &=\mathbbm{E}[Y_{1}|D=1,S_{0}=0,S_{1}=1, Y_{1} > F_{Y_{1}|101}^{-1}(1-p_{NNO1})]\\
    LB^{0}_{ONO1} &=\mathbbm{E}[Y_{0}|D=1,S_{0}=1,S_{1}=1, Y_{0}\leq F_{Y_{0}|111}^{-1}(1-p_{OOO1})]\\
\end{align*}
with mixing proportions point identified as
\begin{align*}
    p_{OOO1}&=\frac{\mathbbm{P}[S_{1}=1|S_{0}=1,D=0]}{\mathbbm{P}[S_{1}=1|S_{0}=1,D=1]}\\
    p_{ONO0}&=1-\frac{\mathbbm{P}[S_1=0|S_0=1,D=1]}{\mathbbm{P}[S_1=0|S_0=1,D=0]}\\
    p_{NOO1}&=\frac{\mathbbm{P}[S_1=1|S_0=0,D=0]}{\mathbbm{P}[S_1=1|S_0=0,D=1]}\\
    p_{NNO1}&=1-\frac{\mathbbm{P}[S_1=1|S_0=0,D=0]}{\mathbbm{P}[S_1=1|S_0=0,D=1]}   
\end{align*}
where $Y_{0}^{LB}$ and $Y_{1}^{LB}$ are the theoretical lower bounds of the potential outcomes in the pre-treatment period and the post-treatment period, respectively and $Y_{0}^{UB}$ and $Y_{1}^{UB}$ are the theoretical upper bounds of the potential outcomes in the pre-treatment period and the post-treatment period, respectively.
\end{proof}


\subsection{Proof of Lemma \ref{lemma:mono_weights_neg}}\label{proof:mono_weights_neg}
\begin{proof}\label{proof:mono_weights_negative} \ 
    Recall that in order to bound $\tau_{OOO}$, we need to identify $p_{OOO0}$ and $p_{OOO1}$. We can express, 
    \begin{align*}
        p_{OOO0} &= \frac{\mathbbm{P}[S_1(0)=1, S_1(1)=1|S_0=1, D=0]}{\mathbbm{P}[S_1(0)=1|S_0=1, D=0]} \\
        & = \frac{\mathbbm{P}[S_1(1)=1|S_0=1, D=0]}{\mathbbm{P}[S_1(0)=1|S_0=1, D=0]}
    \end{align*} where the second equality follows from Assumption \ref{neg_monotone} (negative monotonicity). 
    
    Next, we can also express $p_{OOO1}$ as
    \begin{align*}
        p_{OOO1} &= \frac{\mathbbm{P}[S_1(0)=1, S_1(1)=1|S_0=1, D=1]}{\mathbbm{P}[S_1(1)=1|S_0=1, D=1]} \\
        & = \frac{\mathbbm{P}[S_1(1)=1|S_0=1, D=1]}{\mathbbm{P}[S_1(1)=1|S_0=1, D=1]} =1 
    \end{align*} where, again, the second equality applies Assumption \ref{neg_monotone} (negative monotonicity). Combining the two together, we get the result.
    
\end{proof} 

\subsection{Proof of Theorem \ref{mono_bound_negative}}\label{proof:mono_bound_negative}
\begin{proof}\
The treatment effect of treated for OOO group ($\tau_{OOO}$) can be decomposed as follows,
\begin{align*}
\tau_{OOO} &= \mathbbm{E}[Y_{1}^{\ast}(1)-Y_{1}^{\ast}(0)|D=1, OOO]\\
&= \mathbbm{E}[Y_{1}^{\ast}(1)-Y_{0}^{\ast}(1)+Y_{0}^{\ast}(1)-Y_{1}^{\ast}(0)|D=1, OOO]\\
&= \mathbbm{E}[Y_{1}^{\ast}(1)-Y_{0}^{\ast}(1)|D=1, OOO]-\mathbbm{E}[Y_{1}^{\ast}(0)-Y_{0}^{\ast}(1)|D=1, OOO]\\
&\text{Assumption \ref{no anti}}\\
&= \mathbbm{E}[Y_{1}^{\ast}(1)-Y_{0}^{\ast}(1)|D=1, OOO]-\mathbbm{E}[Y_{1}^{\ast}(0)-Y_{0}^{\ast}(0)|D=1, OOO]\\
&\text{Assumption \ref{PT_OOO}}\\
&= \mathbbm{E}[Y_{1}^{\ast}(1)-Y_{0}^{\ast}(1)|D=1, OOO]-\mathbbm{E}[Y_{1}^{\ast}(0)-Y_{0}^{\ast}(0)|D=0, OOO]
\end{align*}
The decomposition of $\mathbbm{E}[Y_{1}-Y_{0}|D=0,S_{0}=1,S_{1}=1]$ given in equation \ref{mixp2} lead to partial identification of $\mathbbm{E}[Y_{1}^{\ast}(0)-Y_{0}^{\ast}(0)|D=0,OOO]$ which lies within the interval $[LB_{OOO0},UB_{OOO0}]$  
where,
\begin{align*}
    LB_{OOO0} &=\mathbbm{E}[Y_{1}-Y_{0}|D=0,S_{0}=1,S_{1}=1, (Y_{1}-Y_{0})\leq F_{\Delta Y|011}^{-1}(p_{OOO0})]\\
    UB_{OOO0} &=\mathbbm{E}[Y_{1}-Y_{0}|D=0,S_{0}=1,S_{1}=1, (Y_{1}-Y_{0}) > F_{\Delta Y|011}^{-1}(1-p_{OOO0})]
\end{align*}
and $F_{\Delta Y|011}^{-1}(.)$ is the quantile function of the distribution of $Y_{1}-Y_{0}$ given $D=0,S_{0}=1,S_{1}=1$.
Under Assumption \ref{neg_monotone} (negative monotonicity) and \ref{Partialselection}(b) we can point identify $p_{OOO0}=\frac{\mathbbm{P}[S_{1}=1|S_{0}=1,D=1]}{\mathbbm{P}[S_{1}=1|S_{0}=1,D=0]}$ and $p_{OOO1}=1$ (Refer 
Lemma \ref{lemma:mono_weights_neg}). As $p_{OOO1}=1$ the decomposition of $\mathbbm{E}[Y_{1}-Y_{0}|D=1,S_{0}=1,S_{1}=1]$ given in Equation \ref{mixp} lead to point identification of $\mathbbm{E}[Y_{1}^{\ast}(1)-Y_{0}^{\ast}(1)|D=1,OOO]$ as $\mathbbm{E}[Y_{1}-Y_{0}|D=1,S_{0}=1,S_{1}=1]$.

Combining the bounds for $\mathbbm{E}[Y_{1}^{\ast}(1)-Y_{0}^{\ast}(0)|D=0,OOO]$ and point identification of $\mathbbm{E}[Y_{1}^{\ast}(0)-Y_{0}^{\ast}(0)|D=1,OOO]$ we find that the parameter of interest $\tau_{OOO} \in [LB_{\tau_{OOO}}, UB_{\tau_{OOO}}]$ where
\begin{align*}
   LB_{\tau_{OOO}} &= \mathbbm{E}[Y_{1}-Y_{0}|D=1,S_{0}=1,S_{1}=1]-UB_{OOO0}\\
   UB_{\tau_{OOO}} &=\mathbbm{E}[Y_{1}-Y_{0}|D=1,S_{0}=1,S_{1}=1]-LB_{OOO0}
\end{align*}
\end{proof}


\subsection{Proof of Lemma \ref{lemma_dids_rc}}\label{proof:lemma_dids_rc}
\begin{proof}
With repeated cross-sections, the na\"{\i}ve DiD estimand is given as 
{\small \begin{align*}
\tau^{rc}_{\textup{DiDs}} &\equiv \mathbbm{E}[Y|D=1, T=1, S_1=1]-\mathbbm{E}[Y|D=1, T=0, S_0=1] \\
&- \mathbbm{E}[Y|D=0, T=1, S_1=1]+ \mathbbm{E}[Y|D=0, T=0, S_0=1] \\
& = \mathbbm{E}[Y^\ast_1(1)|D=1, T=1, S_1=1]-\mathbbm{E}[Y_0^\ast(1)|D=1, T=0, S_0=1] \\
&- \mathbbm{E}[Y^\ast_1(0)|D=0, T=1, S_1=1]+ \mathbbm{E}[Y_0^\ast(0)|D=0, T=0, S_0=1] \\
& = \mathbbm{E}[Y^\ast_1(1)|D=1, T=1, S_1(1)=1]-\mathbbm{E}[Y_0^\ast(0)|D=1, T=0, S_0=1] \tag{Assumption \ref{no anti} } \\
&- \mathbbm{E}[Y^\ast_1(0)|D=0, T=1, S_1(0)=1]+\mathbbm{E}[Y_0^\ast(0)|D=0, T=0, S_0=1]
\end{align*}}
Consider the first expectation, $\mathbbm{E}[Y^\ast_1(1)|D=1, T=1, S_1(1)=1]$ which is
\begin{align*}
     = q_{OO11}\cdot \mathbbm{E}[Y_1^\ast(1)|D=1, T=1, OO]+(1-q_{OO11})\cdot \mathbbm{E}[Y_1^\ast(1)|D=1, T=1, NO]
\end{align*}
Similarly, the third expectation, $\mathbbm{E}[Y^\ast_1(0)|D=0, T=1, S_1(0)=1]$, can be written as 
\begin{align*}
     &= q_{OO01}\cdot \mathbbm{E}[Y^\ast_1(0)|D=0, T=1, OO]+(1-q_{OO01})\cdot \mathbbm{E}[Y^\ast_1(0)|D=0, T=1, ON] \\
     & = \mathbbm{E}[Y^\ast_1(0)|D=0, T=1, OO] \tag{Assumption \ref{monotone}}.
\end{align*} Finally, because of Assumption \ref{no anti}, 
\begin{align*}
    \mathbbm{E}[Y_0^\ast(0)|D=1, T=0, S_0=1] &= \mathbbm{E}[Y_0^\ast(0)|D=1, T=0, OO] \text{ and } \\
    \mathbbm{E}[Y_0^\ast(0)|D=0, T=0, S_0=1] &= \mathbbm{E}[Y_0^\ast(0)|D=0, T=0, OO]
\end{align*} 
Combining these four expectations together, we get
\begin{align*}
    \tau^{rc}_{\textup{DiDs}}& = q_{OO11}\cdot \mathbbm{E}[Y_1^\ast(1)|D=1, T=1, OO]+(1-q_{OO11})\cdot \mathbbm{E}[Y_1^\ast(1)|D=1, T=1, NO]\\
    &-\mathbbm{E}[Y_0^\ast(0)|D=1, T=0, OO] - \mathbbm{E}[Y^\ast_1(0)|D=0, T=1, OO]+\mathbbm{E}[Y_0^\ast(0)|D=0, T=0, OO] \\
    & = q_{OO11}\cdot \tau_{OO}+(1-q_{OO11})\cdot \{\mathbbm{E}[Y_1^\ast(1)|D=1,T=1, NO]-\mathbbm{E}[Y_1^\ast(0)|D=1, T=1, OO]\} \\
    &+ \bigg\{\mathbbm{E}[Y_1^\ast(0)|D=1, T=1, OO] - \mathbbm{E}[Y_0^\ast(0)|D=1, T=0, OO] \\
    &- \mathbbm{E}[Y_1^\ast(0)|D=0, T=1, OO]+\mathbbm{E}[Y_0^\ast(0)|D=0, T=0, OO]\bigg\} \\
    & = q_{OO11}\cdot \tau_{OO}+(1-q_{OO11})\cdot \{\mathbbm{E}[Y_1^\ast(1)|D=1,T=1, NO]-\mathbbm{E}[Y_1^\ast(0)|D=1, T=1, OO]\}) \tag{Assumption \ref{pt_rc}} \\
    & = q_{OO11}\cdot \tau_{OO}+(1-q_{OO11})\cdot \tau_{NO} \\
    &+(1-q_{OO11})\cdot (\{\mathbbm{E}[Y_1^\ast(0)|D=1,T=1, NO]-\mathbbm{E}[Y_1^\ast(0)|D=1, T=1, OO]\})
\end{align*}
\end{proof}

\subsection{Proof of Lemma \ref{lemma:pipartialid_rcs}}
\begin{proof}\label{proof:pipartialid_rcs} \ 
\begin{itemize}
    \item[(a)] For the treated group, $D=1$ in the post-treatment period, the Fr\'echet bounds are given as
	{\footnotesize	\begin{align}\label{eq:fb1_rcs}
 \mathbbm{P}[S_{1}(0)=1, S_{1}(1)=1|D=1, T=1]
  \in \left[\max\{\mathbbm{P}[S_{1}(0)=1|D=1, T=1]+\mathbbm{P}[S_{1}(1)=1|D=1, T=1]-1,0\},\right.\nonumber\\
\left.\min\{\mathbbm{P}[S_{1}(0)=1|D=1, T=1],\mathbbm{P}[S_{1}(1)=1|D=1, T=1]\}\right].
\end{align}}
Notice that 
\begin{equation}\label{eq:lemma2_11_rcs}
	\mathbbm{P}[S_{1}(1)=1|D=1, T=1]= \mathbbm{P}[S_{1}=1|D=1, T=1]
\end{equation}
which means that it can be directly identified from observed data.  Assumption \ref{Partialselection_rcs}(a) implies that
\begin{align}\label{eq:lemma2_12_rcs}
	\mathbbm{P}[S_{1}(0)=1|D=1, T=1] &=\mathbbm{P}[S_{1}(0)=1|D=0, T=1] \nonumber \\
	& = \mathbbm{P}[S_{1}=1|D=0, T=1] \\
    &\text{ OR } \nonumber 
    \end{align}
    \begin{align}\label{eq:lemma2_12b_rcs}
    \mathbbm{P}[S_{1}(0)=1|D=1, T=1] & = \mathbbm{P}[S_1(0)=1|D=0, T=1]-\mathbbm{P}[S_0(0)=1|D=0, T=0] \nonumber \\
    &+\mathbbm{P}[S_0(0)=1|D=1, T=0]
\end{align} 
Plugging \eqref{eq:lemma2_11_rcs} and \eqref{eq:lemma2_12_rcs} OR \eqref{eq:lemma2_12b_rcs} into \eqref{eq:fb1_rcs}, we get the observable bounds
{\footnotesize	\begin{align}\label{eq:fb1_obs_rcs}
    \mathbbm{P}[S_{1}(0)=1, S_{1}(1)=1|D=1, T=1]&\in \left[\max\{\mathbbm{P}[S_{1}=1|D=0, T=1]+\mathbbm{P}[S_{1}=1|D=1, T=1]-1,0\},\right.\nonumber\\
    &\left.\min\{\mathbbm{P}[S_{1}=1|D=0, T=1],\mathbbm{P}[S_{1}=1|D=1, T=1]\}\right] \\
    &\text{ OR } \nonumber 
    \end{align}
    \begin{align}\label{eq:fb1_obs_rcsb}
    \mathbbm{P}[S_{1}(0)=1, S_{1}(1)=1|D=1, T=1]&\in \big[\max\{\mathbbm{P}[S_1=1|D=0, T=1]-\mathbbm{P}[S_0=1|D=0, T=0] \nonumber\\
    &+\mathbbm{P}[S_0=1|D=1, T=0]+\mathbbm{P}[S_1=1|D=1, T=1]-1,0\},\big.\nonumber\\
    &\big.\min\{\mathbbm{P}[S_1=1|D=0, T=1]-\mathbbm{P}[S_0=1|D=0, T=0] \nonumber\\
    &+\mathbbm{P}[S_0=1|D=1, T=0],\mathbbm{P}[S_{1}=1|D=1, T=1]\}\big]. 
\end{align}}
\item[(b)] Analogously, using Assumption \ref{Partialselection_rcs}(b) for the $D=0$ group in the post-treatment period, gets us the observable bounds
	{\footnotesize	\begin{align}\label{eq:fb0_obs_rcs}
	 \mathbbm{P}[S_{1}(0)=1, S_{1}(1)=1|D=0, T=1]&\in \left[\max\{\mathbbm{P}[S_{1}=1|D=0, T=1]+\mathbbm{P}[S_{1}=1|D=1, T=1]-1,0\},\right.\nonumber\\
	&\left.\min\{\mathbbm{P}[S_{1}=1|D=0, T=1],\mathbbm{P}[S_{1}=1|D=1, T=1]\}\right] \\
    &\text{ OR } \nonumber 
     \end{align}
    \begin{align} \label{eq:fb0_obs_rcsb}
    \mathbbm{P}[S_{1}(0)=1, S_{1}(1)=1|D=0, T=1]&\in \big[\max\{\mathbbm{P}[S_1=1|D=1, T=1]-\mathbbm{P}[S_0=1|D=1, T=0] \nonumber\\
    &+\mathbbm{P}[S_0=1|D=0, T=0]+\mathbbm{P}[S_1=1|D=0, T=1]-1,0\},\big.\nonumber\\
    &\big.\min\{\mathbbm{P}[S_1=1|D=1, T=1]-\mathbbm{P}[S_0=1|D=1, T=0] \nonumber\\
    &+\mathbbm{P}[S_0=1|D=0, T=0],\mathbbm{P}[S_{1}=1|D=0, T=1]\}\big].
\end{align}} Equations \eqref{eq:fb1_obs_rcs}, \eqref{eq:fb1_obs_rcsb}, \eqref{eq:fb0_obs_rcs}, and \eqref{eq:fb0_obs_rcsb} taken together give us the desired result.
\end{itemize}	
\end{proof}

\subsection{Proof of Theorem \ref{nomono_bound_rcs}}\label{proof:nomono_bound_rcs}
\begin{proof}
The treatment effect of treated for OO group ($\tau_{OO}$) can be decomposed as follows,
\begin{align}\label{eq:mono_bound_rcs}
\tau_{OO} &= \mathbbm{E}[Y_{1}^{\ast}(1)-Y_{1}^{\ast}(0)|D=1, T=1, OO] \nonumber \\
&\text{Assumption \ref{pt_rc}} \nonumber \\
& = \mathbbm{E}[Y_1^{\ast}(1)|D=1, T=1, OO]-\mathbbm{E}[Y_1^\ast(0)|D=0, T=1, OO] \nonumber \\
&+ \mathbbm{E}[Y_0^\ast(0)|D=0, T=0, OO] - \mathbbm{E}[Y_0^\ast(0)|D=1, T=0, OO] \nonumber \\
&\text{Assumption \ref{no anti}} \nonumber \\
&= \mathbbm{E}[Y_1^{\ast}(1)|D=1, T=1, OO]-\mathbbm{E}[Y_1^{\ast}(0)|D=0, T=1, OO] \nonumber \\
&+ \mathbbm{E}[Y_0^{\ast}(0)|D=0, T=0, OO] - \mathbbm{E}[Y_0^{\ast}(1)|D=1, T=0, OO]
\end{align}

The decomposition of $\mathbbm{E}[Y|D=1,T=1, S=1]$ given in equation \eqref{mixp_rcs111} lead to partial identification of $\mathbbm{E}[Y_{1}^{\ast}(1)|D=1,T=1, OO]$ which lies within the interval $[LB_{OO11},UB_{OO11}]$ 
where,
\begin{equation}\label{LBUB_rcs}
    \begin{split}
        LB_{OO11} &=\mathbbm{E}[Y|D=1,T=1, S=1, Y\leq F_{ Y|111}^{-1}(q_{OO11})] \\
        UB_{OO11} &=\mathbbm{E}[Y|D=1,T=1, S=1, Y > F_{ Y|111}^{-1}(1-q_{OO11})]
    \end{split}
\end{equation}
and $F_{Y|111}^{-1}(.)$ is the quantile function of the distribution of $Y$ given $D=1,T=1, S=1$. Using Lemma \ref{lemma:pipartialid_rcs}, depending on which version of \ref{Partialselection_rcs}(a) \& (b) we use, we know that for any $\nu_d \in [\max\{\mathbbm{P}[S_1=1|D=d, T=1]+\mathbbm{P}[S_1=1|D=1-d, T=1]-1, 0\}, \min\{\mathbbm{P}[S_1=1|D=d, T=1], \mathbbm{P}[S_1=1|D=1-d, T=1]\}]$ OR $\nu_d \in [\max\{\mathbbm{P}[S_1=1|D=d, T=1]+\mathbbm{P}[S_1=1|D=1-d, T=1]-\mathbbm{P}[S_0=1|D=1-d, T=0]+\mathbbm{P}[S_0=1|D=d, T=0]-1, 0\}, \min\{\mathbbm{P}[S_1=1|D=d, T=1], \mathbbm{P}[S_1=1|D=1-d,T=1]-\mathbbm{P}[S_0=1|D=1-d, T=0]+\mathbbm{P}[S_0=1|D=d, T=0]\}]$, one can obtain the associated mixing probabilities/weights as $q_{gd1}(\nu_d) = \frac{\nu_d}{\mathbbm{P}[S_1=1|D=d, T=1]}$. Evaluating equations (\ref{LBUB_rcs}) for each value of $d$ at the least favorable values for $q_{gd1}(\nu_d)$ yields,
\begin{align}
    LB_{gd1}(\nu^{l}_{d}) &=\mathbbm{E}[Y|D=d,T=1, S=1, Y\leq F_{Y|d11}^{-1}(q_{gd1}(\nu^l_d))] \label{LBlow_rcs}\\
    UB_{gd1}(\nu^{l}_{d}) &=\mathbbm{E}[Y|D=d,T=1, S=1, Y > F_{Y|d11}^{-1}(1-q_{gd1}(\nu^l_d)))]. \label{UBlow_rcs}
\end{align} where $\nu_d^l$ is the lower bound of the identified set $\mathbbm{P}[S_{1}(0)=1, S_{1}(1)=1|D=d, T=1]$ for $d=0,1$. Therefore, combining the bounds for the treated and untreated subpopulations along with partial identification of the weights together with point identification of $\mathbbm{E}[Y_0^{\ast}(1)|D=1, T=0, OO]$ and $\mathbbm{E}[Y_0^{\ast}(0)|D=0, T=0, OO]$ implies that $\tau_{OO} \in [LB_{\tau_{OO}}, UB_{\tau_{OO}}]$ where 
{\small \begin{align*}
    LB_{\tau_{OO}} &= LB_{OO11}(v^l_1)-UB_{OO01}(v^l_0) - \mathbbm{E}[Y|D=1, T=0, S=1]+\mathbbm{E}[Y|D=0, T=0, S=1] \\
    UB_{\tau_{OO}} &= UB_{OOO1}(v^l_1)-LB_{OOO0}(v^l_0)- \mathbbm{E}[Y|D=1, T=0, S=1]+\mathbbm{E}[Y|D=0, T=0, S=1].
\end{align*}}
\end{proof}
\subsection{Proof of Lemma \ref{lemma:mono_weights_rcs}}\label{proof:mono_weights_rcs}
\begin{proof}
    Recall that to bound $\tau_{OO}$, we need to identify $q_{OO11}$ and $q_{OO01}$. We can express, 
    \begin{align*}
        q_{OO01} &= \frac{\mathbbm{P}[S_1(0)=1, S_1(1)=1|D=0, T=1]}{\mathbbm{P}[S_1(0)=1|D=0, T=1]} \\
        & = \frac{\mathbbm{P}[S_1(0)=1|D=0, T=1]}{\mathbbm{P}[S_1(0)=1|D=0, T=1]} \\
        & = 1
    \end{align*} where the second equality follows from positive monotonicity. Next, we can also express $q_{OO11}$ as
    \begin{align*}
        q_{OO11} &= \frac{\mathbbm{P}[S_1(0)=1, S_1(1)=1|D=1, T=1]}{\mathbbm{P}[S_1(1)=1|D=1, T=1]} \\
        & = \frac{\mathbbm{P}[S_1(0)=1|D=1, T=1]}{\mathbbm{P}[S_1(1)=1|D=1, T=1]}. 
    \end{align*} 
    where, again, the second equality applies positive monotonicity. Combining the two together, we get the desired result. 
 
\end{proof}
\subsection{Proof of Theorem \ref{mono_bound_rcs}}\label{proof Theorem mono_bound_rcs}
\begin{proof}
The treatment effect of treated for OO group ($\tau_{OO}$) can be decomposed as follows,
\begin{align}\label{eq:mono_bound_rcs}
\tau_{OO} &= \mathbbm{E}[Y_{1}^{\ast}(1)-Y_{1}^{\ast}(0)|D=1, T=1, OO] \nonumber \\
&\text{Assumption \ref{pt_rc}} \nonumber \\
& = \mathbbm{E}[Y_1^{\ast}(1)|D=1, T=1, OO]-\mathbbm{E}[Y_1^\ast(0)|D=0, T=1, OO] \nonumber \\
&+ \mathbbm{E}[Y_0^\ast(0)|D=0, T=0, OO] - \mathbbm{E}[Y_0^\ast(0)|D=1, T=0, OO] \nonumber \\
&\text{Assumption \ref{no anti}} \nonumber \\
 &= \mathbbm{E}[Y_1^{\ast}(1)|D=1, T=1, OO]-\mathbbm{E}[Y_1^{\ast}(0)|D=0, T=1, OO] \nonumber \\
&+ \mathbbm{E}[Y_0^{\ast}(0)|D=0, T=0, OO] - \mathbbm{E}[Y_0^{\ast}(1)|D=1, T=0, OO]
\end{align}

The decomposition of $\mathbbm{E}[Y|D=1,T=1, S=1]$ given in equation \eqref{mixp_rcs111} lead to partial identification of $\mathbbm{E}[Y_{1}^{\ast}(1)|D=1,T=1, OO]$ which lies within the interval $[LB_{OO11},UB_{OO11}]$ 
where,
\begin{align*}
    LB_{OO11} &=\mathbbm{E}[Y|D=1,T=1, S=1, Y\leq F_{ Y|111}^{-1}(q_{OO11})]\\
    UB_{OO11} &=\mathbbm{E}[Y|D=1,T=1, S=1, Y > F_{ Y|111}^{-1}(1-q_{OO11})]
\end{align*}
and $F_{Y|111}^{-1}(.)$ is the quantile function of the distribution of $Y$ given $D=1,T=1, S=1$.
Under Assumption \ref{monotone} (positive monotonicity) we can point identify $q_{OO01}=1$ (Refer 
Lemma \ref{lemma:mono_weights_rcs}) and depending on which version of Assumption \ref{Partialselection_rcs}(a) we use, $q_{OO11}=\frac{\mathbbm{P}[S_{1}=1|S_{0}=1,D=0]}{\mathbbm{P}[S_{1}=1|S_{0}=1,D=1]}$ OR $q = \frac{\mathbbm{P}[S_1(0)=1|=0, T=1]-\mathbbm{P}[S_0(0)=1|D=0, T=0]+\mathbbm{P}[S_0(0)=1|D=1, T=0]}{\mathbbm{P}[S_{1}=1|S_{0}=1,D=1]}$.
As $q_{OO01}=1$ the decomposition of $\mathbbm{E}[Y|D=0, T=1, S=1]$ given in equation \eqref{mixp_rcs011} leads to point identification of $\mathbbm{E}[Y_{1}^{\ast}(0)|D=0, T=1, OO]$ as $\mathbbm{E}[Y|D=0,T=1, S=1]$.

Combining the bounds for $\mathbbm{E}[Y_{1}^{\ast}(1)|D=1,T=1, OO]$ and point identification of $\mathbbm{E}[Y_{1}^{\ast}(0)|D=0, T=1, OO]$, $\mathbbm{E}[Y_0^\ast(1)|D=1, T=0, OO]$, and $\mathbbm{E}[Y_0^\ast(0)|D=0, T=0, OO]$ we find that the parameter of interest $\tau_{OO} \in [LB_{\tau_{OO}}, UB_{\tau_{OO}}]$ where
{\footnotesize \begin{align*}
     LB_{\tau_{OO}} &= LB_{OO11}-\mathbbm{E}[Y|D=0, T=1, S=1]-\mathbbm{E}[Y|D=1, T=0, S=1]+\mathbbm{E}[Y|D=0, T=0, S=1],\\
     UB_{\tau_{OO}} &=UB_{OO11}-\mathbbm{E}[Y|D=0, T=1, S=1]-\mathbbm{E}[Y|D=1, T=0, S=1]+\mathbbm{E}[Y|D=0, T=0, S=1]
\end{align*}}
\end{proof}

\subsection{Proof of Theorem \ref{Gen_bound_cond}}
\begin{proof}\label{proof:Gen_bound_cond}
    For the positive selection region ($r=+$), the conditional bounds are identified in the same manner as in Theorem \ref{mono_bound}. As we saw in Section \ref{relax mono}, $p_{OOO0}^+(X) = 1$ and $p^{+}_{OOO1}(X) = \dfrac{\mathbbm{P}(S_1=1\mid D=0, X, S_0=1)}{\mathbbm{P}(S_1=1\mid D=1, X, S_0=1)}$. Here, we only need to trim the distribution of $Y_1-Y_0$ for the observed sample of treated units, which contains a mixture of latent types OOO and ONO, whereas the untreated sample point-identifies $\mathbbm{E}[Y_1^\ast(0)-Y_0^\ast(0)\mid D=0, X, OOO]$.

    Then, combining this with the bounds for $\mathbbm{E}[Y_1^\ast(0)-Y_0^\ast(0)\mid D=1, X, OOO]$ presented in \eqref{bounds1_cond_nomono}, we can partially identify $\tau_{OOO}$ in this set as follows:
    \begin{align}\label{lbub_cond+}
        \tau_{OOO}(X) &\in [LB_{OOO1}(X)-\mathbbm{E}[Y_1-Y_0|D=0,X,S_{0}=1,S_{1}=1],  \nonumber \\
        & \qquad UB_{OOO1}(X)-\mathbbm{E}[Y_1-Y_0|D=0,X,S_{0}=1,S_{1}=1] \quad  \text{for } X \in \mathcal{X}_{+}.
    \end{align}
    A similar but reverse argument holds true for the covariate set where selection is negative. In this case, $p^{-}_{OOO1}(X)=1$ and $p^{-}_{OOO0}(X) = \dfrac{\mathbbm{P}(S_1=1\mid D=1, X, S_0=1)}{\mathbbm{P}(S_1=1\mid D=0, X, S_0=1)}$. Here, it suffices to trim the distribution of $Y_1 - Y_0$ for the observed sample of untreated units, which consists of a mixture of latent types OOO and OON whereas the treated sample point-identifies $\mathbbm{E}[Y_1^\ast(0)-Y_0^\ast(0)\mid D=1, X, OOO]$. Combining this with the bounds for $\mathbbm{E}[Y_1^\ast(0)-Y_0^\ast(0)\mid D=0, X, OOO]$ given in \eqref{bounds1_cond_nomono}, we get:
    \begin{align}\label{lbub_cond-}
        \tau_{OOO}(X) &\in \big[\mathbbm{E}[Y_1-Y_0|D=1, X, S_0=1, S_1=1] - UB_{OOO0}(X), \nonumber \\
         &\qquad \mathbbm{E}[Y_1-Y_0|D=1, X, S_0=1, S_1=1]- LB_{OOO0}(X) \big] \quad  \text{for } X \in \mathcal{X}_{-}.
    \end{align}
   In the covariate set with no selection, the treated and untreated samples point identify the expectations, $\mathbbm{E}[Y_1^\ast(0)-Y_0^\ast(0)\mid D=d, X, OOO]$ for both groups and so no trimming is required since $p_{OOOd}(X)=1.$ Therefore, 
    \begin{align}\label{lbub_cond0}
        LB_{OOO1}(X) &=  UB_{OOO1}(X) = \mathbbm{E}[Y_1-Y_0|D=1, X, S_0=1, S_1=1] \text{ and } \nonumber \\
        LB_{OOO0}(X) &=  UB_{OOO0}(X) = \mathbbm{E}[Y_1-Y_0|D=0, X, S_0=1, S_1=1] \quad  \text{for } X \in \mathcal{X}_{0}.
    \end{align}
    implying that $\tau_{OOO}$ is point-identified.

    The unconditional ATT for the OOO type is then derived by aggregating the conditional bounds in equations \eqref{lbub_cond+}–\eqref{lbub_cond0} across covariate regions. Specifically, we integrate over the distribution of $X^r$ for $r \in \{+, -, 0\}$ using the formula provided in equations \eqref{nomonox_lb_true} and \eqref{nomonox_ub_true} for the lower and upper bounds, respectively.
\end{proof}

\section{Additional Material}\label{additional_material}
\subsection{Consistency, Asymptotic Distribution, and Bootstrap Validity for the Bounds Estimators} \label{Asymptotics}
The moment conditions identifying the bounds of interest in the main text involve indicator functions and maximum operators, which introduce non-smoothness into the estimating equations. Nevertheless, uniform convergence and consistency of the estimator can still be established using the arguments in \citet{chen2003estimation}, which provide general results for the consistency and asymptotic distribution M-estimators involving non-smooth functions (like indicator functions) and plug-in estimates (such as the quantiles or trimming proportions). In particular, \cite{chen2003estimation} Theorems 1 and 2 can be applied directly to our setting. Furthermore, in section \ref{appendix:Bootstrap} we show the validity of the bootstrap in our setting by establishing that the conditions on \cite{chen2003estimation} Theorem B holds.

\begin{theorem}[\cite{chen2003estimation} Theorem 1]\label{chen2003theorem1}
Suppose that $\theta_0 \in \Theta$ satisfies $M\left(\theta_0, h_0\right)=0$, and that:
\begin{enumerate}
    \item[(1.1)] $\left\|M_n(\widehat{\theta}, \widehat{h})\right\| \leq \inf _{\theta \in \Theta}\left\|M_n(\theta, \widehat{h})\right\|+o_p(1)$;
    \item[(1.2)] for all $\delta>0$, there exists $\epsilon(\delta)>0$ such that $\inf _{\|\theta-\theta_0 \|>\delta} \| M\left(\theta, h_0\right) \| \geq \epsilon(\delta)>0$;
    \item[(1.3)] uniformly for all $\theta \in \Theta, M(\theta, h)$ is continuous (with respect to the metric $\|\cdot\|_{\mathcal{H}}$ ) in $h$ at $h=h_0$;
    \item[(1.4)] $\left\|\widehat{h}-h_0\right\|_{\mathcal{H}}=o_p(1)$;
    \item[(1.5)] for all sequences of positive numbers $\left\{\delta_n\right\}$ with $\delta_n=o(1)$,
$$
\sup _{\theta \in \Theta,\|h-h_0\|_\mathcal{H} \leq \delta_n} \frac{\left\|M_n(\theta, h)-M(\theta, h)\right\|}{1+\left\|M_n(\theta, h)\right\|+\|M(\theta, h)\|}=o_p(1)
$$    
\end{enumerate}
Then, $\widehat{\theta}-\theta_0=o_p(1)$.
Remark 1: (i) Condition (1.5) is implied by condition (1.5'): For all positive sequences $\delta_n=o(1)$,
$$
\sup _{\theta \in \Theta,\left|h-h_0\right|_\mathcal{H}  \leq \delta_n}\left\|M_n(\theta, h)-M(\theta, h)\right\|=o_p(1) .
$$
\end{theorem}

And for the asymptotic distribution,

\begin{theorem}[\cite{chen2003estimation} Theorem 2]
Suppose that $\theta_0 \in \Theta_\delta\equiv \{\theta \in \Theta:\|\theta-\theta_0\|<\delta\}$ satisfies $M(\theta_0, h_0) = 0$, that $\hat{\theta} - \theta_0 = o_p(1)$, and that:
\begin{align*}
(2.1)\quad & \|M_n(\hat{\theta}, \hat{h})\| = \inf_{\theta \in \Theta_\delta} \|M_n(\theta, \hat{h})\| + o_p(1/\sqrt{n}). \\
(2.2)\quad & \text{(i) The ordinary derivative }  \Gamma_1(\theta, h_0)  \text{   in }\theta \text{ of } M(\theta, h_0) \text{ exists for } \theta \in \Theta_\delta, \text{ and is} \\
& \quad \text{continuous at } \theta = \theta_0; \text{ (ii) the matrix } \Gamma_1 \equiv \Gamma_1(\theta_0, h_0) \text{ is of full (column) rank}. \\
(2.3)\quad & \text{For all } \theta \in \Theta_\delta \text{ the pathwise derivative } \Gamma_2(\theta, h_0)[h - h_0] \text{ of } M(\theta, h_0) \text{ exists in all} \\
& \quad \text{directions } [h - h_0] \in \mathcal{H}; \text{ and for all } (\theta, h) \in \Theta_{\delta_n} \times \mathcal{H}_{\delta_n}, \text{ with a positive sequence } \delta_n = o(1): \\
& \quad \text{(i) } \|M(\theta, h) - M(\theta, h_0) - \Gamma_2(\theta, h_0)[h - h_0]\| \leq c \|h - h_0\|^2_\mathcal{H} \text{ for a constant } c \geq 0; \\
& \quad \text{(ii) } \|\Gamma_2(\theta, h_0)[h - h_0] - \Gamma_2(\theta_0, h_0)[h - h_0]\| \leq o(1)\delta_n. \\
(2.4)\quad & \hat{h} \in \mathcal{H} \text{ with probability tending to one; and } \|\hat{h} - h_0\|_\mathcal{H} = o_p(n^{-1/4}). \\
(2.5)\quad & \text{For all sequences of positive numbers } \{\delta_n\} \text{ with } \delta_n = o(1), \\
& \quad \sup_{\|\theta - \theta_0\| \leq \delta_n, \|h - h_0\|_\mathcal{H} \leq \delta_n} 
\frac{\sqrt{n}\|M_n(\theta, h) - M(\theta, h) - M_n(\theta_0, h_0)\|}{1 + \sqrt{n} \{\|M_n(\theta, h)\| + \|M(\theta, h)\|\}} = o_p(1). \\
(2.6)\quad & \text{For some finite matrix } V_1, \quad \sqrt{n} \left[ M_n(\theta_0, h_0) + \Gamma_2(\theta_0, h_0)[\hat{h} - h_0] \right] \xrightarrow{d} \mathcal{N}[0, V_1]. \\
& \text{Then, } \sqrt{n}(\hat{\theta} - \theta_0) \xrightarrow{d} \mathcal{N}(0, \Omega), \text{ where } \Omega \equiv (\Gamma_1' W \Gamma_1)^{-1} \Gamma_1' W V_1 W \Gamma_1 (\Gamma_1' W \Gamma_1)^{-1}.
\end{align*}
\end{theorem}

Now we proceed to check that the conditions in \cite{chen2003estimation} Theorems 1 and 2 hold in our setting and establish the properties of the estimators.

\subsubsection{Consistency of the Trimmed Mean Estimator}
\begin{lemma}[Consistency]
     Let $\Delta Y=Y_1-Y_0$ have a bounded support, and suppose $\mathbbm{P}[S_1=1|S_0=1,D=1]>0$ and $p_0>0$, then
    \begin{align*}
       \widehat{LB}_{\tau_{OOO}}  \xrightarrow{p}  LB_{\tau_{OOO}} \quad &and \quad \widehat{UB}_{\tau_{OOO}} \xrightarrow{p}   UB_{\tau_{OOO}} 
    \end{align*}
\end{lemma}
\begin{proof}
First, we prove the consistency of the trimmed mean estimator by establishing the arguments in Theorem 1 of \citet{chen2003estimation}.  Proving the consistency of the trimmed mean for the treatment group and only for the lower bound is sufficient, as a symmetric argument will follow for the other trimmed mean estimators, which are of the same form. Let \( Z_i = \big( Y_{i0}, Y_{i1}, D_{i}, S_{i0}, S_{i1} \big) \in \mathcal{Z} \). Let us denote $\theta_0=\mathbbm{E}[Y_{1}-Y_{0}|D=1,S_{0}=1,S_{1}=1, (Y_{1}-Y_{0})\leq F_{\Delta Y|111}^{-1}(p_{OOO1}(v^{l}_{1}))]$ as the true parameter of interest. For notational simplicity, we will denote $q_{p_0}=F_{\Delta Y|111}^{-1}(p_{OOO1}(v^{l}_{1}))$, so that the true parameter of interest can be rewritten simply as $\theta_0=\mathbbm{E}[\Delta Y|D=1,S_{0}=1,S_{1}=1, \Delta Y\leq q_{p_0}]$. Suppose $\theta_0 \in \Theta=[a,b] \subset \mathbb{R}$, where we take the parameter space to be the bounds of the support for the trimmed mean. The consistency follows from applying Theorem 1  of \citet{chen2003estimation}. Define the moment function,
\begin{align*}
     g(Z, \theta, h) = (\Delta Y - \theta) \cdot \mathbbm{1}\{\Delta Y \leq q_p,\, D = 1,\, S_0 = 1,S_1=1\}, \text{ with } h=(p,q_p)
\end{align*}
where $p$ is the trimming proportion and is computed as, $p = \max\left\{ 0, \frac{\alpha+\beta-1}{\beta} \right\}$, with $\alpha=\mathbbm{P}[S_{1}=1|D=0, S_{0}=1]$ and $\beta=\mathbbm{P}[S_{1}=1|D=1, S_{0}=1]$ and $q_p$ the $p^{th}$ quantile of the distribution $\Delta Y|D=1,S_0=1,S_1=1$. The sample moment can be defined as, $M_n(\theta, \hat{h}) = E_n [g(Z_i, \theta, \hat{h})] \equiv \frac{1}{n}\sum_{i=1}^{n}g(Z_i, \theta, \hat{h})$.
The estimator of $\theta_0$ as given in section \ref{estimation} is a solution to $\min_{\theta \in \Theta} \left\| M_n(\theta, \hat{h}) \right\|$. The true parameter value $\theta_0 \in \Theta$ satisfies the population moment condition $M(\theta_0, h_0) = \mathbbm{E}[g(Z, \theta_0, h_0)] = 0$. Now, let us verify each condition in Theorem 1 of \citet{chen2003estimation}.
\begin{condition}[1.1]
     \[
    \left\|M_n(\widehat{\theta}, \widehat{h})\right\| \leq \inf _{\theta \in \Theta}\left\|M_n(\theta, \widehat{h})\right\|+o_p(1)
    \]
\end{condition}
This ensures that the estimator behaves like a reasonable approximation to the true solution of the sample moment problem. In our case $\left\|M_n(\theta, \widehat{h})\right\|$ is a scalar linear in $\theta$. The unique solution to $\left\|M_n(\widehat{\theta}, \widehat{h})\right\|=0$ exactly minimizes the criterion. Thus, this condition holds by construction. In other words, this condition is satisfied since the estimator $\hat{\theta}$ solves the moment condition without approximation.
   
\begin{condition}[1.2]
   For all $\delta>0$, there exists $\epsilon(\delta)>0$ such that $\inf _{\mid \theta-\theta_0 \mid>\delta} \| M\left(\theta, h_0\right) \| \geq \epsilon(\delta)>0$;  
\end{condition}

This identification condition ensures that the population moment function $M\left(\theta, h_o\right) $ has a unique zero at $\theta=\theta_0$. Thus, the parameter is identifiable from the population moment.

The population moment function evaluated at $h_{0}$, is given by
 \[
    M(\theta, h_0) = \mathbbm{E}[(\Delta Y - \theta) \cdot \mathbf{1}\{\ \Delta Y \leq q_{p_0}, D = 1, S_0 = 1,S_1=1\}]
    \]
    
    As  \( \theta_0 = \mathbbm{E}[\Delta Y \mid \Delta Y \leq q_{p_0},  D = 1, S_0 = 1,S_1=1] \), and define $\pi_0 = \mathbbm{P}(\Delta Y \leq  q_{p_0},  D = 1, S_0 = 1,S_1=1) > 0$. Then $M(\theta, h_0) = (\theta_0 - \theta) \cdot \pi_0$. Then, take  $\epsilon(\delta)=(\delta +1)\pi_0$ to obtain
    \[
    \|M(\theta, h_0)\| = |\theta - \theta_0| \cdot \pi_0 \geq \epsilon(\delta) \quad \text{for all } |\theta - \theta_0| > \delta
    \] 
    Thus, condition 1.2 is satisfied.
   \begin{condition}[1.3]
       Uniformly for all $\theta \in \Theta, M(\theta, h)$ is continuous (with respect to the metric $\|\cdot\|_{\mathcal{H}}$ ) in $h$ at $h=h_o$;
   \end{condition} 
   It ensures that the population moment function $M(\theta, h_0)$ has no jumps or discontinuities for $h$ in the neighborhood of $h_0$. Formally,
\[
\forall \varepsilon > 0,\ \exists \ \delta > 0 \text{ such that }
\|h - h_0\|_{\mathcal{H}} < \delta
\;\Longrightarrow\;
\sup_{\theta \in \Theta}\bigl|M(\theta,h) - M(\theta,h_0)\bigr| < \varepsilon.
\]
   \[
\begin{aligned}
\text{Fix any }\theta\in\Theta.\quad 
&|M(\theta,h)-M(\theta,h_0)| \\
&= \Bigl|\mathbbm{E}\bigl[(\Delta Y-\theta)\{\mathbb{I}\{ \Delta Y\le q_p\}-\mathbb{I}\{ \Delta Y\le q_{p_0}\}\}\mathbb{I}\{D=1,S_0=1,S_1=1\}\bigr]\Bigr|\\
&\le \mathbbm{E}\bigl[\,| \Delta Y-\theta|\;\bigl|\mathbb{I}\{ \Delta Y\le q_p\}-\mathbb{I}\{\Delta Y\le q_{p_0}\}\bigr|\;\mathbb{I}\{D=1,S_0=1,S_1=1\}\bigr].
\end{aligned}
\] 
Since $\Theta$ is compact, let $B = \max_{\theta\in\Theta}|\theta| \;<\;\infty$. Then, for every $\theta$, $| \Delta Y-\theta|\,\mathbb{I}\{D=1,S_0=1,S_1=1\}\le |\Delta Y| + B$. Hence,
\[
\sup_{\theta \in \Theta}\bigl|M(\theta,h) - M(\theta,h_0)\bigr|\leq \mathbbm{E}\bigl[\,(|\Delta Y| + B)\;\bigl|\mathbb{I}\{ \Delta Y\le q_p\}-\mathbb{I}\{\Delta Y\le q_{p_0}\}\bigr|\;\bigr]
\]
Then, assuming continuous $F_{\Delta Y|111}$, when $q_{p}\xrightarrow{} q_{p_0}$ we have almost surely, $\mathbb{I}\{ \Delta Y\le q_p\} \xrightarrow{}\mathbb{I}\{\Delta Y\le q_{p_0}\}$. Thus, $|\mathbb{I}\{ \Delta Y\le q_p\}-\mathbb{I}\{\Delta Y\le q_{p_0}\}|\xrightarrow{} 0 \quad a.s$.
Now define $W_q=(|\Delta Y| + B)\;\bigl|\mathbb{I}\{ \Delta Y\le q_p\}-\mathbb{I}\{\Delta Y\le q_{p_0}\}\bigr|$. Then $|W_q|\leq |\Delta Y| + B=G$ and $\mathbbm{E}[G]<\infty$ as by assumption $\mathbbm{E}[|\Delta Y|]< \infty$. Also $W_q \xrightarrow{} 0 \quad a.s $ as $q_{p}\xrightarrow{} q_{p_0}$.  Then by the dominated convergence theorem we have,
\[
\lim_{q_p\to q_{p_0}}\mathbbm{E}\bigl[W_q\bigr]
\;=\;
\mathbbm{E}\Bigl[\lim_{q_p\to q_{p_0}}W_q\Bigr]
\;=\;0.
\]
Our matrix $\|h-h_0\|_{\mathcal{H}}=\max\{|p-p_0|,|q_p-q_{p_0}|\}$ then $\|h-h_0\|_{\mathcal{H}}<\delta \implies |p-p_0|<\delta $ and $|q_p-q_{p_0}|<\delta$

Given any $\epsilon >0$, choose $\delta>0$ so small that whenever $|q_p-q_{p_0}|<\delta$, $\mathbbm{E}\bigl[W_q\bigr] < \epsilon$. Then for all h with $\|h-h_{0}\|_{\mathcal{H}}<\delta$ we have,
\[
\sup_{\theta\in\Theta}\bigl|M(\theta,h)-M(\theta,h_{0})\bigr|
\;\le\;\mathbbm{E}[W_{q}]<\varepsilon
\] 
Thus, condition 1.3 holds.
 \begin{condition}[1.4]
     $\left\|\widehat{h}-h_0\right\|_{\mathcal{H}}=o_p(1)$;
 \end{condition}
This ensures the consistency of the nuisance parameters $h=(p,q)$.\\
   $\hat{p} = \max\left\{ 0, \frac{\hat{\alpha}+\hat{\beta}-1}{\hat{\beta}} \right\}$, with $\hat{\alpha}=\mathbbm{\hat{P}}[S_{1}=1|D=0, S_{0}=1]$ and $\hat{\beta}=\mathbbm{\hat{P}}[S_{1}=1|D=1, S_{0}=1]$  
by law of large numbers $\hat{\alpha} \xrightarrow{p} \alpha_0 $ and $ \hat{\beta} \xrightarrow{p} \beta_0$. Assume $\beta_0=\mathbbm{P}[S_{1}=1|D=1, S_{0}=1]>0$. Then by Slutsky’s Theorem and continuity of the ratio we have $\frac{\hat{\alpha}+\hat{\beta}-1}{\hat{\beta}}\xrightarrow{p}\frac{\alpha_0+ \beta_0-1}{\beta_0}$. As the maximum operator is continuous, by the continuous mapping theorem, we have $\hat{p}\xrightarrow{p} p_0$.  By the Glivenko–Cantelli theorem for conditional samples, the empirical CDF converges uniformly to the true CDF. Then, assuming  $F_{\Delta Y|111}$  continuous in a neighbourhood of $q_{p_0}$ and $f_{\Delta Y|111}(q_{p_0})>0$ , together with the consistency of $\hat{p} $ , by continuous‐mapping theorem (or by standard quantile consistency) \citep{vandervaart1998}, we have $\hat{q}_p\xrightarrow{p}q_{p_0}$. Therefore, $\hat{h}\xrightarrow{p} h_0$. 
Since both components satisfy
\[
|\hat{p} - p_0| = o_p(1),
\qquad
|\hat{q}_p - q_{p_0}| = o_p(1),
\] 
we get
\[
\|\hat h - h_0\|_{\mathcal H}
= \max\bigl\{\,|\hat{p} - p_0| ,\;|\hat{q}_p - q_{p_0}| \,\bigr\}
= o_p(1).
\]
Thus, Condition 1.4 holds.
\begin{condition}[1.5]
    for all sequences of positive numbers $\left\{\delta_n\right\}$ with $\delta_n=o(1)$,
$$
\sup _{\theta \in \Theta,\left|h-h_0\right|_{\mathcal{H}} \leq \delta_n} \frac{\left\|M_n(\theta, h)-M(\theta, h)\right\|}{1+\left\|M_n(\theta, h)\right\|+\|M(\theta, h)\|}=o_p(1)$$
\end{condition}

This is a uniform convergence condition that controls the empirical approximation error $M_n(\theta, h)-M(\theta, h)$ uniformly in $\theta \in \Theta$ and locally in $h$ near $h_o$. This condition ensures that the empirical moment function $M_n(\theta, h)$ converges uniformly to its population counterpart $M(\theta, h)$ over all $\theta \in \Theta$ even when the nuisance function $h$ varies in a neighbourhood of $h_o$.
This condition is implied by,
$$
\sup _{\theta \in \Theta,\left|h-h_0\right|_\mathcal{H}  \leq \delta_n}\left\|M_n(\theta, h)-M(\theta, h)\right\|=o_p(1) .
$$
Define $M(\theta, h) = \mathbbm{E}[g(Z, \theta, h)]$. and $M_n(\theta, h) = \mathbbm{E}_n[g(Z, \theta, h)]\equiv \frac{1}{n}\sum_{i=1}^{n}g(Z_i, \theta, h)$.  
The class of the moment function $\mathcal{F}=\{g(Z, \theta, h):\theta \in \Theta,\|h-h_0\|<\delta_n\}$ is VC-subgraph as the indicator $\mathbbm{1}\{\Delta Y \leq q_p,\, D = 1,\, S_0 = 1,S_1=1\}$ is a subgraph of the threshold class which is VC of finite dimension and multiplying by the linear term $(\Delta Y - \theta)$  over the compact $\Theta$ preserves VC-subgraph property \citep{vandervaart1996}.  Since $\Theta$ is compact and $\Delta Y$ is bounded the envelop $G=\sup_{\theta \in 
 \Theta}|\Delta y-\theta| \leq  C$.  Thus, $\mathcal{F}$ is a VC-subgraph with bounded envelope. Therefore, this class is a uniform Glivenko–Cantelli class, which imply that sample averages converge uniformly to expectations \citep{vandervaart1996}. The class is equicontinuous (i.e this means that small changes in the nuisance component $h$ lead to uniformly small changes in the function values) in $h$ with high probability, as $F_{\Delta Y |111}$ is continuous and strictly increasing. Thus, this ensures that $M_n(\theta,h)$ converges uniformly to $M(\theta,h)$ over $\theta \in \Theta$ and small neighborhoods of $h_0$, thus verifying condition 1.5.
\newline
As all conditions in \cite{chen2003estimation} Theorem 1 are satisfied, it follows $\hat{\theta}\xrightarrow{p}\theta_0$.

Note that the estimator for the bound  $LB_{\tau_{OOO}}$ as defined in Theorem \ref{nomono_bound} is a difference of two trimmed means of the same form. Since we have established the consistency of a general trimmed mean estimator, it follows from Slutsky’s theorem that the difference is also consistent. Therefore,
\[
\widehat{LB}_{\tau_{OOO}}  \xrightarrow{p}  LB_{\tau_{OOO}} 
\]
A symmetric argument can be developed for the upper bound. Hence, $\widehat{UB}_{\tau_{OOO}}  \xrightarrow{p}  UB_{\tau_{OOO}}$

The estimators for the bounds in Theorem \ref{mono_bound} are a difference of trimmed mean and a conditional mean. Therefore, similar arguments follow to conclude consistency of those bound estimators to their true counterpart.

\end{proof}

\subsubsection{Asymptotic Normality of the Trimmed Mean Estimator}
\begin{lemma}[Asymptotic normality]
     Let $\Delta Y=Y_1-Y_0$ have a bounded support, and suppose $\mathbbm{P}[S_1=1|S_0=1,D=1]>0$ and $p_0>0$, then
    \begin{align*}
      \sqrt{n} (\widehat{LB}_{\tau_{OOO}}-LB_{\tau_{OOO}}) \xrightarrow{d} N(0,\Omega_{LB})  \quad &and \quad \sqrt{n} (\widehat{UB}_{\tau_{OOO}}-UB_{\tau_{OOO}}) \xrightarrow{d} N(0,\Omega_{UB})\\
    \end{align*}
\end{lemma}

\begin{proof}
First, we look at the conditions of Theorem 2 of \citet{chen2003estimation} to establish the asymptotic normality of the trimmed mean estimator.

\begin{condition}[2.1]
\[\|M_n(\hat{\theta}, \hat{h})\| = \inf_{\theta \in \Theta_\delta} \|M_n(\theta, \hat{h})\| + o_p(1/\sqrt{n})\]
\end{condition}
This condition requires that the estimator \( \hat{\theta} \) approximately minimizes the sample moment function \( M(\theta, \hat{h}) \), up to an error of order \( o_p(n^{-1/2}) \). It ensures that \( \hat{\theta} \) behaves like a Z-estimator and justifies applying asymptotic linearization.
This condition is satisfied by construction as the estimator $\hat{\theta}$ solves the moment condition without approximation.\\

Recall that our trimmed mean estimator $\hat{\theta}$  is constructed to satisfy the sample moment condition $M_n(\hat{\theta}, \hat{h})=0$.  Then $\|M_n(\hat{\theta}, \hat{h})\|=0$ and $\|M_n(\theta, \hat{h})\|\geq 0$ for all $\theta$, so $\hat{\theta}$ is an exact global minimizer of $\|M_n(\theta, \hat{h})\|$. In particular, restricting the infimum to the local neighborhood $\{\|\theta-\theta_0\|\leq \delta\}$ cannot decrease it, and hence,\[\|M_n(\hat{\theta}, \hat{h})\|=\inf_{\|\theta-\theta_0\|\le\delta} \|M_n(\theta, \hat{h})\|\]
Allowing for any negligible numerical or approximation error of order  $r_n=o_p(n^{-1/2})$ when solving the estimating equations yields,\[\|M_n(\hat{\theta}, \hat{h})\|=\inf_{\|\theta-\theta_0\|\le\delta} \|M_n(\theta, \hat{h})\|+ o_p(n^{-1/2})\]
which is exactly Condition 2.1. 
\begin{condition}[2.2]
   (i)The ordinary derivative $\Gamma_1(\theta, h_0)$ in $\theta $ of $M(\theta, h_0)$ exists for $\theta \in \Theta_\delta$, and is
continuous at $\theta = \theta_0$; (ii) the matrix $\Gamma_1 \equiv \Gamma_1(\theta_0, h_0)$ is of full (column) rank.
\end{condition}

 This condition says that $\Gamma_1(\theta, h_0)$ which is the ordinary derivative of the population moment function $ M(\theta, h_0)=\mathbbm{E}[g(Z,\theta,h_0)]$ with respect to $\theta$ exist for all $\theta$ in a neighborhood of $\theta_0$ (i.e. $\Theta_\delta$) and it is continuous at $\theta=\theta_0$. Further, $\Gamma_1(\theta_0, h_0)$ is of full column rank. This ensures local identification (small changes in $\theta$ induce small changes in $M(\theta, h_0)$).
 \newline
 In our case the population moment $M(\theta, h_0)=\mathbbm{E}[(\Delta Y - \theta) \cdot \mathbbm{1}\{\Delta Y \leq q_{p_0},\, D = 1,\, S_0 = 1,S_1=1\}]$, where $\theta$ is a scalar and appears linearly in the integrand. Its derivative with respect to $\theta$,
 \begin{align*}
     \Gamma_1(\theta, h_0) &= \frac{\partial }{\partial \theta} M(\theta, h_0)\\
    & =\mathbbm{E}[(-1)(\mathbbm{1}\{\Delta Y \leq  q_{p_0},\, D = 1,\, S_0 = 1,S_1=1\}]\\
    &= -\mathbbm{P}(\Delta Y \leq q_{p_0},\, D = 1,\, S_0 = 1,S_1=1)
 \end{align*}
 This derivative exists for all $\theta$ and is continuous in $\theta$ as it is a constant which doesn't depend on $\theta$. Further, this is strictly negative as long as $\mathbbm{P}(\Delta Y \leq q_{p_0},\, D = 1,\, S_0 = 1,S_1=1)> 0$. Thus, it is full rank. 
\begin{condition}[2.3]
    For all $\theta \in \Theta_\delta$ the pathwise derivative $\Gamma_2(\theta, h_0)[h - h_0]$ of $M(\theta, h_0)$ exists in all directions $[h - h_0] \in \mathcal{H}$; and for all $(\theta, h) \in \Theta_{\delta_n }\times \mathcal{H}_{\delta_n}$, with a positive sequence $\delta_n = o(1)$: 
(i) $\|M(\theta, h) - M(\theta, h_0) - \Gamma_2(\theta, h_0)[h - h_0]\| \leq c \|h - h_0\|^2_\mathcal{H}$ for a constant $c \geq 0$; (ii) $\|\Gamma_2(\theta, h_0)[h - h_0] - \Gamma_2(\theta_0, h_0)[h - h_0]\| \leq o(1)\delta_n$. 
\end{condition}
 This deals with the differentiability in the nuisance parameter $h$.
For all $\theta \in \Theta_\delta$ the pathwise derivative  $\Gamma_2(\theta, h_0)[h - h_0]$ of $ M(\theta, h_0)$ exists in all directions $ [h - h_0]$
\newline
 The moment function $M(\theta, h_0)$ is directionally differentiable with respect to the nuisance parameter $h$.

We have to show that for any direction $v \in \mathbbm{R}^2$ the directional derivative $\Gamma_2(\theta, h_0)[v] := \lim_{t \to 0^+} \frac{M(\theta, h_0 + tv) - M(\theta, h_0)}{t}$ exists. This limit exists if we can establish that, 
\[
M(\theta, h_0 + tv) = M(\theta, h_0) + t \cdot \Gamma_2(\theta, h_0)[v] + o(t)
\] 
By standard quantile asymptotics \citep{vandervaart1998} if the conditional density $f_{\Delta Y|111}$ is continuous and strictly positive at a neighborhood of $q_{p_0}$ then the quantile map $\phi:  p \xrightarrow{} q_{p}$ is Hadamard differentiable at $p_0$, meaning it's changes in $p$ can be linearised. The Lemma 21.3 \citep{vandervaart1998} gives the Hadamard directional derivative of the quantile map at $p_0$ in the direction $v$, $\phi'(p_0)[v]=\frac{-v}{f_{\Delta Y|111}(q_{p_0})}$. 

Consider the perturbation $h_t=h_0+tv=(p_0+tv,q_{p_0+tv})$. So for small \( t \) using Taylor approximation, we have:
\begin{align} \label{linerar_change}
    q_{p_0 + t v} &= q_{p_0} - t \cdot \frac{v}{f_{\Delta Y|111}(q_{ p_0})} + o(t)
\end{align}

So we are perturbing the threshold from $q_{p_0}$ to a nearby value,  the moment depends on this threshold inside an indicator function.  Let us expand the moment function,
\begin{align*}
   M(\theta, h) &=\mathbbm{E}[(\Delta Y - \theta) \cdot \mathbbm{1}\{\Delta Y \leq q_p,\, D = 1,\, S_0 = 1,S_1=1\}] \\
   &= \mathbbm{E}[(\Delta Y - \theta) \cdot \mathbbm{1}\{\Delta Y \leq q_p\} \cdot \mathbbm{1}\{ D = 1,\, S_0 = 1,S_1=1\}] 
\end{align*}
This expectation can be written more precisely as an integral over the joint density,
\[
 M(\theta, h)= \int \ (\Delta y-\theta) \mathbbm{1}\{\Delta y \leq q_p\} f_{\Delta y,D=1,S_0=1,S_1=1}(\Delta y).d \Delta y
\]
Because $D=1,S_0=1,S_1=1$ are fixed this reduces to,
\[
M(\theta, h)= \int_{-\infty}^{q_p} (\Delta y-\theta) f_{\Delta Y|111}(\Delta Y). \mathbbm{P}[D=1,S_0=1,S_1=1] \cdot d \Delta y
\] 

by Leibniz’s Rule for differentiation under the integral sign with variable upper bound. we have,
\begin{align} \label{derivate}
   \frac{d}{dq_p} M(\theta, h) &=\mathbbm{P}[D=1,S_0=1,S_1=1]\frac{d}{dq_p}\int_{-\infty}^{q_p} (\Delta y-\theta)f_{\Delta Y|111}(\Delta y)\,d \Delta y \nonumber \\
   &= \mathbbm{P}[D=1,S_0=1,S_1=1](q_p-\theta)f_{\Delta Y|111}(q_p) \nonumber 
\end{align}

So a slight shift in the quantile threshold $q_{p_0+vt}-q_{p_0}\approx - t \cdot \frac{v}{f_{\Delta Y|111}(y_{ p_0})}$ (refer to Equation \ref{linerar_change}) cause a change in the moment of;
\begin{align*}
    &\mathbbm{P}[D=1,S_0=1,S_1=1](q_{p_0}-\theta)f_{\Delta Y|111}(q_{p_0})\cdot \left(-t\right) \cdot \frac{v}{f_{\Delta Y|111}(q_{p_0})}\\
    &=-\mathbbm{P}[D=1,S_0=1,S_1=1]tv(q_{p_0}-\theta)
\end{align*}
Therefore, 
\[
M(\theta, h_t)-M(\theta, h_0) = -tv(q_{p_0}-\theta) \cdot \mathbbm{P}[D=1,S_0=1,S_1=1] +o(t).
\]
Then, we have
\begin{align} \label{directional}
    M(\theta, h_t)-M(\theta, h_0) &=t\cdot \Gamma_2(\theta, h_0)[v]+o(t)
\end{align}
Thus, the directional derivative exists and is,
\begin{align}\label{directderivative}
    \Gamma_2(\theta, h_0)[v]&=-v(q_{p_0}-\theta) \cdot \mathbbm{P}[D=1,S_0=1,S_1=1]\\
&=-v\mathbbm{E}[(q_{p_0}-\theta)\mathbbm{1}\{D=1,S_0=1,S_1=1\}]\nonumber
\end{align}
Now, let's move on to the condition of linear approximation. We have to show that for all directions $v\in \mathcal{H}$ and for all $(\theta,h)$ in a shrinking neighbourhood of $(\theta_0,h_0)$, a linear approximation holds,
\[
 \| M(\theta, h) - M(\theta, h_0) - \Gamma_2(\theta, h_0)[h - h_0] \| \leq c \| h - h_0 \|^2_{\mathcal{H}} \quad \text{for a constant} \quad c\geq 0
\]

This means that the error of the first-order linear approximation should vanish at most quadratically in $\| h - h_0 \|$. Recall that $h_t=h_0+vt$ then $v=\frac{h_t-h_o}{t}$,
\[
\Gamma_2(\theta, h_0)[v]=\Gamma_2(\theta, h_0)\left[\frac{h_t-h_o}{t}\right]
\]
\[
t \Gamma_2(\theta, h_0)[v]=\Gamma_2(\theta, h_0)[h_t-h_o]
\]
Substituting this to the Equation \ref{directional} we have,
\begin{align}\label{directional2}
    M(\theta, h_t)-M(\theta, h_0) &=\Gamma_2(\theta, h_0)[h_t-h_0]+o(t)
\end{align}

further, we know that,

\[
h_t - h_0 = t v \quad \Rightarrow \quad \|h_t - h_0\|_{\mathcal{H}} = \|t v\| = t \cdot \|v\|_{\mathcal{H}}
\]

So:
\[
t = \frac{\|h_t - h_0\|_{\mathcal{H}}}{\|v\|_{\mathcal{H}}} \quad \Rightarrow \quad t \propto \|h_t - h_0\|_{\mathcal{H}}
\]
where the proportionality constant is \( \|v\|_{\mathcal{H}} \), which is fixed. Then as we know that $r(t)=o(t)$ and $t \propto \|h_t - h_0\|_{\mathcal{H}}$,then $r(t)=o(\|h_t - h_0\|_{\mathcal{H}})$. Then, Equation \ref{directional} becomes,
\[
 M(\theta, h_t)-M(\theta, h_0) =\Gamma_2(\theta, h_0)[h_t-h_0]+o(\|h_t - h_0\|_{\mathcal{H}})
\]
Therefore, this establishes the differentiability in $h$ and gives the linear remainder bound
\[
R(h)=o(\|h - h_0\|_{\mathcal{H}}).
 \] 
Next, we establish the Lipschitz continuity of the derivative map directly as Equation \ref{directderivative}
 gives us  $\Gamma_2(\theta, h)$  is just multiplication by the scalar $(q_{p}-\theta) \cdot \mathbbm{P}[D=1,S_0=1,S_1=1]=c(\theta, q_p)$ , where $h=(p,q_p)$. concretely, for any two nuisance values $h'=(p_1,q_{p_1})$ and $h''=(p_2,q_{p_2})$ we have,
\[\Gamma_2(\theta, h')[v]-\Gamma_2(\theta, h'')[v]=(c(\theta, q_{p_1})-c(\theta, q_{p_2}) )\cdot (-v)\]
\[
\|\Gamma_2(\theta, h')-\Gamma_2(\theta, h'')\|= \sup_{\|v\|=1} \quad |\Gamma_2(\theta, h')[v]-\Gamma_2(\theta, h'')[v]|=|c(\theta, q_{p_1})-c(\theta, q_{p_2})|
\]
By regularity assumptions $\theta$ ranges over a compact set, so $|q_{p}-\theta|$ is bounded. Therefore, the function  $q_{p}\mapsto c(\theta, q_p)$ has a bounded derivative in $q_{p}$.  Thus, $|\frac{\partial  }{\partial q_p}c(\theta, q_p)|\leq L$ for all $\theta$ and $q_{p}$ in a small neighbourhood. Therefore, by the mean value theorem, we have $|c(\theta, q_{p_1})-c(\theta, q_{p_2})|\leq L|q_{p_1}-q_{p_2}|$. Finally, since our sup-norm on $h$ satisfies $\|h'-h''\|\geq |q_{p_1}-q_{p_2}|$. Thus, we have,

\begin{align}\label{lipschitz}
    \|\Gamma_2(\theta, h')-\Gamma_2(\theta, h'')\|\leq L |q_{p_1}-q_{p_2}|\leq L \|h'-h''\|
\end{align}
Taylor’s theorem with integral remainder  for any $h$ near $h_0$, we can write,
\[
R(h)= \int_0^1\bigl[\Gamma_2(\theta,\,h_0 + t(h-h_0))-\Gamma_2(\theta,\,h_0)\bigr][h-h_0]\,dt.
\]
\[\|R(h)\|
\le \int_0^1 \bigl\|\Gamma_2(\theta,\,h_0 + t\Delta)-\Gamma_2(\theta,\,h_0)\bigr\|\;\|\Delta\|\;dt, \quad \Delta=h-h_0
 \] 
by Equation \ref{lipschitz}  we have, $\bigl\|\Gamma_2(\theta,\,h_0 + t\Delta)-\Gamma_2(\theta,\,h_0)\bigr\|\leq Lt\|\Delta\|$. Hence,
\[
\|R(h)\|\le\int_{0}^{1} L\,t\,\|\Delta\|\;\|\Delta\|\,dt=L\,\|\Delta\|^{2}\int_{0}^{1} t\,dt=\frac{L}{2}\,\|\Delta\|^{2}.
\]
Thus, we have, 
\[\| M(\theta, h)-M(\theta, h_0) -\Gamma_2(\theta, h_0)[h-h_0]\|=\|R(h)\|\leq \frac{L}{2}\,\|h-h_0\|^{2} \]
satisfying condition 2.3(i).
Now let's consider condition  2.3(ii),

\[
\|\Gamma_2(\theta, h_0)[h - h_0] - \Gamma_2(\theta_0, h_0)[h - h_0]\| \leq o(1)\delta_n.
\]
This is a continuity condition on \( \theta \mapsto \Gamma_2(\theta, h_0)[\cdot] \) that ensures that the linear expansion remains stable when both \( \theta \) and \( h \) are close to their limits.\\
Because \( \delta_n \to 0 \) controls how small the neighbourhoods of \( \theta_0 \) and \( h_0 \) are, the condition ensures that:
\[
\frac{ \| \Gamma_2(\theta, h_0)[h - h_0] - \Gamma_2(\theta_0, h_0)[h - h_0] \| }{\delta_n} \to 0
\]

That is, the derivative map changes slower than the rate at which the inputs are shrinking.
We need to verify that:
\[
\left\| \Gamma_2(\theta, h_0)[h - h_0] - \Gamma_2(\theta_0, h_0)[h - h_0] \right\| \leq o(1) \cdot \delta_n
\]
for all $\theta \in \Theta$, $h \in \mathcal{H}$ satisfying $\|\theta - \theta_0\| \leq \delta_n$, $\|h - h_0\| \leq \delta_n$.

The directional derivative, 
\[
\Gamma_2(\theta, h_0)[v]=-v(q_{p_0}-\theta) \cdot \mathbbm{P}[D=1,S_0=1,S_1=1]+o(\|v\|)
\] 
\[
\Gamma_2(\theta_0, h_0)[v]=-v(q_{p_0}-\theta_0) \cdot \mathbbm{P}[D=1,S_0=1,S_1=1]+o(\|v\|)
\] 
The difference, 
\[
\Gamma_2(\theta, h_0)[v]-\Gamma_2(\theta_0, h_0)[v]=-v \cdot(\theta_0-\theta) \cdot \mathbbm{P}[D=1,S_0=1,S_1=1]+o(\|v\|)
\]
Each of the terms is bounded as follows: 
\begin{align*}
|v|=|h-h_0| &\leq \delta_n, \\
|\theta - \theta_0| &\leq \delta_n
\end{align*}

By assumption, $\|v\| = \|h - h_0\| \leq \delta_n$, and $o(\|v\|)$ means any function vanishing faster than $\|v\|$. Hence:
\[
o(\|v\|) \leq C \cdot \delta_n \Rightarrow o(\|v\|) = o(\delta_n)
\]
Therefore,
\[
\|\Gamma_2(\theta, h_0)[v]-\Gamma_2(\theta_0, h_0)[v]\|=O(\delta_n^2)+o(\delta_n)
\]
Therefore, we have $\|\Gamma_2(\theta, h_0)[v]-\Gamma_2(\theta_0, h_0)[v]\|\leq o(1)\cdot \delta_n$ and condition 2.3(ii) is satisfied.
\begin{condition}[2.4]
   $\hat{h} \in \mathcal{H}$ with probability tending to one; and $\|\hat{h} - h_0\|_\mathcal{H} = o_p(n^{-1/4}).$ 
\end{condition}
We already showed that $\hat{h}\xrightarrow{p}h_0\in \mathcal{H}$. Assume that the parameter space $\mathcal{H}\subset \mathbbm{R}^2$ is closed then $\mathbbm{P}(\hat{h}\in \mathcal{H})\xrightarrow{} 1$

recall that $\hat{p} = \max\left\{ 0, \frac{\hat{\alpha}+\hat{\beta}-1}{\hat{\beta}} \right\}$, with $\hat{\alpha}=\mathbbm{\hat{P}}[S_{1}=1|D=0, S_{0}=1]$ and $\hat{\beta}=\mathbbm{\hat{P}}[S_{1}=1|D=1, S_{0}=1]$\\
each of the $\hat{\alpha}$ and $\hat{\beta}$ is a sample proportion, and from CLT, we know,
\[
\hat{\alpha} - \alpha = O_{p}\bigl(n^{-\tfrac12}\bigr), 
\quad
\hat{\beta} - \beta = O_{p}\bigl(n^{-\tfrac12}\bigr).
\]
The ratio and max operator preserve the convergence rate. Therefore, 
\[
\hat{p} - p_0 = O_{p}\bigl(n^{-\tfrac12}\bigr)
\]
This implies that $n^a(\hat p - p_0) \xrightarrow{p} 0$ for any $a<1/2$. Thus, $(\hat p - p_0)=o_p\bigl(n^{-1/4}\bigr)$. Now decompose $\hat{q}_{\hat p}$,
\[
\hat{q}_{\hat{p}} - q_{p_{0}}
= \left(\hat{q}_{\hat{p}} - q_{\hat{p}}\right)
  + \left(q_{\hat{p}} - q_{p_{0}}\right).
\]
The first term $\bigl(\hat{q}_{\hat p} - q_{\hat p}\bigr)$ is $O_{p}\bigl(n^{-\tfrac12})$ as it is the empirical‐quantile error, and the second term $\bigl(q_{\hat p} - q_{p_0}\bigr)$ is a smooth shift of the true quantile under $\hat{p} - p_0=O_{p}\bigl(n^{-\tfrac12}\bigr)$, hence it is $O_{p}\bigl(n^{-\tfrac12})$. Thus,
\[
\hat{q}_{\hat p} - q_{p_0}=O_{p}\bigl(n^{-\tfrac12}\bigr)
\]
This implies that $n^a(\hat{q}_{\hat p} - q_{p_0} )\xrightarrow{p} 0$ for any $a<1/2$. Thus  $\hat{q}_{\hat p} - q_{p_0}=o_p\bigl(n^{-1/4}\bigr)$. Then by combining the components we have,
\[
\|\hat{h} - h_{0}\|_\mathcal{H}
= \max\bigl\{\lvert\hat p - p_0\rvert,\;\lvert \hat{q}_{\hat p} - q_{p_0}\rvert\bigr\}
= o_{p}\bigl(n^{-\tfrac14}\bigr)
\]
Therefore, condition 2.4 is satisfied. \\

\begin{condition}[2.5]
    For all sequences of positive numbers $\{\delta_n\}$ with $\delta_n = o(1)$, 
\begin{equation*}
\sup_{\|\theta - \theta_0\| \leq \delta_n, \|h - h_0\|_\mathcal{H} \leq \delta_n} 
\frac{\sqrt{n}\|M_n(\theta, h) - M(\theta, h) - M_n(\theta_0, h_0)\|}{1 + \sqrt{n} \{\|M_n(\theta, h)\| + \|M(\theta, h)\|\}} = o_p(1).
\end{equation*}

This condition is implied by,
\begin{align}\label{2.5}
    \sup_{\|\theta - \theta_0\| \leq \delta_n, \|h - h_0\|_\mathcal{H} \leq \delta_n} \|M_n(\theta, h) - M(\theta, h) - M_n(\theta_0, h_0)\|= o_{p}\bigl(n^{-\tfrac12}\bigr)
\end{align}
\end{condition} 
This means over any small shrinking neighborhood of the true $(\theta_0,h_0)$ the increment of the empirical process $\sqrt{n}(M_n(\theta, h) - M(\theta, h))$ when you move from $(\theta_0,h_0)$ to $(\theta,h)$ vanishes in probability. $\|\theta - \theta_0\| \leq \delta_n, \|h - h_0\|_\mathcal{H} \leq \delta_n$ means that we allow $\theta$ and $h$ only within a ball of radius $\delta_n$ around the true parameters.

By the definition of $(\theta_0,h_0)$ we have the population moment $M(\theta_0,h_0)=0$. Therefore, 
\[
M_n(\theta, h) - M(\theta, h) - M_n(\theta_0, h_0)= [M_n(\theta, h) - M(\theta, h)] -[ M_n(\theta_0, h_0)-M(\theta_0,h_0)]
\]
can be defined as the increment of the empirical process as you move from $(\theta_0,h_0)$ to $(\theta,h)$

Let us rewrite this in terms of empirical processes.
\begin{align} \label{2.5 decompose}
    &\sup_{\|\theta - \theta_0\| \leq \delta_n, \|h - h_0\|_\mathcal{H} \leq \delta_n} \|M_n(\theta, h) - M(\theta, h) - M_n(\theta_0, h_0)\| \nonumber \\
    &= \sup_{\|\theta - \theta_0\| \leq \delta_n, \|h - h_0\|_\mathcal{H} \leq \delta_n} \|M_n(\theta, h) - M(\theta, h) -( M_n(\theta_0, h_0)- M(\theta_0, h_0))\| \nonumber\\
    &=\sup_{\|\theta - \theta_0\| \leq \delta_n, \|h - h_0\|_\mathcal{H}\leq \delta_n}  \|(\mathbbm{P}_n g(Z, \theta, h)-\mathbbm{P}g(Z, \theta, h))- (\mathbbm{P}_n g(Z, \theta_0, h_0)-\mathbbm{P}g(Z, \theta_0, h_0))\| \nonumber\\
    &=\sup_{\|\theta - \theta_0\| \leq \delta_n, \|h - h_0\|_\mathcal{H}\leq \delta_n}\|\mathbbm{G}_n[g(Z, \theta, h)-g(Z, \theta_0, h_0)]\|
\end{align}
Where $\mathbbm{P}_nf=\frac{1}{n}\sum f(Z_i)$ is the empirical measure, $\mathbbm{P}f$ is the population counterpart and  $\mathbbm{G}_n=\sqrt{n}(\mathbbm{P}_n-\mathbbm{P})$ is the empirical process. 
Therefore, the condition \ref{2.5} is equivalent to 
\[
\sup_{\|\theta - \theta_0\| \leq \delta_n, \|h - h_0\|_\mathcal{H}\leq \delta_n}\|\mathbbm{G}[g(Z, \theta, h)-g(Z, \theta_0, h_0)]\|=o_p(1)
\]
By the triangular inequality, we have,

\[
\begin{aligned}
&\sup_{\|\theta-\theta_0\|\le\delta_n,\;\|h-h_0\|\le\delta_n}
\bigl\|\mathbb G_n\bigl[g(Z,\theta,h) - g(Z,\theta_0,h_0)\bigr]\bigr\|\\
&\quad\le
\sup_{\|\theta-\theta_0\|\le\delta_n,\;\|h-h_0\|\le\delta_n}
\bigl\|\mathbb G_n\bigl[g(Z,\theta,h)\bigr]\bigr\|
\;+\;
\underbrace{\bigl\|\mathbb G_n\bigl[g(Z,\theta_0,h_0)\bigr]\bigr\|}_{\text{constant in }(\theta,h)}.
\end{aligned}
\]

The class  $\mathcal{G}=\{g(Z, \theta, h):\|\theta-\theta_0\|<\delta_n,\|h-h_0\|<\delta_n\}$ is a VC-subgraph with a bounded envelop since $|g(Z, \theta, h)|\leq G$  for all $(\theta,h)$ in our compact set, then by \cite{vandervaart1996} theorem 2.11.1 ,
\[
\sup_{\|\theta - \theta_0\| \leq \delta_n, \|h - h_0\|_\mathcal{H}\leq \delta_n}\|\mathbbm{G}[g(Z, \theta, h)]\|=O_p(1)
\]
and its stochastic-equicontinuity (modulus-of-continuity) refinement over the shrinking neighborhood  $\delta_n \xrightarrow{} 0$  we have
\begin{align} \label{partA}
    \sup_{\|\theta - \theta_0\| \leq \delta_n, \|h - h_0\|_\mathcal{H}\leq \delta_n}\|\mathbbm{G}[g(Z, \theta, h)]\|=o_p(1)
\end{align}

Notice that $\mathbbm{G}[g(Z, \theta_0, h_0)]$  is just one empirical process evaluation in the same class $\mathcal{G}$ (corresponding to$ (\theta,h)=(\theta_0,h_0)$. Therefore, the same bounds apply,
\begin{align} \label{PartB}
    \mathbbm{G}[g(Z, \theta_0, h_0)]=\|\mathbbm{P}_n g(Z, \theta_0, h_0)-\mathbbm{P}g(Z, \theta_0, h_0)\|=O_{p}\bigl(1\bigr)
\end{align}
Therefore, by combining results in Equation \ref{partA} and \ref{PartB} we have,
\[
\sup_{\|\theta - \theta_0\| \leq \delta_n, \|h - h_0\|_\mathcal{H}\leq \delta_n}\|\mathbbm{G}[g(Z, \theta, h)-g(Z, \theta_0, h_0)]\|=o_p(1)
\]
hence, by dividing $\sqrt{n}$ , we have,
\[
\sup_{\|\theta - \theta_0\| \leq \delta_n, \|h - h_0\|_\mathcal{H} \leq \delta_n} \|M_n(\theta, h) - M(\theta, h) - M_n(\theta_0, h_0)\|= o_{p}\bigl(n^{-\tfrac12}\bigr)
\] 
\begin{condition}[2.6]
    For some finite matrix $V_1, \quad \sqrt{n} \left[ M_n(\theta_0, h_0) + \Gamma_2(\theta_0, h_0)[\hat{h} - h_0] \right] \xrightarrow{d} \mathcal{N}[0, V_1]$.

$ M_n(\theta_0, h_0)$ is a sample mean,
\[
 M_n(\theta_0, h_0)=\mathbbm{E}_n g(Z,\theta_0,h_0)=\frac{1}{n}\sum_{i=1}^n g (Z_i,\theta_0,h_0)
\]
\end{condition}
By the classical CLT,
\[
\sqrt{n}[ M_n(\theta_0, h_0)- M(\theta_0, h_0)]\xrightarrow{d} N(0,Var(g(Z,\theta_0,h_o)))
\]
As $ M(\theta_0, h_0)=0$  we have,
\[
\sqrt{n} M_n(\theta_0, h_0)\xrightarrow{} N(0,\Sigma) \quad where \quad \Sigma=Var(g(Z,\theta_0,h_o))
\]
Now let us look at the distribution of $\hat{h}=(\hat{p},\hat{q}_{\hat{p}})$.  Let's consider the distribution of $\hat{p}$ . As we assume that the true trimming proportion lies strictly within the interior of  (0,1), we can treat $\hat{p}$ as a smooth function of two proportions $\hat{\alpha}$  and $\hat{\beta} $ , each $O_p(n^{-1/2})$, hence we have,
\begin{align}\label{dist_p}
    \sqrt{n}(\hat{p}-p_0) \xrightarrow{d}N(0,\Sigma_p)=\mathcal{Z}_1
\end{align}

 Let us consider the derivation of the asymptotic distribution of $\hat{q}_{\hat{p}}$,
\begin{align} \label{eq_breakdown}
    \sqrt{n}\left(\hat{q}_{\hat{p}}-q_{p_0}\right)=\sqrt{n}\left(\hat{q}_{\hat{p}}-\hat{q}_{p_0}\right)+\sqrt{n}\left(\hat{q}_{p_0}-q_{p_0}\right)
\end{align}
Using Donsker’s Theorem in classical empirical process theory, each empirical CDF satisfies
\begin{equation}\label{Gaussian}
\sqrt{n}\left(F_{n}-F\right) \rightsquigarrow \mathbb{G}
\end{equation}
Where $\mathbb{G}$  are mean-zero Gaussian processes and $n$ is the sample size. Thus, as $q_p$ is Hadamard differentiable and $\sqrt{n}\left(F_{n}-F\right) \rightsquigarrow \mathbb{G}$, we can apply the functional delta method to derive the asymptotic distribution of the sample quantile. Thus, by the classical result on the asymptotic distribution of sample quantiles (\citet{vandervaart1998}), we have,
\begin{align}\label{std_Q_asy}
\sqrt{n}\left(\hat{q}_{p_0}-q_{p_0}\right)\xrightarrow{d} \mathcal{N}\left(0, \frac{p_0(1-p_0)}{f(q_{p_0})^2 }\right)=\mathbb{Z}_2
\end{align}

 $f(q_{p_0})^2$ is the density of the outcome evaluated at the $p_0$-quantile. 
 Now by first-order Taylor expansion, we have,
\begin{align*}
    \sqrt{n}\left(\hat{q}_{\hat{p}}-\hat{q}_{p_0}\right) &\approx\sqrt{n}(\hat{p}-p_0)\cdot \frac{d}{dp_0} \hat{q}_{p_0}
\end{align*}
Asymptotically,
\begin{align*}
    \frac{d}{dp_0} \hat{q}_{p_0} \xrightarrow{} \frac{1}{f(q_{p_0})}
\end{align*}
Therefore, we have,
\begin{align*}
   \sqrt{n}\left(\hat{q}_{\hat{p}}-\hat{q}_{p_0}\right) &\approx\sqrt{n}(\hat{p}-p_0)\cdot \frac{1}{f(q_{p_0})}
\end{align*}
As $\sqrt{n}(\hat{p}-p_0) \xrightarrow{d}\mathbb{Z}_1$ by Slutsky's Theorem, we have,
\begin{align}\label{dist_2}
     \sqrt{n}\left(\hat{q}_{\hat{p}}-\hat{q}_{p_0}\right) &\xrightarrow{d} \frac{\mathbb{Z}_1}{f(q_{p_0})}
\end{align}
Now combining Equation \ref{eq_breakdown} with the distribution results obtained through Equation \ref{std_Q_asy} and \ref{dist_2} by  Slutsky's Theorem, we have,
\begin{align}\label{dist_combine}
     \sqrt{n}\left(\hat{q}_{\hat{p}}-q_{p_0}\right)&\xrightarrow{d} \frac{\mathbb{Z}_1}{f(q_{p_0})}+\mathbb{Z}_2=N(0,\Sigma_q)
\end{align}

Therefore, from Equation \ref{dist_p} and \ref{dist_combine} we have,
\[
\sqrt{n}(\hat{h} - h_0) = \sqrt{n} \begin{pmatrix}
\hat{p} - p_0 \\
\hat{ q}_{\hat{p}} - q_{p_0}
\end{pmatrix}
\xrightarrow{d} \mathcal{N}(0, \Sigma_h).
\]
The mapping $h \mapsto \Gamma_2(\theta_0, h_0)[h - h_0]$ is a fixed matrix (the directional derivative in $h$). Hence
\[
\sqrt{n} \, \Gamma_2(\theta_0, h_0)[\hat{h} - h_0] \xrightarrow{d} \mathcal{N}\left(0, \, \Gamma_2(\theta_0, h_0) \, \Sigma_h \, \Gamma_2(\theta_0, h_0)' \right).
\]
Because \( M_n(\theta_0, h_0) \) is a mean of functions of \( (Y_i, D_i, S_i) \) and \( \hat{h} \) is also asymptotically linear in the same data, a multivariate Lindeberg–Feller argument (or the Cramér–Wold device) gives
\[
\sqrt{n} \begin{pmatrix}
M_n(\theta_0, h_0) \\
\Gamma_2(\theta_0, h_0)[\hat{h} - h_0]
\end{pmatrix}
\xrightarrow{d} \mathcal{N} \left( 0, 
\begin{pmatrix}
\Sigma & \operatorname{Cov} \\
\operatorname{Cov}' & \Gamma_2 \Sigma_h \Gamma_2'
\end{pmatrix}
\right).
\]

Summing the two components satisfies condition 2.6
\[
\sqrt{n} \left[ M_n(\theta_0, h_0) + \Gamma_2(\theta_0, h_0)[\hat{h} - h_0] \right] \xrightarrow{d} \mathcal{N}[0, V_1]
\]
Where, $V_1 = \Sigma + \Gamma_2 \Sigma_h \Gamma_2' +  \operatorname{Cov'}\operatorname{Cov}$.

We have justified all six conditions to establish the asymptotic normality of $\hat{\theta}$.  Therefore, we can conclude that,
\[
\sqrt{n}(\hat{\theta}-\theta)\xrightarrow{d}N(0,\Omega_\theta)
\]

Note that the estimator for the bound  $LB_{\tau_{OOO}}$ as defined in Theorem \ref{nomono_bound} is a difference of two independent trimmed means of the same form. As we proved the asymptotic normality of a general trimmed mean estimator, it follows that the difference is also asymptotically normal. Thus we have,
\[
\sqrt{n} (\widehat{LB}_{\tau_{OOO}}-LB_{\tau_{OOO}}) \xrightarrow{d} N(0,\Omega_{LB})
\]
A symmetric argument can be developed for the upper bound. Hence, $\sqrt{n} (\widehat{UB}_{\tau_{OOO}}-UB_{\tau_{OOO}}) \xrightarrow{d} N(0,\Omega_{UB}) $.  Similar arguments can be made for bounds derived under Theorem \ref{mono_bound}.
\end{proof}
\subsubsection{Bootstrap Consistency}\label{appendix:Bootstrap}
\begin{theorem}[\cite{chen2003estimation} Theorem B: Bootstrap Consistency]
Suppose that \(\{Z_i\}_{i=1}^n\) is i.i.d. and \(\theta_0 \in \operatorname{int}(\Theta)\) satisfies \(\mathbbm{E}[m(Z_i, \theta_0, h_0)] = 0\); that \(\hat{\theta} - \theta_0 = o_{a.s.}(1)\); that conditions (2.1), (2.4), (2.5'), and (2.6) hold with ‘in probability’ replaced by ‘almost surely’; that condition (2.2) holds with \(h_0\) replaced by \(h \in \mathcal{H}_\delta\), while condition (2.3) holds with \(h_0\) replaced by \(h \in \mathcal{H}_{\delta_n}\), and that \(\Gamma_1(\theta, h)\) is continuous (with respect to \(\|\cdot\|_{\mathcal{H}}\)) in \(h\) at \(\theta = \theta_0\), \(h = h_0\). Suppose:
\begin{itemize}
    \item[(2.4B)] With \(P^\ast\)-probability tending to one, \(\hat{h}^\ast \in \mathcal{H}\), and \(\|\hat{h}^\ast - \hat{h}\|_{\mathcal{H}} = o_{P^\ast}(n^{-1/4})\).
    \item[(2.5$^\prime$B)] \(\sup_{(\theta, h) \in \Theta_{\delta_n} \times \mathcal{H}_{\delta_n}} \left\| M_n^\ast(\theta, h) - M_n(\theta, h) - \left\{ M_n^\ast(\theta_0, h_0) - M_n(\theta_0, h_0) \right\} \right\| = o_{P^\ast}(n^{-1/2})\) for all positive values \(\delta_n = o(1)\).
    \item[(2.6B)] \(\sqrt{n} \left[ M_n^\ast(\hat{\theta}, \hat{h}) - M_n(\hat{\theta}, \hat{h}) + \Gamma_2(\hat{\theta}, \hat{h})[\hat{h}^\ast - \hat{h}] \right] = \mathcal{N}(0, V_1) + o_{P^\ast}(1)\).
\end{itemize}
Then, \(\sqrt{n}(\hat{\theta}^\ast - \hat{\theta})\) converges in distribution to a \(\mathcal{N}(0, \Omega)\) distribution in \(P^\ast\)- probability

\end{theorem}
\begin{lemma}[Bootstrap Consistency]
     Let $\Delta Y=Y_1-Y_0$ have a bounded support, and suppose $\mathbbm{P}[S_1=1|S_0=1,D=1]>0$ and $p_0>0$. Let superscript $\ast$ denote the estimators computed under the bootstrap distribution, conditional on the original data set. Then
    \begin{align*}
      \sqrt{n} (\widehat{LB^\ast}_{\tau_{OOO}}-\widehat{LB}_{\tau_{OOO}}) \xrightarrow{d} N(0,\Omega_{LB})  \quad &and \quad \sqrt{n} (\widehat{UB}_{\tau_{OOO}}-\widehat{UB^\ast}_{\tau_{OOO}}) \xrightarrow{d} N(0,\Omega_{UB})
    \end{align*}
\end{lemma}

\begin{proof}
    We first prove the Bootstrap Consistency of the trimmed mean estimator. Suppose superscript $\ast$ denotes a probability or moment computed under the bootstrap distribution conditional on the original data set $\{Z_i\}^n_{i=1}.$  Let $\hat{\theta}^\ast$ and $\hat{h}^\ast$ be the method of moment and nuisance estimators computed based on a bootstrap sample $\{Z^\ast_i\}$ drawn i.i.d. with replacement from the original data. $M^\ast_n(\theta,h)$ is the bootstrap moment function.
   \begin{condition}[2.4B]
       With \(P^\ast\)-probability tending to one, \(h^\ast \in \mathcal{H}\), and \(\|\hat{h}^\ast - \hat{h}\|_{\mathcal{H}} = o_{P^\ast}(n^{-1/4})\)
   \end{condition}
 
 This condition implies that our bootstrap-based estimator $\hat{h}^\ast$ falls into the same function space $\mathcal{H}$ and converges to the original $\hat{h}$ fast enough. Recall $h=(p,q_p)$ and we measure $\big\| h - h'\big\|_{\mathcal H}
= \max\bigl\{\lvert p - p'\rvert,\;\lvert q_{p} - q'_{p}\rvert\bigr\}$

We draw a bootstrap sample of size $n$ from the original data. Let $\hat{\alpha}^\ast=\mathbbm{P_n^\ast}[S_{1}=1|D=0, S_{0}=1]$ and $\hat{\beta}^\ast=\mathbbm{P_n^\ast}[S_{1}=1|D=1, S_{0}=1]$ so that  $\hat{p}^\ast = \max\left\{ 0, \frac{\hat{\alpha}^\ast+\hat{\beta}^\ast-1}{\hat{\beta}^\ast} \right\}$. 
Given the original sample, $\{Z^\ast_i\}^n_{i=1}$ are i.i.d. Bernoulli draws with success probability $\hat{\alpha}$ (or $\hat{\beta}$) in the resample, hence by CLT we have,
\[
\hat{\alpha}^\ast - \hat{\alpha} = O_{P^\ast}\bigl(n^{-\tfrac12}\bigr), 
\quad
\hat{\beta}^\ast - \hat{\beta} = O_{P^\ast}\bigl(n^{-\tfrac12}\bigr).
\]
Assume $p_0\in(\epsilon,1-\epsilon)$ and $\hat{p}\xrightarrow{p}p_0$ implying the kink at zero is never binding. Therefore,  $\hat{p}^\ast= \frac{\hat{\alpha}^\ast+\hat{\beta}^\ast-1}{\hat{\beta}^\ast}$ w.p. $\xrightarrow{}1$ and by delta method we have,
\[
\hat{p}^\ast - \hat{p} = O_{p^\ast}\bigl(n^{-\tfrac12}\bigr)
\]
This implies that $n^a(\hat p^\ast - \hat{p}) \xrightarrow{P^\ast} 0$ for any $a<1/2$.  Thus,  $\hat{p}^\ast - \hat{p}=o_{P^\ast}\bigl(n^{-1/4}\bigr)$

We compute the $\hat{p}\ast$-th quantile of the bootstrap sample, $\hat{q}^\ast_{{\hat p}^\ast}$
 Now decompose  $\hat{q}^\ast_{{\hat p}^\ast}$,
\[
 \hat{q}^\ast_{{\hat p}^\ast} - \hat{q}_{\hat p}
= \bigl( \hat{q}^\ast_{{\hat p}^\ast} - \hat{q}_{\hat p^\ast}\bigr)
  + \bigl( \hat{q}_{\hat p^\ast}- \hat{q}_{\hat p}\bigr).
\]
The first term $( \hat{q}^\ast_{{\hat p}^\ast} - \hat{q}_{\hat p^\ast})$ is $O_{P^\ast}\bigl(n^{-\tfrac12})$ as it is the bootstrap empirical‐quantile error \citep{vandervaart1996}, and the second term $( \hat{q}_{\hat p^\ast}- \hat{q}_{\hat p})$ is a smooth shift of the estimated quantile of the original sample under $\hat{p}^\ast - \hat{p} = O_{P^\ast}\bigl(n^{-\tfrac12}\bigr)$, hence it is $O_{P^\ast}\bigl(n^{-\tfrac12})$. Thus,
\[
\hat{q}^\ast_{{\hat p}^\ast} - \hat{q}_{\hat p}=O_{p^\ast}\bigl(n^{-\tfrac12}\bigr)
\]
This implies that $n^a(\hat{q}^\ast_{{\hat p}^\ast} - \hat{q}_{\hat p})\xrightarrow{P^\ast} 0$ for any $a<1/2$. Thus,  $(\hat{q}^\ast_{{\hat p}^\ast} - \hat{q}_{\hat p})$. Then by combining the components we have,
\[
\|\hat{h}^\ast - \hat{h}\|_\mathcal{H}
= \max\bigl\{\lvert\hat{p}^\ast - \hat p\rvert,\;\lvert\hat{q}^\ast_{\hat{p}^\ast} - \hat{q}_{\hat p}\rvert\bigr\}
= o_{P^\ast}\bigl(n^{-\tfrac14}\bigr)
\]
Therefore, condition 2.4B is satisfied. \\
\begin{condition}[2.5$^\prime$B]
   \(\sup_{(\theta, h) \in \Theta_{\delta_n} \times \mathcal{H}_{\delta_n}} \left\| M_n^\ast(\theta, h) - M_n(\theta, h) - \left\{ M_n^\ast(\theta_0, h_0) - M_n(\theta_0, h_0) \right\} \right\| = o_{P^\ast}(n^{-1/2})\) for all positive values \(\delta_n = o(1)\). 
\end{condition}
This condition means that in a small neighbourhood of the truth $(\theta_0,h_0)$, the difference between the bootstrap moment and the original sample moment after you adjust for the value at the true parameter shrinks faster than $n^{-1/2}$ with probability tending to one in the bootstrap DGP.

Using the same arguments given in the proof of condition 2.5 and following Equation \ref{2.5 decompose} we can establish that 
\begin{align*}
\sup_{(\theta, h) \in \Theta_{\delta_n} \times \mathcal{H}_{\delta_n}} &\left\| M_n^\ast(\theta, h) - M_n(\theta, h) - \left\{M_n^\ast(\theta_0, h_0) - M_n(\theta_0, h_0) \right\} \right\|\\
    &=\sup_{\|\theta - \theta_0\| \leq \delta_n, \|h - h_0\|_\mathcal{H}\leq \delta_n}\|\mathbbm{G_n^\ast}[g(Z, \theta, h)-g(Z, \theta_0, h_0)]\|
\end{align*}

where $\mathbbm{G_n^\ast}=\sqrt{n}(\mathbbm{P^\ast}_n-\mathbbm{P_n})$ is the bootstrap empirical process conditional on data. Therefore, conditional on data $g(Z_i^\ast,\theta,h)$ is a VC subgraph with a bounded envelope. The Bootstrap empirical process satisfies the same equicontinuity and ULLN  conditions \citep{vandervaart1996}. Hence, condition 2.5B follows immediately by the same proof as for the original sample, with each condition holding with $P^\ast$ probability.\\
\begin{condition}[2.6B]
    \(\sqrt{n} \left[ M_n^\ast(\hat{\theta}, \hat{h}) - M_n(\hat{\theta}, \hat{h}) + \Gamma_2(\hat{\theta}, \hat{h})[\hat{h}^\ast - \hat{h}] \right] = \mathcal{N}(0, V_1) + o_{P^\ast}(1)\).
\end{condition}
This implies that bootstrap version of our first order term $\sqrt{n} \left[ M_n(\theta_0, h_0) + \Gamma_2(\theta_0, h_0)[\hat{h} - h_0]\right]$ converges to the same $N(0,V_1)$ up to error $o_{P^\ast}(1)$

\[
M_n^\ast(\hat{\theta},\hat{ h}) - M_n(\hat{\theta},\hat{ h}) =(\mathbbm{P}^\ast_n-\mathbbm{P_n})g(Z,\hat{\theta},\hat{h})
\]
conditional on data, $\{g(Z_i^\ast,\hat{\theta},\hat{h})\}$ are i.i.d. draws from the empirical distribution of $\{g(Z_i,\hat{\theta},\hat{h})\}$. Hence, by bootstrap CLT(i.e. conditional on the observed data, the sample‐with‐replacement average of any fixed function converges to a normal law at the usual $n^{-1/2}$ rate.). We have,
\[
\sqrt{n}[M_n^\ast(\hat{\theta},\hat{ h}) - M_n(\hat{\theta},\hat{ h})] \;{\overset{d}{\underset{P^\ast}{\longrightarrow}}}\;N(0,\hat{\Sigma})
\]
where $\hat{\Sigma}=Var(g(Z,\hat{\theta},\hat{h})|Z_1,...,Z_n)$ and with consistency of $(\hat{\theta},\hat{h})$ together with continuous mapping theorem and LLN we have $\hat{\Sigma}\xrightarrow[p]{}Var(g(Z,\theta_0,h_0)=\Sigma$.  Then,
\begin{align} \label{2.6 MD}
    \sqrt{n}[M_n^\ast(\hat{\theta},\hat{ h}) - M_n(\hat{\theta},\hat{ h})] \;{\overset{d}{\underset{P^\ast}{\longrightarrow}}}\;N(0,\Sigma)
\end{align}

As $\Gamma_2(\theta,h)$ is continuous at $(\theta_0,h_0)$ and $(\hat{\theta},\hat{h})\xrightarrow[p]{}(\theta_0,h_0)$, we have $\|\Gamma_2(\hat{\theta},\hat{h})-\Gamma_2(\theta_0,h_0)\|=o_p(1)$. further, we already showed that $\hat{p}^\ast - \hat{p} = O_{p^\ast}\bigl(n^{-\tfrac12}\bigr)$ in condition 2.4B.  This gives,
\begin{align}\label{2.6 p1}
    \Gamma_2(\hat{\theta}, \hat{h})[\hat{h}^\ast - \hat{h}] 
= \Gamma_2(\theta_0, h_0)[\hat{h}^\ast - \hat{h}] + o_{P^\ast}(n^{-1/2}),
\end{align}

From our original Z--estimation of the nuisance $\hat{h} $ we have, in the original sample,
and each component admits an influence-function expansion in the original  sample:
\[
\hat{p} - p_0 = \frac{1}{n} \sum_{i=1}^n \psi_p(Z_i) + o_p(n^{-1/2}), \qquad 
\hat{q}_{\hat{p}} - q_{p_0} = \frac{1}{n} \sum_{i=1}^n \psi_q(Z_i) + o_p(n^{-1/2}).
\]

 Since $\hat{h}$ itself is fixed once the data are observed, and when you draw a bootstrap sample $\{Z_i^\ast\}$ and the bootstrap draws are i.i.d.\ from the empirical distribution. The same data delta-method arguments show
\[
\hat{p}^\ast - \hat{p} = \frac{1}{n} \sum_{i=1}^n \psi_p(Z_i^\ast) + o_{P^\ast}(n^{-1/2}), \qquad 
\hat{q}_{ \hat{p}}^\ast - \hat{q}_{\hat{p}} = \frac{1}{n} \sum_{i=1}^n \psi_q(Z_i^\ast) + o_{P^\ast}(n^{-1/2}).
\]

Hence the vector 
\[
\hat{h}^\ast - \hat{h}
\]
admits the linear representation
\begin{align} \label{2.6 p2}
    \hat{h}^\ast - \hat{h} = \frac{1}{n} \sum_{i=1}^n \psi_h(Z_i^\ast) + o_{P^\ast}(n^{-1/2}),
\end{align}
where $\psi_h = (\psi_p, \psi_q)'$ is the $2 \times 1$ influence-function vector.

Combining the results from Equations \ref{2.6 p1} and \ref{2.6 p2} we have,

\[
\sqrt{n} \, \Gamma_2(\hat{\theta}, \hat{h})[\hat{h}^\ast - \hat{h}]
= \frac{1}{\sqrt{n}} \sum_{i=1}^{n} \Gamma_2(\theta_0, h_0) \, \psi_h(Z_i^\ast)+o_{P^\ast}(n^{-1/2})
\]

Conditional on the data, $\left\{ \psi_h(Z_i^\ast) \right\}_{i=1}^n$ are i.i.d.\ with mean zero and covariance $\Sigma_h$. 
By the Lindeberg--Feller CLT in the bootstrap DGP (Cramér--Wold),
\[
\frac{1}{\sqrt{n}} \sum_{i=1}^{n} \Gamma_2(\theta_0, h_0) \, \psi_h(Z_i^\ast) 
\;{\overset{d}{\underset{P^\ast}{\longrightarrow}}}\;\mathcal{N}\left( 0, \, \Gamma_2(\theta_0, h_0) \Sigma_h \Gamma_2(\theta_0, h_0)' \right)
\]

Combining the above gives
\begin{align}\label{2.6 NS}
    \sqrt{n} \, \Gamma_2(\hat{\theta}, \hat{h}) \left[ \hat{h}^\ast - \hat{h} \right] 
\xrightarrow{d} \mathcal{N} \left( 0, \, \Gamma_2 \Sigma_h \Gamma_2' \right) + o_{P^\ast}(1),
\end{align}

As both pieces, the moment difference (Equation \ref{2.6 MD}) and nuisance shift (Equation \ref{2.6 NS}) are linear functionals of the same bootstrap sample, the multivariate  bootstrap CLT gives,

\[
\begin{pmatrix}
M_n^\ast(\hat{\theta}, \hat{h}) - M_n(\hat{\theta}, \hat{h}) \\
\Gamma_2(\hat{\theta}, \hat{h})[\hat{h}^\ast - \hat{h}]
\end{pmatrix}
 \;{\overset{d}{\underset{P^\ast}{\longrightarrow}}}\; \mathcal{N} \left( 0, 
\begin{pmatrix}
\Sigma & \operatorname{Cov} \\
\operatorname{Cov}' & \Gamma_2 \Sigma_h \Gamma_2'
\end{pmatrix}
\right).
\]
All of the above linearizations incur only $o_{P^\ast}(1)$ remainders (because 
$\hat{\theta} - \theta_0 = o_p(1)$, $\hat{h} - h_0 = o_p(1)$, and 
$\Gamma_2$ is continuous).  Hence
\[
\sqrt{n} \left[ M_n^\ast(\hat{\theta}, \hat{h}) - M_n(\hat{\theta}, \hat{h}) 
+ \Gamma_2(\hat{\theta}, \hat{h})[\hat{h}^\ast - \hat{h}] \right]
= \mathcal{N}(0, V_1) + o_{P^\ast}(1).
\]
This satisfies the condition 2.6B.

Therefore, \(\sqrt{n}(\hat{\theta}^\ast - \hat{\theta})\) converges in distribution to a \(\mathcal{N}(0, \Omega)\) in \(P^\ast\)- probability

Note that the estimator for the bound  $LB_{\tau_{OOO}}$ as defined in Theorem \ref{nomono_bound} is a difference of two independent trimmed means of the same form. As we proved the bootstrap consistency of a general trimmed mean estimator, the difference is also bootstrap consistent.
\[
\sqrt{n} (\widehat{LB^\ast}_{\tau_{OOO}}-\widehat{LB}_{\tau_{OOO}}) \xrightarrow{d} N(0,\Omega_{LB})
\]
A symmetric argument can be developed for the upper bound. Hence, $\sqrt{n} (\widehat{UB^\ast}_{\tau_{OOO}}-\widehat{UB}_{\tau_{OOO}}) \xrightarrow{d} N(0,\Omega_{UB})$. Similar arguments can be made for the estimators of bounds defined in Theorem \ref{mono_bound}.

Note that the estimator for the bound  $LB_{\tau_{OOO}}$ as defined in Theorem \ref{nomono_bound} is a difference of two independent trimmed means of the same form (say $\hat{\theta}$ and $\hat{\alpha}$). As we proved the Bootstrap consistency of a general trimmed mean estimator, we have in $P^\ast$ probability,
\[
\sqrt{n} (\widehat{\theta^\ast}-\widehat{\theta}) \xrightarrow{d} N(0,\Omega_{\theta}) \quad and \quad \sqrt{n} (\widehat{\alpha^\ast}-\widehat{\alpha}) \xrightarrow{d} N(0,\Omega_{\alpha}).
\]
\end{proof}

\subsection{Estimation of Bounds in Theorems \ref{ONO_bound}, \ref{NNO_bound} and \ref{NOO_bound} }
\label{Estimation_other latent groups}
This section outlines the estimation of the bounds defined in Theorem \ref{ONO_bound}, \ref{NNO_bound} and \ref{NOO_bound}, which are based on the sample analogues of the population counterparts.

\subsubsection{Estimation of Bounds under Theorem \ref{ONO_bound}}
First we estimate the required mixing proportion $p_{OOO1}$ as given in Equation \ref{Estimation_p} and then estimate the bounds for $LB_{ONO1}$, $UB_{ONO1}$ as follows,
\begin{align*}
     \widehat{LB}_{ONO1} &= \frac{\sum_{i=1}^n (Y_{i1}-Y_{i0}) \cdot S_{i0}\cdot S_{i1}\cdot D_{i}  \cdot I\left\{(Y_{i1}-Y_{i0}) \leqslant \hat{y}_{1-\hat{p}_{OOO1}}\right\}}{\sum_{i=1}^n S_{i0}\cdot S_{i1}\cdot D_{i}  \cdot I\left\{(Y_{i1}-Y_{i0}) \leqslant \hat{y}_{1-\hat{p}_{OOO1}}\right\}}\\
    \widehat{UB}_{ONO1} &= \frac{\sum_{i=1}^n (Y_{i1}-Y_{i0}) \cdot S_{i0}\cdot S_{i1}\cdot D_{i}  \cdot I\left\{(Y_{i1}-Y_{i0}) > \hat{y}_{\hat{p}_{OOO1}}\right\}}{\sum_{i=1}^n S_{i0}\cdot S_{i1}\cdot D_{i}  \cdot I\left\{(Y_{i1}-Y_{i0}) > \hat{y}_{\hat{p}_{OOO1}}\right\}}
\end{align*}
Where $\hat{y}_{\hat{p}_{OOO1}}$ and $\hat{y}_{1-\hat{p}_{OOO1}}$ are $\hat{p}_{OOO1}$-th and $(1-\hat{p}_{OOO1})$-th quantile of the conditional distribution $Y_{1}-Y_{0}$ for the treated individuals observed in both time periods.

Next, to estimate the bounds for $LB^{0}_{ONO0}$, $UB^{0}_{ONO0}$ we estimate the required mixing proportion $p_{ONO0}$ as follows,
\begin{align*} 
    \hat{p}_{ONO0} &= \frac{\hat{\pi}_{ONO0}}{\hat{\pi}_{ONN0}+\hat{\pi}_{ONO0}} \quad \text{Lemma \ref{lemma: pro 4(joint)} in Appendix \ref{prop_condi}}\\
    &=\frac{\hat{\mathbbm{P}}[S_0=1, S_1=0, D=0]-\hat{\mathbbm{P}}[S_1=0|S_0=1,D=1]\cdot \hat{\mathbbm{P}}[S_{0}=1,D=0]}{\hat{\mathbbm{P}}[S_0=1, S_1=0, D=0]} \\
    &=1-\frac{\hat{\mathbbm{P}}[S_1=0|S_0=1,D=1]}{\hat{\mathbbm{P}}[S_1=0|S_0=1,D=0]}\\
    \mathbb{\hat{P}}[S_{1}=0|S_{0}=1,D=1]&= \frac{\sum_{i=1}^{n}S_{i0}\cdot (1-S_{i1})\cdot D_{i}}{\sum_{i=1}^{n}S_{i0}\cdot D_{i}} \\
    \mathbb{\hat{P}}[S_{1}=0|S_{0}=1,D=0] &= \frac{\sum_{i=1}^{n}S_{i0}\cdot (1-S_{i1})\cdot (1-D_{i})}{\sum_{i=1}^{n}S_{i0}\cdot (1-D_{i})}
\end{align*}

Then estimate the bounds for $LB^{0}_{ONO0}$, $UB^{0}_{ONO0}$ as follows,
\begin{align} \label{estimation_bound_ONO0t_0}
     \widehat{LB}^{0}_{ONO0} &= \frac{\sum_{i=1}^n Y_{i0} \cdot S_{i0}\cdot (1-S_{i1})\cdot (1-D_{i})  \cdot I\left\{Y_{i0} \leqslant \hat{y}_{\hat{p}_{ONO0}}\right\}}{\sum_{i=1}^n S_{i0}\cdot (1-S_{i1})\cdot (1- D_{i})  \cdot I\left\{Y_{i0} \leqslant \hat{y}_{\hat{p}_{ONO0}}\right\}} \\
    \widehat{UB}^{0}_{ONO0} &= \frac{\sum_{i=1}^n Y_{i0} \cdot S_{i0}\cdot (1-S_{i1})\cdot (1-D_{i})  \cdot I\left\{Y_{i0} > \hat{y}_{1-\hat{p}_{ONO0}}\right\}}{\sum_{i=1}^n S_{i0}\cdot (1-S_{i1})\cdot (1-D_{i})  \cdot I\left\{Y_{i0} > \hat{y}_{1-\hat{p}_{ONO0}}\right\}}\nonumber
\end{align}
Where $\hat{y}_{\hat{p}_{ONO0}}$ and $\hat{y}_{1-\hat{p}_{ONO0}}$ are $\hat{p}_{ONO0}$-th and $(1-\hat{p}_{ONO0})$-th quantile of the control group pre-treatment period outcome distribution conditional on being observed in pre-treatment period and unobserved in post-treatment period.

Next, $\mathbbm{E}[Y_{1}|D=0,S_{0}=1,S_{1}=1]$ can be estimated using its sample analogues as,
\begin{align*}
     \hat{\mathbbm{E}}[Y_{1}|D=0,S_{0}=1,S_{1}=1] &= \frac{\sum_{i=1}^n Y_{i1} \cdot S_{i0}\cdot S_{i1}\cdot (1-D_{i})}{\sum_{i=1}^n S_{i0}\cdot S_{i1}\cdot (1-D_{i}) }
\end{align*}
Then, $Y_{1}^{LB}$ is the theoretical lower bound of potential outcomes in the post-treatment period,

Finally, using these sample analogues instead of the population parameters, the bounds defined for $\tau_{ONO}$ under Theorem \ref{ONO_bound} can be estimated as,
\begin{align*}
    \widehat{LB}_{\tau_{ONO}} &= \widehat{LB}_{ONO1}-\hat{\mathbbm{E}}[Y_{1}|D=0,S_{0}=1, S_{1}=1]+\widehat{LB}^{0}_{ONO0},\\
    \widehat{UB}_{\tau_{ONO}} &= \widehat{UB}_{ONO1}-Y_{1}^{LB}+\widehat{UB}^{0}_{ONO0}
\end{align*}

\subsubsection{Estimation of Bounds under Theorem \ref{NNO_bound}}

First, to estimate the bounds for $LB_{NNO1}$, $UB_{NNO1}$ we estimate the required mixing proportion $p_{NNO1}$ as follows,
\begin{align} \label{estimation_pNNO1}
    \hat{p}_{NNO1} &= \frac{\hat{\pi}_{NNO1}}{\hat{\pi}_{NNO1}+\hat{\pi}_{NOO1}} \quad \text{Lemma \ref{lemma: pro 4(joint)} in Appendix \ref{prop_condi}} \\
    &=\frac{\mathbb{\hat{P}}[S_0=0, S_1=1, D=1]-\mathbb{\hat{P}}[S_1=1|S_0=0,D=0]\cdot \mathbb{\hat{P}}[S_{0}=0,D=1]}{\hat{P}[S_0=0, S_1=1, D=1]} \nonumber\\
    &=1-\frac{\mathbb{\hat{P}}[S_1=1|S_0=0,D=0]}{\mathbb{\hat{P}}[S_1=1|S_0=0,D=1]}\nonumber\\
    \mathbb{\hat{P}}[S_1=1|S_0=0,D=0]&= \frac{\sum_{i=1}^{n}(1-S_{i0})\cdot S_{i1}\cdot (1-D_{i})}{\sum_{i=1}^{n}(1-S_{i0})\cdot (1-D_{i})} \nonumber\\
    \mathbb{\hat{P}}[S_1=1|S_0=0,D=1]&= \frac{\sum_{i=1}^{n}(1-S_{i0})\cdot S_{i1}\cdot D_{i}}{\sum_{i=1}^{n}(1-S_{i0})\cdot D_{i}}\nonumber
\end{align}

Then estimate the bounds for $LB_{NNO1}$, $UB_{NNO1}$ as follows,
\begin{align*}
     \widehat{LB}_{NNO1} &= \frac{\sum_{i=1}^n Y_{i1} \cdot (1-S_{i0})\cdot S_{i1}\cdot D_{i}  \cdot I\left\{Y_{i1} \leqslant \hat{y}_{\hat{p}_{NNO1}}\right\}}{\sum_{i=1}^n (1-S_{i0})\cdot S_{i1}\cdot D_{i}  \cdot I\left\{Y_{i1} \leqslant \hat{y}_{\hat{p}_{NNO1}}\right\}}\\
    \widehat{UB}_{NNO1} &= \frac{\sum_{i=1}^n Y_{i1} \cdot (1-S_{i0})\cdot S_{i1}\cdot D_{i} \cdot I\left\{Y_{i1} > \hat{y}_{1-\hat{p}_{NNO1}}\right\}}{\sum_{i=1}^n (1-S_{i0})\cdot S_{i1}\cdot D_{i}  \cdot I\left\{Y_{i1} > \hat{y}_{1-\hat{p}_{NNO1}}\right\}}
\end{align*}
Where $\hat{y}_{\hat{p}_{NNO1}}$ and $\hat{y}_{1-\hat{p}_{NNO1}}$ are $\hat{p}_{NNO1}$-th and $(1-\hat{p}_{NNO1})$-th quantile of the treatment group post-treatment period outcome distribution conditional on being unobserved in pre-treatment period and observed in post-treatment period.

Next, estimate  mixing proportion $p_{OOO1}$ as given in Equation \ref{Estimation_p}. Then, estimate the lower bound $LB^{0}_{ONO1}$ as follows,
\begin{align*}
     \widehat{LB}^{0}_{ONO1} &= \frac{\sum_{i=1}^n Y_{i0} \cdot S_{i0}\cdot S_{i1}\cdot D_{i}  \cdot I\left\{Y_{i0} \leqslant \hat{y}_{1-\hat{p}_{OOO1}}\right\}}{\sum_{i=1}^n S_{i0}\cdot S_{i1}\cdot D_{i}  \cdot I\left\{Y_{i0} \leqslant \hat{y}_{1-\hat{p}_{OOO1}}\right\}}
\end{align*}
Where $\hat{y}_{1-\hat{p}_{OOO1}}$ is $(1-\hat{p}_{OOO1})$-th quantile of the treated group pre-treatment period outcome distribution conditional on being observed in both periods.
Estimate $\widehat{LB}^{0}_{ONO0}$ as detailed in Equation \ref{estimation_bound_ONO0t_0}.
Next, $\mathbbm{E}[Y_{1}|D=0,S_{0}=0,S_{1}=1]$ can be estimated using its sample analogues as,
\begin{align}\label{E_1|001}
     \hat{\mathbbm{E}}[Y_{1}|D=0,S_{0}=0,S_{1}=1] &= \frac{\sum_{i=1}^n Y_{i1} \cdot (1-S_{i0})\cdot S_{i1}\cdot (1-D_{i})}{\sum_{i=1}^n (1-S_{i0})\cdot S_{i1}\cdot (1-D_{i}) }
\end{align}

The $Y_{0}^{LB}$ and $Y_{1}^{LB}$ are the theoretical lower bounds of the potential outcomes in the pre-treatment period and the post-treatment period.
Finally, using these sample analogues instead of the population parameters, the bounds defined for $\tau_{NNO}$ under Theorem \ref{NNO_bound} can be estimated as,
\begin{align*}
    \widehat{LB}_{\tau_{NNO}} &= \widehat{LB}_{NNO1}-\widehat{LB}^{0}_{ONO1}-\hat{\mathbbm{E}}[Y_{1}|D=0,S_{0}=0, S_{1}=1]+Y_{0}^{LB},\\
    \widehat{UB}_{\tau_{NNO}} &= \widehat{UB}_{NNO1}-Y_{0}^{LB}- Y_{1}^{LB}+\widehat{LB}^{0}_{ONO0}
\end{align*}

\subsubsection{Estimation of Bounds under Theorem \ref{NOO_bound}}

First, we estimate the required mixing proportion $p_{NNO1}$ as given in Equation \ref{estimation_pNNO1} and then estimate the bounds for $LB_{NOO1}$, $UB_{NOO1}$ as follows,
\begin{align*}
     \widehat{LB}_{NOO1} &= \frac{\sum_{i=1}^n Y_{i1} \cdot (1-S_{i0})\cdot S_{i1}\cdot D_{i}  \cdot I\left\{Y_{i1} \leqslant \hat{y}_{1-\hat{p}_{NNO1}}\right\}}{\sum_{i=1}^n (1-S_{i0})\cdot S_{i1}\cdot D_{i}  \cdot I\left\{Y_{i1} \leqslant \hat{y}_{1-\hat{p}_{NNO1}}\right\}}\\
    \widehat{UB}_{NOO1} &= \frac{\sum_{i=1}^n Y_{i1} \cdot (1-S_{i0})\cdot S_{i1}\cdot D_{i} \cdot I\left\{Y_{i1} > \hat{y}_{\hat{p}_{NNO1}}\right\}}{\sum_{i=1}^n (1-S_{i0})\cdot S_{i1}\cdot D_{i}  \cdot I\left\{Y_{i1} > \hat{y}_{\hat{p}_{NNO1}}\right\}}
\end{align*}
Where $\hat{y}_{\hat{p}_{NNO1}}$ and $\hat{y}_{1-\hat{p}_{NNO1}}$ are $\hat{p}_{NNO1}$-th and $(1-\hat{p}_{NNO1})$-th quantile of the treatment group's post-treatment period outcome distribution conditional on being unobserved in pre-treatment period and observed in post-treatment period.

Next, estimate  mixing proportion $p_{OOO1}$ as given in Equation \ref{Estimation_p}. Then, estimate the lower bound $LB^{0}_{OOO1}$ as follows,
\begin{align*}
    \widehat{LB}^{0}_{OOO1} &= \frac{\sum_{i=1}^n Y_{i0} \cdot S_{i0}\cdot S_{i1}\cdot D_{i}  \cdot I\left\{Y_{i0} \leqslant \hat{y}_{\hat{p}_{OOO1}}\right\}}{\sum_{i=1}^n S_{i0}\cdot S_{i1}\cdot D_{i}  \cdot I\left\{Y_{i0} \leqslant \hat{y}_{\hat{p}_{OOO1}}\right\}}
\end{align*}
Where $\hat{y}_{\hat{p}_{OOO1}}$ is $\hat{p}_{OOO1}$-th quantile of the treatment group's pre-treatment period outcome distribution conditional on being observed in both periods.

Then $\mathbbm{E}[Y_{1}|D=0,S_{0}=0, S_{1}=1]$  can be estimated  using its sample analogues as in Equation  \ref{E_1|001}, while and $\mathbbm{E}[Y_{0}|D=0,S_{0}=1,S_{1}=1]$ can be estimated using its sample analogues as,
\begin{align*}
     \hat{\mathbbm{E}}[Y_{0}|D=0,S_{0}=1,S_{1}=1] &= \frac{\sum_{i=1}^n Y_{i0} \cdot S_{i0}\cdot S_{i1}\cdot (1-D_{i})}{\sum_{i=1}^n S_{i0}\cdot S_{i1}\cdot (1-D_{i}) }
\end{align*}
The $Y_{0}^{LB}$ is the theoretical lower bound of the potential outcomes in the pre-treatment period.
Finally, the bounds defined for $\tau_{NOO}$ under Theorem \ref{NOO_bound} can be estimated as,
\begin{align*}
    \widehat{LB}_{\tau_{NOO}} &= \widehat{LB}_{NOO1}-\widehat{LB}^{0}_{OOO1}-\hat{\mathbbm{E}}[Y_{1}|D=0,S_{0}=0, S_{1}=1]+Y_{0}^{LB}\\
    \widehat{UB}_{\tau_{NOO}} &= \widehat{UB}_{NOO1}-Y_{0}^{LB}-\hat{\mathbbm{E}}[Y_{1}|D=0,S_{0}=0, S_{1}=1]+\hat{\mathbbm{E}}[Y_{0}|D=0,S_{0}=1, S_{1}=1]]
\end{align*}
\subsection{Derivation of Numerical Bounds}
\label{numerical bounds}

The true ATT for this DGP can be defined as,
\begin{align*}
    \tau &= \mathbbm{E}\left[ Y_{i1}^\ast(1) - Y_{i1}^\ast(0) \mid D_{i}=1 \right] \nonumber \\
    &=\mathbbm{E}[\ \tau_i\mid  D_{i}=1]\\
    &=\sum_g\tau_g \cdot \mathbbm{P}[g_i=g \mid D_{i}=1]
\end{align*}

For this DGP, we have set $\tau_{OOO}=5$, $\tau_{ONO}=4$, $\tau_{ONN}=1$, $\tau_{NOO}=3$, $\tau_{NNO}=1$, and $\tau_{NNN}=-1$. Let $\pi_g^1\equiv \mathbbm{P}[g_i=g \mid D_i=1]$
\begin{align*}
    \tau= -1\cdot \pi_{NNN}^1+1\cdot \pi_{NNO}^1+3\cdot \pi_{NOO}^1+1\cdot \pi_{ONN}^1+4\cdot \pi_{ONO}^1+5\cdot \pi_{OOO}^1
\end{align*}
 \begin{align*}
     \pi_g^1=\mathbbm{P}[g_i=g \mid D_i=1]=\frac{\pi_{g1}}{\mathbbm{P}[D_i=1]}=2 \cdot \pi_{g1}
 \end{align*}
 As in this DGP $ D_i=\mathbbm{1}\bigl\{\,a_i + w_{i1}\ge0\bigr\}$ and $a_i\sim N(0,1)$ and $w_{i1}\sim N(0,1)$ so $a_i + w_{i1}\sim (0,2)$ which gives $\mathbbm{P}[D_i=1]=\mathbbm{P}[a_i + w_{i1}\geq 0]=0.5$.
 \begin{align*}
     \pi_{g1}=\int_{\mathbbm{R}^5} \mathbbm{1}\{g\}\cdot \mathbf1\{a + w_{1}\geq 0\}\Phi_5(a,b,v_0,v_1,w_1) da,db,dv_0,dv_1,dw_1
 \end{align*}
where, $\Phi_X(a,b,v_0,v_1,w_1)$ is the joint $N_5(0,\Sigma_X)$ density of $(a,b,v_0,v_1,w_1)$. Even though $(a,v_0,v_1)$ are drawn from  truncated normals, they can be approximated by standard normals as the truncation threshold ($M=5$) is chosen such that the probability mass we are cut off is negligible.
$$\begin{pmatrix}a_i\\[4pt]b_i\\[2pt]v_{i0}\\[2pt]v_{i1}\\[2pt]w_{i1}\end{pmatrix}
   \sim
   N\!\Bigl(\mathbf0,\;
     \Sigma_X\Bigr),
   \quad
   \Sigma_X = 
   \begin{pmatrix}
     1 & 0 & 0 & 0&0\\
     0 & 1 & 0 & 0& 0\\
     0 & 0 & 1 & 0& 0\\
     0 & 0 & 0 & 1& 0\\
     0 & 0 & 0 & 0& 1
   \end{pmatrix},$$
In this setting an individual will belong to $OOO1$ group if $S_{i0}(0)=1,S_{i1}(1)=1,S_{i1}(0)=1,D_i=1$. This implies,
\begin{align*}
    S_{i0}(0) &= \mathbbm{1}\{ b_i+ v_{i0}>0\}=1\implies b_i+ v_{i0}>0\\
    S_{i1}(0) &= \mathbbm{1}\{b_i+ v_{i1}>0\}=1\implies b_i+ v_{i1}>0 \\
    S_{i1}(1) &= \mathbbm{1}\{1.5+ b_i+ v_{i1}>0\}=1\implies b_i+ v_{i1}>-1.5\\
    D_i &=\mathbbm{1}\{a_i+w_{i1}>0\}=1 \implies a_i+w_{i1}>0
\end{align*}

 Let $z_1=b_i+ v_{i0}$ ,$z_2=b_i+ v_{i1}$, $z_3=b_i+ v_{i1}$ and $z_4=a_i+w_{i1}$ then, $\pi_{OOO1}=\mathbbm{P}[z_1>0,z_2>0,z_3>-1.5,z_4>0]$. Thus, we need the joint distribution of $Z=(z_1,z_2,z_3,z_4)$.
 Define 
 $$
    A = 
   \begin{pmatrix}
     0& 1 & 1 & 0&0\\
     0 & 1 & 0 & 1& 0\\
   0 & 1 & 0 & 1& 0\\
     1 & 0 & 0 & 0& 1
   \end{pmatrix}
   $$
   Then $Z=A\cdot (a,b,v_0,v_1,w_1)\sim N_4(0, A\Sigma_X A')$,
   \begin{align*}
		\begin{pmatrix}
			z_1 \\
			z_2\\
   z_3\\
   z_4
		\end{pmatrix}&\sim N\left(\begin{pmatrix}0\\0\\0\\0\end{pmatrix},\begin{pmatrix}2&1&1&0\\1&2&2&0\\1&2&2&0\\0&0&0&2\end{pmatrix}\right)
	\end{align*}
    
   Then we can obtain $\pi_{OOO1}=\mathbbm{P}[z_1>0,z_2>0,z_3>-1.5,z_4>0]$ from this distribution as,
$$
\pi_{OOO1}=\mathbbm{P}[z_1>0,z_2>0,z_3>-1.5,z_4>0]= \int_{0}^{\infty}\int_{0}^{\infty}\int_{-1.5}^{\infty}\int_{0}^{\infty} f_{z_1z_2z_3z_4}(z_1,z_2,z_3,z_4) dz_4 dz_3 dz_2 dz_1
$$
Similarly, an individual will belong to $ONO1$ group if $S_{i0}(0)=1,S_{i1}(1)=0,S_{i1}(0)=1,D_i=1$. This implies,
$$
\pi_{ONO1}=\mathbbm{P}[z_1>0,z_2<0,z_3>-1.5,z_4>0]= \int_{0}^{\infty}\int_{-\infty}^{0}\int_{-1.5}^{\infty}\int_{0}^{\infty} f_{z_1z_2z_3z_4}(z_1,z_2,z_3,z_4) dz_4 dz_3 dz_2 dz_1
$$
Similarly,
$$
\pi_{ONN1}=\mathbbm{P}[z_1>0,z_2<0,z_3<-1.5,z_4>0]= \int_{0}^{\infty}\int_{-\infty}^{0}\int_{-\infty}^{-1.5}\int_{0}^{\infty} f_{z_1z_2z_3z_4}(z_1,z_2,z_3,z_4) dz_4 dz_3 dz_2 dz_1
$$
$$
\pi_{NOO1}=\mathbbm{P}[z_1<0,z_2>0,z_3>-1.5,z_4>0]= \int_{-\infty}^{0}\int_{0}^{\infty}\int_{-1.5}^{\infty}\int_{0}^{\infty} f_{z_1z_2z_3z_4}(z_1,z_2,z_3,z_4) dz_4 dz_3 dz_2 dz_1
$$
$$
\pi_{NNO1}=\mathbbm{P}[z_1<0,z_2<0,z_3>-1.5,z_4>0]= \int_{-\infty}^{0}\int_{-\infty}^{0}\int_{-1.5}^{\infty}\int_{0}^{\infty} f_{z_1z_2z_3z_4}(z_1,z_2,z_3,z_4) dz_4 dz_3 dz_2 dz_1
$$
$$
\pi_{NNN1}=\mathbbm{P}[z_1<0,z_2<0,z_3<-1.5,z_4>0]= \int_{-\infty}^{0}\int_{-\infty}^{0}\int_{-\infty}^{-1.5}\int_{0}^{\infty} f_{z_1z_2z_3z_4}(z_1,z_2,z_3,z_4) dz_4 dz_3 dz_2 dz_1
$$

Then $\tau_{true}=\sum_g\tau_g \cdot \pi^1_{g1}=\sum_g\tau_g \cdot 2\cdot \pi_{g1}=2.850444$

Latent group-specific parallel trend holds as, in this DGP, the mean difference in untreated potential outcomes $Y^\ast_{i1}(0)-Y^\ast_{i0}(0)$ for any $g$ and treatment status $D$ is,
\begin{align*}
    \mathbbm{E}[Y^\ast_{i1}(0)-Y^\ast_{i0}(0)|g,D]&=t^{g_i}_1-t^{g_i}_0+\mathbbm{E}[u_{i1}-u_{i0}|g,D]
\end{align*}
As $\Delta u_i$ is uncorrelated with $D$, $\mathbbm{E}[u_{i1}-u_{i0}|g,D=0]=\mathbbm{E}[u_{i1}-u_{i0}|g,D=1]$. Therefore,  group-specific parallel trend holds (Assumption \ref{PT_group}). Furthermore, the group-specific time effects are chosen such that the outcome mean dominance assumptions (Assumption \ref{mean dominance}) hold.

\subsubsection{True ATT Bounds for OOO Group}\label{TrueOOO}

In order to derive Theoretical bounds for $\tau_{OOO}$ as defined in Theorem \ref{mono_bound}. we need to derive the distribution of $Y_{1}-Y_{0}$ given $D_{1}=1,S_{0}=1,S_{1}=1$ which we denote by $w$.
\begin{align} \label{w}
    w &=(Y_{i1}-Y_{i0})|D_i=1,S_{i0}=1,S_{i1}=1\\ \nonumber
    &=(S_{i1}\cdot Y^\ast_{i1}-S_{i0}\cdot Y^\ast_{i0})\mid D_i=1,S_{i0}=1,S_{i1}=1\\ \nonumber
    &=(Y^\ast_{i1}(1)-Y^\ast_{i0}(0))|D_i=1,S_{i0}(0)=1,S_{i1}(1)=1\\ \nonumber
    &=(t^{g_i}_1 + \tau^{g_i}+ c_i+ u_{i1}-(t^{g_i}_0+c_i+ u_{i0})\mid D_i=1,S_{i0}(0)=1,S_{i1}(1)=1\\\nonumber
    &=(t^{g_i}_1-t^{g_i}_0+\tau^{g_i}+\Delta u_i)\mid D_i=1,S_{i0}(0)=1,S_{i1}(1)=1\nonumber
\end{align}
Under the joint selection $D_{1}=1,S_{0}=1,S_{1}(1)=1$ only latent group $OOO: (S_{0}=1,S_{1}(0)=1,S_{1}(1)=1)$ and $ONO: (S_{0}=1,S_{1}(0)=0,S_{1}(1)=1)$ exists. Therefore, $w$ will be a two‐component selection Gaussian mixture, which is approximated by Gaussian distributions for simplicity.
$$
f_W(w)=\frac{\pi_{OOO1}}{\pi_{OOO1}+\pi_{ONO1}}\cdot \varphi\bigl(w;\mu_{OOO},\sigma^2_{OOO}\bigr)+\frac{\pi_{ONO1}}{\pi_{OOO1}+\pi_{ONO1}}\cdot \varphi\bigl(w;\mu_{ONO},\sigma^2_{ONO}\bigr)
$$
where for this DGP $\mu_{OOO}=t^{OOO}_1-t^{OOO}_0+\tau^{OOO}+\mathbbm{E}[\Delta u_i|OOO]=6-5+5+\mathbbm{E}[\Delta u_i|OOO]=6+\mathbbm{E}[\Delta u_i|OOO]$, $\sigma^2_{OOO}=var(\Delta u_i|OOO)$, $\mu_{ONO}=t^{ONO}_1-t^{ONO}_0+\tau^{ONO}+\mathbbm{E}[\Delta u_i|ONO]=5-4+4+\mathbbm{E}[\Delta u_i|ONO]=5+\mathbbm{E}[\Delta u_i|ONO]$ and $\sigma^2_{ONO}=var(\Delta u_i|ONO)$. As $u_{i1}\sim N(0,1)$ and $u_{i0}\sim N(0,1)$ are uncorrelated, then $\Delta u_{i}\sim N(0,2)$. We have shown that the joint distribution of $Z_1$,$Z_2$ and $Z_3$ is multivariate  normal. Further,
\begin{align*}
    cov[\Delta u_i,z_1]&=E[(u_{i1}-u_{i0})(b_i+ v_{i0})]-E[u_{i1}-u_{i0}]E[b_i+ v_{i0}]\\
    &=E[u_{i1}b_i+u_{i1}v_{i0}-u_{i0}b_i-u_{i0}v_{i0}]\\
    &=-E[u_{i0}v_{i0}]=-0.7\\
    cov[\Delta u_i,z_2]&=E[(u_{i1}-u_{i0})(b_i+ v_{i1})]-E[u_{i1}-u_{i0}]E[b_i+ v_{i1})]\\
    &=E[u_{i1}b_i+u_{i1}v_{i1}-u_{i0}b_i-u_{i0}v_{i1}]\\
    &=E[u_{i1}v_{i1}]=0.6\\
     cov[\Delta u_i,z_3]&=E[(u_{i1}-u_{i0})(b_i+ v_{i1})]-E[u_{i1}-u_{i0}]E[b_i+ v_{i1}]\\
    &=E[u_{i1}b_i+u_{i1}v_{i1}-u_{i0}b_i-u_{i0}v_{i1}]\\
    &=E[u_{i1}v_{i1}]=0.6
\end{align*}
Then it follows that the joint distribution $f_{\Delta u_iz_1 z_2 z_3}(\Delta u_i,z_1,z_2,z_3)$ is multivariate normal as,
\begin{align*}
		\begin{pmatrix}
  \Delta u_i\\
			z_1 \\
                z_2 \\
			z_3
		\end{pmatrix}&\sim N\left(\begin{pmatrix}0\\0\\0\\0\end{pmatrix},\begin{pmatrix}2&-0.7&0.6&0.6\\-0.7&2&1&1\\0.6&1&2&2\\0.6&1&2&2\end{pmatrix}\right)
	\end{align*}
Then, the required conditional expectations can be evaluated as follows,
\begin{align*}
    E[\Delta u_i|OOO] &= E[\Delta u_i|z_1>0,z_2>0,z_3>-1.5]\\
    &=E[\Delta u_i|z_1>0,z_2>0]\\
    &=\frac{E[\Delta u_i\cdot I(z_1>0) \cdot I(z_2>0)]}{\mathbbm{P}[z_1>0,z_2>0]}\\
    &=\frac{\int_{-\infty}^{\infty}\int_{0}^{\infty}\int_{0}^{\infty} \Delta u_i f_{\Delta u_iz_1z_2}(\Delta u_i,z_1,z_2)dz_2 dz_1 d\Delta u_i}{\int_{0}^{\infty}\int_{0}^{\infty} f_{z_1z_2}(z_1,z_2) dz_2 dz_1}=\mu_1
\end{align*}
\begin{align*}
    E[\Delta u^2_i|OOO] &= E[\Delta u^2_i|z_1>0,z_2>0,z_3>-1.5]\\
    &=E[\Delta u^2_i|z_1>0,z_2>0]\\
    &=\frac{E[\Delta u^2_i\cdot I(z_1>0) \cdot I(z_2>0)]}{\mathbbm{P}[z_1>0,z_2>0]}\\
    &=\frac{\int_{-\infty}^{\infty}\int_{0}^{\infty}\int_{0}^{\infty} \Delta u^2_i f_{\Delta u_iz_1z_2}(\Delta u_i,z_1,z_2)dz_2 dz_1 d\Delta u_i}{\int_{0}^{\infty}\int_{0}^{\infty} f_{z_1z_2}(z_1,z_2) dz_2 dz_1}=\mu_2
\end{align*}
\begin{align*}
    E[\Delta u_i|ONO] &= E[\Delta u_i|z_1>0,z_2<0,z_3>-1.5]\\
    &=E[\Delta u_i|z_1>0,0>z_3>-1.5]\\
    &=\frac{E[\Delta u_i\cdot I(z_1>0) \cdot I(0>z_3>-1.5)]}{\mathbbm{P}[z_1>0,0>z_3>-1.5]}\\
    &=\frac{\int_{-\infty}^{\infty}\int_{0}^{\infty}\int_{-1.5}^{0} \Delta u_i f_{\Delta u_iz_1z_3}(\Delta u_i,z_1,z_3)dz_3 dz_1 d\Delta u_i}{\int_{0}^{\infty}\int_{-1.5}^{0} f_{z_1z_3}(z_1,z_3) dz_3 dz_1}=\mu_3
\end{align*}
\begin{align*}
    E[\Delta u^2_i|ONO] &= E[\Delta u^2_i|z_1>0,z_2<0,z_3>-1.5]\\
    &=E[\Delta u^2_i|z_1>0,0>z_3>-1.5]\\
    &=\frac{E[\Delta u^2_i\cdot I(z_1>0) \cdot I(0>z_3>-1.5)]}{\mathbbm{P}[z_1>0,0>z_3>-1.5]}\\
    &=\frac{\int_{-\infty}^{\infty}\int_{0}^{\infty}\int_{-1.5}^{0} \Delta u^2_i f_{\Delta u_iz_1z_3}(\Delta u_i,z_1,z_3)dz_3 dz_1 d\Delta u_i}{\int_{0}^{\infty}\int_{-1.5}^{0} f_{z_1z_3}(z_1,z_3) dz_3 dz_1}=\mu_4
\end{align*}
For notation simplicity let $p_{OOO1}=\frac{\pi_{OOO1}}{\pi_{OOO1}+\pi_{ONO1}}$ then,
$$
w \approx  p_{OOO1}N(\mu_{OOO},\sigma^2_{OOO})+(1-p_{OOO1}) N(\mu_{ONO},\sigma^2_{ONO})
$$
$$
w \approx  p_{OOO1}N(6+\mu_1,\mu_2-\mu^2_1)+(1-p_{OOO1}) N(5+\mu_3,\mu_4-\mu^2_3)
$$
Then the bounds defined in the Theorem 2,
\begin{align*}
 LB^{\prime}_{OOO1} &= E[w|w\leq F_{w}^{-1}(p_{OOO1})]\\
 UB^{\prime}_{OOO1}  &= E[w|w > F_{w}^{-1}(1-p_{OOO1})]
\end{align*}
Let \(c = F_w^{-1}(p_{OOO1})\) be the unique solution \(x\) to \(p_{OOO1}\,\Phi\bigl(\tfrac{x-\mu_{OOO}}{\sigma_{OOO}}\bigr)+(1-p_{OOO1})\,\Phi\bigl(\tfrac{x-\mu_{ONO}}{\sigma_{ONO}}\bigr)= p_{OOO1}\), which can be found numerically via R’s \texttt{uniroot} or a Newton--Raphson routine. A similar procedure can be used to get the numerical value for $c'=F_{w}^{-1}(1-p_{OOO1})$. With these known values for $c$ and $c'$ we can write,
\begin{align*}
    E[w|w\leq c]&=\frac{\int_{-\infty}^{c} w\,f(w)\,dw}{\int_{-\infty}^{c} f(w)\,dw}\\
    &=\frac{p_{OOO1} \int_{-\infty}^{c} w\,\varphi(w;\mu_{OOO},\sigma_{OOO}^{2})\,dw+(1-p_{OOO1}) \int_{-\infty}^{c} w\,\varphi(w;\mu_{ONO},\sigma_{ONO}^{2})\,dw}{p_{OOO1}\Phi\!\Bigl(\tfrac{c-\mu_{OOO}}{\sigma_{OOO}}\Bigr)+(1-p_{OOO1})\Phi\!\Bigl(\tfrac{c-\mu_{ONO}}{\sigma_{ONO}}\Bigr)}
\end{align*}
Note for any $N(\mu,\sigma^2)$, $\int_{-\infty}^{c} w\,\varphi(w;\mu,\sigma^{2})\,dw= \mu \Phi\!\Bigl(\tfrac{c-\mu}{\sigma}\Bigr)-\sigma\phi\Bigl(\tfrac{c-\mu}{\sigma}\Bigr) $ where $\phi$ is standard normal pdf. Denote $\alpha_{OOO}=\tfrac{c-\mu_{OOO}}{\sigma_{OOO}}$ and $\alpha_{ONO}=\tfrac{c-\mu_{ONO}}{\sigma_{ONO}}$. Then,
\begin{align}\label{c}
    E[w|w\leq c]&=\frac{p_{OOO1}[\mu_{OOO}\Phi(\alpha_{OOO})-\sigma_{OOO} \phi(\alpha_{OOO})]+(1-p_{OOO1})[\mu_{ONO}\Phi(\alpha_{ONO})-\sigma_{ONO} \phi(\alpha_{ONO})]}{p_{OOO1}\Phi(\alpha_{OOO})+ (1-p_{OOO1})\Phi(\alpha_{ONO})}
    \end{align}
Similarly, we can write,
    \begin{align*}
    E[w|w> c']&=\frac{\int_{c'}^{\infty} w\,f(w)\,dw}{\int_{c'}^{\infty}  f(w)\,dw}\\
    &=\frac{p_{OOO1} \int_{c'}^{\infty}  w\,\varphi(w;\mu_{OOO},\sigma_{OOO}^{2})\,dw+(1-p_{OOO1}) \int_{c'}^{\infty}  w\,\varphi(w;\mu_{ONO},\sigma_{ONO}^{2})\,dw}{p_{OOO1}(1-\Phi\!\Bigl(\tfrac{c'-\mu_{OOO}}{\sigma_{OOO}}\Bigr))+(1-p_{OOO1})(1-\Phi\!\Bigl(\tfrac{c'-\mu_{ONO}}{\sigma_{ONO}}\Bigr))}
\end{align*}
Note for any $N(\mu,\sigma^2)$, $\int_{c'}^{\infty}  w\,\varphi(w;\mu,\sigma^{2})\,dw= \mu \bigl[1-\Phi\!\Bigl(\tfrac{c'-\mu}{\sigma}\Bigr)\bigl]+\sigma\phi\Bigl(\tfrac{c'-\mu}{\sigma}\Bigr) $ where $\phi$ is standard normal pdf. Denote $\alpha'_{OOO}=\tfrac{c'-\mu_{OOO}}{\sigma_{OOO}}$ and $\alpha'_{ONO}=\tfrac{c'-\mu_{ONO}}{\sigma_{ONO}}$. Then,
{\footnotesize
\begin{align}\label{c'}
    E[w|w> c']&=\frac{p_{OOO1}[\mu_{OOO}[1-\Phi(\alpha'_{OOO})]+\sigma_{OOO} \phi(\alpha'_{OOO})]+(1-p_{OOO1})[\mu_{ONO}[1-\Phi(\alpha'_{ONO})]+\sigma_{ONO} \phi(\alpha'_{ONO})]}{p_{OOO1}[1-\Phi(\alpha'_{OOO})]+ (1-p_{OOO1})[1-\Phi(\alpha'_{ONO})]}
    \end{align}
}
    We already showed $w\approx  p_{OOO1}N(6+\mu_1,\mu_2-\mu^2_1)+(1-p_{OOO1}) N(5+\mu_3,\mu_4-\mu^2_3)$. Then $ \mathbbm{E}(w) =\mathbbm{E}[(Y_{i1}-Y_{i0})|D_i=1,S_{i0}=1,S_{i1}=1]=p_{OOO1}(6+\mu_1)+(1-p_{OOO1})(5+\mu_3)$. Further,
    {\footnotesize
    \begin{align*}
    \tau_{\textup{DiDs}} &=\mathbbm{E}[Y_{i1}-Y_{i0}|D_i=1,S_{i0}=1,S_{i1}=1]-\mathbbm{E}[Y_{i1}-Y_{i0}|D_i=0,S_{i0}=1,S_{i1}=1]\\
    &= p_{OOO1}(t^{OOO}_1-t^{OOO}_0+\tau^{OOO}+\mu_1)+(1-p_{OOO1})(t^{ONO}_1-t^{ONO}_0+\tau^{ONO}+\mu_3)-(t^{OOO}_1-t^{OOO}_0+\mu_1)
\end{align*}}
Finally, we can obtain the numerical bounds for $\tau_{OOO}$ as defined in Theorem \ref{mono_bound}
\begin{equation*}
            \begin{split}
                LB^{\prime}_{\tau_{OOO}} &= LB^{\prime}_{OOO1}-\mathbbm{E}[Y_{1}-Y_{0}|D=0,S_{0}=1,S_{1}=1],\\
                UB^{\prime}_{\tau_{OOO}} &=UB_{OOO1}^{\prime}-\mathbbm{E}[Y_{1}-Y_{0}|D=0,S_{0}=1,S_{1}=1]
            \end{split}
        \end{equation*}
			where,
			{\footnotesize
	\begin{align*}
					&LB_{OOO1}^{\prime} =\frac{p_{OOO1}[\mu_{OOO}\Phi(\alpha_{OOO})-\sigma_{OOO} \phi(\alpha_{OOO})]+(1-p_{OOO1})[\mu_{ONO}\Phi(\alpha_{ONO})-\sigma_{ONO} \phi(\alpha_{ONO})]}{p_{OOO1}\Phi(\alpha_{OOO})+ (1-p_{OOO1})\Phi(\alpha_{ONO})}\\
					&UB_{OOO1}^{\prime} =\frac{p_{OOO1}[\mu_{OOO}[1-\Phi(\alpha'_{OOO})]+\sigma_{OOO} \phi(\alpha'_{OOO})]+(1-p_{OOO1})[\mu_{ONO}[1-\Phi(\alpha'_{ONO})]+\sigma_{ONO} \phi(\alpha'_{ONO})]}{p_{OOO1}[1-\Phi(\alpha'_{OOO})]+ (1-p_{OOO1})[1-\Phi(\alpha'_{ONO})]}\\
                    &\alpha_{OOO} =\tfrac{c-\mu_{OOO}}{\sigma_{OOO}} \quad \alpha_{ONO}=\tfrac{c-\mu_{ONO}}{\sigma_{ONO}}\\
                    &\alpha'_{OOO} =\tfrac{c'-\mu_{OOO}}{\sigma_{OOO}} \quad \alpha'_{ONO}=\tfrac{c'-\mu_{ONO}}{\sigma_{ONO}}\\
                    &\mu_{OOO}=t^{OOO}_1-t^{OOO}_0+\tau^{OOO}+\mu_1=6+\mu_1\quad \sigma^2_{OOO}=\mu_2-\mu^2_1\\  &\mu_{ONO}=t^{ONO}_1-t^{ONO}_0+\tau^{ONO}+\mu_3=5+\mu_3\quad \sigma^2_{ONO}=\mu_4-\mu^2_3\\
                    &c = F_w^{-1}(p_{OOO1}) \quad  c' = F_w^{-1}(1-p_{OOO1})\\
                    &p_{OOO1}=\frac{\pi_{OOO1}}{\pi_{OOO1}+\pi_{ONO1}}\\
     &\mathbbm{E}[Y_{i1}-Y_{i0}|D_i=0,S_{i0}=1,S_{i1}=1] = t^{OOO}_1-t^{OOO}_0+\mu_1
				\end{align*}
    }
  \subsubsection{True ATT Bounds for ONO Group} \label{TrueONO}
  To derive Theoretical bounds for $\tau_{ONO}$ as defined in Theorem \ref{ONO_bound}, we can follow a similar procedure as in \ref{TrueOOO}. Following Equation \ref{c} and \ref{c'} we have 
       
{\footnotesize      
\begin{align*}
    LB_{ONO1} &=\mathbbm{E}[Y_{1}-Y_{0}|D=1,S_{0}=1,S_{1}=1, (Y_{1}-Y_{0})\leq F_{\Delta Y|111}^{-1}(1-p_{OOO1})]= E[w|w\leq c']\\
    &=\frac{p_{OOO1}[\mu_{OOO}\Phi(\alpha'_{OOO})-\sigma_{OOO} \phi(\alpha'_{OOO})]+(1-p_{OOO1})[\mu_{ONO}\Phi(\alpha'_{ONO})-\sigma_{ONO} \phi(\alpha'_{ONO})]}{p_{OOO1}\Phi(\alpha'_{OOO})+ (1-p_{OOO1})\Phi(\alpha'_{ONO})}\\
    UB_{ONO1} &=\mathbbm{E}[Y_{1}-Y_{0}|D=1,S_{0}=1,S_{1}=1, (Y_{1}-Y_{0}) > F_{\Delta Y|111}^{-1}(p_{OOO1})]=E[w|w> c] \\
    &=\frac{p_{OOO1}[\mu_{OOO}[1-\Phi(\alpha_{OOO})]+\sigma_{OOO} \phi(\alpha_{OOO})]+(1-p_{OOO1})[\mu_{ONO}[1-\Phi(\alpha_{ONO})]+\sigma_{ONO} \phi(\alpha_{ONO})]}{p_{OOO1}[1-\Phi(\alpha_{OOO})]+ (1-p_{OOO1})[1-\Phi(\alpha_{ONO})]}
\end{align*}}

$Y_{01}^{LB}$, which is the true lower bound of the untreated potential outcome's support in the post-treatment period, is $\min[Y^\ast_{i1}(0)]= \min[t^{g_i}_1 + c_i+ u_{i1}]=1-5-5=-9$
\begin{align*}
    \mathbbm{E}[Y_{i1}|D_i=0,S_{i0}=1,S_{i1}=1]
    &=\mathbbm{E}[S_{i1}\cdot Y^\ast_{i1}|D_i=0,S_{i0}=1,S_{i1}=1]\\
    &=\mathbbm{E}[Y^\ast_{i1}(1)|D_i=0,S_{i0}(0)=1,S_{i1}(0)=1] \quad \text{Assumption \ref{monotone}}\\
     &=\mathbbm{E}[Y^\ast_{i1}(1)|D_i=0,S_{i0}(0)=1,S_{i1}(0)=1,S_{i1}(1)=1] \\
     &=\mathbbm{E}[t^{OOO}_1 + c_i+ u_{i1}|D_i=0,S_{i0}(0)=1,S_{i1}(0)=1,S_{i1}(1)=1]\\
     &=t^{OOO}_1+\mathbbm{E}[c_i+u_{i1}|z_1>0,z_2>0,z_3>-1.5,z_4\leq 0]\\
     &=6+\mathbbm{E}[c_i+u_{i1}|z_1>0,z_2>0,z_4<0]
\end{align*}
Let $v_1=c_i+u_{i1}$ and define \[
Z_{124}=(z_1,z_2,z_4)'=(b_i+v_{i0},\, b_i+v_{i1},\, a_i+w_{1i})' 
   \sim N_3\!\left(0,\Sigma_{Z_{124}}\right),
\qquad
\Sigma_{Z_{124}}=
\begin{pmatrix}
2 & 1 & 0\\
1 & 2 & 0\\
0 & 0 & 2
\end{pmatrix}.
\]
 Then the joint $(V_1,Z_{124})$ is Gaussian with $\Sigma_{(V_1,Z_{124})}=(0,\,0.6,\,0.7)$. The event $\{D=0,\; S_0=1,\; S_1=1\}\ $ is equivalent to $Z_{124}\in\mathcal{R}$ with $\mathcal{R}_{OOO0}=\{\,z_1>0,\; z_2>0,\; z_4\le 0\,\}$. Under this truncation $V_1|Z_{124}\sim selection-normal(\mu_{V_1},\sigma^2_{V_1})$ where $\mu_{V_1}=E\bigl[V_1 \mid Z_{124}\in\mathcal R\bigr]=\Sigma_{V_1,Z_{124}}\,\Sigma_{Z_{124},Z_{124}}^{-1}\,E\bigl[Z_{124}\mid \mathcal R\bigr].$ Where $E[Z_{124}\mid \mathcal R]
=\frac{1}{\mathbbm{P}(Z_{124}\in\mathcal R)}\int_{\mathcal R} z_{124}\,\varphi_3\bigl(z_{124};0,\Sigma_{Z_{124}}\bigr)\,dZ_{124}$. Then $\mathbbm{E}[Y_{i1}|D_i=0,S_{i0}=1,S_{i1}=1]=6+\mu_{V_1R_{OOO0}}$.
 Now to derive true values of $LB^{0}_{ONO0}$ and $UB^{0}_{ONO0}$ we need to derive the distribution of $Y_{0}|D=0,S_{0}=1,S_{1}=0$

 \begin{align*} 
    Y_{i0}|D_i=0,S_{i0}=1,S_{i1}=0 
    &=S_{i0}\cdot Y^\ast_{i0}|D_i=0,S_{i0}=1,S_{i1}=0\\ 
    &=Y^\ast_{i0}(0)|D_i=0,S_{i0}(0)=1,S_{i1}(0)=0\\ 
    &=t^{g_i}_0+c_i+ u_{i0}|D_i=0,S_{i0}(0)=1,S_{i1}(0)=0\\
\end{align*}

Under the joint selection $D_{1}=0,S_{0}=1,S_{1}(1)=0$ only latent group $ONN: (S_{0}=1,S_{1}(0)=0,S_{1}(1)=0)$ and $ONO: (S_{0}=1,S_{1}(0)=0,S_{1}(1)=1)$ exists. Therefore, $Y_{0}|D=0,S_{0}=1,S_{1}=0$ will be a two‐component selection Gaussian mixture. Let $v_0=c_i+u_{i0}$. Then the joint $(V_0,Z_{124})$ is Gaussian with $\Sigma_{(V_0,Z_{124})}=(0.7,\,0,\,0.7)$. Under the event $Z_{124}\in\mathcal{R}$, $V_0|Z_{124}\sim selection-normal(\mu_{V_0},\sigma^2_{V_0})$ where $\mu_{V_0R}=E\bigl[V_0 \mid Z_{124}\in\mathcal R\bigr]=\Sigma_{V_0,Z_{124}}\,\Sigma_{Z_{124},Z_{124}}^{-1}\,E\bigl[Z_{124}\mid \mathcal R\bigr].$ Where $E[Z_{124}\mid \mathcal R]
=\frac{1}{\mathbbm{P}(Z_{124}\in\mathcal R)}\int_{\mathcal R} z_{124}\,\varphi_3\bigl(z_{124};0,\Sigma_{Z_{124}}\bigr)\,dZ_{124}$.
$\sigma^2_{V_0R}=\Sigma_{v_0v_0}-\Sigma_{(v_0,Z_{124})}\Sigma_{Z_{124}}^{-1}\Sigma_{(Z_{124},v_0)}+\Sigma_{(v_0,Z_{124})}\Sigma_{Z_{124}}^{-1}Var(Z_{124}\mid \mathcal R)\Sigma_{Z_{124}}^{-1}\Sigma_{(Z_{124},v_0)}$. Define the selection region for $ONN$ as $\mathcal R_{\mathrm{ONN0}}=\{\,z_1 > 0,\;z_2 < 0,\;z_4 < 0\}$ and the selection region for $ONO$ as $\mathcal R_{\mathrm{ONO0}}
=\{\,z_1 > 0,\;-1.5 < z_2 \le 0,\;z_4 < 0\}$. For notation simplicity let $p_{ONO0}=\frac{\pi_{ONO0}}{\pi_{ONO0}+\pi_{ONN0}}$, $\mu_{ONN0}=t_0^{ONN}+\mu_{v_0R_{ONN}}$, $\sigma^2_{ONN0}=\sigma^2_{V_0R_{ONN0}}$,$\mu_{ONO0}=t_0^{ONO0}+\mu_{v_0R_{ONO0}}$ and $\sigma^2_{ONO0}=\sigma^2_{V_0R_{ONO0}}$. 
$$
Y_0|010=Y_{0}|D=0,S_{0}=1,S_{1}=0 \approx  p_{ONO0}N(\mu_{ONO0},\sigma^2_{ONO0})+(1-p_{ONO0}) N(\mu_{ONN0},\sigma^2_{ONN0})
$$
 Following Equation \ref{c} and \ref{c'} we have,
 {\footnotesize
 \begin{align*}
    LB^{0}_{ONO0} &=\mathbbm{E}[Y_{0}|D=0,S_{0}=1,S_{1}=0, Y_{0}\leq F_{Y_0|010}^{-1}(p_{ONO0})]\\
    &=\mathbbm{E}[(Y_0|010)|(Y_0|010)<c_1] \quad c_1=F_{Y_0|010}^{-1}(p_{ONO0})\\
    &=\frac{p_{ONO0}[\mu_{ONO0}\Phi(\alpha_{ONO0})-\sigma_{ONO0} \phi(\alpha_{ONO0})]+(1-p_{ONO0})[\mu_{ONN0}\Phi(\alpha_{ONN0})-\sigma_{ONN0} \phi(\alpha_{ONN0})]}{p_{ONO0}\Phi(\alpha_{ONO0})+ (1-p_{ONO0})\Phi(\alpha_{ONN0})}\\
    &\alpha_{ONO0}=\frac{c_1-\mu_{ONO0}}{\sigma_{ONO0}} \quad \alpha_{ONN0}=\frac{c_1-\mu_{ONN0}}{\sigma_{ONN0}}\\
    UB^{0}_{ONO0} &=\mathbbm{E}[Y_{0}|D=0,S_{0}=1,S_{1}=0, Y_{0} > F_{Y_0|010}^{-1}(1-p_{ONO0})] \\
    &=\mathbbm{E}[(Y_0|010)|(Y_0|010)>c_2] \quad c_2=F_{Y_0|010}^{-1}(1-p_{ONO0})\\
    &=\frac{p_{ONO0}[\mu_{ONO0}[1-\Phi(\alpha'_{ONO0})]+\sigma_{ONO0} \phi(\alpha'_{ONO0})]+(1-p_{ONO0})[\mu_{ONN0}[1-\Phi(\alpha'_{ONN0})]+\sigma_{ONN0} \phi(\alpha'_{ONN0})]}{p_{ONO0}[1-\Phi(\alpha'_{ONO0})]+ (1-p_{ONO0})[1-\Phi(\alpha'_{ONN0})]}\\
    &\alpha'_{ONO0}=\frac{c_2-\mu_{ONO0}}{\sigma_{ONO0}} \quad \alpha'_{ONN0}=\frac{c_2-\mu_{ONN0}}{\sigma_{ONN0}}\\
\end{align*}}
Now we can get the true bounds for $\tau_{ONO}$ as,
{\footnotesize\begin{equation*}
            \begin{split}
                LB_{\tau_{ONO}} &= LB_{ONO1}-\mathbbm{E}[Y_{1}|D=0,S_{0}=1, S_{1}=1]+LB^{0}_{ONO0},\\
                UB_{\tau_{ONO}} &= UB_{ONO1}-Y_{01}^{LB}+UB^{0}_{ONO0}
            \end{split}
        \end{equation*}}
\subsubsection{True ATT Bounds for NNO Group} \label{trueNNO}

To derive Theoretical bounds for $\tau_{ONO}$ as defined in Theorem \ref{NNO_bound}, we can follow a similar procedure as in Section \ref{TrueONO}. First to derive true values of $LB_{NNO1}$ and $UB_{NNO1}$ we need to derive the distribution of $Y_{1}|D=1,S_{0}=0,S_{1}=1$
\begin{align*} 
    Y_{i1}|D_i=1,S_{i0}=0,S_{i1}=1
    &=S_{i1}\cdot Y^\ast_{i1}|D_i=1,S_{i0}=0,S_{i1}=1\\ 
    &=Y^\ast_{i1}(1)|D_i=1,S_{i0}(0)=0,S_{i1}(0)=0,S_{i1}(1)=1\\ 
    &=t^{g_i}_1+\tau^{g_i}+c_i+ u_{i1}|D_i=1,S_{i0}(0)=0,S_{i1}(1)=1
\end{align*}
Under the joint selection $D_{1}=0,S_{0}=0,S_{1}(1)=1$ only latent group $NOO: (S_{0}=0,S_{1}(0)=1,S_{1}(1)=1)$ and $NNO: (S_{0}=0,S_{1}(0)=0,S_{1}(1)=1)$ exists. Therefore, $Y_{1}|D=1,S_{0}=0,S_{1}=1$ will be a two‐component selection Gaussian mixture. As, we already defined $v_1=c_i+u_{i1}$. We already showed that $V_1|Z_{124}\sim selection-normal(\mu_{V_1R},\sigma^2_{V_1R})$ where $\mu_{V_1R}=E\bigl[V_1 \mid Z_{124}\in\mathcal R\bigr]=\Sigma_{V_1,Z_{124}}\,\Sigma_{Z_{124},Z_{124}}^{-1}\,E\bigl[Z_{124}\mid \mathcal R\bigr]$ and $\sigma^2_{V_1R}=\Sigma_{v_1v_1}-\Sigma_{(v_1,Z_{124})}\Sigma_{Z_{124}}^{-1}\Sigma_{(Z_{124},v_1)}+\Sigma_{(v_1,Z_{124})}\Sigma_{Z_{124}}^{-1}Var(Z_{124}\mid \mathcal R)\Sigma_{Z_{124}}^{-1}\Sigma_{(Z_{124},v_1)}$. Define the selection region for $NOO$ as $\mathcal R_{\mathrm{NOO}}=\{\,z_1 \le 0,\;z_2 > 0,\;z_4 > 0\}$ and the selection region for $NNO$ as $\mathcal R_{\mathrm{NNO}}
=\{\,z_1 \le 0,\;-1.5 < z_2 \le 0,\;z_4 > 0\}$. For notation simplicity let $p_{NNO1}=\frac{\pi_{NNO1}}{\pi_{NNO1}+\pi_{NOO1}}$, $\mu_{NNO1}=t_1^{NNO}+\tau^{NNO}+\mu_{v_1R_{NNO}}$, $\sigma^2_{NNO1}=\sigma^2_{V_1R_{NNO}}$, $\mu_{NOO1}=t_1^{NOO}+\tau^{NOO}+\mu_{v_1R_{NOO}}$ and $\sigma^2_{NOO1}=\sigma^2_{V_1R_{NOO}}$. Then,
$$
Y_1|101=Y_{1}|D=1,S_{0}=0,S_{1}=1 \approx  p_{NNO1}N(\mu_{NNO1},\sigma^2_{NNO1})+(1-p_{NNO1}) N(\mu_{NOO1},\sigma^2_{NOO1}).
$$
Then following Equation \ref{c}  and \ref{c'} we have,
\footnotesize{  \begin{align*}
    LB_{NNO1} &=\mathbbm{E}[Y_{1}|D=1,S_{0}=0,S_{1}=1, Y_{1}\leq F_{Y_{1}|101}^{-1}(p_{NNO1})]\\
    &=\mathbbm{E}[(Y_1|101)|(Y_1|101)<c_3] \quad c_3=F_{Y_{1}|101}^{-1}(p_{NNO1})\\
    &=\frac{p_{NNO1}[\mu_{NNO1}\Phi(\alpha_{NNO1})-\sigma_{NNO1} \phi(\alpha_{NNO1})]+(1-p_{NNO1})[\mu_{NOO1}\Phi(\alpha_{NOO1})-\sigma_{NOO1} \phi(\alpha_{NOO1})]}{p_{NNO1}\Phi(\alpha_{NNO1})+ (1-p_{NNO1})\Phi(\alpha_{NOO1})}\\
    &\alpha_{NNO1}=\frac{c_3-\mu_{NNO1}}{\sigma_{NNO1}} \quad \alpha_{NOO1}=\frac{c_3-\mu_{NOO1}}{\sigma_{NOO1}}\\
    UB_{NNO1} &=\mathbbm{E}[Y_{1}|D=1,S_{0}=0,S_{1}=1, Y_{1} > F_{Y_{1}|101}^{-1}(1-p_{NNO1})]\\
    &=\mathbbm{E}[(Y_1|101)|(Y_1|101)>c_4] \quad c_4=F_{Y_{1}|101}^{-1}(1-p_{NNO1})\\
    &=\frac{p_{NNO1}[\mu_{NNO1}[1-\Phi(\alpha'_{NNO1})]+\sigma_{NNO1} \phi(\alpha'_{NNO1})]+(1-p_{NNO1})[\mu_{NOO1}[1-\Phi(\alpha'_{NOO1})]+\sigma_{NOO1} \phi(\alpha'_{NOO1})]}{p_{NNO1}[1-\Phi(\alpha'_{NNO1})]+ (1-p_{NNO1})[1-\Phi(\alpha'_{NOO1})]}\\
    &\alpha'_{NNO1}=\frac{c_3-\mu_{NNO1}}{\sigma_{NNO1}} \quad \alpha'_{NOO1}=\frac{c_4-\mu_{NOO1}}{\sigma_{NOO1}}
\end{align*}}
From section \ref{TrueONO} we have,
 \begin{align*}
    LB^{0}_{ONO0} &=\mathbbm{E}[Y_{0}|D=0,S_{0}=1,S_{1}=0, Y_{0}\leq F_{Y_0|010}^{-1}(p_{ONO0})]\\
    &=\mathbbm{E}[(Y_0|010)|(Y_0|010)<c_1] \quad c_1=F_{Y_0|010}^{-1}(p_{ONO0})\\
    &=\frac{p_{ONO0}[\mu_{ONO0}\Phi(\alpha_{ONO0})-\sigma_{ONO0} \phi(\alpha_{ONO0})]+(1-p_{ONO0})[\mu_{ONN0}\Phi(\alpha_{ONN0})-\sigma_{ONN0} \phi(\alpha_{ONN0})]}{p_{ONO0}\Phi(\alpha_{ONO0})+ (1-p_{ONO0})\Phi(\alpha_{ONN0})}\\
    &\alpha_{ONO0}=\frac{c_1-\mu_{ONO0}}{\sigma_{ONO0}} \quad \alpha_{ONN0}=\frac{c_1-\mu_{ONN0}}{\sigma_{ONN0}}
\end{align*}
To derive true values of $LB^{0}_{ONO1}$ we need to derive the distribution of $Y_{0}|D=1,S_{0}=1,S_{1}=1$
\begin{align*} 
    Y_{i0}|D_i=1,S_{i0}=1,S_{i1}=1
    &=S_{i0}\cdot Y^\ast_{i0}|D_i=1,S_{i0}=1,S_{i1}=1\\ 
    &=Y^\ast_{i0}(1)|D_i=1,S_{i0}(0)=1,S_{i1}(0)=1,S_{i1}(1)=1\\ 
    &=t^{g_i}_0+c_i+ u_{i0}|D_i=1,S_{i0}(0)=1,S_{i1}(1)=1
\end{align*}
Under the joint selection $D_{1}=1,S_{0}=1,S_{1}(1)=1$ only latent group $OOO: (S_{0}=1,S_{1}(0)=1,S_{1}(1)=1)$ and $ONO: (S_{0}=1,S_{1}(0)=0,S_{1}(1)=1)$ exists. Therefore, $Y_{0}|D=1,S_{0}=1,S_{1}=1$ will be a two‐component selection Gaussian mixture. As, we already defined $v_0=c_i+u_{i0}$. We already showed that $V_0|Z_{124}\sim selection-normal(\mu_{V_0R},\sigma^2_{V_0R})$ where $\mu_{V_0R}=E\bigl[V_0 \mid Z_{124}\in\mathcal R\bigr]=\Sigma_{V_0,Z_{124}}\,\Sigma_{Z_{124},Z_{124}}^{-1}\,E\bigl[Z_{124}\mid \mathcal R\bigr]$ and $\sigma^2_{V_0R}=\Sigma_{v_0v_0}-\Sigma_{(v_0,Z_{124})}\Sigma_{Z_{124}}^{-1}\Sigma_{(Z_{124},v_0)}+\Sigma_{(v_1,Z_{124})}\Sigma_{Z_{124}}^{-1}Var(Z_{124}\mid \mathcal R)\Sigma_{Z_{124}}^{-1}\Sigma_{(Z_{124},v_0)}$. Define the selection region for $OOO$ as $\mathcal R_{\mathrm{OOO}}=\{\,z_1 > 0,\;z_2 > -1.5,\;z_4 > 0\}$ and the selection region for $ONO$ as $\mathcal R_{\mathrm{ONO}}
=\{\,z_1 > 0,\;-1.5 < z_2 \le 0,\;z_4 > 0\}$. For notation simplicity let $\mu_{ONO1}=t_0^{ONO}+\mu_{v_0R_{ONO}}$, $\sigma^2_{ONO1}=\sigma^2_{V_0R_{NNO}}$, $\mu_{OOO1}=t_0^{OOO}+\mu_{v_0R_{OOO}}$ and $\sigma^2_{OOO1}=\sigma^2_{V_0R_{OOO}}$. Then,
$$
Y_0|111=Y_{0}|D=1,S_{0}=1,S_{1}=1 \approx  p_{OOO1}N(\mu_{OOO1},\sigma^2_{OOO1})+(1-p_{OOO1}) N(\mu_{ONO1},\sigma^2_{ONO1}).
$$
Then following Equation \ref{c} we have,
{\footnotesize  \begin{align*}
     LB_{ONO1} &=\mathbbm{E}[Y_{0}|D=1,S_{0}=1,S_{1}=1, Y_{0}\leq F_{Y_{0}|111}^{-1}(1-p_{OOO1})]\\
     &=\mathbbm{E}[(Y_0|111)|(Y_0|111)<c_5] \quad c_5=F_{Y_{0}|111}^{-1}(1-p_{OOO1})\\
     &=\frac{(1-p_{OOO1})[\mu_{ONO1}\Phi(\alpha_{ONO1})-\sigma_{ONO1} \phi(\alpha_{ONO1})]+p_{OOO1}[\mu_{OOO1}\Phi(\alpha_{OOO1})-\sigma_{OOO1} \phi(\alpha_{OOO1})]}{(1-p_{OOO1})\Phi(\alpha_{ONO1})+ p_{OOO1}\Phi(\alpha_{OOO1})}\\
    &\alpha_{ONO1}=\frac{c_5-\mu_{ONO1}}{\sigma_{ONO1}} \quad \alpha_{NOO1}=\frac{c_5-\mu_{OOO1}}{\sigma_{OOO1}}
\end{align*}}
$Y_{01}^{LB}$, which is the true lower bound of the untreated potential outcome's support in the post-treatment period, is $\min[Y^\ast_{i1}(0)]= \min[t^{g_i}_1 + c_i+ u_{i1}]=1-5-5=-9$. $Y_{00}^{LB}$, which is the true lower bound of the untreated potential outcome's support in the pre-treatment period, is $\min[Y^\ast_{i0}(0)]= \min[t^{g_i}_0 + c_i+ u_{i0}]=0-5-5=-10$. $Y_{10}^{LB}$, which is the true lower bound of the treated potential outcome's support in the pre-treatment period, is $\min[Y^\ast_{i0}(1)=Y^\ast_{i0}(0)]= -10$.
\begin{align*}
    \mathbbm{E}[Y_{i1}|D_i=0,S_{i0}=0,S_{i1}=1]
    &=\mathbbm{E}[S_{i1}\cdot Y^\ast_{i1}|D_i=0,S_{i0}=0,S_{i1}=1]\\
    &=\mathbbm{E}[Y^\ast_{i1}(1)|D_i=0,S_{i0}(0)=0,S_{i1}(0)=1] \quad \text{Assumption \ref{monotone}}\\
     &=\mathbbm{E}[Y^\ast_{i1}(1)|D_i=0,S_{i0}(0)=0,S_{i1}(0)=1,S_{i1}(1)=1] \\
     &=\mathbbm{E}[t^{NOO}_1 + c_i+ u_{i1}|D_i=0,S_{i0}(0)=0,S_{i1}(0)=1,S_{i1}(1)=1]\\
     &=t^{NOO}_1+\mathbbm{E}[c_i+u_{i1}|z_1>0,z_2>0,z_3>0,Z_4\leq 0]\\
     &=4+\mathbbm{E}[c_i+u_{i1}|z_1<0,z_2>0,Z_4<0]
\end{align*}
As, we already defined $v_1=c_i+u_{i1}$. We already showed that $V_1|Z_{124}\sim selection-normal(\mu_{V_1R},\sigma^2_{V_1R})$ where $\mu_{V_1R}=E\bigl[V_1 \mid Z_{124}\in\mathcal R\bigr]=\Sigma_{V_1,Z_{124}}\,\Sigma_{Z_{124},Z_{124}}^{-1}\,E\bigl[Z_{124}\mid \mathcal R\bigr]$. Here the selection region $\mathcal R$ for $NOO$ as $\mathcal R_{\mathrm{NOO0}}=\{\,z_1 < 0,\;z_2 > 0,\;z_4 < 0\}$. Where $E[Z_{124}\mid \mathcal R_{\mathrm{NOO0}}]=\frac{1}{\mathbbm{P}(Z_{124}\in \mathcal R_{\mathrm{NOO0}})}\int_{\mathcal R_{\mathrm{NOO0}}} z_{124}\,\varphi_3\bigl(z_{124};0,\Sigma_{Z_{124}}\bigr)\,dZ_{124}$. Then $\mathbbm{E}[(Y_{i1}|D_i=0,S_{i0}=0,S_{i1}=1]=4+\mu_{V_1R_{NOO0}}$.

Combining this, we can derive true theoretical bounds for NNO group as 
{\footnotesize\begin{equation*}
    \begin{split}
        LB_{\tau_{NNO}} &= LB_{NNO1}-LB^{0}_{ONO1}-\mathbbm{E}[Y_{1}|D=0,S_{0}=0, S_{1}=1]+Y_{00}^{LB},\\
         UB_{\tau_{NNO}} &= UB_{NNO1}-Y_{10}^{LB}- Y_{01}^{LB}+LB^{0}_{ONO0}
    \end{split}
\end{equation*}}

\subsubsection{True ATT Bounds for NOO Group} 
To derive Theoretical bounds for $\tau_{NOO}$ as defined in Theorem \ref{NOO_bound}, we can follow a similar procedure as in section \ref{TrueONO}. We already showed in the sections \ref{TrueONO} and \ref{trueNNO} that $\mathbbm{E}[(Y_{i1}|D_i=0,S_{i0}=1,S_{i1}=1]=6+\mu_{V_1R_{OOO0}}$ , $\mathbbm{E}[Y_{1}|D=0,S_{0}=0, S_{1}=1]=4+\mu_{V_1R_{NOO}}$ and $Y_{00}^{LB}=Y_{10}^{LB}=-10$. To derive true values of $LB_{NOO1}$ and $UB_{NOO1}$ we need to derive the distribution of $Y_{1}|D=1,S_{0}=0,S_{1}=1$, which we derived in section \ref{trueNNO}. Thus, following that derivation, we have,
 \begin{align*}
    LB_{NOO1} &=\mathbbm{E}[Y_{1}|D=1,S_{0}=0,S_{1}=1, Y_{1}\leq F_{Y_{1}|101}^{-1}(1-p_{NNO1})]\\
    &=\mathbbm{E}[(Y_1|101)|(Y_1|101)<c_4]\\
    &=\frac{p_{NNO1}[\mu_{NNO1}\Phi(\alpha'_{NNO1})-\sigma_{NNO1} \phi(\alpha'_{NNO1})]+(1-p_{NNO1})[\mu_{NOO1}\Phi(\alpha'_{NOO1})-\sigma_{NOO1} \phi(\alpha'_{NOO1})]}{p_{NNO1}\Phi(\alpha'_{NNO1})+ (1-p_{NNO1})\Phi(\alpha'_{NOO1})}\\
     &\alpha'_{NNO1}=\frac{c_4-\mu_{NNO1}}{\sigma_{NNO1}} \quad \alpha'_{NOO1}=\frac{c_4-\mu_{NOO1}}{\sigma_{NOO1}}\\
    UB_{NOO1} &=\mathbbm{E}[Y_{1}|D=1,S_{0}=0,S_{1}=1, Y_{1} > F_{Y_{1}|101}^{-1}(p_{NNO1})] \\
    &=\mathbbm{E}[(Y_1|101)|(Y_1|101)>c_3]\\
    &=\frac{p_{NNO1}[\mu_{NNO1}[1-\Phi(\alpha_{NNO1})]+\sigma_{NNO1} \phi(\alpha_{NNO1})]+(1-p_{NNO1})[\mu_{NOO1}[1-\Phi(\alpha_{NOO1})]+\sigma_{NOO1} \phi(\alpha_{NOO1})]}{p_{NNO1}[1-\Phi(\alpha_{NNO1})]+ (1-p_{NNO1})[1-\Phi(\alpha_{NOO1})]}\\
    &\alpha_{NNO1}=\frac{c_3-\mu_{NNO1}}{\sigma_{NNO1}} \quad \alpha_{NOO1}=\frac{c_3-\mu_{NOO1}}{\sigma_{NOO1}}
\end{align*}
To derive true values of $LB^{0}_{OOO1}$ we need to derive the distribution of $Y_{0}|D=1,S_{0}=1,S_{1}=1$, which we derived in section \ref{trueNNO}. Thus, following that derivation, we have,
\begin{align*}
    LB^{0}_{OOO1} &=\mathbbm{E}[Y_{0}|D=1,S_{0}=1,S_{1}=1, Y_{0}\leq F_{Y_{0}|111}^{-1}(p_{OOO1})]\\
    &=\mathbbm{E}[(Y_0|111)|(Y_0|111)<c_6] \quad c_6=F_{Y_{0}|111}^{-1}(p_{OOO1})\\
     &=\frac{(1-p_{OOO1})[\mu_{ONO1}\Phi(\alpha_{ONO1})-\sigma_{ONO1} \phi(\alpha_{ONO1})]+p_{OOO1}[\mu_{OOO1}\Phi(\alpha_{OOO1})-\sigma_{OOO1} \phi(\alpha_{OOO1})]}{(1-p_{OOO1})\Phi(\alpha_{ONO1})+ p_{OOO1}\Phi(\alpha_{OOO1})}\\
    &\alpha_{ONO1}=\frac{c_6-\mu_{ONO1}}{\sigma_{ONO1}} \quad \alpha_{NOO1}=\frac{c_6-\mu_{OOO1}}{\sigma_{OOO1}}
\end{align*} 
\begin{align*}
    \mathbbm{E}[Y_{i0}|D_i=0,S_{i0}=1,S_{i1}=1]
    &=\mathbbm{E}[S_{i0}\cdot Y^\ast_{i1}|D_i=0,S_{i0}=1,S_{i1}=1]\\
    &=\mathbbm{E}[Y^\ast_{i0}(0)|D_i=0,S_{i0}(0)=1,S_{i1}(0)=1] \quad \text{Assumption \ref{monotone}}\\
     &=\mathbbm{E}[Y^\ast_{i0}(0)|D_i=0,S_{i0}(0)=1,S_{i1}(0)=1,S_{i1}(1)=1] \\
     &=\mathbbm{E}[t^{OOO}_0 + c_i+ u_{i0}|D_i=0,S_{i0}(0)=1,S_{i1}(0)=1,S_{i1}(1)=1]\\
     &=t^{OOO}_0+\mathbbm{E}[c_i+u_{i0}|z_1>0,z_2>0,z_3>-1.5,z_4\leq 0]\\
     &=5+\mathbbm{E}[c_i+u_{i0}|z_1>0,z_2>0,z_4<0]
\end{align*}  
Under the joint selection $D_{1}=0,S_{0}=1,S_{1}(1)=1$ only latent group $OOO: (S_{0}=1,S_{1}(0)=1,S_{1}(1)=1)$  exists.  As we already defined $v_0=c_i+u_{i0}$. We already showed that $V_0|Z_{124}\sim selection-normal(\mu_{V_0R},\sigma^2_{V_0R})$ where $\mu_{V_0R}=E\bigl[V_0 \mid Z_{124}\in\mathcal R\bigr]=\Sigma_{V_0,Z_{124}}\,\Sigma_{Z_{124},Z_{124}}^{-1}\,E\bigl[Z_{124}\mid \mathcal R\bigr]$ and $\sigma^2_{V_0R}=\Sigma_{v_0v_0}-\Sigma_{(v_0,Z_{124})}\Sigma_{Z_{124}}^{-1}\Sigma_{(Z_{124},v_0)}+\Sigma_{(v_1,Z_{124})}\Sigma_{Z_{124}}^{-1}Var(Z_{124}\mid \mathcal R)\Sigma_{Z_{124}}^{-1}\Sigma_{(Z_{124},v_0)}$. Define the selection region for $OOO$ as $\mathcal R_{\mathrm{OOO}}=\{\,z_1 > 0,\;z_2 > 0,\;z_4 < 0\}$ Then  $\mathbbm{E}[Y_{i0}|D_i=0,S_{i0}=1,S_{i1}=1]=\mu_{OOO1}=t_0^{OOO}+\mu_{v_0R_{OOO}}$.  
Combining this, we can derive true theoretical bounds for NOO group as 
{\footnotesize\begin{equation*}
    \begin{split}
        LB_{\tau_{NOO}} &= LB_{NOO1}-LB^{0}_{OOO1}-\mathbbm{E}[Y_{1}|D=0,S_{0}=0, S_{1}=1]+Y_{00}^{LB},\\
        UB_{\tau_{NOO}} &= UB_{NOO1}-Y_{10}^{LB}-\mathbbm{E}[Y_{1}|D=0,S_{0}=0, S_{1}=1]+\mathbbm{E}[Y_{0}|D=0,S_{0}=1, S_{1}=1]],
    \end{split}
\end{equation*}}
 
\subsection{Simulation Results for Other Latent Groups}\label{simu_other}
\begin{figure}[h]
     \centering
         \centering
         \caption{Distribution of estimated mixture proportion, $\hat{p}_{ONO0}$, with monotonicity}
         \includegraphics[width=0.5\linewidth]{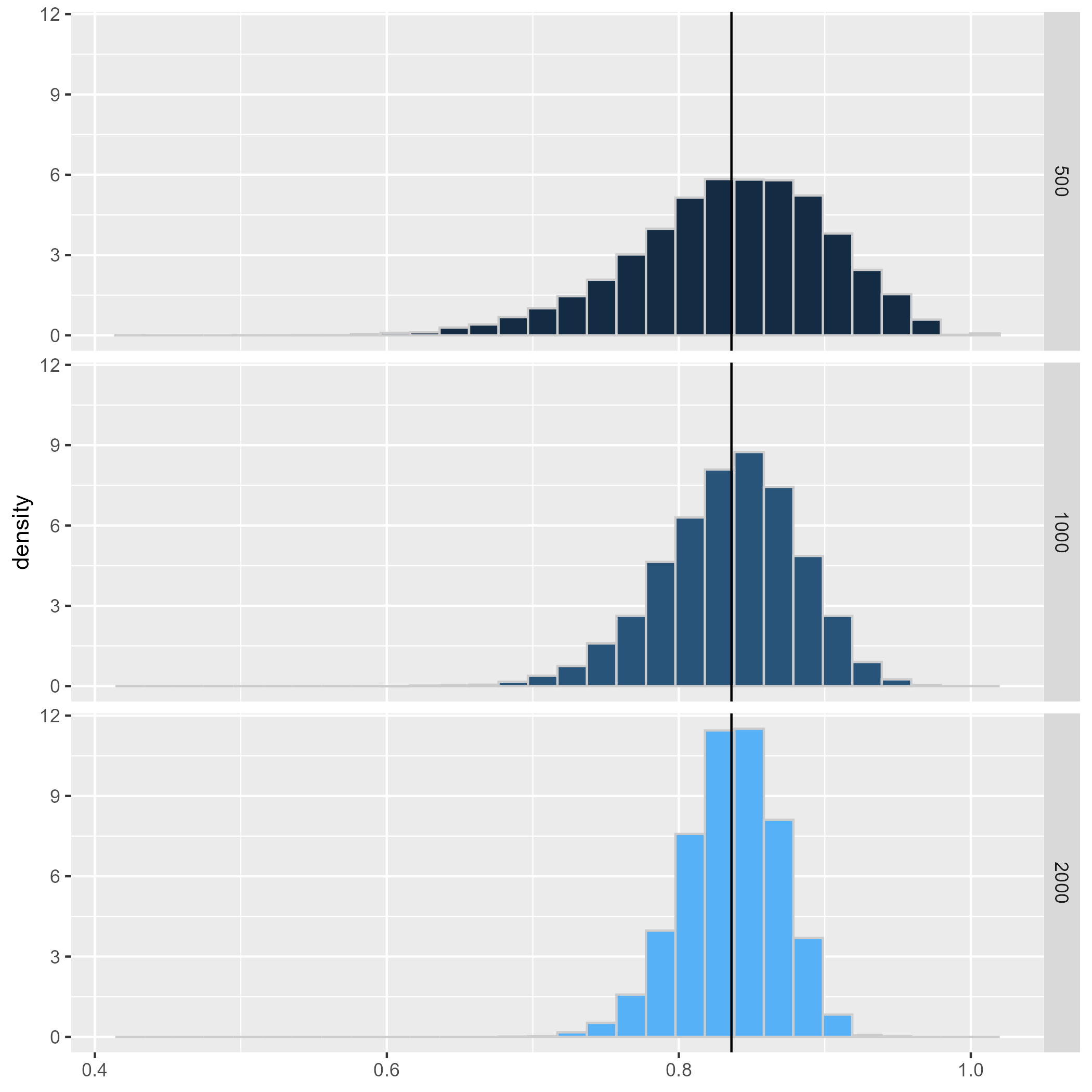} 
        \label{fig: mixp_ONO0}
        \begin{minipage}{0.7\textwidth}
        \footnotesize \textit{Note:} This figure presents the empirical sampling distribution of $\hat{p}_{ONO0}$ for sample sizes $n \in \{500, 1000, 2000\}$, based on 10,000 replications. The vertical black line is the true value of the mixing proportion, $p_{ONO0} = 0.8360$.
        \end{minipage}
\end{figure}
\begin{figure}[h]
     \centering
         \centering
         \caption{Distribution of estimated mixture proportion, $\hat{p}_{NNO1}$, with monotonicity}
         \includegraphics[width=0.5\linewidth]{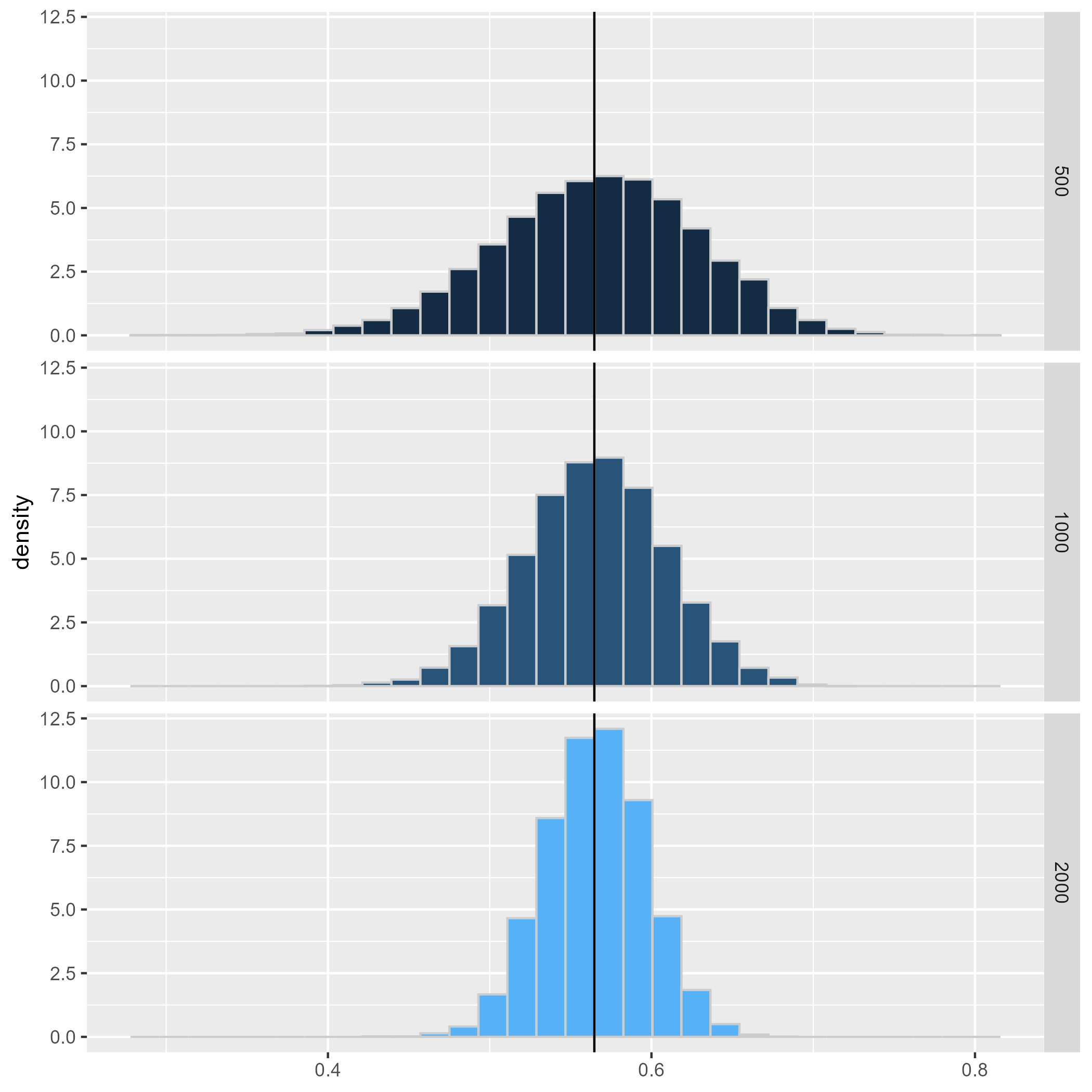} 
        \label{fig: mixp_NNO1}
        \begin{minipage}{0.7\textwidth}
        \footnotesize \textit{Note:} This figure presents the empirical sampling distribution of $\hat{p}_{NNO1}$ for sample sizes $n \in \{500, 1000, 2000\}$, based on 10,000 replications. The vertical black line is the true value of the mixing proportion, $p_{NNO1} = 0.5647$.
        \end{minipage}
\end{figure}

\begin{table}[H]
\footnotesize
  \caption{Estimated bounds for other latent groups} 
  \centering
  \label{simu bounds_other groups}%
  \begin{threeparttable}
    \begin{tabular}{lccccccc}
    \toprule
         & $[\widehat{LB}_{\tau_{ONO}}, \widehat{UB}_{\tau_{ONO}}]$ & Coverage & $[\widehat{LB}_{\tau_{NOO}}, \widehat{UB}_{\tau_{NOO}}]$ & Coverage & $[\widehat{LB}_{\tau_{NNO}}, \widehat{UB}_{\tau_{NNO}}]$ & Coverage \\
    \midrule
    $n=500$    & [0.9673,20.3819] && [-15.0792,18.7858] && [-13.7444,29.6986] &\\
    95\% CI & (0.0353,21.1901) & 0.9197 &(-15.8642,19.6515)  & 0.9324&(-14.5928,30.6509)  & 0.9291\\
    IM 95\% CI &  (0.1850,21.0603) & 1 & (-15.7380,19.5124)  & 1  &(-14.4565,30.4978)  &  1\\
    \hline
    $n=1000$   & [0.9609,20.3909] && [-15.0970,18.7926] &&[-13.7642,29.7208] & \\
    95\% CI &  (0.3138,20.9597) & 0.9219 &(-15.6439,19.4053)    & 0.9356&(-14.3596,30.3888) & 0.9265\\
    IM 95\% CI &   (0.4178,20.8683) & 1 &(-15.5561,19.3068)    & 1&(0.4178,20.8683) & 1\\
    \hline
    $n=2000$   & [0.9684,20.3974] & & [-15.1076,18.7972] &&  [-13.7724,29.7317] &  \\
    95\% CI    & (0.5148,20.7984) & 0.9294 & (-15.4910,19.2305)   & 0.9392& (-14.1919,30.2028) & 0.9263\\
    IM 95\% CI   & (0.5877,20.7340) & 1 & (-15.4294,19.1609)   & 1& (-14.1245,30.1271) & 1\\
\hline 
\end{tabular}%
\end{threeparttable}
    \begin{tablenotes}[flushleft]
    \item\footnotesize Notes: The table presents estimated bounds $\tau_{ONO}$,$\tau_{NOO}$ and $\tau_{NNO}$ for three sample sizes, $n \in \{500,1000,2000\}$. The estimated bounds have been averaged across 10,000 replications.  The true bounds for the simulated data for the latent groups are as follows: $[LB_{\tau_{ONO}}, UB_{\tau_{ONO}}]=[0.9661,20.4025]$, , $[LB_{\tau_{NOO}}, UB_{\tau_{NOO}}] = [-15.1222,18.8031]$, and  $[LB_{\tau_{NNO}}, UB_{\tau_{NNO}}]  = [-13.7919,29.7483]$, respectively. The 95\% confidence interval (CI) reflects the coverage probability of the true interval. The Imbens and Manski (IM) 95\% CI gives the coverage probability of the true parameter.
    \end{tablenotes}
\end{table}%

\subsection{DGP with Mixing Proportion Approaching One}\label{change in mixp}
 \begin{table}[H]
\centering
\caption{Na\"{i}ve DiD estimates and estimated bounds for $\tau_{OOO}$ when $p_{OOO1}\rightarrow 1$}\label{tbl:change Mixp}
\begin{threeparttable}
\begin{tabular}{ccccccc} 
\toprule
$\zeta$  & $p_{OOO1}$ & $\tau_{OOO}$ & $n$ & $\hat{p}_{OOO1}$  & $\hat{\tau}_{\textup{DiDs}}$ & $[\widehat{LB}_{\tau_{OOO}}, \widehat{UB}_{\tau_{OOO}}]$ \\
\midrule
\multirow{3}{*}{0.1}  & \multirow{3}{*}{0.9601} & \multirow{3}{*}{2.0703} & 500  & 0.9166 & 4.9388 & [4.7199, 5.1623] \\
&&& 1000 & 0.9353 & 4.9374 & [4.7627, 5.1166] \\
&&& 2000 & 0.9469 & 4.9368 & [4.7896, 5.0881] \\
\midrule
\multirow{3}{*}{0.4}  & \multirow{3}{*}{0.8649} & \multirow{3}{*}{2.2747} & 500  & 0.8600 & 4.7712 & [4.4181, 5.1387] \\
&&& 1000 & 0.8651 & 4.7688 & [4.4268, 5.1252] \\
&&& 2000 & 0.8656 & 4.7673 & [4.4258, 5.1240] \\
\midrule
\multirow{3}{*}{0.6}  & \multirow{3}{*}{0.8178} & \multirow{3}{*}{2.4029} & 500  & 0.8173 & 4.6812 & [4.2261, 5.1532] \\
&&& 1000 & 0.8186 & 4.6791 & [4.2268, 5.1505] \\
&&& 2000 & 0.8184 & 4.6802 & [4.2277, 5.1522] \\
\midrule
\multirow{3}{*}{0.8}  & \multirow{3}{*}{0.7807} & \multirow{3}{*}{2.5220} & 500  & 0.7816 & 4.6046 & [4.0624, 5.1659] \\
&&& 1000 & 0.7815 & 4.6040 & [4.0610, 5.1667] \\
&&& 2000 & 0.7808 & 4.6066 & [4.0622, 5.1717] \\
\midrule
\multirow{3}{*}{1.0}  & \multirow{3}{*}{0.7517} & \multirow{3}{*}{2.6306} & 500  & 0.7528 & 4.5387 & [3.9231, 5.1743] \\
&&& 1000 & 0.7519 & 4.5403 & [3.9228, 5.1789] \\
&&& 2000 & 0.7518 & 4.5434 & [3.9255, 5.1823] \\
\midrule
\multirow{3}{*}{1.5}  & \multirow{3}{*}{0.7052} & \multirow{3}{*}{2.8506} & 500  & 0.7052 & 4.4297 & [3.6864, 5.1922] \\
&&& 1000 & 0.7053 & 4.4301 & [3.6871, 5.1937] \\
&&& 2000 & 0.7053 & 4.4311 & [3.6875, 5.1954] \\
\midrule
\multirow{3}{*}{2.0}  & \multirow{3}{*}{0.6824} & \multirow{3}{*}{2.9979} & 500  & 0.6827 & 4.3677 & [3.5592, 5.1959] \\
&&& 1000 & 0.6825 & 4.3657 & [3.5560, 5.1960] \\
&&& 2000 & 0.6826 & 4.3679 & [3.5573, 5.1996] \\
\midrule
\multirow{3}{*}{3.0}  & \multirow{3}{*}{0.6685} & \multirow{3}{*}{3.1313} & 500  & 0.6687 & 4.3216 & [3.4669, 5.1976] \\
&&& 1000 & 0.6681 & 4.3234 & [3.4657, 5.2030] \\
&&& 2000 & 0.6689 & 4.3212 & [3.4651, 5.1998] \\
\bottomrule
\end{tabular}
\begin{tablenotes}[flushleft]
    \item\footnotesize Notes: The table reports the estimated bounds for $\tau_{OOO}$, and point estimates of the na\"{i}ve DiD estimand $\tau_{\textup{DiDs}}$ for values of the true mixing proportion $p_{OOO1}$ approaching one, across sample sizes $n \in \{500, 1000, 2000\}$. All estimates are averaged across 10,000 replications. 
    \end{tablenotes}
    \end{threeparttable}
\end{table}

 \subsection{DGP that Violates Monotonicity} \label{without mono DGP}

The DGP given in section \ref{simulation} was modified to violate monotonicity as follows. For each individual, the parameter $\zeta$ is drawn at random from the set \{-1.5, 0, 1.5\}, with equal probability assigned to each value. This generates positive selection for some units, negative selection for others, and no selection for the remainder, such that the monotonicity assumption is violated, in aggregate, for the overall sample. This construction also ensures that all eight latent groups are present in the population. Recall that latent group $g$ is defined by the tuple $(S_{i0}, S_{i1}(0), S_{i1}(1)) \in \{NNN,NNO,NOO,ONN,ONO,OOO,NON,OON\}$. For this modified DGP, we set $\tau_g=(-1,1,3,1,4,5-1,-2)'$, $t_0=(0,2,3,1,4,5,0,1)'$ and $t_1=(1,3,4,2,5,6,1,2)'$, while all other parameter values are kept the same as in Section \ref{simulation}. With this modification, the true overall ATT is $\tau = 1.5947$. The true latent-group-specific ATTs also change for two groups: $\tau_{NON} = -1$ and $\tau_{OON} = -2$, while the ATTs for all remaining latent groups remain the same as in Section \ref{simulation}. Table \ref{simu bounds nomono OOO} reports the simulated bounds for this DGP based on Theorem \ref{nomono_bound}. The bounds contain the true ATT for the OOO group but they are wider than the corresponding bounds obtained under the DGP satisfying monotonicity in Section \ref{simulation}.

\begin{table}[H]
     \caption{Estimated bounds for $\tau_{OOO}$ with monotonicity violations}
     \centering
     \label{simu bounds nomono OOO}
     \small
     \begin{threeparttable}
    \begin{tabular}{ccccc}
    \toprule
       $n$   &$\hat{\tau}_{\textup{DiDs}}$     &$[\widehat{LB}_{\tau_{OOO}}, \widehat{UB}_{\tau_{OOO}}]$  & 95\% CI & IM 95\% CI\\  
    \hline
    $500$               &4.7703        & [2.2637,7.2626] & (1.2469,8.1949)  & (1.4103,8.0450)     \\
    $1000$               &4.7469       & [2.3403,7.1448] & (1.6721,7.7522)   & (1.7795,7.6546)       \\
    $2000$               &  4.7652    & [2.3309,7.1902]    & (1.8574,7.6170) & (1.9335,7.5484)       \\
    \bottomrule
    \end{tabular}
    \begin{tablenotes}[flushleft]
    \item\footnotesize Notes: This table reports estimated bounds for $\tau_{OOO}$ for three sample sizes, $n=\{500,1000,2000\}$. The bounds have been averaged across 10,000 replications. The 95\% confidence interval (CI) reflects the coverage probability of the true interval. The Imbens and Manski (IM) 95\% CI gives the coverage probability of the true parameter. 
    \end{tablenotes}
    \end{threeparttable}
\end{table}

\subsection{DGP with Assumption \ref{Partialselection}(a), \ref{Partialselection}(b), and no monotonicity} \label{DGP_4ab}

The post-treatment counterfactual selection indicators in the DGP from Section \ref{simulation} are modified to construct a setting where Assumptions \ref{Partialselection}(a) and \ref{Partialselection}(b) hold, while monotonicity and \ref{inde_conditional} fail. In particular, we let the correlation between the selection shocks at $t=1$ to differ by treatment status. Define $\rho_i=\rho_0\cdot (1-D_i)+ \rho_1\cdot D_i $. Then define $v_{i1}(0)=v_{i1}$ and $v_{i1}(1)=\rho_i\cdot v_{i1}+\sqrt{1-\rho_i^2}\cdot k_{i1}$ where $k_{i1}\sim N(0,1)$. Selection at time $t=1$ is then given by $S_1(0)=\mathbbm{I}(b_i+v_{i1}(0)>0)$ and  $S_1(1)=\mathbbm{I}(b_i+v_{i1}(1)>0)$. This structure would ensure 4(a) and 4(b) hold while the joint distribution of $(S_{i1}(0),S_{i1}(1))$ depends on treatment status $D_i$ through $\rho_i$. Here we let $\rho_0=0.3$ and $\rho_1=0.9$ to impose a strong correlation when treated and low correlation when untreated. As we do not impose monotonicity here, this will allow the population to consist of all eight latent groups.  We set $\tau_g=(-1,1,3,1,4,5-1,-2)'$, $t_0=(0,2,3,1,4,5,0,1)'$ and $t_1=(1,3,4,2,5,6,1,2)'$, while all other parameter values are kept the same as in Section \ref{simulation}. With this modification, the true $\tau_{OOO}$ remains unchanged while the new overall ATT is $\tau=1.8494$. 

\begin{table}[H]
    \caption{Estimated bounds for $\tau_{OOO}$ under Assumption \ref{Partialselection}(a) \& (b)}
    \centering
    \label{simu bounds_OOO_4ab}
    \small
    \begin{threeparttable}
    \begin{tabular}{cccc}
    \toprule
    $n$      & $[\widehat{LB}_{\tau_{OOO}}, \widehat{UB}_{\tau_{OOO}}]$  & 95\% CI & IM 95\% CI\\
    \hline
    $500 $                      & [2.7831,7.0075]  & (1.9383,7.8162) & (2.0741,7.6862) \\
    $1000 $                   & [2.7837,7.0090]   & (2.1968,7.5703)& (2.2911,7.4801) \\
    $2000$                     & [2.7898,7.0081] & (2.3807,7.4051)& (2.4465,7.3413)\\
    \bottomrule
    \end{tabular}
    \begin{tablenotes}[flushleft]
    \item\footnotesize Notes: The table presents estimated bounds for $\tau_{OOO}$ for three sample sizes $n\in \{500,1000,2000\}$. The bounds have been averaged across 10,000 replications. The 95\% confidence interval (CI) reflects the coverage probability of the true interval. The Imbens and Manski (IM) 95\% CI gives the coverage probability of the true parameter. 
    \end{tablenotes}
    \end{threeparttable}
\end{table}
\subsection{DGP with Assumption \ref{Partialselection}(a) and Monotonicity } \label{DGP_4a} 

The post-treatment counterfactual selection indicators in the DGP from section \ref{simulation} are modified to allow Assumption \ref{Partialselection}(a) and monotonicity to hold, while assumptions \ref{Partialselection}(b) and \ref{inde_conditional} fail. First, we draw a latent selection index with a treatment shift as follows: $s_0=b_i+v_{i1}$ and $s_1=s_0+\delta D$, where $\delta>0$ and both variables have thresholds at zero: $S_1(0)=\mathbbm{1}(s_0>0)$, $S_1(1)=\mathbbm{1}(s_1>0)$. Here we let $\delta=5$ to impose a strong treatment effect on selection. With this modification, true $\tau_{OOO}$ remains unchanged, while the new overall ATT is $\tau=3.1659$. 

\begin{table}[H]
    \caption{Estimated  bounds for $\tau_{OOO}$ under Assumption \ref{Partialselection}(a) and Monotonicity}
    \centering
    \label{simu bounds_OOO_4a}
    \small
    \begin{threeparttable}
    \begin{tabular}{ccccc}
    \toprule
    $n$      & $[\widehat{LB}_{\tau_{OOO}}, \widehat{UB}_{\tau_{OOO}}]$ &95\% CI  &IM 95\% CI  \\
    \hline
    $500$                      & [3.4530,5.1974]   &(2.9515,5.6639)  & (3.0321,5.5889)       \\
    $1000 $                    & [3.4497,5.2007]  & (3.0916,5.5298)  &(3.1492,5.4770)        \\
    $2000$                      & [3.4490,5.2018]  &(3.1976,5.4331) & (3.2380,5.3959)       \\
    \bottomrule
    \end{tabular}
    \begin{tablenotes}[flushleft]
    \item\footnotesize Notes: The table presents estimated bounds for $\tau_{OOO}$ for three sample sizes $n\in \{500,1000,2000\}$. The bounds have been averaged across 10,000 replications. The 95\% confidence interval (CI) reflects the coverage probability of the true interval. The Imbens and Manski (IM) 95\% CI gives the coverage probability of the true parameter.
    \end{tablenotes}
    \end{threeparttable}
\end{table}
\subsection{NSW Training Program: Sample Selection Based on Unemployment}
    \label{NSW_unemployed}
\begin{table}[H]
\caption{Unemployment rates in NSW survey }
 \label{unemployment rate}
\begin{adjustbox}{center=\textwidth}
\begin{tabular}{lllllll}
\toprule
                      & \multicolumn{3}{l}{Pre-treatment period}                                                           & \multicolumn{3}{l}{Post-treatment period}                                                      \\ \cmidrule{2-7}
\multicolumn{1}{l}{} & \multicolumn{1}{l}{Controls} & \multicolumn{1}{l}{Treated} & \multicolumn{1}{l}{Total} & \multicolumn{1}{l}{Controls} & \multicolumn{1}{l}{Treated} & \multicolumn{1}{l}{Total} \\ 
\midrule
Unemployed            & 437                           & 440                          & 877                        & 270                           & 270                          & 540                        \\
Total                 & 585                           & 600                          & 1185                       & 585                           & 600                          & 1185                       \\
Percentage unemployed & 74.7                          & 73.3                         & 74.0                       & 46.2                          & 45.0                         & 45.6  \\
\bottomrule
\end{tabular}
\end{adjustbox}
\end{table}

\begin{table}[H]
\caption{Observed counts for NSW survey employed/unemployed }
\label{NSW missing obs groups}
\begin{adjustbox}{center=\textwidth}
\begin{tabular}{ccccc}
\toprule
\textbf{$S_{0}$} & \textbf{$S_1$} & \textbf{$D=0$} & \textbf{$D=1$} &  Total\\ \hline
0           & 0           & 216  & 212   & 428  \\
0           & 1           & 221   & 228   & 449   \\
1           & 0           & 54 & 58 & 112 \\
1           & 1           & 94 & 102  & 196 \\
\midrule
\multicolumn{2}{l}{Total} & 585 & 600 & 1185   \\ 
\bottomrule         
\end{tabular}
\end{adjustbox}
\end{table}


\begin{table}[H]
\caption{Covariate means and p-values from the test of equality of two means for the employed and unemployed samples in pre-treatment period }
\label{summary NSW unemployed base}
\begin{threeparttable}
\scriptsize
\begin{tabularx}{0.9\textwidth}{lcccccc}
\toprule
\multirow{2}{*}{Covariates}        & \multicolumn{3}{l}{Control}                  & \multicolumn{3}{l}{Treatment}                \\ 
\cmidrule{2-7}
& Unemployed & Employed &  $\small \mathbbm{P}[|T|>|t|]$
& Unemployed & Employed &  $\small \mathbbm{P}[|T|>|t|]$  \\ 
\midrule
Real earnings in 1975       &            & 3475.5180  &                         &            & 3231.1860  &                         \\
                            &            & (3170.9910)  &                         &            & (2729.1130)  &                         \\
Proportion Hispanic         & 0.1579       & 0.0405     & 0.0002                    & 0.1295       & 0.07500     & 0.0642                    \\
                            & (0.3651)       & (0.1979)     &                         & (0.3362)       & (0.2642)     &                         \\
Age, years                  &  33.9977      & 32.9865    & 0.1374                     & 34.3841      & 32.0625    & 0.0006                    \\
                            & (7.1953)       & (7.0014)     &                         & (7.4406)       &  (7.0199)     &                         \\
Proportion Black            & 0.7918       & 0.8919     & 0.0064                    & 0.8250       & 0.8688     & 0.2005                    \\
                            & (0.4065)       & (0.3116)     &                         & (0.3804)       & (0.3387)     &                         \\
Years of education          & 10.1373      & 10.6216    & 0.0121                    & 10.2364      & 10.5063    & 0.1201                    \\
                            & (2.1661)       & (1.5226)     &                         & (1.8944)       & (1.8326)     &                         \\
Proportion of high          & 0.7002       & 0.6351     & 0.1415                    & 0.7091       & 0.6750     & 0.4212                    \\
school dropouts             & (0.4587)       & (0.4830)     &                         & (0.4547)       & (0.4698)     &                         \\
Proportion married          & 0.0320       & 0.0608     & 0.1199                    & 0.0227       & 0.0188     & 0.7678                    \\
                            & (0.1763)       & (0.2398)     &                         & (0.1492)       & (0.1361)     &                         \\
Number of children   & 2.2220       & 2.2568     & 0.7845                    & 2.2636       & 1.9688     & 0.0134                    \\
    in 1975                        & (1.3233)       & (1.3758)     &                         & (1.3376)       & (1.1406)     &                         \\
\midrule    
Observations                                    & 437        & 148      & 585        & 440        & 160      & 600             \\
\bottomrule
\end{tabularx}
\begin{tablenotes}[flushleft]
    \item \footnotesize{Notes: Standard deviations are in parentheses. Reported p-values are from the equality test for two means between the observed and missing samples. Real earnings in 1975 are expressed in terms of 1982 dollars.}
\end{tablenotes}
\end{threeparttable}
\end{table}

We observe a significant difference in the share of Hispanics between the employed and unemployed samples in both the treated and control groups. Furthermore, unemployed controls differ significantly from employed controls in education and racial self-identification. At the same time, unemployed-treated individuals are older and have slightly more children living in the household than employed-treated individuals.


\begin{table}[H]
\caption{Covariate means and p-values from the test of equality of two means for the employed and unemployed samples  in post-treatment period }
\label{summary NSW unemployed follow up}
\begin{threeparttable}
\scriptsize
\begin{tabularx}{0.95\textwidth}{lcccccc}
\toprule
\multirow{2}{*}{Covariates}        & \multicolumn{3}{l}{Control}                  & \multicolumn{3}{l}{Treatment}                \\ \cmidrule{2-7} 
& Unemployed & Employed & \small $\mathbbm{P}[|T|>|t|]$
& Unemployed & Employed & \small $\mathbbm{P}[|T|>|t|]$  \\ 
\midrule
Real earnings in 1979       &            & 7119.2490  &                         &            & 8463.5030  &                         \\
                            &            & (4873.4320)  &                         &            & (4880.2950)  &                         \\
Proportion Hispanic         & 0.1333       & 0.1238     & 0.7318                    & 0.1259       & 0.1061     & 0.4488                    \\
                            & (0.3406)       & (0.3299)     &                         & (0.3324)       & (0.3084)     &                         \\
Age, years                  & 34.3963      & 33.1810    & 0.0404                    & 33.8037      & 33.7333    & 0.9078                    \\
                            & (7.2734)       & (7.0138)     &                         & (7.7399)       & (7.1153)     &                         \\
Proportion Black            & 0.8148       & 0.8190     & 0.8952                    & 0.8296       & 0.8424     & 0.6738                    \\
                            & (0.3892)       & (0.3856)     &                         & (0.3767)       & (0.3649)     &                         \\
Years of education          & 9.9852       & 10.4952    & 0.0024                    & 10.1148      & 10.4667    & 0.0225                    \\
                            & (2.1504)       & (1.8977)     &                         & (1.9486)       & (1.8102)     &                         \\
Proportion of high          & 0.7519       & 0.6254     & 0.0010                    & 0.7296       & 0.6758     & 0.1525                    \\
school dropouts             & (0.4327)       & (0.4848)     &                         & (0.4450)       & (0.4688)     &                         \\
Proportion married          & 0.0259       & 0.0508     & 0.1233                    & 0.0259       & 0.0182     & 0.5177                    \\
                            & (0.1592)       & (0.2199)     &                         & (0.1592)       & (0.1338)     &                         \\
Number of children          & 2.3593       & 2.1206     & 0.0311                    & 2.1889       & 2.1818     & 0.9470                    \\
in 1975                     & (1.4223)       & (1.2485)     &                         & (1.2867)       & (1.3013)     &                         \\
\midrule
Observations                                     & 270        & 315      & 585 & 270        & 330      & 600 \\                  
\bottomrule
\end{tabularx}
    \begin{tablenotes}[flushleft]
        \item \footnotesize Notes: Standard deviations are in parentheses. Reported p-values are from the test of equality for two means between the observed and missing samples. Real earnings in 1979 are expressed in terms of 1982 dollars.
    \end{tablenotes}
\end{threeparttable}
\end{table}

\subsection{Impact of Work From Home on Employee Performance: Sample Selection Based on Attrition}
\label{WFH}
\begin{table}[H]
\caption{Observed counts WFH experiment }
\label{WFH_counts}
\begin{adjustbox}{center=\textwidth}
\begin{tabular}{ccccc}
\toprule
$S_{0}$ & $S_1$ & $D=0$ & $D=1$ &  Total\\ 
\midrule
0           & 0           & 0  & 0   & 0  \\
0           & 1           & 0   & 0   & 0   \\
1           & 0           & 41 & 21 & 62 \\
1           & 1           & 77 & 110  & 187 \\
\midrule
\multicolumn{2}{l}{Total} & 118 & 131 & 249   \\ 
\bottomrule         
\end{tabular}
\end{adjustbox}
\end{table}
\setlength{\tabcolsep}{3pt}
\renewcommand{\arraystretch}{0.8}
{\scriptsize
\begin{table}[H]
\caption{Covariate means and p-values from the test of equality of two means for employees in the base period in the WFH application}
\label{summary WFH base}
\centering
\begin{threeparttable}
\begin{tabular}{lccc} 
\toprule
Covariates & Control & Treatment & $\mathbbm{P}[|T|>|t|]$ \\
\midrule
Prior performance z score       & -0.0401 & -0.0285 & 0.8746 \\
          & (0.5396) &(0.6187) &\\
Male        & 0.4661 & 0.4656 & 0.9943 \\
          & (0.5010) &(0.5007) &\\
Age        & 24.3475 & 24.4351 & 0.8460 \\
          & (3.5358) &(3.5672) &\\
Cost of commute        &  8.3378 & 7.8917  &  0.6144  \\
          & (5.5542) &(8.0313) &\\  
 Children        &  0.2373 & 0.1145  &   0.0104  \\
          & (0.4272) &(0.3196) &\\          
Married       &  0.3220 & 0.2214   &  0.0742    \\
          & (0.4692) &(0.4168) &\\  
Prior experience       &  16.7533 & 18.9618   &  0.5025     \\
          & (23.8181) &(27.6582) &\\ 
Tenure  &  28.2542  & 26.1374   &  0.4479       \\
          & (21.9382) &(21.9424) &\\ 
Age of Youngest Child  &  0.7119   & 0.5267   &  0.4507        \\
          & (1.9220) &(1.9388) &\\ 
Rent apartment  &  0.2034   & 0.2443    &  0.4424         \\
          & (0.4042) &(0.4313) &\\ 
Own bedroom  &   0.9915    & 0.9695    &  0.2169         \\
          & (0.0921 ) &(0.1727 ) &\\ 
Secondary technical school  &   0.4746    & 0.4580    &   0.7946        \\
          & (0.5019) &(0.5001) &\\ 
High school  &   0.1356    & 0.1756    &   0.3881        \\
          & (0.3438) &(0.38192) &\\ 
Tertiary,University  &   0.3559    & 0.3511   &    0.9374        \\
          & (0.4808) &(0.4792) &\\ 
Internet  &   1  & 0.9924  &     0.3436         \\
          & (0) &(0.0076) &\\ 
\midrule
Observations       &118 &131 &\\  
\bottomrule
\end{tabular} 
\begin{tablenotes}[flushleft]
    \item \footnotesize{Notes: Standard deviations are in parentheses. Reported p-values are from the equality test for two means between the controls and treated. }
\end{tablenotes}
\end{threeparttable}
\end{table}
}

\setlength{\tabcolsep}{3pt}
\renewcommand{\arraystretch}{0.8}
{\scriptsize
\begin{table}[H]
\caption{Covariate means and p-values from the test of equality of two means for employees who stayed and who left in the treatment period}
\label{summary WFH follow up}
\begin{threeparttable}
\small
\begin{tabular}{lcccccc}
\toprule
\multirow{2}{*}{Covariates}        & \multicolumn{3}{l}{Control}                  & \multicolumn{3}{l}{Treatment}                \\ \cmidrule{2-7}
& Left & Stayed & \small $\mathbbm{P}[|T|>|t|]$
& Left & Stayed & \small $\mathbbm{P}[|T|>|t|]$  \\ 
\midrule
Prior performance z score & -0.1619 &  0.0247 & 0.0736 & -0.1756 & -0.0004 & 0.2358 \\
                          & (0.5080) & (0.5479) &        &        & (0.5734) & (0.6254) \\
Male                      &  0.5122 &  0.4416 & 0.4682 &  0.5714 &  0.4455 & 0.2925 \\
                          & (0.5061) & (0.4998) &        & (0.5071) & (0.4993) &        \\
Age                       & 23.6342 & 24.7273 & 0.1101 & 22.9524 & 24.7182 & 0.0371 \\
                          & (3.2921) & (3.6224) &        & (3.2631) & (3.5662) &        \\
Cost of commute           &  8.9401 &  8.0171 & 0.3923 & 10.7619 &  7.3438 & 0.0738 \\
                          & (6.3002) & (5.1284) &        & (13.2360) & (6.5494) &        \\
Children                  &  0.2195 &  0.2468 & 0.7431 &  0.1429 &  0.1091 & 0.6591 \\
                          & (0.4191) & (0.4339) &        & (0.3586) & (0.3132) &        \\
Married                   &  0.2683 &  0.3506 & 0.3662 &  0.1905 &  0.2273 & 0.7123 \\
                          & (0.4486) & (0.4803) &        & (0.4024) & (0.4210) &        \\
Tenure                    & 20.7195 & 32.2662 & 0.0060 & 16.7143 & 27.9364 & 0.5200 \\
                          & (17.5841) & (23.0489) &      & (12.3537) & (22.9315) &        \\
Prior experience          & 17.9756 & 16.1025 & 0.6860 & 18.2381 & 19.1000 & 0.0312 \\
                          & (21.9608) & (24.8652) &      & (18.8677) & (29.1018) &        \\
Age of Youngest Child     &  0.6829 &  0.7273 & 0.9056 &  0.4286 &  0.5455 & 0.8013 \\
                          & (1.8363) & (1.9777) &        & (1.2071) & (2.0527) &        \\
Rent apartment            &  0.1707 &  0.2208 & 0.5242 &  0.3333 &  0.2273 & 0.3036 \\
                          & (0.3809) & (0.4175) &        & (0.4830) & (0.4210) &        \\
Secondary technical school&  0.5122 &  0.4545 & 0.5544 &  0.4762 &  0.4545 & 0.8566 \\
                          & (0.5061) & (0.5012) &        & (0.5118) & (0.5002) &        \\
Own bedroom               &  1.0000 &  0.9870 & 0.4679 &  1.0000 &  0.9636 & 0.9500 \\
                          & (0.0000) & (0.1140) &        & (0.0000) & (0.1881) &        \\
High school               &  0.0976 &  0.1558 & 0.3829 &  0.0476 &  0.2000 & 0.0940 \\
                          & (0.3004) & (0.3651) &        & (0.2182) & (0.4018) &        \\
Tertiary, University      &  0.3902 &  0.3377 & 0.5739 &  0.4286 &  0.3364 & 0.4212 \\
                          & (0.4939) & (0.4760) &        & (0.5071) & (0.4746) &        \\
Internet                  &  1.0000 &  1.0000 &   --   &  1.0000 &  0.9909 & 0.6639 \\
                          & (0.0000) & (0.0000) &        & (0.0000) & (0.0953) &        \\
\midrule                          
Observations              & 41 & 77 & 118 & 21 & 110 & 131 \\
\bottomrule
\end{tabular}
    \begin{tablenotes}[flushleft]
    \item \footnotesize Notes: Standard deviations are in parentheses. Reported p-values are from the test of equality for two means between the observed and missing samples. 
    \end{tablenotes}
\end{threeparttable}
\end{table} }

\end{document}